\def\firstsub{first}%
\def\arxivsub{arxiv}%
\def\submissiontype{arxiv}
\ifx\submissiontype\undefined%
\def\submissiontype{final}%
\fi%
\ifx\submissiontype\firstsub%
\PassOptionsToPackage{twoside}{geometry}
\fi%
\documentclass[titlepage]{hwthesis}
\usepackage{lscape}
\usepackage{extarrows}
\usepackage{collectbox}
\usepackage{pdfpages}
\usepackage{mathrsfs, mathtools}
\usepackage{dsfont}
\usepackage[]{hyperref}
\usepackage{slashed}
\usepackage[percent]{overpic} 
\usepackage{amsthm,graphicx
	,a4wide,exscale,relsize,setspace,lineno}
\usepackage{subcaption}
\usepackage[square,numbers,sort&compress]{natbib}
\usepackage{geometry}
\usepackage{algorithm2e}
\usepackage[nottoc]{tocbibind}
\usepackage[title,titletoc]{appendix}
\usepackage{tikz} 
\usetikzlibrary{matrix,cd,arrows} 
\usepackage{mathtools} 
\usepackage{amsmath,amssymb} 
\usepackage{tensor}   
\usepackage{mathrsfs} 
\usepackage{bbm}      
\usepackage{booktabs} 
\usepackage{enumitem} 
\setlist{nosep}        
				       
\usepackage{lipsum}    
\usepackage{tikz}
\usetikzlibrary{matrix,arrows,decorations.pathmorphing}
\usepackage{tikz-cd}			       
\tikzcdset{%
    triple line/.code={\tikzset{%
        double distance = 3pt, 
        double=\pgfkeysvalueof{/tikz/commutative diagrams/background color}}},
    quadruple line/.code={\tikzset{%
        double distance = 5.3pt, 
        double=\pgfkeysvalueof{/tikz/commutative diagrams/background color}}},
    Rrightarrow/.code={\tikzcdset{triple line}\pgfsetarrows{tikzcd implies cap-tikzcd implies}},
    RRightarrow/.code={\tikzcdset{quadruple line}\pgfsetarrows{tikzcd implies cap-tikzcd implies}}
}

\newcommand*{\tarrow}[2][]{\arrow[Rrightarrow, #1]{#2}\arrow[dash, shorten >= 0.5pt, #1]{#2}}

\ifx\submissiontype\arxivsub%
\else%
\fi%

\definecolor{MyDarkBlue}{rgb}{0.15,0.25,0.45}

\usepackage{hyperref}
\ifx\submissiontype\arxivsub%
\hypersetup{
hypertexnames=false,
colorlinks=true,
citecolor=MyDarkBlue,
linkcolor=MyDarkBlue,
urlcolor=MyDarkBlue,
pdfauthor={Lukas Müller},                     
pdftitle={PhD Thesis Lukas Müller},            
pdfsubject={hep-th math-ph},              
breaklinks=true
}
\else%
\hypersetup{
hypertexnames=false,
colorlinks=true,
citecolor=black,
linkcolor=black,
urlcolor=black,
pdfauthor={Lukas Müller},                     
pdftitle={PhD Thesis Lukas Müller},            
pdfsubject={hep-th math-ph},              
breaklinks=true
}
\fi%

\ifx\submissiontype\arxivsub%
\relax
\else%
\fi

\usepackage{macros} 

\begin{document}


\title{Extended Functorial Field Theories and Anomalies in Quantum Field Theories}



\author{Lukas M{\"u}ller} 


\maketitle
\pagenumbering{gobble}
\ifx\submissiontype\firstsub%
\mbox{}
\fi%

\begin{abstract}
We develop a general framework for the description of anomalies using extended functorial 
field theories extending previous work by Freed and Monnier. In this framework, anomalies are described by invertible
field theories in one dimension higher and anomalous field theories live on their boundaries. 

We provide precise mathematical definitions for all concepts involved using the language of symmetric monoidal bicategories.
In particular, field theories with anomalies will be described by symmetric monoidal transformations. The use of higher 
categorical concepts is necessary to capture the Hamiltonian picture of anomalies. 
The relation to the path integral and the Hamiltonian description of anomalies will be explained in 
detail. 
Furthermore, we discuss anomaly inflow in detail. 

We apply the general framework to the parity anomaly in fermionic systems coupled to background gauge and gravitational fields
on odd-dimensional spacetimes. We use the extension of the Atiyah-Patodi-Singer index theorem to manifolds with corners due to Loya and Melrose
to explicitly construct an extended invertible field theory encoding the anomaly. 
This allows us to compute explicitly the 2-cocycle
of the projective representation of the gauge symmetry on the quantum state space, which
is defined in a parity-symmetric way by suitably augmenting the standard chiral fermionic
Fock spaces with Lagrangian subspaces of zero modes of the Dirac Hamiltonian that naturally
appear in the index theorem. 

As a second application, we study discrete symmetries of Dijkgraaf-Witten theories and their gauging. 
Non-abelian group cohomology is used to describe discrete symmetries and we derive concrete conditions for such a symmetry to admit 't Hooft anomalies in terms of the Lyndon-Hochschild-Serre spectral sequence. 
We give an explicit realization of a discrete gauge theory with 't Hooft anomaly as a state on the boundary of a higher-dimensional Dijkgraaf-Witten theory using a relative version of 
the pushforward construction of Schweigert and Woike. 
\end{abstract}
\ifx\submissiontype\firstsub%
\mbox{}
\fi%

\clearpage
\vspace*{12cm}
\hfil
\textit{To everyone who has sparked, encouraged or supported my interest in mathematics and natural sciences.}
\clearpage
\ifx\submissiontype\firstsub%
\mbox{}
\clearpage
\fi%

\begin{acknowledgements}
First and foremost, I would like to thank my advisor Richard Szabo for his guidance and support throughout my PhD,  his valuable advice and
for giving me the opportunity to follow my own ideas without ever feeling left alone.    
  
Furthermore, I would like to thank Christoph Schweigert for helpful discussions related to my research, inviting me to Hamburg multiple times
and his constant support of my academic career.  
   
I am grateful to my other collaborators: to Severin Bunk for explaining various results and concepts related to bundle gerbes, interesting  discussions on
mathematical physics in general and organizing the reading group which introduced me to model categories, 
to  L\'{o}r\'{a}nt  Szegedy for a fun collaboration on area dependent field theories and 2-dimensional Yang-Mills theory and 
to Lukas Woike, for listening to all my ideas on physics and mathematics and his valuable input into all my projects. 

I would like to thank Ehud Meir for inviting me to Aberdeen and 
helpful discussions on spectral sequences and topological field theories, and Nathalie Wahl for showing 
interest in my work, inviting me to Copenhagen and many interesting discussions on graph models and 
equivariant field theories.  

I would like to thank Simon Lentner, Ingo Runkel, Lennart Schmidt, Thomas Wasserman and Tim Weelinck for helpful discussions on my research projects.  

I am grateful to Lennart Schmidt and Lukas Woike for helpful comments on a draft of this thesis.  
I would like to thank all PhD students in the mathematical physics group at Heriot-Watt University for organising and 
participating in regular reading groups on a variety of topics.

I am thankful to Christian Sämann and Konrad Waldorf who agreed to be my examiners.

My work is supported by the Doctoral Training Grant ST/N509099/1 from the UK Science and Technology
Facilities Council (STFC).
\end{acknowledgements}
\ifx\submissiontype\firstsub%
\mbox{}
\fi%

\ifx\submissiontype\arxivsub%
\relax
\else%
\includepdf[pages={1}]{researchthesissubmission.pdf}
\ifx\submissiontype\firstsub%
\mbox{}
\fi%
\fi%

\clearpage
\setcounter{page}{1}
\pagenumbering{roman}

\tableofcontents

\clearpage
\pagenumbering{arabic}
\chapter{Introduction}
This thesis is concerned with the mathematical description of anomalies in quantum field theories. In the introduction we sketch parts of the physical background motivating our work and the mathematical
formalism of functorial field theories used throughout. Afterwards, we give an informal summary of our results.   

\section{Anomalies in quantum field theories}\label{Sec: Anomalies}
Usually quantum field theories are constructed from classical field 
theories via a process called ``quantization".  A classical field theory on a
manifold $M$ is specified by a collection of fields $\mathcal{F}(M)$ and a ``local"
action functional $S\colon \mathcal{F}(M) \longrightarrow \R$ which usually depends
on additional geometric structures on $M$ such as a metric (see e.g.~\cite[Chapter 5]{ClassicalFields}).  
The physical states of the system are described by those field configurations which extremise the
action and are usually solutions to a set of partial differential equations. 
The set of classical observables $\mathcal{O}$ consists of functions on the space of solutions
to these differential equations.\footnote{In general, the space of solution is a geometrically
ill-behaved object and should be replaced by the derived space of solutions~\cite{FA2}. In physical terms this is related to the BV description of classical field theories~\cite{Batalin:1984jr,FA2,BVLoo}.}

In the path integral formulation of quantum field theory~\cite[Chapter 9]{PSQFT} the quantum theory corresponding 
to a classical field theory is formally described by the integral over the space of classical fields
\begin{align}\label{Eq: path integral}
Z(M)=\int\displaylimits_{\psi \in \mathcal{F}(M)}\exp \left( \tfrac{\iu}{\hbar} S(\psi)\right) \mathcal{D} \psi \ \
\end{align}
 and the expectation value of an observable $f \in \mathcal{O}$ can be computed by inserting
 the observable into~\eqref{Eq: path integral}:
\begin{align}
\langle f \rangle = \int\displaylimits_{\psi \in \mathcal{F}(M)} f(\psi) \exp \left( \tfrac{\iu}{\hbar} S(\psi)\right) \mathcal{D} \psi \ \ .
\end{align} 
These expressions are in general mathematically ill-defined. In perturbative quantum field
theory the path integral can be interpreted as a formal power series in 
$\hbar$~\cite{RGCostello}. However, the terms appearing in the perturbation series 
are divergent in general and need to be renormalized. A choice of a way to perform this renormalization is called
a \emph{renormalization scheme}.  

The need to renormalize the divergences makes it unclear whether all properties of a classical 
field theory continue to hold in the quantum theory. Indeed, it turns out that
symmetries are in general not preserved under quantization. 
Symmetries of a classical field theories 
are transformations of the space of fields which leave the action invariant. A symmetry is 
called \emph{anomalous} if there is no renormalization scheme which extends the symmetry
to the quantum theory. By Noether's theorem~\cite{Noether:1918zz} there is for every classical symmetry a 
conserved quantity $Q^\mu$ satisfying (we use Einstein summation convention)
\begin{align}\label{EQ: Sym}
\partial_\mu Q^\mu = 0 
\end{align} 
on solutions to the equations of motion. Anomalies now manifest by a violation of
Equation \eqref{EQ: Sym} 
\begin{align}
\langle \partial_\mu Q^\mu \rangle = \alpha \ \ 
\end{align} 
in the quantized theory.
The function $\alpha$ does depend on the choice of renormalization scheme, and an
anomaly is present if $\alpha$ is non zero in every renormalization scheme. 
One of the first examples of an anomaly is the Adler-Bell-Jackiw (ABJ) anomaly~\cite{Adler:1969gk,Bell1969}. 
It appears in a quantum field theory defined on $\R^4$ equipped with its canonical Lorentzian metric $\eta_{\mu \nu}$ of
signature $(-,+,+,+)$ and a $U(1)$ background gauge field $A_\mu$. The dynamical field of the theory is a four component 
spinor $\psi_\mu$. The classical field theory is described by the action 
\begin{align}\label{EQ: Action QED}
S = \int_{\R^4} \bar{\psi}  \gamma^\mu (\iu \partial_\mu + e A_\mu )  \psi  \ \ ,  
\end{align}   
where $e$ is the electric charge, $\gamma_0, \dots , \gamma_3$ are the usual gamma matrices~\cite[Section 3.4]{PSQFT}, $\bar{\psi}=
\psi^{\dagger}\gamma^0$ and we used the metric $\eta_{\mu\nu}$ to raise and lower indices. The action~\eqref{EQ: Action QED} admits 
the symmetry 
\begin{align}
\psi \mapsto \exp(\iu \lambda \gamma_5) \psi \ \ ,
\end{align} 
where $\gamma_5= \iu \gamma_0\gamma_1 \gamma_2 \gamma_3 $ and $\lambda$ is a real constant. The corresponding conserved current
is
\begin{align}
j_\mu = \bar{\psi}\gamma_\mu \gamma_5 \psi \ \ .
\end{align}
A direct Feynman diagram computation shows~\cite[Section 4.3.1]{Bertlmann:1996xk} that this current is not preserved at the quantum level:
\begin{align}
\langle \partial_\mu j^\mu \rangle = \frac{e^2}{16 \pi^2} \epsilon^{\mu,\nu,\alpha, \beta} F_{\mu \nu} F_{\alpha, \beta} \ \ ,
\end{align} 
where $F=\dd A$ is the field strength. 
The ABJ anomaly can be used to explain the observed decay of a neutral pion into two photons which would be forbidden by 
the symmetry otherwise.  
This shows that anomalies of global symmetries are necessary to accurately describe physical observations. 
In contrast gauge symmetries can never be anomalous in a consistent 
quantum field theory, since 
the formulation of gauge theories depends on the existence of gauge 
symmetry and it is not possible (even at a physical level of rigour) to quantize gauge
theories with anomalies. Hence all gauge anomalies need to cancel each other out in
physically reasonable theories. This poses some restrictions on the properties of 
the matter content in the standard model, see e.g.~\cite[Section 4.9]{Bertlmann:1996xk}. 

In Chapter~\ref{Sec: General Theory} we will develop a framework for the mathematical description of anomalies using functorial field theories.
We will illustrate this general framework in Chapter~\ref{Cha: Index} and Chapter~\ref{Chapter: t Hooft} by applying it 
to two examples. The first example will be the parity anomaly of chiral gauge theories
on odd dimensional spacetimes. On the 3 dimensional manifold $M=\R^3$ the parity anomaly takes the following 
form: Equation~\eqref{EQ: Action QED} 
can be adapted to 3 dimensions by replacing the 4-dimensional gamma matrices with their 3-dimensional 
version:  
\begin{align}
S = \int_{\R^3} \bar{\psi} \D   \psi \ \ ,
\end{align} 
with the \emph{Dirac operator} $\D =  \gamma^\mu (\iu \partial_\mu + e A_\mu )$.
This theory admits a time-reversal symmetry $T\psi(t,x,y)= \gamma_0 \psi(-t,x_1,x_2)$ which turns out 
to be anomalous. 
So far our discussion took place in Lorentzian signature. However, for the derivation of
the parity anomaly on arbitrary spacetimes Euclidean signature is more convenient and we will use it from now on. 
Formally, the path integral on a spin manifold $M$ equipped with 
a background gauge field is the determinant of the Dirac operator 
\begin{align}\label{EQ: partition function parity} 
Z(M, A_\mu)= \det (\D)= \prod_{\lambda \in \text{spec} (\D) } \lambda \ \ .    
\end{align}
Orientation reversal acts on the partition function by complex conjugation.
Since all eigenvalues of $\D$ are real, the path integral is formally invariant with respect to this symmetry. 
However, the determinant in Equation~\eqref{EQ: partition function parity} is divergent and hence
needs to be renormalized. The parity anomaly now manifests itself as the statement that it is impossible 
to renormalise the determinant in such a way that the path integral is real and invariant under gauge transformations. 
Let us fix a renormalization such that $Z(M,A_\mu)$ is real and positive for a fixed background gauge field $A_\mu$. 
Furthermore, let $A'_\mu$ be an equivalent gauge field. We can compute the sign difference between $Z(M,A_\mu)$ and
$Z(M,A'_\mu)$ by following the path $(1-t)A_\mu +t A'_\mu$ in the space of background fields and changing the sign
every time an eigenvalue of $\D$ passes through zero. In general, this ``spectral flow" might lead to a non-trivial 
sign change breaking gauge invariance. The spectral flow can be computed from the index of a Dirac operator 
on the mapping torus constructed from $[0,1]\times M $ by gluing the ends together using the gauge transformation
to construct a background gauge field on $S^1\times M$~\cite{WittenFermionicPathInt}.      
 
The second class of anomalies we study are 't Hooft anomalies of Dijkgraaf-Witten
theories. Dijkgraaf-Witten
theories are gauge theories with finite gauge group $D$. The space of fields 
is the space $\Bun_D(M)$ of principal $D$-bundles which automatically carry a unique 
connection, since the Lie algebra of $D$ is $0$. The space of gauge equivalence classes 
of principal $D$-bundles is finite which makes it possible to define the path integral
\eqref{Eq: path integral} without the need to remove infinities. For this reason
any classical symmetry can be extended to the quantum theory. However, these 
symmetries can still have 't Hooft anomalies~\cite{Kapustin:Symmetries}. 't Hooft anomalies are slightly different in flavour to the anomalies discussed above.
There is no problem in extending a symmetry with 't Hooft anomaly to a global symmetry of the quantum theory.
However, the anomaly is an obstruction to gauging the symmetry, i.e.\  coupling the theory
to non-trivial background gauge fields for the symmetry group. Hence these anomalies 
are only visible in non-trivial backgrounds. Dijkgraaf-Witten theories 
can be mathematically rigorously defined and hence provide an ideal toy model 
for the mathematical study of 't Hooft anomalies.   

In the Hamiltonian picture of quantum field theory anomalies manifest themself in terms of 
a projective representation of the symmetry group on the Hilbert space $\mathcal{H}$ of the theory~\cite{Nelson:1984gu}.
In principle, projective symmetry actions do not cause any problems in a quantum theory since a physical state is described
by a ray in the Hilbert space. However, 
a projective action of a gauge symmetry makes it impossible to implement the Gauss law, which states that physical states
need to be invariant under the action of the gauge symmetry on the Hilbert space. This is the Hamiltonian analogue
of the statement that gauge theories with anomalous gauge symmetry are inconsistent. 
 
\section{Anomalies and symmetry-protected topological phases}
At the end of the last century it was realised that quantum phases of matter
exist which cannot be described by Landau's theory of symmetry breaking. 
Instead, these phases can be distinguished by `topological order'
parameters which prevent them from being deformed to a trivial system
whose ground state is a factorized state.
Recently, there has been renewed interest in anomalies, because of their relation to these
topological phases of matter. The structures studied in this thesis are 
closely related to and motivated by the point of view on anomalies in the context of condensed matter 
physics.
 
Since their discovery over 30 years ago, immense progress in understanding and 
classifying topological phases has been made. 
For instance, there exists a classification for non-interacting gapped
systems in 
terms of twisted equivariant K-theory~\cite{FreedMoore} (see also~\cite{Bunk2017}). 

A fruitful approach to the study of gapped interacting systems is to 
consider the effective low-energy (long-range) continuum theory of 
a lattice Hamiltonian model. Usually these field theories are topological.
A famous example is the effective 
description of the integer quantum Hall effect in terms of 
Chern-Simons gauge theory. In this sense a gapped quantum phase may be thought of as a path-connected
component of the moduli stack of topological quantum field theories.
However, in the interacting case no complete classification exists. For this reason 
one usually restricts to tractable subclasses. Anomalies in $n$-dimensional field theories 
are closely related to 
`short-range entangled' phases \cite{Chen:2010gda} in $n+1$-dimensions. 
A gapped phase $\Psi$ is short-range entangled if there exists a 
phase $\Psi^{-1}$ such that the `stacked' phase $\Psi \otimes \Psi^{-1}$
can be deformed by an adiabatic transformation of the Hamiltonian to a
trivial product state without closing the energy gap between the
ground state and the first excited state. 
The stacking operation of topological phases corresponds to the tensor product   
of their low-energy effective topological field theories. 
A topological field theory is called invertible if it admits an inverse with respect to 
the tensor product. This observation motivates a classification of
short-range entangled phases in terms of invertible topological field 
theories~\cite{Freed:2014eja, FreedHopkins}. 

Let $G$ be a group.
In the case of an additional global $G$-symmetry, a non-trivial short-range entangled
phase may be trivial when the symmetry is ignored. 
Such a phase is called 
`$G$-symmetry-protected'~\cite{Chen:2010gda,Chen:2011pg}.   
A $G$-symmetry-protected phase can be understood by studying its topological
response to non-trivial background $G$-gauge fields, which is called
`gauging' the $G$-symmetry. 
For a finite symmetry group $G$, the low-energy effective field theories 
are $G$-equivariant topological field theories~\cite{Kapustin:2015uma}. 
The correspondence between topological field theories and symmetry-protected topological phases of matter is discussed in e.g.~\cite{Wen:2013oza,Kapustin:2015uma,Gaiotto:2015zta,WittenFermionicPathInt}.
Classical Dijkgraaf-Witten theories provide a particularly tractable
class of invertible $G$-equivariant topological field theories. The corresponding lattice Hamiltonian
models have been constructed in e.g.~\cite{PhysRevB.92.045101,Cong2017,Cong:2017ffh,Wang2018}. 
They are classified by group cohomology. The corresponding classification 
of topological phases is called group cohomological classification~\cite{Levin2012}.
However, this is not a 
complete classification of symmetry-protected phases and more refined
classifications have been proposed, see e.g.~\cite{Kapustin:2015uma, FreedHopkins, Gaiotto:2017zba, PhysRevX.8.011055}.

An essential feature of symmetry-protected topological phases is that they exhibit
`topologically protected' boundary states. These boundary states can 
be effectively described by anomalous quantum field theories. The gapped bulk system
is then characterised by gapless boundary states, such as the chiral
quantum Hall edge states, which exhibit gauge or gravitational
anomalies; conversely, an $n-1$-dimensional system whose ground state topological
order is anomalous can only exist as the boundary of an $n$-dimensional
topological phase. While the boundary quantum field theory on its own
suffers from anomalies, the symmetry-protected boundary states are described by considering the anomalous 
theory `relative' to the higher-dimensional bulk theory, where it becomes a non-anomalous quantum field theory 
under the `bulk-boundary correspondence'~\cite{Ryu:2010ah,Ryu:2012he} in which the boundary states undertake anomaly 
inflow from the bulk field theory~\cite{Callan:1984sa,Faddeev:1986pc}. 
The standard examples are provided by topological insulators corresponding to the parity anomaly which are protected by fermion number conservation and time-reversal symmetry ($G=U(1)\times\mathbb{Z}_2$)~\cite{Hasan:2010xy,Qi:2011zya}. 
The presence of non-trivial global anomalies forces the boundary theory 
to be non-trivial and topologically protected.
Reversing this logic, it follows that $n{+}1$-dimensional invertible field 
theories should classify the possible anomalies in $n$ dimensions~\cite{Wen:2013oza}.

\section{Modern perspective on anomalies}
In the situations just described, a field theory with anomaly is 
well-defined as a theory living on the boundary of an invertible quantum field
theory in one dimension higher. 
The modern perspective on anomalies is that anomalies in $n-1$-dimensional quantum field theories
can be completely described by invertible $n$-dimensional field theories~\cite{FreedAnomalies,MonnierHamiltionianAnomalies, doi:10.1002/prop.201910012, WittenFermionicPathInt}. We have already seen a glance of this correspondence when we related the 
parity anomaly to the index of the Dirac operator on a mapping torus in one dimension higher. This 
perspective has the advantage that tools from the study of invertible field theories can be applied to anomalies.
For example, the partition function of an 
invertible field theory is a cobordism invariant~\cite{Freed:2004yc,Yonekura:2018ufj}: two $n$-dimensional manifolds $M$ and $M'$ with background fields 
are \emph{cobordant} if there exist an $n+1$-dimensional manifold $X$ such that $\partial X= M \sqcup -M'$ and the  
background fields on $M$ and $M'$ can be extended to $X$. The cobordism group $\Omega_n(\mathcal{F})$, depending on the
background fields $\mathcal{F}$, consists of equivalence classes of cobordant manifolds. The multiplication is given by the disjoint
union of manifolds. The partition function of an invertible field theory can now be described by a group homomorphism 
$\Omega_n(\mathcal{F})\longrightarrow U(1)$. 
This perspective is extremely helpful in arguing that certain anomaly field theories are trivial, because one only 
needs to check that the partition function vanishes on a hopefully simple set of generators of $\Omega_n(\mathcal{F})$. 
In practice it is even often possible to show that $\Omega_n(\mathcal{F})$ is zero and hence all anomalies 
vanish. Furthermore there are powerful algebraic tools for the computation of cobordism groups, such as the 
Atiyah-Hirzebruch spectral sequence~\cite{MR0139181}. We refer to~\cite{G-Etxebarria2019} for recent applications of these ideas to particle physics and to~\cite{MonnierMAnomalies,freed2019mtheory} for applications to M-theory. 

Knowledge of the partition function of the invertible field theory is enough to show the absence of anomalies. However,
to provide a concrete description of a theory with anomaly, such as the boundary states of a topological phase of matter, further information about the theory is required. In the next section we introduce the mathematical formulation of 
functorial field theories used throughout this thesis and sketch how anomalies can be described in this framework.

\section{Extended functorial field theories}
Functorial (quantum) field theory is one attempt to mathematically rigorously capture 
(some of) the structures of quantum field theory.  
The idea is to give an axiomatic framework for the partition function and state space of a
quantum field theory. Recall, that formally the partition function $Z(M)$ on an 
$n$-dimensional manifold $M$ is calculated by
the Feynman path integral of an exponentiated action functional over the space of
dynamical field configurations on $M$; so far there is no mathematically 
well-defined theory of such path integration in general. 
The axioms of functorial field theories are derived from the properties that
such an integration would satisfy in the case that the action
functional is an integral of a local Lagrangian density on $M$.\footnote{In the case of discrete gauge theory there exists
	a well-defined integration 
	theory (see for example Appendix~\ref{Sec: Homotop Groupoid}) which satisfies the axioms of a functorial field theory.}
A quantum field theory should also assign 
a Hilbert space of states $Z(\Sigma)$ 
to every $n{-}1$-dimensional manifold $\Sigma$. They satisfy
$Z(\Sigma\sqcup\Sigma')\cong Z(\Sigma)\otimes_\C Z(\Sigma')$, i.e.\  the
state space of non-interacting systems is given by the tensor product
of the corresponding Hilbert spaces.
In any quantum field theory there exists a time evolution 
operator (propagator)
\begin{align}
Z([t_0,t_1]\times \Sigma)\colon Z(\Sigma) \longrightarrow Z(\Sigma)
\end{align}         
from time $t_0$ to $t_1$. We think of this operator as associated to the 
cylinder $[t_0,t_1]\times \Sigma$. They satisfy 
$Z([t_1,t_2]\times \Sigma) \circ Z([t_0,t_1]\times \Sigma)=Z([t_0,t_2]\times \Sigma)$. 
The path integral should also allow for the construction of a more general operator 
$Z(M)\colon Z(\Sigma_-)\longrightarrow Z(\Sigma_+)$ for every manifold $M$
with boundary $\Sigma_-\sqcup \Sigma_+$, such that the gluing of
manifolds corresponds to the composition of linear maps. Such a manifold is called a \emph{cobordism} 
from $\Sigma_-$ to $\Sigma_+$.

These considerations motivate the definition of a functorial field theory,
generalising Atiyah's definition of topological field 
theories~\cite{Atiyah1988} and Segal's definition of conformal field
theories~\cite{Segal1988}, as a symmetric monoidal functor 
\begin{align}
Z \colon \Cob_n^\mathscr{F} \longrightarrow \Hilb \ ,
\end{align}    
where $\Cob_n^\mathscr{F}$ is a category modelling physical spacetimes with
non-dynamical background fields $\mathscr{F}$ and $\Hilb$ is the category of
complex Hilbert spaces. 
Roughly speaking, $\Cob_n^\mathscr{F}$ contains closed $n{-}1$-dimensional manifolds 
with background fields as objects, 
and as morphisms the $n$-dimensional cobordisms as well as additional limit morphisms corresponding
to diffeomorphisms which are compatible with the background fields.
The additional morphisms encode symmetries.  We give a precise definition in Section~\ref{Sec: Non-extended field theories}.  
Evaluating $Z$ on a closed $n$-dimensional manifold $M$ gives rise to a linear
map $\C \cong Z(\varnothing)\longrightarrow Z(\varnothing) \cong \C$ which 
can be identified with a complex number $Z(M)$, the partition function of 
$Z$ on $M$. 
From a mathematical standpoint this definition can be thought of as a prescription for computing a manifold invariant 
$Z(M)$ by cutting manifolds into simpler pieces and studying the quantum field theory
on these pieces.

We now turn our attention to the description of anomalies.
The partition function of an $n{-}1$-dimensional 
quantum field theory $Z$ with anomaly
described by an invertible field theory $\Aa \colon \Cob_n^\mathscr{F} \longrightarrow \Hilb$ 
on an $n{-}1$-dimensional manifold $\Sigma$ takes values in the one-dimensional vector space $\Aa(\Sigma)$, instead of $\C$. 
It is possible to pick a non-canonical isomorphism $\Aa(\Sigma)\cong \C$
to identify the partition function with a complex number. The group of
symmetries acts non-trivially on $\Aa(\Sigma)$ encoding the breaking of 
the symmetry in the quantum field theory $Z$. 

To also incorporate the description of the state space of $Z$ on an 
$n{-}2$-dimensional manifold $S$ we need to promote $\Aa$ to an extended field
theory which assigns $\C$-linear categories to $n{-}2$-dimensional manifolds
such that the anomalous state space $Z(S)$ can be considered as an object of $\Aa(S)$. 
In other words, $\Aa$ should be an extended functorial field theory, i.e.\  a 
symmetric monoidal 2-functor 
\begin{align}
\Aa \colon \Cob_{n,n-1,n-2}^\mathscr{F} \longrightarrow \Tvs \ ,
\end{align}  
where $\Cob_{n,n-1,n-2}^\mathscr{F}$ is a bicategorical extension of
$\Cob_{n}^\mathscr{F}$; see Section~\ref{Sec: Cob bicat} for details. 
There are different possible choices for the target bicategory. For simplicity we 
restrict ourselves to Kapranov-Voevodsky(KV) 2-vector spaces \cite{KV94}. 

Requiring that $\Aa$ is an invertible field theory implies that there is 
a non-canonical equivalence of categories $\Aa(S)\cong \fvs$ which
allows one
to identify the state space of the anomalous theory with a vector space. 
We can subsume the ideas sketched above in the following concise definition:
A \emph{quantum field theory with anomaly} is a natural symmetric monoidal 2-transformation
\begin{align}
Z\colon 1 \Longrightarrow \tr\, \Aa
\end{align}
between a trivial field theory and a certain truncation of $\Aa$; see 
Section~\ref{Sec: GT Anomalies} for details.
This formalism allows one to compute the 2-cocycle twisting the projective
representation of the symmetry group on the state space completely in 
terms of the extended field theory $\Aa$~\cite{MonnierHamiltionianAnomalies}.
Anomalous theories formulated in this way are a special case of relative field
theories~\cite{RelativeQFT} and are closely related to twisted quantum field 
theories~\cite{Stolz:2011zj,Johnson-Freyd:2017ykw}. 
The present thesis develops this framework for the description of anomalies via 
functorial field theories following previous work by~\cite{FreedAnomalies, MonnierHamiltionianAnomalies} and applies
it to the two examples mentioned in Section~\ref{Sec: Anomalies}. We shall now give an overview of our constructions
and findings. 

\section{Summary of results and outline}\label{Sec: Outline}
In Chapter~\ref{Sec: General Theory} which is based on the publications~\cite{parity, tHooft}, we develop the general theory underlying the following examples. 
In the first part of the chapter we define extended functorial field theories. 
One of the technical difficulties related to giving a concrete definition
is the construction of the higher cobordism category equipped with additional structure.
For cobordisms with tangential structure there exists an $(\infty,n)$-categorical definition~\cite{Lurie2009a,CS}. A categorical version with arbitrary background fields
taking into account families of manifolds has been defined by Stolz and Teichner~\cite{Stolz:2011zj}.
A bicategory of cobordisms equipped with elements of topological stacks is constructed in \cite{schommer2011classification}.
One of the main technical accomplishments of this part of Chapter~\ref{Sec: General Theory} is the explicit
construction of a geometric cobordism bicategory
$\ECobF$ which includes arbitrary
background fields in the form of a general stack\footnote{In~\cite{RelativeQFT} the more general situation of simplicial sheaves is considered.} $\mathscr{F}$
(Section~\ref{Sec: Cob bicat}). Although this is only
a slight generalisation of previous constructions, it is still quite
technically complicated, and its explicit form makes all of our
statements precise. We show how the definition of 
an extended functorial field theory as a 2-functor 
\begin{align}
\Za\colon \ECobF \longrightarrow \Tvs
\end{align}
encodes some of the properties of quantum field theories with a focus on the implementation of symmetries via
limit morphisms. 
We use the symmetric monoidal bicategory $\Tvs$ of KV 2-vector spaces as the target for our field theories.
A KV 2-vector space is an additive semi-simple $\C$-linear category with finitely many isomorphism classes of simple objects.   

Building on this bicategory, we then use
the theory of symmetric monoidal bicategories
following~\cite{schommer2011classification} and the ideas
of~\cite{MonnierHamiltionianAnomalies} to work out the concrete
form of the anomalous quantum field theories sketched above; this is described in
Section~\ref{Sec: GT Anomalies}. 
We will show in Corollary~\ref{Coro: Relative for 1 induces ordinary} that a field theory with trivial 
anomaly $\Aa \cong \mbf1$ is nothing but an ordinary field theory. 
In Section~\ref{Sec: P anomaly actions} we explain the relation to more traditional approaches to 
the description of anomalies in terms of non-trivial line bundles and gerbes over the space of 
field configurations.  
Furthermore, the relation to projective representations of the gauge group
in~\cite{MonnierHamiltionianAnomalies}, and its extension to projective groupoid
representations following~\cite{BoundaryConditionsTQFT}, will be explained. 
We conclude Chapter~\ref{Sec: General Theory} with an abstract 
description of anomaly inflow, i.e.\  the coupling of bulk and boundary degrees of freedom, at the level of partition functions and state spaces in Section~\ref{Sec: Anomaly inflow}: 
let $M$ be a manifold with boundary $\partial M= \Sigma$ and $\Za\colon 1 \Longrightarrow \tr \Aa$ an anomalous quantum 
field theory with anomaly theory $\Aa \colon \ECobF \longrightarrow \Tvs$. The partition function of $\Za$ on $\Sigma$
is an element $\Za(\Sigma)\in \Aa(\Sigma)$. Evaluating the anomaly field theory $\Aa$ on $M$ induces a linear map 
$\Aa(M)\colon \Aa(\Sigma)\longrightarrow \Aa(\emptyset)\cong \C $. This allows us to define the partition function of the
combined system as the complex number $\Aa(M)[\Za(\Sigma)]\in \C$.
Similarly, the state space of $\Za$ on an $n-2$-dimensional manifold $S$ is an element of the linear category $\Aa(S)$ and
the choice of a manifold $\Sigma$ with boundary $\partial \Sigma =-S$ allows us to define a combined state space
$\Aa(\Sigma)[\Za(S)]\in \Aa(\emptyset)\cong \fvs$. 
Proposition~\ref{Prop: bulk boundary} and Theorem~\ref{Thm: Bulk boundary} show that the combined system does 
not suffer from any anomalies. The coupling of bulk and boundary degrees of freedom depends on the full quantum field theory 
$\Aa$ describing the anomaly and not just its truncation.   

Chapter~\ref{Cha: Index} is concerned with the construction of a
concrete example of this general formalism describing the parity
anomaly in odd spacetime dimensions and is based on the publication~\cite{parity}. 
As the parity anomaly is related
to an index in one
dimension higher~\cite{Niemi,AlvarezGaume,WittenFermionicPathInt}, this
suggests that quantum field theories with parity anomaly should take
values in an extended field theory constructed from index theory.\footnote{This
naturally fits in with the classification of topological insulators
and superconductors using index theory and K-theory,
see~\cite{Ertem:2017lni} for a recent exposition of this.} We
build such a theory using the index theory for manifolds with corners
developed in~\cite{LoyaMelrose,Loyaindex}, which extends the well-known
Atiyah-Patodi-Singer index theorem~\cite{APS} to manifolds with
corners of codimension~$2$. Our construction produces an extended
quantum field theory $\mathcal{A}_{\rm parity}^\zeta$ depending on a
complex parameter $\zeta\in\C^\times$ in any even spacetime dimension
$n$; for $\zeta=-1$ this theory describes the parity anomaly in odd
spacetime dimensions.  

To exemplify how our constructions fit into the
usual treatments of the parity anomaly from the path integral
perspective, we first consider in
Section~\ref{Sec: path integral parity} the simpler construction of an ordinary
(unextended) invertible quantum field theory $\CobF\longrightarrow \fvs$ from the usual
Atiyah-Patodi-Singer index theorem for even-dimensional manifolds with boundary reviewed in Section~\ref{Sec: APS}. We show that the definition of the
partition function $Z_{\rm parity}^\zeta$ as a natural symmetric
monoidal transformation implies that the complex number $Z_{\rm
  parity}^\zeta(M)$ transforms under a gauge transformation $\phi:M\longrightarrow
M$ by multiplication with a 1-cocycle of the gauge group given by $\zeta$ to a power determined by the index
of the Dirac operator on the corresponding mapping cylinder
$\mathfrak{M}(\phi)$. For
$\zeta=-1$ this is precisely the same gauge anomaly that arises from the spectral flow of edge states under
adiabatic evolution signalling the presence
of the global parity
anomaly~\cite{Redlich,AlvarezGaume,WittenFermionicPathInt}, which is
a result of the sign ambiguity in the definition of the fermion
path integral in odd spacetime dimension as explained in Section~\ref{Sec: Anomalies}. We further illustrate how
the bulk-boundary correspondence from Section~\ref{Sec: Anomaly inflow} recovers a construction from~\cite{WittenFermionicPathInt}.

A key feature of the Hamiltonian formalism defined by our construction
of the extended quantum field theory $\mathcal{A}_{\rm parity}^\zeta$
is that the index of a Dirac operator on a manifold with corners (which we recall in Section~\ref{Sec: Index2}) 
depends on the choice of a unitary self-adjoint chirality-odd
endomorphism of the kernel of the induced Dirac operator on all
corners, whose positive eigenspace defines a Lagrangian subspace of
the kernel with respect to its natural symplectic structure. We assemble
all possible choices into a linear category $\Aa_{\rm
  parity}^\zeta(S)$ assigned to
$n-2$-dimensional manifolds $S$ by $\Aa^\zeta_{\rm parity}$. The
index theorem for manifolds with corners splits into a sum of a bulk integral over
the Atiyah-Singer curvature form and boundary contributions depending on
the endomorphisms. We use these boundary
terms to define the theory $\Aa_{\rm parity}^\zeta$ on 1-morphisms,
i.e.\  on $n-1$-dimensional manifolds $\Sigma$; the general idea is to use
categorical coends\footnote{We will define coends in Section~\ref{Sec: Extended Indext FQFT}.} to treat all possible boundary conditions at the same
time. The index theorem then induces a natural transformation between
linear functors, defining the theory $\Aa_{\rm parity}^\zeta$ on
2-morphisms, i.e.\  on $n$-dimensional manifolds.

A crucial ingredient
in the construction of the invertible extended field theory $\Aa_{\rm
  parity}^\zeta$ in Section~\ref{Sec: Extended Indext FQFT} is a natural linear map
\begin{align} 
\Phi_{T_0,T_1}(\Sigma_0,\Sigma_1):\Aa_{\rm parity}^\zeta(\Sigma_1)\circ \Aa_{\rm
  parity}^\zeta(\Sigma_0)(T_0) \longrightarrow \Aa_{\rm
  parity}^\zeta(\Sigma_1\circ \Sigma_0)(T_0)
\end{align} 
for every pair of 1-morphisms $\Sigma_0:S_0\longrightarrow S_1$ and
$\Sigma_1:S_1 \longrightarrow S_2$ with corresponding endomorphisms $T_i$ on
the corner manifolds $S_i$; it forms the components of a natural
linear isomorphism $\Phi$ which is associative. A lot of information
about the parity anomaly is contained in this map. The
construction of $\Aa_{\rm parity}^\zeta$ allows us to fix endomorphisms
for concrete calculations and still have a theory which is independent of these
choices. Viewing a field
theory with parity anomaly as a theory $Z_{\rm parity}^\zeta$ relative
to $\Aa_{\rm parity}^\zeta$ in the sense explained before, we then get a
vector space of quantum states $Z_{\rm parity}^\zeta(S) \in \Aa_{\rm parity}^\zeta$ for
every $n-2$-dimensional manifold $S$; the groupoid of gauge
transformations $\mathsf{Sym}(S)$ only acts projectively on this
space. Denoting this projective representation by $\rho$, for any pair
of gauge transformations $\phi_1,\phi_2:S \longrightarrow S$ one finds
\begin{align}
\rho(\phi_2)\circ\rho(\phi_1)=\Phi_{T_1,T_2}\big(\mathfrak{M}(\phi_1),
\mathfrak{M}(\phi_2)\big) \ \rho(\phi_2\circ\phi_1) \ ,
\end{align} 
where $\mathfrak{M}(\phi_i)$ is the mapping cylinder of $\phi_i$. Using
results of~\cite{RelationEtaInvariants,LeschWojciechowski}, we can
calculate the corresponding 2-cocycle $\alpha_{\phi_1,\phi_2}$
appearing in the conventional Hamiltonian description of anomalies~\cite{Faddeev:1984jp,Faddeev:1985iz,Mickelsson:1983xi} in
terms of the action of gauge transformations on Lagrangian subspaces
of the kernel of the Dirac Hamiltonian on $S$; the explicit
expression can be found in Equation~\eqref{eq:2cocycleexplicit}. 

In Chapter~\ref{Chapter: t Hooft}, which is based on the publications~\cite{tHooft,EHQFT}, we give a mathematical description of 
symmetries of Dijkgraaf-Witten theories and their gauging 
in the framework of functorial 
field theory which is motivated by physical considerations~\cite{Kapustin:Symmetries}. 
Let $D$ be a finite group and $n$ a natural number.
The possible topological actions for $n$-dimensional Dijkgraaf-Witten
theories with gauge group $D$ are classified by the group cohomology
of $D$ or equivalently by the singular cohomology of 
the classifying space $BD$ with coefficients 
in $U(1)$ \cite{DijkgraafWitten,FreedQuinn}. Let $\omega \in Z^n(BD;U(1))$ be an $n$-cocycle
and $M$ an $n$-dimensional manifold. 
Let $P$ be a $D$-gauge field on $M$ with classifying map 
$\psi_P\colon M \longrightarrow BD$.
The action of the Dijkgraaf-Witten theory $E_\omega$ evaluated at $P$
is given by
\begin{align}
\exp(2\pi \,\iu\, S_{\text{DW}})\coloneqq \int_M\, \psi_P^*\, \omega \ .
\end{align} 
The quantum theory can be defined by appropriately summing 
over isomorphism classes of $D$-bundles. We present a construction of (equivariant) Dijkgraaf-Witten theories
as extended functorial field theories 
from the parallel transport of higher flat gerbes in Section~\ref{Sec: Parallel transport} using 
the orbifold construction of Schweigert and Woike~\cite{OFK, EOFK}. 
The classical field theory will be defined on the bordism category $\EDCob$ constructed by specifying the background 
fields $\F$ to consist of orientations and principal $D$-bundles described by maps into $BD$. The only background
structure required for the quantum theory is an orientation.
We also analyse the algebraic structures underlying the resulting extended functorial field theory in 2 and 3-dimensions. 
We will show that in 2-dimensions the gauge theory is completely described 
by the $\omega$-twisted representation theory of $D$. In 3-dimensions the theory can be described 
in terms of the representation theory of the $\omega$-twisted Drinfeld double of $D$~\cite{Twisted_Drinfeld_double,TwistedDWandGerbs}. 

In Section~\ref{Sec: DW Sym} we study discrete symmetries of Dijkgraaf-Witten theories. 
In general, a physical symmetry group $G$ acts on gauge fields only up to gauge transformations. 
Since for finite gauge groups, gauge transformations can be naturally identified with 
homotopies of classifying maps, we define such an action as a
homotopy coherent action of $G$ on 
$BD$ (Definition~\ref{Def: homotopy Coherent action}). We show that,
up to equivalence, homotopy 
coherent actions on $BD$ are described by non-abelian group 2-cocycles. If $D$ is abelian,
this description agrees with~\cite{Kapustin:Symmetries}. 
Non-abelian 2-cocycles classify extensions of $G$ by $D$:
\begin{align}
1\longrightarrow D \overset{\iota}{\longrightarrow} \widehat{G} \overset{\lambda}{\longrightarrow} G \longrightarrow 1 \ .
\end{align} 
This extension has a natural physical interpretation: It describes how to 
combine $D$- and $G$-gauge fields into a single $\widehat{G}$-gauge field. When the extension is non-trivial, i.e.\  $\widehat{G}$ is not a product group $D\times G$, one says that the $G$-symmetry is `fractionalized'~\cite{Wang:2017loc}.

A homotopy coherent action on $BD$ induces a homotopy coherent action on the 
collection of classical $D$-gauge theories. 
Homotopy fixed points of this action are defined to be classical 
field theories with $G$-symmetry 
(Definition~\ref{Def: Field theory with k. symmetry}).
An essential feature of homotopy fixed points is that they are a
structure, not a property.  
In Proposition~\ref{Prop: Classical symmetry}
we show that if the topological action is preserved by the action of $G$ 
(Definition~\ref{Def: Preserved}), 
then the corresponding Dijkgraaf-Witten theory can be equipped with 
a homotopy fixed point structure.

An internal symmetry of a quantum field theory acts on its Hilbert space
of states. This motivates the definition of a functorial quantum
  field theory with {internal $G$-symmetry} as a functor 
\begin{align}
\Cob_n^\mathscr{F} \longrightarrow G\text{-}\mathsf{Rep}
\end{align}  
to the category $G\text{-}\mathsf{Rep}$ of representations 
of $G$.   
We show in Proposition \ref{Prop: symmetry extends to quantum} that 
classical symmetries of Dijkgraaf-Witten theories induce internal symmetries 
of the quantized theory. 
This shows that discrete gauge theories are anomaly-free in the sense that all symmetries 
extend to the quantum level. 

{'t Hooft anomalies} appear 
as an obstruction to gauging the $G$-symmetry, 
i.e.\  to coupling it to 
non-trivial background gauge fields (Definition \ref{Def: Gauging}). 
Gauging the $G$-symmetry can be achieved by finding a 
topological action $\widehat{\omega}$ for a $\widehat{G}$-gauge theory which restricts to $\omega$ and 
performing a path integral over
$D$-gauge fields. 
Mathematically, this can be described by       
the equivariant Dijkgraaf-Witten theories introduced in Section~\ref{Sec: Equivariant DW}. 
In Theorem \ref{Theorem: Gauging} we prove that 
the equivariant Dijkgraaf-Witten theory corresponding to $\widehat{\omega}$ gauges
the $G$-symmetry. We discuss the gauging of discrete symmetries
in Section \ref{Sec: Gauging}.

However, in general it might be impossible to find a
topological action which restricts correctly. In this 
case we say that the corresponding symmetry has a 
't Hooft anomaly. 
The obstructions for $\widehat{\omega}$ to exist are encoded
in the Lyndon-Hochschild-Serre spectral sequence.
For an $n$-dimensional field theory there are $n$ obstructions
which need to vanish. 
In Proposition \ref{Prop: Obstructions} we show that if all 
obstructions except the last one vanish then there exists an $n{+}1$-dimensional     
topological action $\theta $ for a discrete $G$-gauge theory,
together with an $n$-cochain $\omega'$ in $C^n(B\widehat{G}; U(1))$ satisfying
$\iota^* \omega' = \omega$ and $\delta \omega' = \lambda^* \theta$.
These obstructions are studied in Section \ref{Sec: Obstruction}. In Section~\ref{Sec: Fully extended}
we present some ideas on possible interpretations of the obstructions from the point of view of fully extended
topological quantum field theories and defects.

Based on a relative version of the push construction from~\cite{OFK} we construct a boundary 
quantum field theory $Z_{\omega'}$ encoding the anomaly in Section~\ref{Sec: Boundary DW} if 
all obstructions beside the last one vanish. 
Let $\lambda \colon \widehat{G} \longrightarrow G$ be a group homomorphism and 
$\Za_1, \Za_2 \colon \EGCob \longrightarrow \Tvs$ extended field theories. The group homomorphism $\lambda$ induces a 2-functor 
$\lambda \colon \EWGCob \longrightarrow \EGCob$. The relative push construction allows us to construct from a relative 
field theory $Z\colon \tr \lambda^*\Za_1 \Longrightarrow \tr \lambda^* \Za_2$ a relative field theory 
$\lambda_* Z \colon \tr \Za_1 \Longrightarrow \tr  \Za_2 $. In the case that $\Za_1=\Za_2=\mbf1$ 
this construction reduces to the one given in~\cite{OFK}. 
The relative push construction is a bit involved and hence we refer to Section~\ref{Sec: State space} for details.
However, let us give an informal description of the partition function of $Z_{\omega'}$ here.
The fact that $\omega'$ is not closed implies that 
$\int_M\, \psi_{\widehat P}^*\,{\omega'}$
is not gauge-invariant for a general $\widehat{G}$-bundle $\widehat P $ on $M$. 
Under a gauge transformation 
$\widehat{h}\colon \widehat P \longrightarrow \widehat P'$
the value of $\int_M\, \psi_{\widehat P}^*\,{\omega'}$
changes by multiplication with\footnote{This integral is not actually
well-defined as a complex number, see Section \ref{Sec: DW via Orbifold} for details.
We ignore this subtlety in the present section.} 
$\int_{[0,1]\times M}\, \widehat{h}^* \delta {\omega'}$, 
where we consider $\widehat{h}$ as a homotopy $[0,1]\times M \longrightarrow
B\widehat G$.
We can rewrite this integral as $\int_{[0,1]\times M}\, (\lambda_*\widehat{h}\,)^* {\theta}$.
This is exactly the value of $E_\theta$ evaluated on $\lambda_*\widehat{h}$,
which shows that the anomaly is controlled by the bulk classical gauge theory 
$E_\theta$. 
The rough idea for the construction of $Z_{\omega'}$ is to modify the partial 
orbifold construction used in Section \ref{Sec: Gauging} in a way suited
to the construction of boundary states. 
Let us fix a $G$-bundle $P$ on $M$. To define the partition function 
we want to perform an integration over the preimage of $P$ under $\lambda_*$.
However, in the presence of gauge transformations, requiring two bundles to be
the same is not natural. Hence we use the homotopy fibre 
$\lambda_*^{-1}[P]$ as a groupoid with objects given by pairs $(\widehat P, h)$ of a 
$\widehat G$-bundle $\widehat P$ and a gauge transformation 
$h\colon \lambda_*\widehat{P} \longrightarrow P $. Morphisms are gauge 
transformations $\widehat{h} \colon \widehat{P} \longrightarrow \widehat{P}'$
which are compatible
with $h$ and $h'$.
We show that 
\begin{align}
Z_{\omega'}(M)\coloneqq \int_M\, \psi_{\widehat{P}}^*\,\omega' \ \int_{[0,1]\times M}\, h^*\theta
\end{align}           
is gauge-invariant with respect to morphisms in $\lambda_*^{-1}[P]$. 
We define the partition function of $Z_{\omega'}$ on $M$ as\footnote{We recall integration over essentially finite groupoids in Appendix~\ref{Sec: Homotop Groupoid}.} 
\begin{align}
Z_{\omega'}(M) := \int_{\lambda_*^{-1}[P]}\, L_{\omega'}(M) \ .
\end{align}  
The state space of the anomalous theory is constructed in Section~\ref{Sec: State space}.
The groupoid of symmetries acts only projectively on this state space.
Using Theorem~\ref{Theorem: Transgression} we show
that the 2-cocycle twisting this
projective representation is the transgression of $\theta$ to the groupoid
of $G$-bundles. With this construction we provide an explicit
demonstration of the anomaly inflow mechanism developed in Section~\ref{Sec: Anomaly inflow} at the level of both
partition functions and state spaces, which renders the composite bulk-boundary 
field theory free from anomalies.

Appendix~\ref{Chap: Bicat} outlines the basic definitions and our conventions related to symmetric monoidal bicategories.
In Appendix~\ref{Sec: Homotop Groupoid} we collect some background material on the canonical model structure on the 
category of groupoids, homotopy (co)limits, stacks and integration over finite groupoids. Model categorical language 
is not required to understand the main part of the thesis as long as the reader is willing to accept the 
concrete models for homotopy fibres and pullbacks provided without derivation.  

\section{Conventions and notations}
For the convenience of the reader, we summarise here our notation and
conventions which are used throughout this paper.
\begin{itemize}
\item We denote by $\fvs$ the symmetric monoidal category of finite-dimensional
complex vector spaces.
\item We denote by $\Tvs$ the symmetric monoidal bicategory of KV 2-vector spaces.
\item We denote by $\Hilb$ the symmetric monoidal category of complex Hilbert spaces.

\item We denote by $\Grp$ the category of groups. 
\item We denote by $\mathsf{Grpd}$ the 2-category of (small) groupoids.

\item We denote by $\Cat$ the bicategory of (small) categories.

\item We denote by $\CobF$ the symmetric monoidal category of
$n$-dimensional geometric cobordisms with background fields
$\mathscr{F}$.
\item We denote by $\ECobF$ the symmetric monoidal bicategory of $n$-dimensional
geometric cobordisms with background fields $\mathscr{F}$.

\item Let $G$ be a group.
We denote by $B G$ the classifying space of $G$, which for finite $G$
is an Eilenberg-MacLane space $K(G,1)$, i.e.\  its only non-trivial
homotopy group is $\pi_1(BG)=G$.
Let $P$ be a principal $G$-bundle on a manifold $M$. 
We denote by $\psi_P\colon M \longrightarrow BG$ the corresponding
classifying map. 

\item Let $T$ be a topological space, $n$ a positive integer 
and $A$ an abelian group. 
We denote the pairing of chains and cochains on $T$ by 
$\langle \,\cdot\, , \,\cdot\, \rangle \colon C^n(T;A)\times C_n(T) \longrightarrow A$.

\item 
Let $G$ be a group.
The groupoid of $G$-bundles on a manifold $M$ is denoted by $\BunG(M)$. 

\item Most constructions in this thesis are done in a fixed dimension $n$. 
For such a fixed dimension, we use $M$, $\Sigma$ and $S$ to denote manifolds of dimensions $n$, $n-1$ and $n-2$, respectively.

\item 
Let $M$ be an oriented manifold. We denote by $-M$ the same manifold
equipped with the opposite orientation.

\item 
Let $\Ca$ be a monoidal category. We denote by $*\Ds \Ca$ the bicategory 
with one object and $\Ca$ as endomorphisms. 

\item 
Let $\lambda \colon G\longrightarrow G'$ be a homomorphism of groups. 
We denote the induced maps $BG\longrightarrow BG'$ and $*\Ds G \longrightarrow *\Ds G'$ again by $\lambda$. 
 
\item Let $F\colon \Ca \longrightarrow \Ca'$ be a functor. We write the 
limit of $F$ as an end $\int_\Ca\, F$. 

\item Let $G$ be a finite group. We denote by
$\EGCob$ the bicategory of cobordisms equipped with maps into $BG$.

\end{itemize}
\chapter{Description of anomalies via extended functorial quantum field theories}\label{Sec: General Theory}
In this chapter we will develop a general framework for the description of anomalies
using extended functorial quantum field theories following and extending previous work~\cite{FreedAnomalies, MonnierHamiltionianAnomalies}. In Section~\ref{Sec: EFQFT} we provide the necessary definitions before
turning to anomalies in Section~\ref{Sec: GT Anomalies}. We focus on the abstract description and postpone the discussion of
non-trivial examples to later chapters.  We will heavily 
use the language of symmetric monoidal bicategories. For our conventions concerning (and details about) symmetric monoidal 
bicategories see Appendix~\ref{Chap: Bicat}. 

\section{Extended functorial field theories}\label{Sec: EFQFT}
Quantum field theories are usually defined on smooth manifolds equipped with additional structure, 
such as an orientation or a metric. We call the additional non-dynamical structures required to define
a specific quantum field theory \emph{background fields}. 
These fields should be local, i.e.\ they form a sheaf (or higher versions thereof) on
the category of manifolds, and can have (higher) internal symmetries such as gauge
symmetries. We thus incorporate all data of background fields such as bundles with connections, spin structures and metrics into a stack
$
\mathscr{F}\colon \mathsf{Man}_n^{\mathrm{op}}\longrightarrow \mathsf{Grpd}
$
on the category $\mathsf{Man}_n$ of $n$-dimensional manifolds with
corners and local diffeomorphisms; we regard $\mathsf{Man}_n$ as a
2-category with only trivial 2-morphisms, and $\mathsf{Grpd}$ denotes
the 2-category of small groupoids, functors and natural
isomorphisms. One should think of elements of $\mathscr{F}(M)$ as the
collection of classical background fields on $M$, which in particular
satisfies the sheaf condition, i.e.\ for every open cover $\{U_a\}$ of
a manifold $M$, the diagram
$$
\Fscr(M)\longrightarrow \prod_a\, \Fscr(U_a)\doublearrow
\prod_{a,b}\, \Fscr(U_a\cap U_b) \triplearrow \prod_{a,b,c}\, \Fscr(U_a\cap U_b \cap U_c)
$$
is a weak/homotopy equalizer diagram in $\mathsf{Grpd}$. 
In Appendix~\ref{Sec: Stacks} we review some basic facts and definitions related to stacks.
The site we use to define stacks differs from the one usually used in the context of differential 
geometry. It contains local diffeomorphisms as morphisms instead of arbitrary smooth maps.
The reason for this is that we want to be able to describe structures such as metrics or orientations which 
cannot be pulled back along general smooth maps.    
Even-though desirable for applications to higher gauge theory, we decided not to work with 
higher stacks since they are not necessary to capture the examples discussed in this thesis. 

Throughout this chapter we fix a spacetime dimension $n\in \N_{>0}$ and a stack 
$\mathscr{F}$ defined on $n$-dimensional manifolds describing the classical background fields of
the theory under consideration. Usually, $\mathscr{F}$ is a product of different types of fields. Common types
appearing in physics are orientations, Riemannian metrics, principal bundles with connections
and spin structures. Note that most of these structures are actually sheaves considered as stacks
with only identity automorphisms.
\begin{remark}
Throughout this thesis we will ignore the following technical point: 
background fields on a manifold usually form smooth (infinite dimensional) manifolds and hence one
should work with smooth stacks and at various points demand constructions to depend 
smoothly on the background gauge fields. This can be achieved by working with (bi)categories fibred over the
category of smooth manifolds~\cite{Stolz:2011zj}.   
\end{remark}

\subsection{Non-extended field theories}\label{Sec: Non-extended field theories}
Before defining extended functorial quantum field theories we
briefly recall the non-extended definition~\cite{Stolz:2011zj,FreedAnomalies,FelixKleinSegal}. 
To model physical spacetimes with background fields, we introduce a symmetric 
monoidal category $\CobF$ whose objects are triples 
$(\Sigma,\epsilon, f_{\Sigma}\in \mathscr{F}((-\epsilon,\epsilon)\times \Sigma))$, where 
$\Sigma$ is an $n-1$-dimensional closed manifold with connected components $\Sigma_1,\dots , \Sigma_l$, $\epsilon=(\epsilon_1,\dots \epsilon_l)$ is an 
$l$-tuple of positive real numbers 
and $f$ an element constant along $(-\epsilon,\epsilon)$ of $\mathscr{F}((-\epsilon,\epsilon)\times \Sigma)$. Here we use the shorthand notation 
\begin{align}
(-\epsilon,\epsilon)\times \Sigma \coloneqq \sqcup_{i=1}^l (-\epsilon_i,\epsilon_i)\times \Sigma_i \ \ .
\end{align}
The precise definition of a constant element is not so important  
for the moment; it is enough to have the example of a structure which is pulled back from $\Sigma$ or a metric of the form $dt^2+g_\Sigma$ in mind. We will give a definition in
Definition~\ref{Def:Constant}.

Before defining the morphisms of $\CobF$ we introduce:
\begin{definition}
Let $M_1$ and $M_2$ be manifolds equipped with background fields $f_1\in \mathscr{F}(M_1)$
and $f_2\in \mathscr{F}(M_2)$.
A \emph{$\mathscr{F}$-diffeomorphism} 
$(M_1,f_1)\longrightarrow (M_2,f_2)$ consists of a diffeomorphism $\varphi\colon 
M_1\longrightarrow M_2$ together with a morphism $\varphi^*f_2\longrightarrow f_1$ in 
$\mathscr{F}(M_1)$.
\end{definition} 
There are two types of morphisms in $\CobF$. The first type is given by equivalence 
classes of $n$-dimensional cobordisms with background fields. A
\emph{cobordism with background fields} $(\Sigma_-,\epsilon_-,f_{\Sigma_-}) \longrightarrow 
(\Sigma_+,\epsilon_+,f_{\Sigma_+})$ is a 4-tuple 
\begin{equation}
(M,f_M\in \mathscr{F}(M), \varphi_-\colon [0,\epsilon_-)\times \Sigma_- \longrightarrow M_-, \varphi_+\colon (-\epsilon_+,0]\times \Sigma_+ \longrightarrow M_+ )
\end{equation} 
where $M$ is an $n$-dimensional compact manifold with boundary, 
$f_M$ is a background field on $M$, $M_-\cup M_+$ is a collar the boundary of $M$, i.e. an open neighbourhood which is 
diffeomorphic to $[0,1)\times \partial M$ of and 
$\varphi_-$ and $\varphi_+$ are $\mathscr{F}$-diffeomorphisms.     
Two cobordisms with background fields $(M,f,\varphi_-,\varphi_+)$ and $(M',f',\varphi'_-,\varphi'_+)$ from
$(\Sigma_-,\epsilon_-,f_{\Sigma_-})$ to 
$ (\Sigma_+,\epsilon_+,f_{\Sigma_+})$
are \emph{equivalent} if there exists a $\F$-diffeomorphism 
$\psi\colon M \longrightarrow M'$ compatible with the collars. Concretely, this means that the diagram
\begin{equation}
\begin{tikzcd}
 & M \ar[dd,"\psi"] & \\
 {[0,\epsilon_-) \times \Sigma_- }\ar[ru,"\varphi_-"]\ar[rd,"\varphi'_-",swap]& & {(-\epsilon_+,0] \times \Sigma_+}\ar[lu,"\varphi_+",swap]\ar[ld,"\varphi'_+"]\\
 & M' &  
\end{tikzcd}  
\end{equation}
of $\F$-diffeomorphisms commutes.
\begin{figure}[hbt]
\begin{overpic}[width=14.0cm,
scale=0.3]{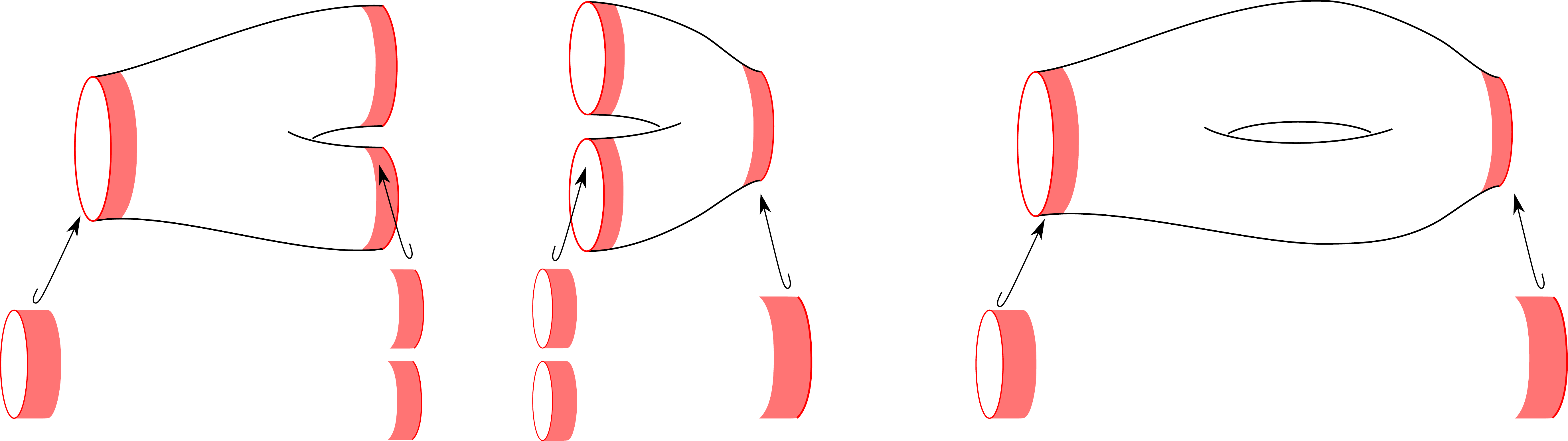}
\put(13,22){\scriptsize{$M$}}
\put(42,22){\scriptsize$M'$}
\put(5,11){\scriptsize$\varphi_{-}$}
\put(26,15){\scriptsize$\varphi_{+}$}
\put(33,15){\scriptsize$\varphi'_{-}$}
\put(51,12){\scriptsize$\varphi'_{+}$}
\put(67,11){\scriptsize$\varphi_{-}$}
\put(99,12){\scriptsize$\varphi'_{+}$}
\put(78,22){\scriptsize$M'  \circ M$}
\put(0,-2){\scriptsize$  [0,\epsilon_a) \hspace{-0.14cm} \times\hspace{-0.08cm}\Sigma_a$}
\put(17,-2){\scriptsize$ (-\epsilon_b,0] \hspace{-0.14cm} \times\hspace{-0.08cm} \Sigma_b $}
\put(31,-2){\scriptsize$ [0,\epsilon_b) \hspace{-0.14cm} \times\hspace{-0.08cm} \Sigma_b $}
\put(45,-2){\scriptsize$ (-\epsilon_c,0] \hspace{-0.14cm} \times\hspace{-0.08cm}\Sigma_c $}
\put(60,-2){\scriptsize$ [0,\epsilon_a) \hspace{-0.14cm} \times\hspace{-0.08cm} \Sigma_a $}
\put(90,-2){\scriptsize$ (-\epsilon_c,0] \hspace{-0.14cm} \times\hspace{-0.08cm}\Sigma_c $}
\end{overpic}
\bigskip
\caption{\small Illustration of two composable regular morphisms in $\CobF$ (on the left) and their composition (on the right).}
\label{Fig:Sketch CompositionA}
\end{figure}
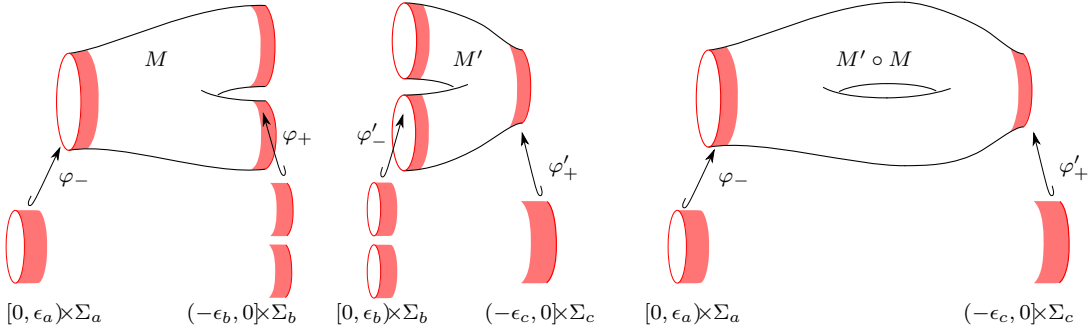
We refer to morphisms described by equivalence classes of cobordisms with background fields as 
\emph{regular morphisms}. 

The second class of morphisms $(\Sigma_a,\epsilon_a,f_{\Sigma_a})\longrightarrow (\Sigma_b,\epsilon_b,f_{\Sigma_b})$ is only defined for $\epsilon_a=\epsilon_b$ 
and consist of a diffeomorphism $\varphi\colon \Sigma_a \longrightarrow \Sigma_b$
together with the structure of a constant $\F$-diffeomorphism on  $(-\epsilon_a,\epsilon_a)\times \varphi\colon
(-\epsilon_a,\epsilon_a)\times \Sigma_a \longrightarrow (-\epsilon_a,\epsilon_a)\times \Sigma_b$.
We call these morphisms \emph{limit morphisms}. 
\\
The composition of limit morphisms is induced from the composition of $\F$-diffeomor-phisms.
From a limit morphism $\varphi$ we can construct a regular 
morphism $\mathcal{M}_\epsilon(\varphi)$ for every $\epsilon > \epsilon_a$ by 
gluing $[0,\tfrac{3}{4}\epsilon)\times \Sigma_a$ and 
$(\tfrac{1}{4}\epsilon ,\epsilon]\times \Sigma_b$ together along the diffeomorphism 
$(\tfrac{1}{4}\epsilon ,\tfrac{3}{4}\epsilon )\times \varphi$. The constant 
$\F$-diffeomorphism $\varphi$  can be used  to equip 
$\mathcal{M}_\epsilon(\varphi)$ with a $\F$-background gauge field via the descent property of stacks. 
We call $\mathcal{M}_\epsilon(\varphi)$ the \emph{mapping cylinder of length $\epsilon$}
of $\varphi$. Note that $\mathcal{M}_\epsilon (\varphi)$ can be defined as a manifold for all 
$\epsilon>0$. 
Informally, we think of limit morphism $\varphi$ as the $\epsilon \longrightarrow 0$ limit of $\mathcal{M}_\epsilon(\varphi)$.

Let $\varphi$ be a limit morphism and $(M,f,\varphi_-,\varphi_+)$ a regular morphism. 
The composition $(M,f,\varphi_-,\varphi_+) \circ \varphi$ (when defined) 
is the regular morphism $(M,f,\varphi_-\circ [0,\epsilon)\times \varphi ,\varphi_+)$.
The composition $\varphi \circ (M,f,\varphi_-,\varphi_+)$ 
is the regular morphism $(M,f,\varphi_- ,  \varphi_+ \circ (-\epsilon,0]\times \varphi^{-1})$.
\\
The composition of regular morphisms $(M,f,\varphi_-,\varphi_+)\colon(\Sigma_a,\epsilon_a,f_{\Sigma_a}) \longrightarrow 
(\Sigma_b,\epsilon_b,f_{\Sigma_b})$ and $(M',f',\varphi'_-,\varphi'_+)\colon(\Sigma_b,\epsilon_b,f_{\Sigma_b}) \longrightarrow 
(\Sigma_C,\epsilon_c,f_{\Sigma_c})$ is given by gluing the underlying manifolds along collars (see Figure~\ref{Fig:Sketch CompositionA} for a sketch).
To equip the composition with a background field cover $M\sqcup_{\Sigma_b} M'$ with three 
open sets
\begin{align}\label{Eq: Cover gluing}
\begin{split}
U_-& = M \setminus \Sigma_b, \\ 
U_+ & = M'\setminus\Sigma_b, \\ 
U_{-+} & = (-\epsilon,\epsilon)\times \Sigma_b \ \ .
\end{split} 
\end{align} 
The background gauge field on $M\sqcup_{\Sigma_b} M'$ is now defined via the descent
property of the stack $\F$ for the cover $\{U_-,U_{-+},U_+\}$ from $f|_{U_-}$, $f'|_{U_+}$ and $f_{\Sigma_b}$. The required morphisms in 
$\F(U_-\cap U_{-+} )$ and $\F(U_{-+}\cap U_+ )$ are part of the structure of the $\F$-diffeomorphisms 
$\varphi_+$ and $\varphi_-$.  
The monidal structure on $\CobF$ is the disjoint union of manifolds. 
\begin{remark}\label{Rem: Composition in terms of mapping cylinders}
The composition of regular morphisms can also be described in terms of mapping cylinders. It corresponds to removing 
the collars for the boundaries along which the gluing takes place and then replacing them with mapping cylinders 
$\mathcal{M}_{\epsilon_b}(\varphi_+^{-1})$ and $\mathcal{M}_{\epsilon_b}(\varphi'_-)$, respectively. The composed morphism corresponds
to gluing the resulting manifolds along $\Sigma_b$.  
\end{remark}
\begin{remark}\label{Rem: Shortcommings CobF}
The definition of $\CobF$ is far from perfect. There are two important points one 
should improve to describe the actual structure of quantum field theories:
\begin{itemize}
\item The objects we use depend on the actual value of $\epsilon$. It would be better 
to work with germs of geometric structures. 

\item Our definition does not contain any information about the smooth structure of the 
background fields.
\end{itemize}
Both problems are solved in the approach by
Stolz and Teichner~\cite{Stolz:2011zj} using germs and categories fibred over manifolds. 
We do not work with their approach, because it is harder to generalize
to the bicategorical setting in Section~\ref{Sec: Cob bicat}.   
\end{remark}
Based on the symmetric monoidal category $\CobF$ describing the 
geometry of $n$-dimensional manifolds, we can give a definition of $n$-dimensional 
functorial quantum field theories.
\begin{definition}\label{Def: QFT}
An \emph{$n$-dimensional functorial quantum field theory} is a symmetric monoidal 
functor 
\begin{align}
Z\colon \CobF \longrightarrow \Hilb\ ,
\end{align}
where $\Hilb$ is the symmetric monoidal category of complex Hilbert spaces and bounded linear operators. 
\end{definition} 
The simplest example of a functorial quantum field theory is the trivial theory 
\begin{align}
1\colon \CobF \longrightarrow \Hilb
\end{align} 
sending every object to the one-dimensional Hilbert space $\C$ and every morphism 
to the identity map. 

Definition~\ref{Def: QFT} is an effective way to (partially) encode the structure generally expected to be part of a 
quantum field theory.
For example, the \emph{time evolution operator} ``$\exp(-\iu H t)$" on the Hilbert space $Z((\Sigma,\epsilon,f))$ for an 
object $(\Sigma,\epsilon,f)\in \CobF$
is the bounded linear operator 
assigned to the mapping cylinder $\mathcal{M}_t(\id_\Sigma)$. Note that in our formalism the time evolution operator 
is only defined for $t>  \epsilon$. This is a consequence of the shortcomings mentioned in 
Remark~\ref{Rem: Shortcommings CobF}. For a more detailed discussion of the functorial approach to
quantum field theory we refer to \cite{FelixKleinSegal}.

Let $\Sigma$ be an $n-1$-dimensional manifold.
We denote by $\SymF(\Sigma)$ the \emph{symmetry subgroupoid} of $\CobF$ with objects having $\Sigma$ as 
underlying manifold and only limit morphisms as morphisms. Furthermore, we denote by 
$\SymF$ the subgroupoid of $\CobF$ containing only limit morphisms. 
\begin{example}
\begin{itemize}
\item
Let $\F$ be the stack (sheaf) of (Riemannian) metrics. Every metric $g$ on $\Sigma$
defines an object $(\Sigma, 1, dt^2+g )\in \CobF$, where we denote the coordinate along $(-1,1)$ by $t$.
A morphism $(\Sigma, 1, dt^2+g )\longrightarrow (\Sigma', 1, dt^2+g' )$ in $\SymF(\Sigma)$ is an isometry $\Sigma 
\longrightarrow \Sigma'$. 

\item  
Let $G$ be a Lie group and $\F$ the stack of principal $G$-bundles with connection.
A bundle on $\Sigma$ defines again an object of $\SymF(\Sigma)$ via pullback along the projection
$(-1,1)\times \Sigma \longrightarrow \Sigma$. The 
morphism in $\SymF(\Sigma)$ with trivial underlying diffeomorphism corresponding to connection preserving gauge 
transformations. 
\end{itemize}
\end{example}
A \emph{representation of a groupoid $\G$} is a functor $\G \longrightarrow \Hilb \ (\text{or }\fvs)$.
A representation is called \emph{unitary} if all morphisms in its image are unitary operators.   
Every functorial quantum field theory induces by restriction to $\SymF$ a representation 
of its symmetry groupoid. However, in general there is no reason for this representation to be 
unitary. To get unitary representations from functorial field theories one needs to work
with so called reflection positive field theories~\cite{FreedHopkins}.
  
In physics, a symmetry is an automorphism of the state space
which commutes with the time evolution operator or equivalently the Hamiltonian of 
the theory. We show that this is captured by the framework functorial field theories. 
\begin{proposition}
Let $Z\colon \CobF \longrightarrow \Hilb$ be a functorial quantum field theory 
and $\varphi\colon (\Sigma,\epsilon,f)\longrightarrow (\Sigma',\epsilon,f')$ a limit morphism. 
The symmetry $Z(\varphi)\colon Z(\Sigma,\epsilon,f)\longrightarrow Z(\Sigma',\epsilon,f')$ 
commutes with the time evolution operator.
\end{proposition}
\begin{proof}
The mapping cylinder $\mathcal{M}_t(\id_{\Sigma,f})$ is $([0,t]\times \Sigma,f) \in \CobF$. 
We prove the statement by showing that
the morphisms $ ([0,t]\times \Sigma',f') \circ \varphi $ and $ \varphi \circ ([0,t]\times \Sigma,f)  $
agree as morphisms in $\CobF$. 
For this consider the diagram 
\begin{equation}
\begin{tikzcd}
 & {[0,t]\times \Sigma'} \ar[dd,"\id\times \varphi^{-1}"] & \\ 
{[0,\epsilon)\times \Sigma} \ar[ru,"\id \times \varphi"]\ar[rd] & & {(-\epsilon,0]\times \Sigma'} \ar[dl,"\cdot+t \times \varphi^{-1}"] \ar[lu] \\ 
 &{[0,t]\times \Sigma}& 
\end{tikzcd}
\end{equation}
of $\F$-manifolds which shows that both morphism lie in the same equivalence class.
\end{proof}   
Let $Z_1\colon \CobF \longrightarrow \Hilb$ and $Z_2\colon \CobF \longrightarrow \Hilb$ be $n$-dimensional functorial
quantum field theories. We define the tensor product $Z_1\otimes Z_2$ locally,
i.e.\ via 
\begin{align}
Z_1\otimes Z_2(\Sigma,\epsilon,f)\coloneqq Z_1(\Sigma,\epsilon,f)\otimes Z_2(\Sigma,\epsilon,f) \ \ .
\end{align} 
There is a simple class of functorial quantum field theories which is 
central for the description of anomalies in qunatum field theories.
\begin{definition}\label{Def: Invertible}
	A functorial quantum field theory $Z\colon \CobF \longrightarrow \Hilb$ is \emph{invertible}
	if there exists a functorial quantum field theory $Z^{-1}\colon \CobF \longrightarrow \Hilb$
	such that $Z\otimes Z^{-1}\cong 1$. 
\end{definition}
Recall that a Picard groupoid is a symmetric monoidal groupoid in which every object and 
morphism has an inverse with respect to the tensor product.  
Every invertible field theory factors through the maximal Picard subgroupoid of $\Hilb$.

\subsection{A cobordism bicategory}\label{Sec: Cob bicat}
In this section we make the first step towards extending the definition of
a functorial quantum field to include manifolds of codimension $2$ by 
defining a bicategorical version of $\CobF$. This requires the use of manifolds 
with corners. We introduce the necessary mathematical definitions now. 

\subsubsection*{Manifolds with corners}
It is possible to define manifolds with corners of arbitrary codimension. We restrict ourself
to codimension $2$ for notational convenience. 
Roughly speaking, a manifold of dimension $n$ is a topological space which locally looks like open subsets of $\R^n$. The idea behind manifolds 
with corners of codimension $2$ is to replace $\R^n$ by $\R^{n-2}\times \R_{\geq0}^2$; we denote by $\mathrm{pr}_{\R_{\geq0}^2}:\R^{n-2}\times \R_{\geq0}^2\longrightarrow \R_{\geq0}^2$ the obvious projection. 
A map between subsets of $\R^{n-2}\times \R_{\geq0}^2$ is \emph{smooth} if there exists a smooth extension to open subsets of $\R^n$. 
A \emph{chart} for a subset $U$ of a topological space $X$ is then a homeomorphism $\varphi \colon U \longrightarrow V \subset \R^{n-2}\times \R_{\geq0}^2$. 
Two charts $\varphi_1 \colon U_1 \longrightarrow V_1$ and $\varphi_2 \colon U_2 \longrightarrow V_2$ are \emph{compatible} if $\varphi_2 \circ \varphi_1^{-1}\colon \varphi_1(U_1\cap U_2)\longrightarrow \varphi_2(U_1\cap U_2)$ is a diffeomorphism. 
As for manifolds, a collection of charts covering $X$ is called an \emph{atlas}. 
An atlas is \emph{maximal} if it contains all compatible charts. 
\begin{definition}
A \emph{manifold with corners} of codimension $2$ is a second countable Hausdorff space $M$ together with a maximal atlas.
\end{definition}    
\begin{remark}
Closed manifolds and manifolds with boundary are, in particular, manifolds with corners.
\end{remark}

We define the \emph{tangent space} $T_x M$ at a point $x\in M$ as the space of derivatives on the real-valued functions $C^\infty(M)$ at $x$. We define \emph{embeddings} in the same way as for manifolds without corners, i.e.\ as smooth injective maps which are injective on all tangent spaces. 
We are now able to introduce an essential concept used throughout this thesis.
\begin{definition}\label{def:collar}
A \emph{collar} for a submanifold $Y\subset \partial M$ is a diffeomorphism $\varphi \colon [0,\epsilon) \times Y \longrightarrow   U_Y$ for some fixed $\epsilon>0$ and a neighbourhood $U_Y$ of $Y$.  
\end{definition}
\begin{figure}
\begin{center}
\includegraphics[scale=2]{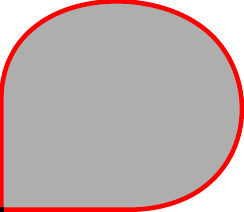}
\caption{A manifold with corners such that the red submanifold of the boundary does not admit a collar.}
\label{Fig: Collar}
\end{center}
\end{figure}
\begin{remark}
In general collars do not exist as can be seen from the simple example in Figure~\ref{Fig: Collar}.
\end{remark}
Given $x\in M$ we define the \emph{index of $x$} $\mathrm{index}(x)\in\{ 0,1,2\}$ to be the number of coordinates of $\big(\mathrm{pr}_{\R_{\geq0}^2}\circ \varphi \big) (x)$ equal to $0$ for a (and hence every) chart $\varphi$ of a neighbourhood of $x$. 
The \emph{corners} of $M$ are the collection of all points of index~$2$. A \emph{connected face} of $M$ is the closure of a maximal connected subset of points of index~$1$. 
\begin{definition}
A manifold with corners is a \emph{manifold with faces} if each $x\in M$ belongs to exactly $\mathrm{index}(x)$ connected faces.
\end{definition}
In this case we define a \emph{face} of $M$ to be a disjoint union of connected faces,  which is a manifold with  boundary.  
A \emph{boundary defining function} for a face $H_i$ is a function $\rho_i \in C^\infty(M)$ such that
$\rho_i(x)\geq 0$ and $\rho_i(x)=0$ if and only if $x\in H_i$.
\begin{definition}
A \emph{$\langle 2 \rangle$-manifold} is a manifold $M$ with faces together with two faces $\partial_0 M$ and $\partial_1 M$ such that $\partial M = \partial_0 M \cup \partial_1 M$ and $\partial_0 M \cap \partial_1 M$ are the corners of $M$. 
\end{definition}  
Denote by $\Delta_1$ the category corresponding to the ordered set $\{0,1\}$. Concretely $\Delta_1$
has two objects $0$ and $1$ and one non identity morphism $0 \longrightarrow 1$. A $\langle 2 \rangle$-manifold $M$ then defines a diagram $M \colon \Delta_1^2 \longrightarrow \Manc $ of shape $\Delta_1^2$ in the category $\Manc $ of manifolds with corners and smooth embeddings:
\begin{equation}
\begin{tikzcd}
& M & \\
\partial_0 M \ar[ru] & & \partial_1 M \ar[lu] \\
& \partial_0 M\cap \partial_1 M \ar[ru] \ar[lu] & 
\end{tikzcd}
\end{equation}
Moreover, $\langle 2 \rangle$-manifolds always admit a compatible set of collars:
\begin{proposition}
Let $M$ be an $n$-dimensional $\langle 2 \rangle$-manifold and $\epsilon_0,\epsilon_1$ positive real 
numbers. There are collars, $ [0,\epsilon_0) \times (\partial_0 M  \cap \partial_1 M)  \hookrightarrow \partial_0 M $, $ [0,\epsilon_1) \times (\partial_0 M  \cap \partial_1 M)  \hookrightarrow \partial_1 M $,
$ [0,\epsilon_1) \times \partial_0 M \hookrightarrow M $ and $ [0,\epsilon_0) \times \partial_1 M \hookrightarrow M $ such that the following diagram commutes
\begin{equation}
\begin{tikzcd}[column sep= -4]
& M & \\
 {[0,\epsilon_1) \times \partial_0 M}  \ar[ru] & &  {[0,\epsilon_0) \times \partial_1 M} \ar[lu] \\
&  {[0,\epsilon_0)\times [0,\epsilon_1) \times (\partial_0 M\cap \partial_1 M)}  \ar[ru] \ar[lu] & 
\end{tikzcd} \ \ 
\end{equation}
\end{proposition}
For a proof see for example~\cite[Lemma 2.1.6]{LauresCorners}. 

\subsubsection*{Invariant background fields}
Informally, we want to think of objects in the bicategory $\ECobF$ as 
$(n-2)$-dimensional manifolds equipped with $\F$-background fields. However,
it might not be possible to evaluate the stack $\F$ on manifolds of lower dimensions
directly. For this reason our objects will be $(n-2)$-dimensional manifolds $S$ together 
with an $\F$-background field on 
$(-\epsilon_1,\epsilon_1)\times (-\epsilon_2,\epsilon_2)\times S$. We still want
to keep the intuition that these fields only depend on the manifold $S$ and hence 
require the background fields to be ``constant" in the 
$(-\epsilon_1,\epsilon_1)\times (-\epsilon_2,\epsilon_2)$ direction. To formulate 
what we mean by constant we make the following definition:     
\begin{definition}
\label{Def:Invariant element stack}
Let $\mathscr{F}$ be a stack, $M$ a manifold, and $\mathscr{S}(M)$ a groupoid which consists of a collection of open subsets of $M$ including $M$ and a collection of diffeomorphisms as morphisms. 
\begin{itemize}
\item[(a)] An \emph{invariant structure} with respect to $\mathscr{S}(M)$ for an element $\mathsf{f} \in \mathscr{F}(M)$ is a natural 2-transformation
\begin{equation}
\begin{tikzcd}
 &\; \ar[dd, Leftarrow, shorten >= 20 ,shorten <= 15, " \mathsf{f}\,",swap] & \\
\mathscr{S}(M)^{\mathrm{op}} \ar[rr, bend left, "{\mathscr{F}}"] \ar[rr, swap, bend right, "\mathsf{1}"]& & \mathsf{Grpd} \\
 & \; &
\end{tikzcd}
\end{equation}
from the constant 2-functor sending every object to the groupoid $\mathsf{1}$ with one object and one morphism, such that the natural transformation is induced on objects by $f$ (see the following remark for an explaination). Here we regard $\mathscr{S}(M)$ as a 2-category with trivial 2-morphisms.\footnote{This is the same as a higher fixed point for the groupoid action corresponding to $\mathscr{F}$, as discussed for example in \cite{HesseSchweigertValentino}.} We call an element $\mathsf{f}\in \mathscr{F}(M)$ together with the choice of an invariant structure an \emph{invariant element}.

\item[(b)] A morphism $\mit\Theta\colon \mathsf{f}\longrightarrow \mathsf{f}'$ between invariant elements $\mathsf{f},\mathsf{f}' \in \mathscr{F}(M)$ is \emph{invariant} under $\mathscr{S}(M)$ if it induces a modification (see Definition~\ref{Definition Modification Bicategory}) between the natural 2-transformations corresponding to $\mathsf{f}$ and~$\mathsf{f}'$.
\end{itemize}    
\end{definition}
\begin{remark}
Let us spell out in detail what we mean by saying that $\mathsf{f}\in\mathscr{F}(M)$ induces a natural 2-transformation on objects. A map $\mathsf{f}_U \colon \mathsf{1} \longrightarrow \mathscr{F}(U)$ is an element of $\mathscr{F}(U)$. We set $\mathsf{f}_U = \mathsf{f}|_U$ for all $U \in \mathrm{Obj}\big(\mathscr{S}(M)\big)$. To equip this with the structure of a natural 2-transformation we have to fix natural transformations (see Definition~\ref{Definition transformation Bicategory})
 \begin{equation}
\begin{tikzcd}
\Hom_{\mathscr{S}(M)^{\rm op}}(U_1,U_2) \arrow{r}{\mathsf{1} }\arrow[d,swap,"\mathscr{F}(\, \cdot\, )"] & \Hom_{\mathsf{Grpd}}(\mathsf{1},\mathsf{1})\arrow{d}{\, \mathsf{f}_{U_2\ast}}\\
\Hom_{\mathsf{Grpd}}\big(\mathscr{F}(U_1),\mathscr{F}(U_2)\big)\arrow[ru,Rightarrow, "\mathsf{f}_{U_1 U_2}"] \arrow[r,swap, "\mathsf{f}_{U_1}^\ast"] & \Hom_{\mathsf{Grpd}}\big(\mathsf{1},\mathscr{F}(U_2)\big)
\end{tikzcd}
 \end{equation}  
They can be described by a collection of morphisms
 $\mathsf{f}_{U_1 U_2}(t)\colon t^\ast \, \mathsf{f}_{U_1} \longrightarrow \mathsf{f}_{U_2}$ 
for every morphism $t\colon U_2 \longrightarrow U_1$ of $\mathscr{S}(M)$ which have to 
satisfy the coherence conditions \eqref{Equation1: Definition Transformation} and \eqref{Equation2: Definition Transformation}:
\begin{align}
\label{EQ1:Structure Invariant}
\mathsf{f}_{U_2 U_3}(t_2) \circ t_2^\ast\, \mathsf{f}_{U_1 U_2}(t_1) &= \mathsf{f}_{U_1U_3}(t_1\circ t_2) \circ \Phi_{\mathscr{F}(U_3)\mathscr{F}(U_2)\mathscr{F}(U_1)}(t_1\times t_2) \ , \\[4pt]
\label{EQ2:Structure Invariant}
\mathsf{f}_{U U}(\id_U)\circ \Phi_{\mathscr{F}(U)}(\id_\star) &= \id\,_{\mathsf{f}_U} 
\end{align}
for morphisms $t_2 \colon U_3 \longrightarrow U_2$ and $t_1\colon U_2
\longrightarrow U_1$. For a sheaf considered as a stack, the maps $\mathsf{f}_{U_1U_2}(t)$ must be identity maps and we reproduce, for example, the definition of an invariant function. 
\end{remark}
\begin{example}
Let $\mathscr{F}=\mathsf{Bun}_G$ be the stack of principal $G$-bundles, and let $M$ be a manifold equipped with an action $\rho\colon \Gamma \longrightarrow \Diff(M)$ of a group $\Gamma$ by diffeomorphisms of $M$. We can encode the action into a groupoid $\mathscr{S}(M)$ as in Definition \ref{Def:Invariant element stack} with one object $M$ and morphisms $\{\rho(\gamma)\mid \gamma \in \Gamma \}$. A $G$-bundle $P$ which is invariant under $\mathscr{S}(M)$ comes with gauge transformations $\mit\Theta_\gamma \colon \rho(\gamma)^\ast P \longrightarrow P$ satisfying \eqref{EQ1:Structure Invariant} and \eqref{EQ2:Structure Invariant}. This is just a $\Gamma$-equivariant $G$-bundle.
An invariant morphism between two $\Gamma$-equivariant $G$-bundles is then a $\Gamma$-equivariant gauge transformation.    
\end{example}
\begin{definition}\label{Def:Constant}
Let $k<n$ be a positive number and $\Sigma $ an $n-k$-dimensional manifold (with boundary). For every collection of (not necessarily open) intervals $I_1,\dots , I_k \subset \R$, we say that an element $\mathsf{f} \in \mathscr{F}(I_1 \times \dots \times I_k\times \Sigma)$ is \emph{constant} along $I_1 \times \dots \times I_k$ if it is invariant under translations in the direction along $I_1 \times \dots \times I_k$, i.e.\ invariant with respect to the groupoid with open subsets of $I_1 \times \dots \times I_k\times \Sigma$ as objects and translations along $I$ as morphisms.
\end{definition}
\begin{example}
For example metrics of the form $dt_1^2+\dots dt_k^2 + g_\Sigma $ on $I_1 \times \dots \times I_k \times \Sigma$ are constant
along $I_1 \times \dots \times I_k$. Furthermore, every pullback of a structure defined on $\Sigma$ to 
$I_1 \times \dots \times I_k \times \Sigma$ along the projection will be constant in a canonical way.   
\end{example}
Now we have all the mathematical prerequisites at our disposal to define 
\subsubsection*{The bicategory $\ECobF$}

Inspired by \cite{schommer2011classification} and the sketch of~\cite[Appendix~A]{MonnierHamiltionianAnomalies}, we introduce a bicategory of manifolds equipped with geometric fields. For the definition of a Dirac operator, a metric on the underlying manifold is crucial, whence we cannot assume that the field content is topological. This leads to technical problems in defining 2-morphisms. We make the assumption that the field content is constant near gluing boundaries and use a specific choice of collars to get around these problems.   

We define the bicategory $\ECobF$ with objects given by quadruples 
\begin{align}
\big(S,\mathsf{f}^{n-2}, \epsilon_1, \epsilon_2\big)
\end{align}
consisting of a closed $(n-2)$-dimensional manifold $S$ with $l$ connected components $S_1,\dots ,S_l$, $l$-tuples $\epsilon_{1},\epsilon_{2}\in \R_{>0}^l$ and an element $\mathsf{f}^{n-2}\in \mathscr{F}\big((-\epsilon_1,\epsilon_1)\times (-\epsilon_2,\epsilon_2) \times S \big)$ which is constant along $(-\epsilon_1,\epsilon_1)\times (-\epsilon_2,\epsilon_2)$. Here we introduced the notation
\begin{align}
(-\epsilon_1,\epsilon_1)\times (-\epsilon_2,\epsilon_2) \times S = \bigsqcup_{i=1}^l \,  (-\epsilon_{1,i},\epsilon_{1,i})\times(-\epsilon_{2,i},\epsilon_{2,i})\times S_i \ ,
\end{align} 
which we will continue to use throughout this section.

There are two different kinds of 1-morphisms in $\ECobF$:
\begin{itemize}
\item[(1a)]
Regular 1-morphisms 
\begin{align}
\Sigma \colon \big(S_-,\mathsf{f}^{n-2}_-, \epsilon_{-1},\epsilon_{-2}\big) \longrightarrow \big(S_+,\mathsf{f}^{n-2}_+,\epsilon_{+1},\epsilon_{+2}\big)
\end{align}
consist of 7-tuples
\begin{align}
\big(\Sigma,\varphi^{n-1}_-,\varphi^{n-1}_+,\mathsf{f}^{n-1},{\mit\Theta}^{n-1}_-,{\mit\Theta}^{n-1}_+, \epsilon\big) \ ,
\end{align}
where $\Sigma$ is a $(n-1)$-dimensional manifold with boundary and $k$ connected components together with a decomposition of a collar of its boundary into $N^{n-1}_-$ and $N^{n-1}_+$, 
$\varphi^{n-1}_- \colon [0,\epsilon_{-1}) \times S_- \longrightarrow N^{n-1}_-$  and $\varphi^{n-1}_+ \colon  (-\epsilon_{+2},0]\times S_+ \longrightarrow N^{n-1}_+$ are diffeomorphisms, $\epsilon \in \R_{>0}^k$, $\mathsf{f}^{n-1}\in \mathscr{F}\big( (-\epsilon,\epsilon) \times \Sigma \big)$ is constant along $(-\epsilon,\epsilon)$ and\footnote{For this statement to make sense we require that $\epsilon$ is compatible with $\epsilon_{\pm 2}$ on the boundary.} ${\mit\Theta}^{n-1}_\pm \colon \mathsf{f}^{n-2}_\pm\longrightarrow \varphi_\pm^\ast \mathsf{f}^{n-1}$ are constant morphisms. Here we use ${\mit\Theta}^{n-1}_\pm$ to implicitly define the structure of a constant object on $N^{n-1}_\pm$. 

\item[(1b)]
Limit 1-morphisms consist of diffeomorphisms $\phi \colon S \longrightarrow S'$ together with the structure of a constant $\F$-diffeomorphism on $(-\epsilon_1, \epsilon_1)\times (-\epsilon_2, \epsilon_2)\times \phi$.\footnote{For this to make sense we require that all $l$-tuples $\epsilon$ are equal.}
\end{itemize}
For the composition of regular 1-morphisms $\Sigma_1\colon S\longrightarrow S'$ and 
$\Sigma_2\colon S'\longrightarrow S''$ we glue the underlying manifolds using their collars. The manifold $\Sigma_2 \circ \Sigma_1$ comes with a natural open cover (compare also Equation~\eqref{Eq: Cover gluing})
$U_1 = \Sigma_1\setminus S'$, 
$U_{12} =  (-\epsilon_{1+},\epsilon_{2-}) \times S'$ and
$U_2 = \Sigma_2 \setminus S' $.
We equip the resulting manifold with a $\F$-background field using covers $(-\epsilon,\epsilon) \times U_1 $, $(-\epsilon,\epsilon) \times U_2 $ and $(-\epsilon,\epsilon) \times U_{12} $ 
and the descent property of the stack $\mathscr{F}$. 
Note that the isomorphisms required for the descent are part of the regular 1-morphisms. 
Composition of limit 1-morphisms is given by composition of $\F$-diffeomorphisms. The composition of a limit 1-morphism with a regular 1-morphism is given by changing the identification of the collars and ${\mit\Theta}_\pm^{n-1} $ using the limit 1-morphism.

There are also two different kinds of 2-morphisms in $\ECobF$:
\begin{itemize}
\item[(2a)]
Regular 2-morphisms (see Figure~\ref{Fig: 2-Morphism}) 
\begin{eqnarray*}
&& M \colon \big(\Sigma_1,\varphi^{n-1}_{1-},\varphi^{n-1}_{1+},\mathsf{f}^{n-1}_1,{\mit\Theta}^{n-1}_{1-},{\mit\Theta}^{n-1}_{1+}, \epsilon_1\big) \\ && \qquad \qquad \qquad \qquad \qquad \qquad  \Longrightarrow \big(\Sigma_2,\varphi^{n-1}_{2-},\varphi^{n-1}_{2+},\mathsf{f}^{n-1}_2,{\mit\Theta}^{n-1}_{2-},{\mit\Theta}^{n-1}_{2+}, \epsilon_2\big) \ ,
\end{eqnarray*}
with $\Sigma_{i} \colon\big(S_-,\mathsf{f}^{n-2}_-, \epsilon_{-1},\epsilon_{-2}\big) \longrightarrow \big(S_+,\mathsf{f}^{n-2}_+, \epsilon_{+1},\epsilon_{+2}\big)$ regular 1-morphisms for $i=1,2$, consist of equivalence classes of 6-tuples 
\begin{align}
\big(M,\mathsf{f}^n,\varphi^n_-,\varphi^n_+,{\mit\Theta}^n_-,{\mit\Theta}^n_+\big) \ ,
\end{align}
where $M$ is an $n$-dimensional $\langle 2 \rangle $-manifold with corners equipped with
collars $N^{n}_-$ and $N^{n}_+$ of part of the $0$-boundary such that the closure of $N^{d}_\pm$ contains the $1$-boundary, $\mathsf{f}^{n}$ is an element of $\mathscr{F}(M)$, $\varphi^n_- \colon   [0,\epsilon_{-1}) \times \Sigma_1 \longrightarrow N^n_-$ and $\varphi^n_+ \colon  (-\epsilon_{+2},0] \times \Sigma_2  \longrightarrow N^d_+$ are diffeomorphisms, and ${\mit\Theta}^n_- \colon \mathsf{f}^{n-1}_1 \longrightarrow \varphi^{n\, \ast}_- \, \mathsf{f}^n$ and ${\mit\Theta}^n_+ \colon \mathsf{f}^{n-1}_2 \longrightarrow \varphi^{n\, \ast}_+ \, \mathsf{f}^n$ are constant morphisms. All of these structures have to be compatible, in the sense that the diagram 
\begin{equation}
\begin{footnotesize}
\label{EQ: Compatibility condition 2morphisms}
\begin{tikzcd}
& &  {(-\epsilon_2,0) \times \Sigma_2}  \ar[dd,"\varphi_+^d"]& & \\ 
& &  & & \\
 {(-\epsilon_{-2},0)\times [0,\epsilon_{-1}) \times  S_-} \ar[rr, "{\imath_-}"]\ar[rruu, "{\id \times \varphi_{2-}^{n-1} }"] \ar[rrdd,"{(\cdot +\epsilon_{-2}) \times \varphi_{1-}^{n-1} }",swap] & & M & & {(-\epsilon_{+2},0) \times [0,\epsilon_{+1}) \times  S_+}  \ar[ll,swap, "{\imath_+}"] \ar[lldd,"{(\cdot + \epsilon_{+2}) \times \varphi_{1+}^{n-1} }"] \ar[lluu,"{\id \times \varphi_{2+}^{n-1}}",swap] \\ 
& & & & \\
& & {(0, \epsilon_1) \times \Sigma_1 } \ar[uu,swap,"\varphi_-^n"] & &
\end{tikzcd}
\end{footnotesize}
\end{equation}
commutes, where $\imath_\pm$ are inclusions. We change the sign of the coordinates corresponding to both intervals in the lower embedding. This induces a diagram of functors in groupoids and we require that all morphisms $\mit\Theta$ are compatible with this diagram.  Note that the collars of the 1-morphisms induce collars for the 1-boundaries which agree by \eqref{EQ: Compatibility condition 2morphisms}. Two such 6-tuples are equivalent if they are $\mathscr{F}$-diffeomorphic relative to half of the collars.
\begin{figure}
\footnotesize
\begin{center}
\begin{overpic}[width=12cm,
scale=1]{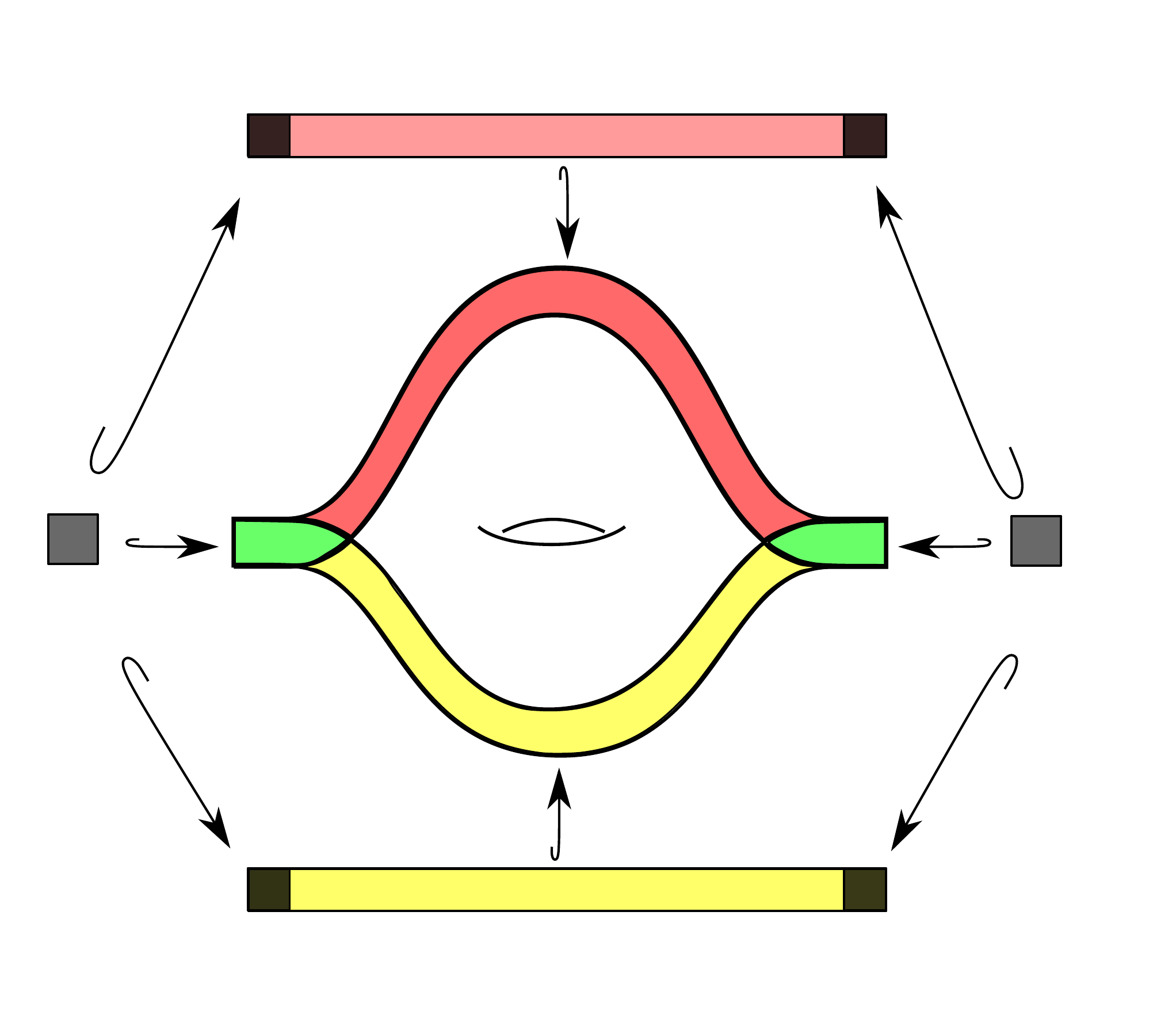}
\put(40,5){$ [0,\epsilon_1) \times \Sigma_1 $}
\put(40,80){$ (-\epsilon_2,0] \times \Sigma_2 $}
\put(-10,35){$(-\epsilon_{-2},0)\!\times  \! [0,\epsilon_{-1})\!\times \! S_- $}
\put(75,35){$(-\epsilon_{+2},0) \!\times\![0,\epsilon_{+1})\!\times \! S_+ $}
\put(46,50){$M$}
\put(49,16){$\varphi^n_-$}
\put(50,69){$\varphi^n_+$}
\end{overpic}
\end{center}
\caption{\small Illustration of a regular 2-morphism in $\ECobF$.}
\normalsize
\label{Fig: 2-Morphism}
\end{figure}

\item[(2b)]
Limit 2-morphisms consist of pairs $(\phi, {\mit\Theta})$, where $\phi: \Sigma_1 \longrightarrow \Sigma_2$ is a diffeomorphism relative to the collars and  ${\mit\Theta}$ is the structure of a constant $\F$-diffeomorphism on $(-\epsilon,\epsilon)\times \phi$. There are no non-trivial 2-morphisms between limit 1-morphisms. 
\end{itemize}
We define horizontal and vertical composition of 2-morphisms as follows:
\begin{itemize}
\item[(Ha)]
Horizontal composition of regular 2-morphisms is given by gluing along 1-boundaries.

\item[(Hb)]
Horizontal composition of limit 2-morphisms is defined by ``gluing together'' diffeomorphisms and the descent condition for morphisms in the stack $\mathscr{F}$. This uses the open cover defined in \eqref{Eq: Cover gluing}.

\item[(Hc)]
Horizontal composition of a limit 2-morphism with a regular 2-morphism is defined by the attachment of a mapping cylinder to the 1-boundary. 

\item[(Va)]
Vertical composition of limit 2-morphisms is given by composition of diffeomorphisms, pullback and composition of morphisms in the stack $\mathscr{F}$. 

\item[(Vb)]
Vertical composition of regular 2-morphisms is a slightly more complicated. Simple gluing of $M$ and $M'$ along a common 1-morphism does not return a 2-morphism, since the resulting 1-boundaries are ``too long". In the context of topological field theories a solution to this problem \cite{schommer2011classification} consists of picking once and for all a diffeomorphism $[0,2] \longrightarrow [0,1]$. We are unable to use this trick here, since the stacks we consider in this thesis may contain metrics. Instead, we will use collars to circumvent this problem. 
Given two regular 2-morphisms $M_1 \colon \Sigma_1 \Longrightarrow \Sigma_2$ and $M_2 \colon \Sigma_2 \Longrightarrow \Sigma_3$, we define 
\begin{align}
&\tilde{M}_1 = M_1 \setminus \varphi^n_{1+}\big((-\tfrac{\epsilon}{2},0] \times \Sigma_2 \big) \ ,\\[4pt]
&\tilde{M}_2 = M_2 \setminus \varphi^n_{2-}\big( [0,\tfrac{\epsilon}{2}) \times \Sigma_2  \big) \ .
\end{align}
We define the vertical composition $M_{2}\circ M_{1}$ to be the manifold resulting from gluing $\tilde{M}_1$ and $\tilde{M}_2$ along $\Sigma_2$. We have to equip this manifold with appropriate collars: 
Write $\tilde{N}^n_1= N^n_1 \cap \varphi^n_{1+}\big([-\tfrac{\epsilon}{2},0] \times \Sigma_2 \big)$, where $N^n_1$ is the incoming collar of $M_{1}$. We set  
\begin{align}
C = \big( \varphi^n_{2-} \circ (  \, \cdot \, +\epsilon \times \id) \circ (\varphi^n_{1+})^{-1}\big)\big(\tilde{N}^n_1\big) \ .
\end{align} 
We can glue $C$ to the remainder of the collar of $M_1$ to get
a new collar; this is only possible because we assumed that the
corresponding elements of $\mathscr{F}$ are constant along the
collars. It is possible to define a new collar for $\Sigma_3$ in
the same way.

\item[(Vc)]
Vertical composition of a limit 2-morphism with a regular 2-morphism
is defined by changing the identification of collars as in Section~\ref{Sec: Non-extended field theories}.
\end{itemize}
This completes the definition of the geometric cobordism bicategory $\ECobF$.
We will show that the disjoint union of manifolds turns $\ECobF$ into a symmetric monoidal bicategory.

We prove that $\ECobF$ is a symmetric monoidal bicategory using a method developed by Schulman~\cite{Schulman, FunctorialConstructionBicat} 
following a similar proof in the topological case by Schommer-Pries~\cite[Section 3.1.4]{schommer2011classification}. 
For this we recall 
\begin{definition}[\cite{Schulman}]
A \emph{symmetric monoidal pseudo-double category} $\cD$ consists of a symmetric monoidal category of ``objects" $\cD_0$,
a symmetric monoidal category of ``arrows" $\cD_1$, symmetric monoidal functors 
\begin{align}
\begin{split}
\id \colon \cD_0 \longrightarrow \cD_1,  \\
S,T \colon \cD_1 \longrightarrow \cD_0, \\ 
\odot \colon \cD_1 \times_{\cD_0}\cD_1 \longrightarrow \cD_1,  
\end{split}
\end{align}
and symmetric monoidal natural transformations 
\begin{align}\label{Eq: Natural transformations for double bicat}
\begin{split}
\alpha \colon (-\odot-)\odot - &\Longrightarrow -\odot(-\odot-), \\ 
\lambda \colon (\id(-) \odot - ) &\Longrightarrow \id_{\cD_1}, \\ 
\rho \colon ( - \odot \id(-)  ) &\Longrightarrow \id_{\cD_1},
\end{split}
\end{align}
such that $S$ and $T$ are strict, $S(\id(A))=T(\id(A))=A$ for all $A\in \cD_0$, 
$S(M\odot N)=S(N)$ and $T(M\odot N)=T(M)$ 
for all $N,M\in \cD_1$, applying $S$ and $T$ to the 
natural transformations in Equation~\eqref{Eq: Natural transformations for double bicat} gives identities, 
$\alpha, \lambda , \rho$ satisfy the usual coherence conditions for monoidal categories (Pentagon and Triangle relation) 
and $\id$ strictly preserves the unit object of $\cD_0$.  
\end{definition}
We interpret the data of a symmetric monoidal pseudo-double category as follows:
\begin{itemize}
\item 
Its objects are the elements of $\cD_0$. 
\item The morphisms $f\colon A \longrightarrow B $ in $\cD_0$ are its vertical 1-morphisms. 
\item An object $g$ of $\cD_1$ is a horizontal 1-morphisms $g\colon S(g)\longrightarrow T(g)$. 
\item A morphisms $F\colon g \longrightarrow g' \in \cD_1$ is a 2-morphisms filling the square 
\begin{equation}
\begin{tikzcd}[row sep=0.3cm, column sep=0.3cm]
S(g) \ar[rr,"g"] \ar[dd,"S(F)", swap] & & T(g) \ar[dd,"T(F)"] \\
 & F & \\  
S(g') \ar[rr,"g'",swap] & & T(g')
\end{tikzcd}
\end{equation}    
\end{itemize}
2-morphism admit a strict vertical composition and a coherent horizontal composition. Let $\cD$ be a symmetric monoidal pseudo-double category. Restricting to 
2-morphisms $F$ such that $S$ and $T$ applied to it are identities defines a bicategory $H(\cD)$~\cite{Schulman}. 
The advantage of using symmetric monoidal pseudo-double categories is that there is simple criterion for the bicategory 
$H(\cD)$ to be symmetric monoidal. Before we can state this result we need to recall one more definition from~\cite{Schulman}:
\begin{definition}
Let $\cD$ be a symmetric monoidal pseudo-double category. Furthermore, let $f\colon A \longrightarrow B$ be a vertical 1-morphism. 
A \emph{companion} of $f$ is a horizontal 1-morphism $\widehat{f}\colon A \longrightarrow B $ together with 2-morphisms
\begin{equation}
\begin{tikzcd}[row sep=0.3cm, column sep=0.3cm]
A \ar[rr,"\widehat{f}"] \ar[dd,"f", swap] & & B \ar[dd,"\id"] \\
& F_B & \\  
B \ar[rr,"\id(B)",swap] & & B
\end{tikzcd} \hspace{1cm}
\text{   and   } \hspace{1cm} 
\begin{tikzcd}[row sep=0.3cm, column sep=0.3cm]
A \ar[rr,"\id(A)"] \ar[dd,"\id", swap] & & A \ar[dd,"f"] \\
& F_A & \\  
A \ar[rr,"\widehat{f}",swap] & & B
\end{tikzcd}
\end{equation}
such that 
\begin{equation}
\begin{tikzcd}[row sep=0.3cm, column sep=0.3cm]
A \ar[rr,"\id(A)"] \ar[dd,"\id", swap] & & A \ar[dd,"f"] \\
& F_A & \\  
A \ar[rr,"\widehat{f}"] \ar[dd,"f", swap] & & B \ar[dd,"\id"] \\
& F_B & \\  
B \ar[rr,"\id(B)",swap] & & B
\end{tikzcd} 
= 
\begin{tikzcd}[row sep=0.3cm, column sep=0.3cm]
A \ar[rr,"\id(A)"] \ar[dd,"f", swap] & & A \ar[dd,"f"] \\
& \id(f) & \\  
B \ar[rr,"\id(B)",swap] & & B
\end{tikzcd}
\end{equation}
and
\begin{equation}
\begin{tikzcd}[row sep=0.3cm, column sep=0.3cm]
A \ar[rr, "\id(A)"] \ar[dd,"\id",swap] & & A \ar[rr,"\widehat{f}"] \ar[dd,"f", swap] & & B \ar[dd,"\id"]  \\
 & F_A  & & F_B & \\  
A \ar[rr, "\widehat{f}",swap] & & B \ar[rr,"\id(B)",swap] & & B
\end{tikzcd} 
= 
\begin{tikzcd}[row sep=0.3cm, column sep=0.3cm]
A \ar[rr,"\widehat{f}"] \ar[dd,"\id", swap] & & B \ar[dd,"\id"] \\
& \id_{\widehat{f}} & \\  
A \ar[rr,"\widehat{f}",swap] & & B
\end{tikzcd} 
\end{equation}
A \emph{conjoint} of $f$ is a companion of $f$ in the symmetric monoidal pseudo-double category constructed from 
$\cD$ by reversing the direction of the horizontal 1-morphisms but not of the vertical 1-morphisms. 
\end{definition}  
\begin{theorem}[Theorem 5.1 of \cite{Schulman}]
Let $\cD$ be a symmetric monoidal pseudo-double category such that every vertical 1-morphism admits a companion and 
a conjoint. Then $H(\cD)$ is a symmetric monoidal bicategory. 
\end{theorem}
\begin{corollary}
$\ECobF$ is a symmetric monoidal bicategory.
\end{corollary}
\begin{proof}
Note that $\ECobF$ can be constructed from the following symmetric monoidal pseudo-double category:
\begin{itemize}
\item The category $\cD_0$ consists of the objects of $\ECobF$ and limit 1-morphisms between them. 
\item The horizontal 1-morphisms (objects of $\cD_1$) are all 1-morphisms in $\ECobF$. 
\item 2-morphism are an extension of the 2-morphisms in $\ECobF$ allowing the morphisms to be non constant at the corners. 
More concretely, there are the following two types of 2-morphisms: 
\begin{itemize}
\item[(2a$'$)] Regular 2-morphisms are as above with the Condition~\ref{EQ: Compatibility condition 2morphisms} replaced with the weaker 
condition that the failure 
of the Diagram~\ref{EQ: Compatibility condition 2morphisms} to commute is given by constant $\F$-diffeomorphisms.   
\item[(2b$'$)]Limit 2-morphisms are $\F$-diffeomorphism which do not need to be relative to the boundaries. However, they still 
need to be constant on the collars.  
\end{itemize}  
\end{itemize}
The monoidal structure is induced by the disjoint union of manifolds. We need to show that every vertical morphism $f$ in $\cD$ admits companions and conjoints. These are just $f$ and $f^{-1}$ considered as horizontal 1-morphisms.  
\end{proof}
\begin{remark}\label{Rem: Shortcomings bicategory}
The bicategory $\ECobF$ has the same shortcomings as its categorical analogue $\CobF$. Concretely, the explicit dependence 
on the 
collar size $\epsilon$ and the fact that we ignore the smooth structure of the background fields are unsatisfactory.
The smoothness problem could be solved by working with bicategories fibred over the category of smooth manifolds. However,
we are not aware of an approach which solves the $\epsilon$-dependence in general. The problem is that as soon as metrics are 
involved the 1-boundary has a length and it seems hard to glue 1-boundaries of different length together. 
In the context of 2-dimensional extended conformal field theories a similar problem can be solved by allowing specific kinds of singular 
manifolds~\cite{AH}. 
A possible general solution is to work with double bicategories instead of bicategories~\cite{double}. All the results in this thesis should carry over to other definitions of the cobordism bicategory.       
\end{remark}

\subsection{The definition}
There is still one ingredient missing for the definition of extended functorial 
field theories: the target bicategory. This should be an appropriate categorification
of the category of Hilbert spaces. To our knowledge there is no universally accepted 
target. In this thesis we will restrict ourself to a simple target bicategory $\Tvs$
which is 
standard in the context of extended topological field theories, see the appendix of \cite{bartlett2015modular} for
a discussion of possible targets for topological field theories. 
The elements of $\Tvs$ are called Kapranov-Voevodsky 2-vector spaces and can be understood
as a categorification of finite dimensional vector spaces.  
\begin{definition}
\label{Def:2Vect}
A \emph{Kapranov-Voevodsky 2-vector space}~\cite{KV94} is a $\C$-linear semi-simple additive category $\Vscr$
with finitely many isomorphism classes of simple objects; in particular, a 2-vector space is also an abelian category.  
There is a 2-category $\Tvs$ of 2-vector spaces, $\C$-linear functors and natural transformations. 
Given two 2-vector spaces $\Vscr_1$ and $\Vscr_2$ we can define their tensor product $\Vscr_1 \boxtimes \Vscr_2$~\cite[Definition~1.15]{bakalov2001lectures} to 
be the category with objects given by finite formal sums 
\[
\bigoplus_{i=1}^n\, V_{1i}\boxtimes V_{2i} \ , 
\]
with $V_{1i}\in \mathrm{Obj}(\Vscr_1)$ and $V_{2i}\in \mathrm{Obj}(\Vscr_2)$. The space of morphisms is given by
\small
\[
\mathrm{Hom}_{\Vscr_1\boxtimes \Vscr_2}\Big(\, \mbox{$\bigoplus\limits_{i=1}^n\, V_{1i}\boxtimes V_{2i}\,,\,\bigoplus\limits_{j=1}^m \, V'_{1j}\boxtimes V'_{2j} $} \, \Big)= \bigoplus_{i=1}^n \ \bigoplus_{j=1}^m\, \mathrm{Hom}_{\Vscr_1}(V_{1i},V'_{1j})\otimes_\C \mathrm{Hom}_{\Vscr_2}(V_{2i},V'_{2j}) \ .
\] 
\normalsize
This tensor product coincides with the Deligne product of abelian categories. It furthermore satisfies the universal property with respect to bilinear functors that one would expect from a tensor product.
We can also take tensor products of $\C$-linear functors and of natural transformations.
 Then the 2-category $\Tvs$ with $\boxtimes$ is a symmetric monoidal bicategory with monoidal unit~$1$ given by the category of finite-dimensional vector spaces $\fvs$.  
\end{definition}
\begin{remark}
It would be more appropriate to work with ``2-Hilbert space" as defined for example in~\cite{2Hilb}. We choose here to work with $\Tvs$ since it reduces some of the technical complexity while still capturing all essential features.  
\end{remark}
Let $\Vscr$ be a 2-vector space. The natural functor 
\begin{align}
\Hom_{\Vscr}(\cdot,\cdot) \colon \Vscr^{\opp} \times \Vscr \longrightarrow \fvs 
\end{align}
is part of a two variable adjunction 
\begin{align}
\Hom_{\fvs}(S,\Hom_{\Vscr}(V_1,V_2))\cong \Hom_{\Vscr}(S*V_{1},V_2) \cong \Hom_{\Vscr}(V_1, V_2^S)
\end{align}
for all $V_1,V_2 \in \Vscr$ and $S\in \fvs$. The linear functor 
\begin{align}
* \colon \fvs \boxtimes \Vscr \longrightarrow \Vscr 
\end{align}
equips $\Vscr$ with the structure of a $\fvs$-module category. This explains in which sense KV 2-vector spaces are
a categorification of vector spaces: The ground field $\C$ is replaced with the category $\fvs$ and 
vector spaces ($\C$-modules) are replaced by a particularly nice class of $\fvs$-module categories.

The following remark provides a concrete description for the objects, 1-morphism
and 2-morphisms in $\Tvs$.
\begin{remark}\label{Rem: 2VS}
Being semi-simple with a finite number of simple objects means that
every 2-vector space $\Vscr$ is equivalent as a linear category to $\fvs^k$ for 
a natural number $k$.
Linear functors preserve direct sums. This allows us to describe a linear functor
$\fvs^k\longrightarrow \fvs^l$ up to equivalence by a matrix 
$(V_{ij})_{i=1,\dots k, j=1,\dots,l}$ 
of vector spaces. Composition of linear functors is given by matrix multiplication
replacing the multiplication of numbers with the tensor product $\otimes$ of vector space.   

A natural transformation $(V_{ij})\Longrightarrow (V'_{ij})$ can be described by 
a matrix of linear maps $f_{ij}\colon V_{ij}\longrightarrow V'_{ij}$. Vertical composition
is the component-wise composition of linear maps and horizontal composition is given by
matrix multiplication combined with the tensor product of linear maps. 
\end{remark}
All examples of extended functorial field theories considered in this thesis can be formulated using the 
2-category $\Tvs$.
However, it is too restrictive for a general quantum field theory. 
The 2-category $\Tvs$ should embed into any reasonable target bicategory and hence it
seems reasonable to work with the 2-category $\Tvs$ whenever possible. 
Even though 
$\Tvs$ is ``finite dimensional" in nature, it is still able to capture crucial properties
about the infinite dimensional state space of quantum field theories with anomalies,
see \cite{MonnierHamiltionianAnomalies} for more details.  

Now we can state the central definition for this thesis.
\begin{definition}\label{Def: EQFT}
An \emph{$n$-dimensional extended functorial quantum field theory} with background fields $\mathscr{F}$ (or \emph{extended quantum field theory} for short) is a symmetric monoidal 2-functor 
$$
\mathcal{Z} \colon \ECobF\longrightarrow \Tvs
$$
from a geometric cobordism bicategory to the 2-category of 2-vector spaces.
\end{definition}  
The simplest example is again the trivial theory $1\colon \ECobF\longrightarrow \Tvs$
which assigns to every object $S\in \ECobF$ the category of finite dimensional 
vector spaces, to every 1-morphism in $\ECobF$ the identity functor $\fvs\longrightarrow \fvs$ and to every 2-morphism the identity natural transformation.

The following remark explains how Definition~\ref{Def: EQFT} extends Definition~\ref{Def: QFT}.
\begin{remark}\label{Rem: EQFT induces QFT}
The endomorphism category of the monoidal unit $\emptyset \in \ECobF$ has as objects pairs 
$(\Sigma,\mathsf{f}^{n-1},\epsilon)$ where $\Sigma$ is a closed $(n-1)$-dimensional 
manifold and $\mathsf{f}^{n-1}$ is an element of $\F((-\epsilon,\epsilon)\times \Sigma)$. 
A morphism is a cobordism equipped with a compatible $\F$-background field. 
This category does not agree with $\CobF$ directly since the composition in 
$\End_{\ECobF}(\emptyset)$ involves deleting half of the collar. However, there is a functor 
$\CobF \longrightarrow \End_{\ECobF}(\emptyset)$ which sends an object $(\Sigma,\mathsf{f}^{n-1},\epsilon)$ to the object $(\Sigma,\mathsf{f}^{n-1},2\epsilon)$ 
and sends a regular morphism to the cobordism with a cylinder of length $\epsilon$ added to the ingoing and outgoing boundary. 

By restriction and pullback along this functor every extended functorial field theory $\mathcal{Z}\colon \ECobF
\longrightarrow \Tvs$ induces a functor $\CobF \longrightarrow \End_\Tvs(\fvs)\cong \fvs$,
i.e.\ an ordinary quantum field theory with values in the category of finite dimensional 
vector spaces.   
\end{remark}   
We define the \emph{tensor product} $\mathcal{Z}_1\boxtimes \mathcal{Z}_2$ of two extended 
quantum field theories $\mathcal{Z}_1,\mathcal{Z}_2\colon \ECobF \longrightarrow \Tvs$ pointwise, i.e.\
by $(\mathcal{Z}_1\boxtimes \mathcal{Z}_2)[S]=\mathcal{Z}_1(S) \boxtimes \mathcal{Z}_2(S)$ for 
all objects $S \in \ECobF$.
We generalize Definition~\ref{Def: Invertible} to the extended case:
\begin{definition}\label{Def: E Invertible}
Let $\mathcal{Z}\colon \ECobF \longrightarrow \Tvs$ be an extended functorial quantum field theory. The theory $\Za$ is \emph{invertible}
if there exists a functorial quantum field theory $\mathcal{Z}^{-1}\colon \ECobF \longrightarrow \Tvs$
such that $\mathcal{Z}\boxtimes \mathcal{Z}^{-1}\cong 1$. 
\end{definition}
For a symmetric monoidal bicategory $(\Ba,\otimes_\Ba)$ we call the maximal 2-subgroupoid of $\Ba$
containing only elements, 1-morphisms and 2-morphisms which are (weakly) invertible with respect 
to the tensor product the \emph{maximal Picard 2-subgroupoid of $\Ba$}. Every invertible extended quantum field theory factors through the maximal Picard 2-subgroupoid of $\Tvs$, which is equivalent to 
a bicategory $B\C^\times$ with one object, one 1-morphism and elements of $\C^\times$ as 2-morphisms. 
  
A \emph{2-representation} of a groupoid $\mathcal{G}$ is a 2-functor $\mathcal{G}\longrightarrow \Tvs$.
Let $\ESymF$ be the groupoid consisting of all objects of $\ECobF$ and limit 1-morphisms. Note that this
is actually a groupoid, since limit morphisms compose strictly. Every extended functorial quantum 
field theory induces a 2-representation of $\ESymF$.    
We will provide examples of extended field theories in Chapter~\ref{Cha: Index} and \ref{Chapter: t Hooft}, but for the moment focus on 
the abstract description of anomalies in this framework. 
 
\section{Description of anomalies}\label{Sec: GT Anomalies}
We will now give a general description of anomalies in the framework of functorial quantum field theory. The point of view we take in this thesis is that anomalies in $n-1$ dimensions can be described by invertible extended field theories in $n$ dimensions \cite{FreedAnomalies, BoundaryConditionsTQFT, MonnierHamiltionianAnomalies}. This is naturally formulated in the language of symmetric monoidal bicategories (or $(\infty,n)$-categories, see \cite{BoundaryConditionsTQFT}). 

\subsection{Relative field theories}\label{Sec: RFT}
Let $\Ba$ be a bicategory. We denote by $\tr \Ba$, the \emph{truncation of $\Ba$}, the bicategory 
with the same objects 
and 1-morphisms, but only invertible 2-morphisms. 
Let $\Fa\colon \Ba \longrightarrow \Ba'$ be a 2-functor. 
We denote by $\tr \Fa \colon \tr \Ba \longrightarrow \Ba'$ its restriction to $\tr \Ba$. Note that
$\tr \Fa$ factors through $\tr \Ba'$. However, it will be important that we consider $\tr \Fa$ as a functor to
$\Ba'$.  
The truncation operation is essential for the following definition~\cite{RelativeQFT}:
\begin{definition}\label{Def: relative field theory}
Let $\mathcal{A}_1,\mathcal{A}_2\colon \ECobF \longrightarrow \Tvs$ be two extended quantum field 
theories. A \emph{functorial quantum field theory relative to $\mathcal{A}_1$ and $\mathcal{A}_2$}
is a symmetric monoidal natural transformation 
\begin{align}
\mathcal{Z}\colon \tr \Aa_1 \Longrightarrow \tr \Aa_2 \ \  . 
\end{align}
\end{definition}
\begin{remark}
This definition depends on what we mean by a `symmetric monoidal natural transformation'. We use 
Definition~\ref{Def:2sym transformation}. In particular, we do not require the natural transformations 
in the definition to be invertible and hence its direction is important. The kind of transformations we consider are sometimes called lax in the
literature. We refer to~\cite{Johnson-Freyd:2017ykw} for a discussion of the different definitions in the more general framework of 
$(\infty,n)$-categories. Our choice is justified by Corollary~\ref{Coro: Relative for 1 induces ordinary} below,
see also \cite[Theorem 1.5]{Johnson-Freyd:2017ykw}. 
\end{remark}
The picture one should have in mind is that the theory $\mathcal{Z}$ lives on an interface 
or codimension 1 defect between the theories $\Aa_1$ and $\Aa_2$. 
A field theory living on an interface between a theory $\Aa$ and the trivial theory 
is the same thing as a theory living on the boundary of $\Aa$. 
Anomalies can be described by a special type of relative field theories.
\begin{definition}
\label{Definition Anommalous field theory}
An \emph{anomalous quantum field theory} with anomaly described
by an invertible extended quantum field theory $\mathcal{A}\colon
\ECobF\longrightarrow \Tvs$ is a natural symmetric monoidal 2-transformation 
\begin{align}
\mathcal{Z} \colon \mbf1 \Longrightarrow \mathsf{tr}\mathcal{A} \ .
\end{align} 
We call $\mathcal{A}$ the \emph{anomaly quantum field theory} describing the anomaly of $\mathcal{Z}$.
\end{definition} 
In practice, anomaly field theories are of topological nature. For the description of some anomalies 
the invertibility condition has to be relaxed, e.g.\ in the context of 2-dimensional 
conformal field theories and 6-dimensional super conformal field theories~\cite{MonnierHamiltionianAnomalies, MonnierMAnomalies}. 
These generalisations will not be discussed in this thesis.     

Before we unpack the definition, we will look at the corresponding categorical definition.
We define the truncation $\tr \Ca$ of a category $\Ca$ to be its maximal subgroupoid. 
\begin{definition}\label{Def: Anomalous partition function}
Let $\Aa\colon \CobF \longrightarrow \fvs$ be an $n$-dimensional invertible 
functorial quantum field theory. 
An \emph{anomalous partition function with anomaly $\Aa$} is a natural symmetric monoidal 
transformation 
\begin{align}
\mathcal{Z}\colon \mbf1 \Longrightarrow \tr \Aa \ \ .
\end{align}   
\end{definition} 
Unpacking this definition, we get for every object $\Sigma\in \CobF$ a linear map 
$\mathcal{Z}(\Sigma)\colon \C
=\mbf1(\Sigma) \longrightarrow \mathcal{A}(\Sigma) $, such that the diagram
\begin{equation}\label{Eq: Diagram APF}
\begin{tikzcd}
\C \ar[d, swap, "{\id}"] \ar{rrr}{\Za(\Sigma)} & & & \mathcal{A}(\Sigma) \ar{d}{\mathcal{A}(\phi)} \\
 \C \ar[rrr, swap, "{\mathcal{Z}(\phi( \Sigma))}"] & & & \mathcal{A}\big(\phi( \Sigma)\big)
\end{tikzcd}
\end{equation} 
commutes for all limit morphisms $\phi$. Since
$\mathcal{A}$ is an invertible field theory, $\mathcal{A}(\Sigma)$ is
a one-dimensional vector space and, as such, isomorphic to $\C$, though not
necessarily in a canonical way. This translates into an ambiguity in
the definition of the partition function as a complex number, which is
the simplest manifestation of an anomaly.

We now unpack Definition~\ref{Definition Anommalous field theory} using the conventions and
notations outlined in Appendix~\ref{Chap: Bicat}.
For every
object $S\in \ECobF$ we get a $\C$-linear functor
\begin{equation}
\mathcal{Z}(S) \colon \fvs=\mbf1(S) \longrightarrow
\mathcal{A}(S) \ \ ,
\end{equation} which can be (non-canonically) identified with
a complex vector space in $\fvs$. 
For $S_-,S_+\in \ECobF$ we get a natural transformation 
\begin{equation}
\begin{tikzcd}
\Hom_{\tr \ECobF} \big(S_- , S_+\big) \ar[dd,swap,"\mathcal{A}"] \ar{rr}{\mbf1}& &\Hom_\Tvs (\fvs,\fvs)\ar{dd}{\mathcal{Z}(S_+)_\ast} \\ 
 & & \\
 \Hom_\Tvs \big(\mathcal{A}( S_-) ,\mathcal{A} (S_+)\big) \ar[uurr, Rightarrow, "\mathcal{Z}",shorten <= 2em, shorten >= 2em] \ar[rr,swap,"\mathcal{Z}(S_-)^\ast"] & & \Hom_\Tvs \big(\fvs, \mathcal{A}(S_+)\big)  
\end{tikzcd}
\end{equation}
which consists of a natural linear transformation
\[ 
\mathcal{Z}\big(\Sigma\big): \mathcal{A}\big(\Sigma\big)\circ \mathcal{Z}\big(S_-\big) \Longrightarrow A\big(S_+\big) 
\] 
for every 1-morphism $\Sigma \colon S_- \longrightarrow
S_+$. 
The definition further includes a modification $\mit\Pi_\mathcal{Z}$
consisting of natural isomorphisms 
\[
{\mit\Pi}_\mathcal{Z}\big(S_-,S_+\big)\colon \chi_\mathcal{A}\circ \mathcal{Z}\big(S_-\big)\boxtimes \mathcal{Z}\big(S_+\big)\Longrightarrow \mathcal{Z}\big(S_- \sqcup S_+\big) \circ \lambda_{\tr \ECobF}
\]
and a natural isomorphism
\[
M_\mathcal{Z}^{-1}\colon A(\emptyset)\Longrightarrow \iota_\mathcal{A} \ .
\]  
All of these structures have to satisfy appropriate compatibility conditions, which we summarize in
\begin{proposition}
\label{Lemma Compatibility conditions}
For every anomalous quantum field theory $\mathcal{Z}$ with anomaly $\mathcal{A}$, there are identities
\begin{align}
\label{EQ1: Lemma Compatibility conditions}
\mathcal{Z} \big(\Sigma_2\big) \circ \mathcal{Z}\big(\Sigma_1\big) &= \mathcal{Z}\big(\Sigma_2 \circ \Sigma_1\big) \circ \Phi_{\mathcal{A}}\big(\Aa(\Sigma_2) \circ \Aa(\Sigma_1)\big) \ , \\[4pt]
\label{EQ2: Lemma Compatibility conditions}
\mathcal{Z}\big(\id_{S}\big)\circ\Phi_\Aa\big(\Aa(\id_{S}) \big)&= \id_{\mathcal{Z}(S)} \ , \\[4pt]
\label{EQ3: Lemma Compatibility conditions}
\Za\big(\Sigma_1'\big) &= \Za\big(\Sigma_2'\big)\circ \big(\mathcal{A}(f)\bullet \id_{\Za(S_-')}\big) \ , 
\end{align}
for any 2-isomorphism $f\colon \Sigma_1' \Longrightarrow \Sigma_2'$, together with the following commutative diagrams wherein we suppress obvious structure 2-morphisms and identity 2-morphisms:
\begin{equation}
\label{EQ3': Lemma Compatibility conditions}
\begin{footnotesize}
\begin{tikzcd}
\mathcal{A}(\Sigma_1 \sqcup \Sigma_2)\bullet \chi_\mathcal{A} \bullet A(S_{1-})\boxtimes A(S_{2-}) \ar[rrr, Rightarrow,"{{\mit\Pi}_\Za(S_{1-}, S_{2-})}"]
 \ar[dd,swap,Rightarrow, "\Za{(\Sigma_1 \sqcup \Sigma_2)}"] 
& &  &\mathcal{A}(\Sigma_1 \sqcup \Sigma_2)\bullet \Za(S_{1-}\sqcup S_{2-})
\ar[dd,Rightarrow, "{\Za(\Sigma_1) \boxtimes \Za(\Sigma_2)}"] \\ 
 & & &
\\
\chi_\mathcal{A} \bullet \Za(S_{1+})\boxtimes \Za(S_{2+})
\ar[rrr,swap,Rightarrow, "{{\mit\Pi}_\Za(S_{1+},S_{2+})}"] 
& & & \Za(S_{1+}\sqcup S_{2+})
\end{tikzcd}
\end{footnotesize}
\end{equation}
\begin{equation}
\label{EQ4: Lemma Compatibility conditions}
\begin{scriptsize}
\begin{tikzcd}
 \mathcal{A}(\alpha_{\tr\ECobF})\!\bullet\! \Za((S_1\sqcup S_2)\sqcup S_3)
 \ar[r, Rightarrow,"{\Za(\alpha_{\tr \ECobF})}"]& 
 \Za((S_1\sqcup S_2)\sqcup S_3)
  \\
   & \\
\mathcal{A}(\alpha_{\tr \ECobF})\!\bullet\! \chi_\mathcal{A} \!\bullet\! ( A( S_1\sqcup S_2)\boxtimes A(S_3) )
\ar[uu,Rightarrow,"{{\mit\Pi}_\Za(S_1\sqcup S_2, S_3 )} "] &  \chi_\mathcal{A} \!\bullet\! (\Za(S_1)\boxtimes \Za(S_2 \sqcup S_3 ) )\ar[uu,swap,Rightarrow,"{{\mit\Pi}_\Za( S_1,S_2 \sqcup S_3) }"] \\
   & \\
\mathcal{A}(\alpha_{\tr\ECobF})\!\bullet\! \chi_\mathcal{A} \!\bullet\!  ((\chi_\mathcal{A}\!\bullet\! \Za(S_1)\boxtimes \Za(S_2))\boxtimes \Za(S_3)) \ar[uu,Rightarrow,"{{\mit\Pi}_\Za(S_1,S_2)} "] \ar[r,swap,Rightarrow, "{{\mit\Omega}_\mathcal{A} }"] & \chi_\mathcal{A} \!\bullet\!  ( \Za(S_1)\boxtimes (\chi_\mathcal{A}\!\bullet\! \Za( S_2)\boxtimes \Za(S_3) )) \ar[uu,swap,Rightarrow, "{{\mit\Pi}_\Za( S_2,S_3) }"]
\end{tikzcd}
\end{scriptsize}
\end{equation}
\begin{equation}
\label{EQ5: Lemma Compatibility conditions}
\begin{tikzcd}
 \mathcal{A}\big(\lambda_{\tr\ECobF}\big)\bullet \Za(\emptyset \sqcup S) \arrow[rr,Rightarrow,"\Za(\lambda_\ECobF)"] & & \Za\big(S\big)\ar[d, Rightarrow,"{{\mit\Gamma}_\mathcal{A}^{-1} }"] \\
\mathcal{A}\big(\lambda_{\tr\ECobF}\big)\bullet \chi_\mathcal{A} \bullet \Za(\emptyset )\boxtimes \Za\big(S\big) \ar[u,Rightarrow, "{{\mit\Pi}_\Za(\emptyset,S) }"]\arrow[rr,swap, Rightarrow,"M^{-1}_\Za\boxtimes \id"] & &  \iota_\mathcal{A} \boxtimes \Za\big(S\big)
\end{tikzcd}
\end{equation}
\begin{equation}
\label{EQ6: Lemma Compatibility conditions}
\begin{tikzcd}
\Za\big(S \sqcup \emptyset\big) \arrow[rr,Rightarrow,"\Za(\rho_\ECobF)"] & & \mathcal{A}\big(\rho_\ECobF\big)\bullet \Za\big(S\big)\ar[d,Rightarrow,"{{\mit\Delta}_\mathcal{A} }"] \\
\chi_\mathcal{A} \bullet \Za\big(S\big) \boxtimes \Za(\emptyset )\ar[u,Rightarrow,"{{\mit\Pi}_\Za(S,\emptyset)}"]\arrow[rr,swap, Rightarrow,"\id \boxtimes M^{-1}_\Za"] & & \Za \big(S\big)\boxtimes \iota_\mathcal{A}  
\end{tikzcd}
\end{equation}
and
\begin{equation}
\label{EQ7: Lemma Compatibility conditions}
\begin{small}
\begin{tikzcd}
\mathcal{A}( \beta_{\tr\ECobF})\bullet\chi_\mathcal{A}\bullet \Za(S_1)\boxtimes \Za(S_2) \ar[d,swap,Rightarrow,"{{\mit\Upsilon}^{-1}_\mathcal{A} }"] \ar[rrr,Rightarrow,"{{\mit\Pi}_\Za(S_1, S_2)}"] & & & \mathcal{A}( \beta_{\tr \ECobF})\bullet \Za(S_1\sqcup S_2)\ar[d,Rightarrow, "{\Za(\beta_{\tr \ECobF})}"]\\
\chi_\mathcal{A}\bullet  \Za(S_2)\boxtimes \Za(S_1) \ar[rrr,swap,Rightarrow, "{{\mit\Pi}_\Za(S_2,S_1)}"] & & & \Za(S_2\sqcup S_1)
\end{tikzcd}
\end{small}
\end{equation}
\end{proposition}
\begin{proof}
Writing out the coherence diagrams \eqref{Equation1: Definition
  Transformation} and \eqref{Equation2: Definition Transformation} for
$\Za$ implies \eqref{EQ1: Lemma Compatibility conditions} and
\eqref{EQ2: Lemma Compatibility conditions}. The identity \eqref{EQ3:
  Lemma Compatibility conditions} is the naturality condition for the
natural symmetric monoidal 2-transformation $\Za$. The diagram \eqref{EQ3': Lemma Compatibility conditions} follows from the diagram \eqref{EQ: Modification} for the modification ${\mit\Pi}_\Za$.
The diagrams \eqref{EQ4: Lemma Compatibility conditions}--\eqref{EQ7: Lemma Compatibility conditions} follow from writing out the coherence conditions \eqref{EQ:1 Definition s.m. transformation}--\eqref{EQ:4 Definition s.m. transformation} for $\Za$.
\end{proof}
\begin{remark}
These conditions should be understood as a projective (or twisted) version of the
definition of a symmetric monoidal functor. For this reason we have
drawn the diagrams \eqref{EQ4: Lemma Compatibility
  conditions}--\eqref{EQ7: Lemma Compatibility conditions} in close
analogy to the diagrams appearing in the definition of a braided monoidal
functor. 
\end{remark}
\begin{remark}\label{Rem: concrete description of anomalous field theories }
This remark provides a more concrete description of the data and 
axioms of an anomalous field theory using Remark~\ref{Rem: 2VS}.
The functor $\Za(S)\colon \fvs \longrightarrow \Aa(S)$ can be
described by an object $Z(S)[\C]\in \Aa(S)$ which by a slight
abuse of notation we denote again by $\Za(S)$. The natural
transformation $\Za(\Sigma)\colon \mathcal{A}(\Sigma) \circ \Za(S_1) \Longrightarrow \Za(S_2)$ 
can be described by a morphism
$\Za(\Sigma)\colon
\Aa(\Sigma)[\Za(S_1)]\longrightarrow \Za(S_2)$ in
$\Aa(S_2)$. Requiring $\Za$ to be a natural 2-transformation explicitly reduces
to the following:
Let $S$, $S_1$, $S_2$ and $S_3$ be
objects of $\ECobF$, and $\Sigma_a\colon S_1
\longrightarrow S_2$ and $\Sigma_b \colon S_2 \longrightarrow S_3$ be 1-morphisms in $\ECobF$. Then the diagrams
\begin{equation}\label{Condition twisted functoriality}
\begin{tikzcd}
\Aa(\Sigma_b)\circ \Aa(\Sigma_a)[\Za(S_1)] \ar[dd,
"{\Aa(\Sigma_b)[\Za(\Sigma_a)]}",swap] \ar[r] &
\Aa(\Sigma_b \circ \Sigma_a)[\Za(S_1)]
\ar[dd, "{\Za(\Sigma_b \circ \Sigma_a)}"] \\ & \\
\Aa(\Sigma_b)[\Za(S_2)] \ar[r, swap,"{\Za(\Sigma_b)}"] & \Za(S_3)
\end{tikzcd} 
\end{equation}
and
\begin{equation}\label{Condition twisted preservation of identities}
\begin{tikzcd}[column sep=tiny]
\Za(S)\ar[rr] \ar[rd, "\id", swap] &  & \Aa(\id_{S})[\Za(S)] \ar[ld, "{\Za(\id_{S})}"] \\
 &\ \ \ \Za(S)
\end{tikzcd}
\end{equation}
commute, where the unlabelled morphisms are part of the structure of the extended field theory $\Aa$. 
The modification $M^{-1}_\Za$ can be described explicitly by specifying natural morphisms 
\[ 
M^{-1}\colon Z(\varnothing) \longrightarrow \iota_E (\C)  
\] 
in $\Aa (\varnothing)$. The modification $\Pi_\mathcal{Z}$ is described by natural morphisms 
\[
\Pi_{\Za}\big(S_1\,,\,S_2\big)\colon \chi_{\Aa}[ \Za(S_1)\boxtimes \Za(S_2)] \longrightarrow \Za(S_1 \sqcup S_2) \ 
\]
in $\Aa(S_1 \sqcup S_2)$. We do not spell out the condition corresponding to the monoidal structure 
explicitly. 
\end{remark}
We have the following important corollary of Proposition~\ref{Lemma Compatibility conditions} justifying
Definition~\ref{Definition Anommalous field theory}.
\begin{corollary}\label{Coro: Relative for 1 induces ordinary}
An anomalous quantum field theory with trivial anomaly $\Aa \colon \mbf1
\Longrightarrow \mbf1$ is an $n-1$-dimensional quantum field theory.
\end{corollary}
\begin{proof}
 We can canonically identify the functor $\Za(S)\colon \fvs\longrightarrow \fvs$ with the vector space $\Za(S)(\C )$ and the natural transformation $\Za(\Sigma)\colon \text{id}_{\fvs}\circ \Za(S_-)\Rightarrow \Za(S_+)$ with a linear map $\Za(\Sigma) \colon  \Za(S_-)(\C)\longrightarrow \Za(S_+)(\C)$. The compatibility conditions summarised by Proposition~\ref{Lemma Compatibility conditions} then imply that the vector spaces and linear maps defined in this way form a quantum field theory.
\end{proof} 
\begin{remark}
Let $\Aa \colon \ECobF \longrightarrow \Tvs$ be an invertible extended field theory and 
$A\colon \CobF \longrightarrow \fvs$ 
the non-extended field theory induced by restricting $\Aa$ to the endomorphisms of the monoidal
unit in $\ECobF$ as explained in Remark~\ref{Rem: EQFT induces QFT}.  
Restricting Definition~\ref{Definition Anommalous field theory} to the endomorphisms of
$\emptyset$ induces a natural transformation $\tr\, A \Longrightarrow
1$ as can be seen from Remark~\ref{Rem: concrete description of anomalous field theories }. 
The direction for this transformation does not agree with the one in 
Definition~\ref{Def: Anomalous partition function}. 
Since in most physically relevant examples all vector spaces are
Hilbert spaces, this discrepancy can be resolved by taking the
adjoint. Here we stick to Definition~\ref{Def: Anomalous partition function}, 
because it has a natural geometric interpretation in terms of line bundles as 
explained in the next section.     
\end{remark}

\subsection{Projective anomaly actions}\label{Sec: P anomaly actions}

In the Lagrangian description of anomalies the partition function of an anomalous field theory can be identified with 
a section of a line bundle over the space of background gauge fields. The anomaly 
is then realised as the non-triviality of this line bundle~\cite{NashBook,CM94}. 
In the Hamiltonian description of anomalies~\cite{Mickelsson:1983xi, CM96,Carey1997,CM94} the state space of the theory cannot 
be defined in a gauge invariant way. The obstruction to this is the non-triviality of 
a gerbe over the space of background fields. Gerbes are higher analogues of line bundles. 
To recover these perspectives from the functorial approach we briefly review the theory 
of line bundles and 2-line bundles over groupoids~\cite{TwistedDWandGerbs,SWParallel}.

A \emph{vector bundle over a groupoid $\cG$} is a functor $\rho : \cG \longrightarrow \fvs$. Although in algebraic terms this is just a representation of $\cG$, the geometric viewpoint has proven to be profitable, see e.g.\ \cite{TwistedDWandGerbs} or \cite{OFK}. 

A \emph{line bundle} $L :  \cG \longrightarrow \fvs$ is a vector bundle for which all fibres, i.e.\ images $\rho(x)$ for $x\in \cG$, are 1-dimensional vector spaces. 
Formulated differently, a line bundle takes values in the maximal Picard subgroupoid of $\fvs$. 
The simplest example is the trivial line bundle $\mbf1 \colon \cG \longrightarrow \fvs$ sending 
every object to $\C$ and every morphism to the identity map. A \emph{section} of a line bundle
$L\colon \cG \longrightarrow \fvs$ is a natural transformation $\mbf1 \Longrightarrow L$.

The Picard groupoid corresponding to $\fvs$ is equivalent to the category $\C \DS \C^\times$ with one object $\C$ and $\C^\times$ as endomorphisms. Hence, we
can factor every line bundle up to a natural isomorphism as
\begin{equation}\label{EQ: Classification of line bundles}
\begin{tikzcd}
\, \cG  \ar{r}{L}\ar{d}{\tilde{L}} & \fvs \\
\C\DS \C^\times \ar[hookrightarrow]{ru}{} &
\end{tikzcd}
\end{equation} 
This shows that the groupoid of line bundles over $\cG$ is equivalent to the groupoid $[\cG,\C\DS\C^\times]$ of functors $\cG \longrightarrow \C\DS\C^\times$. Hence, line bundles over $\cG$ are classified by 
\begin{align}
\pi_0[\cG,\C\DS\C^\times] =  [|B\cG|,K(\C^\times,1)] = H^1(\cG;\C^\times) \, 
\end{align}
where 
\begin{itemize}
	\item $|B\cG|$ is the geometric realisation of the nerve $B\cG$ of $\cG$,
	\item  $K(\C^\times,1)$ is the aspherical Eilenberg-MacLane space with fundamental group $\C^\times$
	\item and $H^1(\cG;\C^\times)$ the first groupoid cohomology with coefficients in $\C^\times$.
\end{itemize}
The restriction of every invertible field theory $\Aa \colon \CobF \longrightarrow \fvs$ 
to $\tr \CobF$ induces a line bundle over the groupoid $ \tr \CobF$. Note that $\tr \CobF$ is nothing else
than the groupoid $\SymF$ of symmetries. An anomalous partition function $\Za \colon \mbf1 \Longrightarrow \tr \Aa$ is then just a section of this line bundle. 

\begin{remark}
If the stack $\F$ is 
a smooth stack, i.e.\ the groupoid of background field configurations is a (infinite dimensional)
Lie groupoid, such as the stack of connections on a principal $G$-bundle, then it is natural 
to require the line bundle $\Aa|_\SymF$ and the section corresponding to the anomalous partition function
to be smooth in an appropriate sense. We do not develop this idea in any detail in this thesis.         
\end{remark}
Categorifying the definition of a vector bundle over a groupoid we arrive at the following notion, see \cite{SWParallel}:
\begin{definition}
A \emph{2-vector bundle over a groupoid $\cG$} is a 2-functor
\begin{align}
\rho : \cG \longrightarrow \Tvs \ ,
\end{align}
where we consider $\cG$ as a 2-category with only trivial 2-morphisms. 
\end{definition}
This is again just a 
2-representation of $\cG$. However, the geometric perspective makes the relation to classical 
approaches to the description of anomalies more apparent.
    
A \emph{2-line bundle over $\cG$} is a 2-vector bundle which takes values in the full Picard sub-2-groupoid of $\Tvs$. A \emph{section $s$} of a 2-line bundle $L\colon \cG \longrightarrow \Tvs$ is
a natural 2-transformation $s \colon \mbf1 \Longrightarrow L$. 
 
A 2-vector space is invertible with respect to the Degline tensor product if and only if it is 1-dimensional, i.e.\ equivalent to the category of vector spaces. Every linear functor between two 1-dimensional 2-vector spaces can be described up to natural isomorphism by a vector space, see
Remark~\ref{Rem: 2VS}. A functor is invertible if and only if this vector space is 1-dimensional.
This shows that the Picard 2-subgroupoid of $\Tvs$ is equivalent to $\fvs\DS \id \DS \C^\times$.  
Hence, for every 2-line bundle $L: \cG \longrightarrow \Tvs$ there is a diagram 
\begin{equation}\label{EQ: Classification of 2-line bundles}
\begin{tikzcd}
\, \cG \ar{r}{L}  \ar[swap]{d}{\tilde{L}} & \fvs \\
\fvs\DS \id \DS \C^\times \ar[hookrightarrow]{ru}{}  &
\end{tikzcd}
\end{equation} 
commutative up to a 2-isomorphism.

Using the higher categorical analogue of \eqref{EQ: Classification of line bundles}, we arrive at a classification of 2-line bundles in terms of
\begin{align}
\pi_0[\cG,\fvs\DS \id \DS \C^\times] = [|B\cG|, K(\C^\times,2)] = H^2(\cG;\C^\times)\ .
\end{align} 
Let us be a bit more explicit on how 2-functors $  \cG \longrightarrow \fvs\DS \id 
\DS \C^\times=\mathsf{B}^2\C^\times$
are related to 2-cocycles and their coboundaries.
\begin{remark}\label{rem:2cocycle}
\begin{itemize}
\item[(a)]
Let $\alpha\colon \cG \longrightarrow \mathsf{B}^2\C^\times$ be a 2-functor. Writing out Definition \ref{Definition Morphism Bicategory} we get for every pair $(g,g')\in \mathrm{Hom}_\cG(G_1,G_2)\times \mathrm{Hom}_\cG(G_2,G_3)$ a non-zero complex number $\alpha_{g,g'}$ such that 
\begin{align}
\alpha_{g_3\circ g_2, g_1} \, \alpha_{g_3,g_2} = \alpha_{g_3, g_2\circ g_1} \,  \alpha_{g_2,g_1} \ , 
\end{align}
for all composable morphisms $g_1,g_2,g_3$, and 
\begin{align}
\alpha_{\id_{\sft(g)},g}=\alpha_{\id_{\sft(g)},\id_{\sft(g)}}=\alpha_{g,\id_{\sfs(g)}} \ . 
\end{align}
Note that the 2-morphism $\alpha_1 \colon \alpha(\id)\Longrightarrow \id$ is completely fixed by the coherence condition \eqref{EQ2: Definition 2Functor} and takes the value $\alpha_{\id,\id}^{-1}$.

\item[(b)]
The data contained in a natural 2-transformation $\sigma :\alpha\Longrightarrow\alpha'$ between two 2-cocycles is given by a collection $\sigma_g \in \C^\times$ for all morphisms $g$ in $\cG$ such that
\begin{align}
\sigma_{g_2\circ g_1} \ \alpha'_{g_2,g_1}= \alpha_{g_2,g_1} \ \sigma_{g_1} \, \sigma_{g_2} 
\end{align}
for all composable morphisms $g_1,g_2$. This is the coherence condition \eqref{Equation1: Definition Transformation} which also implies \eqref{Equation2: Definition Transformation}.
We see that natural 2-transformations restrict to the usual coboundaries on endomorphisms of an object.

\item[(c)]
The data contained in a modification $\theta:\sigma\Rrightarrow\sigma'$ between two natural 2-transformations is an assignment of an element $\theta_G \in \C^\times $ to every $G\in \mathrm{Obj}(\mathscr{G})$ such that
\begin{align}
\theta_{\sft(g)} \ \sigma_g = \sigma'_g \ \theta_{\sfs(g)} \ ,
\end{align}
which is the condition \eqref{EQ: Modification}.
\end{itemize}
\end{remark} 
For an anomalous quantum field theory the symmetry group (groupoid) only acts projectively.
Projective representations of groupoids have the following concrete definition 
(see e.g.~\cite[Section 2.3.1]{TwistedDWandGerbs}). 

\begin{definition}
A \emph{projective representation} $\rho $ of a groupoid $\cG$ twisted by a 2-cocycle $\alpha\colon \cG \longrightarrow \sfB^2\C^\times$ consists of the following data:
\begin{itemize}
\item[(a)]
A complex vector space $V_G$ for all $G\in \Obj(\cG)$.

\item[(b)]
A linear map $\rho (g)\colon V_{\sfs(g)}\longrightarrow V_{\sft(g)}$ for each morphism $g$ of $\cG$ such that 
\begin{align}
\rho(g_2) \circ \rho(g_1) = \alpha_{g_2,g_1} \ \rho (g_2 \circ g_1)
\end{align}
for all composable morphisms $g_1,g_2$.
\end{itemize}
\end{definition}
There is a reformulation of this definition adapted to the study of anomalous functorial field theories. 
\begin{proposition}\label{prop:projrep}
A projective groupoid representation with 2-cocycle $\alpha\colon \cG\longrightarrow \sfB^2\C^\times \subset \Tvs$ is the same as a natural 2-transformation $\mbf1 \Longrightarrow \alpha$, where $\alpha$ is considered as a 2-functor to $\Tvs$.\footnote{This is the same as a higher fixed point for the representation $\alpha $ of $\cG$.}  
\end{proposition}
\begin{proof}
This follows immediately from spelling out Definition \ref{Definition transformation Bicategory}. 
\end{proof}
\begin{remark}
We can use Proposition~\ref{prop:projrep} to define intertwiners between projective representations as modifications between the corresponding 2-transformations. 
\end{remark}
To apply this general formalism to the anomalous field theories at
hand, first note that by restriction every extended invertible quantum field theory 
$\Aa \colon \ECobF \longrightarrow \Tvs$ induces 
a 2-line bundle over $\ESymF$. An anomalous field theory $\Za \colon \mbf1 \Longrightarrow \tr\Aa$
induces a section of this 2-line bundle. 
We denote by $\Pic_2(\Tvs)$ the Picard 2-groupoid of the 2-category of 2-vector 
spaces; there is a
canonical embedding $\Pic_2(\Tvs)\longrightarrow\Tvs$. 
An extended quantum field theory $\Aa$ is invertible if and only if 
it factors uniquely through $\Aa:\CobF\longrightarrow\mathsf{Pic}_2(\Tvs)\hookrightarrow\Tvs$. 
 
We can pick an equivalence of 2-categories $\Pic_2(\Tvs)\longrightarrow \mathsf{B}^2\C^\times$ by choosing a non-canonical equivalence between every invertible 2-vector space and $\fvs$. This identifies every 
invertible linear functor between invertible 2-vector spaces with a one dimensional vector space. 
Picking a linear isomorphism from every 1-dimensional vector space to $\C$ induces this 
equivalence. An inverse to this equivalence is given by the embedding 
$\iota \colon \mathsf{B}^2\C^\times \longrightarrow \mathsf{Pic}_2(\Tvs)$.
The invertibililty of the anomaly quantum field theory $\mathcal{A}$ and this equivalence induces a 2-cocycle of the symmetry groupoid with values in $\C^\times$: 
\begin{align}
\alpha^\Aa:\ESymF  \longrightarrow \mathsf{B}^2\C^\times \ .
\end{align}
The cohomology class of this 2-cocycle is independent of the choices involved: let $\phi,\phi' \colon \Pic_2(\Tvs)\longrightarrow \mathsf{B}^2\C^\times$ be two equivalences of bicategories which both 
have the inclusion $\iota\colon \mathsf{B}^2\C^\times \longrightarrow \Pic_2(\Tvs)$ as weak inverse. 
We then get 
a chain of natural isomorphisms
\begin{align}
\phi \Longrightarrow \phi \circ \iota \circ \phi' \Longrightarrow \phi' 
\end{align} 
inducing a natural isomorphism $\alpha^\Aa \Longrightarrow \alpha'^\Aa$, which is by 
Remark~\ref{rem:2cocycle} a coboundary. 
Combining these facts with Proposition~\ref{prop:projrep} we can then infer
\begin{proposition}
Every anomalous quantum field theory $\Za:\mbf1\Longrightarrow\tr\Aa$ induces a projective
representation of the symmetry groupoid $\ESymF$. The
2-cocycle $\alpha^\Aa$ corresponding to this representation is unique up to coboundary.
\end{proposition}
We have seen in Proposition~\ref{prop:projrep} that natural 2-transformations $\mbf1\Longrightarrow \alpha$ are the same as projective representations of groupoids, so it should come as no surprise that these cocycles appear in the description of anomalies. The interesting prospect is that we can extend these cocycles to invertible extended field theories. This allows us to calculate quantities related to anomalies using the machinery of extended quantum field theories. Furthermore, we can couple such a theory to a bulk theory cancelling the anomaly as we explain in the next section. It is not clear that every anomaly admits such an extension, but all anomalies should give a projective representation of the symmetry groupoid.

\subsection{Anomaly inflow in functorial field theories}\label{Sec: Anomaly inflow}
The definition of an anomalous functorial quantum field theory (Definition~\ref{Definition Anommalous field theory}) 
only depends on the truncation of the extended invertible field theory $\Aa$. One could 
wonder why we still require $\Aa$ to be a full field theory and indeed the definition of 
a twisted field theory due to Stolz and Teichner~\cite{Stolz:2011zj} only requires a functor from 
the truncated cobordism category. The advantage of having a full quantum field theory is
that the evaluation on $n$-dimensional manifolds allows for interesting constructions 
related to the anomalous field theory. The most important example is the coupling of bulk 
and boundary degrees of freedom to produce anomaly-free field theories.   
We shall now discuss in more detail how to couple the bulk field
theory $\Aa$ and boundary field theory $\Za$ to construct an anomaly-free theory. 
We start with the unextended framework corresponding to 
Definition~\ref{Def: Anomalous partition function}.
This involves the full quantum field theory $\Aa \colon \CobF \longrightarrow \fvs$ 
and not just its truncation. 

Let $(\Sigma,\mathsf{f}^{n-1})$ be an object of $\CobF$
and $(M,\mathsf{f}^n)$ a morphism $(\Sigma,\mathsf{f}^{n-1})\longrightarrow \emptyset$. In
particular, $M$ is an $n$-dimensional manifold 
with boundary $\partial M = -\Sigma $, and background gauge fields $\mathsf{f}^n \in \F(M)$ extending
the background gauge field $\mathsf{f}^{n-1}$ on the boundary. 
An anomalous field theory $\Za \colon \mbf1 \Longrightarrow \tr \Aa$ defines an element $\Za(\Sigma,\mathsf{f}^{n-1})\in \Aa(\Sigma,\mathsf{f}^{n-1})$. 
The partition function of the composite system can now be defined as
\begin{align}\label{Def: Combined partition function}
\Za_{\rm bb}(M,\mathsf{f}^{n}, \Sigma)= \Aa(M,\mathsf{f}^n)[\Za(\Sigma,\mathsf{f}^{n-1})] \ \in \ \Aa(\varnothing)\cong \C \ .
\end{align}     
This definition does not depend on any additional choices. 
\begin{proposition}\label{Prop: bulk boundary}
The combined partition function $Z_{\rm bb}(M,\mathsf{f}^{n}, \Sigma)$ is invariant under
$\F$-diffeomorphism, which do not need to be the identity at the boundary.\footnote{However, they need to be 
compatible with the collars, i.e.\ invariant.} 
\end{proposition}
\begin{proof}
Let $(M',\mathsf{f}'^{n})$ be a morphism $(\Sigma',\mathsf{f}'^{n-1})\longrightarrow \emptyset$ and $\nu\colon (M,\mathsf{f}^{n})\longrightarrow (M',\mathsf{f}'^{n})$ a $\F$-diffeomorphism.
We then calculate
\begin{align}
\begin{split} 
Z_{\rm bb}(M',\mathsf{f}'^{n},\Sigma' )&= \Aa(M',\mathsf{f}'^{n}\,)[\Za(\Sigma',\mathsf{f}'^{n-1})] \\[4pt]
&= \Aa(M',\mathsf{f}'^{n}\,)\circ \Aa( \nu|_{\Sigma})[\Za(\Sigma,\mathsf{f}^{n-1})]\\[4pt]
&=\Aa(M,\mathsf{f}^{n})[\Za(\Sigma,\mathsf{f}^{n-1})]\\[4pt]
&=\Za_{\rm bb}(M,\mathsf{f}^{n}, \Sigma)\ , 
\end{split} 
\end{align}
where in the second
equality we used~\eqref{Eq: Diagram APF} and in the third equality the fact that $\Aa$ is invariant
under $\F$-diffeomorphisms relative to the boundary. 
This shows that the composite partition function is anomaly-free. 
\end{proof}
Definition \ref{Definition Anommalous field theory} also allows us to formulate the
composite system at the level of state spaces. Let $\Aa \colon \ECobF
\longrightarrow \Tvs$ be an invertible extended field theory. Consider
a $1$-morphism $(\Sigma,\mathsf{f}^{n-1}) \colon (S,\mathsf{f}^{n-2})\longrightarrow \emptyset$. An
anomalous field theory $\Za \colon \mbf1 \longrightarrow \tr \Aa$ defines an element 
$\Za(S,\mathsf{f}^{n-2})\in \Aa(S,\mathsf{f}^{n-2})$. The composite state space is given by
\begin{align}
\Za_{\rm bb}(\Sigma,\mathsf{f}^{n-1} , S)= \Aa(\Sigma,\mathsf{f}^{n-1} )[Z(S,\mathsf{f}^{n-2})] \
  \in \ \Aa(\varnothing) \cong \fvs \ . \label{Eq: Def combined state space}
\end{align} 
This vector space does not depend on any additional choices. 
To construct an honest action on the combined state space, the following observation (which we will also use
later on) is helpful:
\begin{lemma}\label{Lem: non-constant morphism}
Let 
\begin{align}
(\Sigma,f, \epsilon, \varphi_-,\varphi_+)\colon (S_-,\epsilon_-,\epsilon, f_-)\longrightarrow (S_+,\epsilon_+,\epsilon,f_+)
\end{align}
and 
\begin{align} 
(\Sigma',f', \epsilon, \varphi'_-,\varphi'_+)\colon (S'_-,\epsilon_-,\epsilon, f'_-)\longrightarrow (S'_+,\epsilon_+,\epsilon , f'_+)
\end{align} 
be 1-morphisms in $\ECobF$. Furthermore, let $\id \times \nu\colon ((-\epsilon,\epsilon)\times \Sigma,f)\longrightarrow ((-\epsilon,\epsilon)\times \Sigma',f')$ be an $\F$-diffeomorphism
constant on the collars. 
Then 
\begin{equation}
\begin{tikzcd}
(S_-,\epsilon_-,\epsilon, f_-) \ar[d, "{\nu|_{S_-}}",swap] \ar[r,"{(\Sigma, f)}"] & (S_+,\epsilon_+,\epsilon,f_+) \ar[d, "{\nu|_{S_+}}"] \\ 
(S'_-,\epsilon_-,\epsilon,f'_-) \ar[r,"{(\Sigma', f')}",swap] & (S'_+,\epsilon_+,\epsilon,f'_+) 
\end{tikzcd}
\end{equation}
commutes up to a limit 2-morphism corresponding to $\nu$.
\end{lemma}
\begin{proof}
This follows directly from the commutativity of 
\begin{equation}
\begin{tikzcd}
 & (-\epsilon,\epsilon)\times\Sigma  \ar[dd, "{\id\times\nu}"] & \\ 
 {(-\epsilon,\epsilon)\times [0,\epsilon_-)\times S_-} \ar[ru, "{\id\times\varphi_-}"] \ar[rd, "{\id\times (\varphi_-' \circ \nu|_{S_-})}", swap] & & {(-\epsilon,\epsilon)\times(-\epsilon_+,0]\times S_+' } \ar[lu,swap, " \id\times (\varphi_+ \circ {\nu|^{-1}_{S_+'})}"] \ar[ld, "{\id\times \varphi_+'}"] \\ 
  & (-\epsilon,\epsilon)\times\Sigma' &
\end{tikzcd} 
\end{equation}
\end{proof}
In the previous lemma we were careful to include all the necessary $\epsilon$'s. In the reminder of this section we
return to suppressing them in the notation. 
Let $ (\Sigma,\mathsf{f}^{n-1}) \colon (S,\mathsf{f}^{n-2})\longrightarrow \emptyset $ and $ (\Sigma',\mathsf{f}'^{n-1})\colon (S',\mathsf{f}'^{n-2})\longrightarrow \emptyset$ be a 1-morphisms in $\ECobF$ and $\nu
\colon (\Sigma,\mathsf{f}^{n-1}) \longrightarrow (\Sigma', \mathsf{f}'^{n-1})$ a $\F$-diffeomorphism. 
Then there is an induced linear map $\Za_{\rm bb}(\nu)$
\begin{align}
\begin{split} 
\Za_{\rm bb}(\Sigma,\mathsf{f}^{n-1} , S)=  \Aa(\Sigma,\mathsf{f}^{n-1} )&[\Za(S,\mathsf{f}^{n-2})]\xrightarrow{\Aa(\nu)} \Aa\big((\Sigma', \mathsf{f}'^{n-1}\,)\circ \nu|_{S} \big)[\Za(S,\mathsf{f}^{n-2})]\\[4pt]
&\xrightarrow{\hspace{0.65cm}} \Aa(\Sigma', \mathsf{f}'^{n-1}\,)\circ \Aa(\nu|_{S})[Z(S,\mathsf{f}^{n-2})]\\[4pt]
&\xrightarrow{\Za(\nu|_{S})} \Aa(\Sigma',\mathsf{f}'^{n-1}\,)[\Za(S',
  \mathsf{f}'^{n-2})]=\Za_{\rm bb}(\Sigma',\mathsf{f}'^{n-1} , S') \ , \label{Eq: Def action combined state space}
\end{split} 
\end{align}
where the first map is induced by the limit 2-morphism constructed in Lemma~\ref{Lem: non-constant morphism}.  
\begin{theorem}\label{Thm: Bulk boundary}
The linear maps $\Za(\nu)$ provide an honest representation of the groupoid of $\F$-background fields on 1-morphisms 
to $\emptyset$ and their symmetries,
i.e.\ $\F$-diffeomorphisms.
\end{theorem}
\begin{proof}
Let $\nu_a \colon (\Sigma,\mathsf{f}^{n-1})\longrightarrow (\Sigma',\mathsf{f}'^{n-1})$ and $\nu_b \colon (\Sigma',\mathsf{f}'^{n-1})\longrightarrow (\Sigma'',\mathsf{f}''^{n-1})$
be $\F$-diffeomorphisms as above. The composition
\begin{align} 
\begin{split} 
  \Aa(\Sigma,\mathsf{f}^{n-1} )[\Za(S,\mathsf{f}^{n-2})]&\xrightarrow{\Aa(\nu_a)} \Aa\big((\Sigma', \mathsf{f}'^{n-1}\,)\circ \nu_a|_{S} \big)[\Za(S,\mathsf{f}^{n-2})]\\[4pt]
	&\xrightarrow{\hspace{0.65cm}} \Aa(\Sigma', \mathsf{f}'^{n-1}\,)\circ \Aa(\nu_a|_{S})[Z(S,\mathsf{f}^{n-2})]\\[4pt]
	&\xrightarrow{\Za(\nu_a|_{S})} \Aa(\Sigma',\mathsf{f}'^{n-1}\,)[\Za(S',
	\mathsf{f}'^{n-2})]  \\
	& \xrightarrow{\Aa(\nu_b)} \Aa\big((\Sigma'', \mathsf{f}''^{n-1}\,)\circ \nu_b|_{S'} \big)[\Za(S',\mathsf{f}'^{n-2})]\\[4pt]
	&\xrightarrow{\hspace{0.65cm}} \Aa(\Sigma'', \mathsf{f}''^{n-1}\,)\circ \Aa(\nu_b|_{S'})[Z(S',\mathsf{f}'^{n-2})]\\[4pt]
	&\xrightarrow{\Za(\nu_b|_{S'})} \Aa(\Sigma'',\mathsf{f}''^{n-1}\,)[\Za(S'',
	\mathsf{f}''^{n-2})]
\end{split} 
\end{align}
can be rewritten using the commutativity of linear maps corresponding to natural transformations applied at different positions
in the composition of functors as 
\begin{align} 
\begin{split} 
 \Aa(\Sigma,\mathsf{f}^{n-1} )&[\Za(S,\mathsf{f}^{n-2})]\xrightarrow{\Aa(\nu_a)} \Aa\big((\Sigma', \mathsf{f}'^{n-1}\,)\circ \nu_a|_{S} \big)[\Za(S,\mathsf{f}^{n-2})]\\[4pt]
&\xrightarrow{\hspace{0.65cm}} \Aa(\Sigma', \mathsf{f}'^{n-1}\,)\circ \Aa(\nu_a|_{S})[Z(S,\mathsf{f}^{n-2})]\\[4pt]
&\xrightarrow{\Aa(\nu_b)} \Aa(\Sigma'', \mathsf{f}''^{n-1}\,)\circ \Aa(\nu_b|_{S'})\circ \Aa(\nu_a|_{S})[\Za(S,
\mathsf{f}^{n-2})]  \\ 
&\xrightarrow{\Za(\nu_a|_{S})} \Aa(\Sigma'', \mathsf{f}''^{n-1}\,)\circ \Aa(\nu_b|_{S'})[\Za(S',
\mathsf{f}'^{n-2})] \\ 
&\xrightarrow{\Za(\nu_b|_{S'})} \Aa(\Sigma'', \mathsf{f}''^{n-1}\,)[\Za(S'',
\mathsf{f}''^{n-2})] \ \ .
\end{split}  
\end{align}
Next we use that $\Aa$ is a 2-functor and Equation~\eqref{EQ1: Lemma Compatibility conditions} to rewrite this as 
\begin{align}
\begin{split}  \Aa(\Sigma,\mathsf{f}^{n-1} )[\Za(S,\mathsf{f}^{n-2})]&\xrightarrow{\Aa(\nu_b \circ \nu_a)} \Aa\big((\Sigma'', \mathsf{f}''^{n-1}\,)\circ (\nu_b \circ \nu_a)|_{S} \big)[\Za(S,\mathsf{f}^{n-2})]\\[4pt]
&\xrightarrow{\hspace{0.65cm}} \Aa(\Sigma'', \mathsf{f}''^{n-1}\,)\circ \Aa((\nu_b \circ \nu_a)|_{S})[Z(S,\mathsf{f}^{n-2})]\\[4pt]
&\xrightarrow{\Za((\nu_b \circ \nu_a)|_{S})} \Aa(\Sigma'',\mathsf{f}''^{n-1}\,)[\Za(S'',
\mathsf{f}''^{n-2})] \ . 
\end{split} 
\end{align} 
This finishes the proof. 
\end{proof}
This theorem describes a way of coupling bulk and boundary degrees of freedom to an anomaly-free state space. In condensed matter physics applications the invertible field theory $\Aa$ arises as the low-energy effective theory of the bulk system. 
\chapter{The parity anomaly}\label{Cha: Index}
In this chapter we deploy the general theory developed in Chapter~\ref{Sec: General Theory} to
describe the parity anomaly of fermionic gauge theories defined on odd-dimensional manifolds. 
The parity anomaly has recently received renewed interest due to its
relation to topological insulators and more generally topological phases of matter~\cite{WittenFermionicPathInt}.   
The anomaly field theory can be constructed using the index of a Dirac operator on even dimensional 
manifolds. 

In Section~\ref{Sec: APS} we collect some results about index theory on manifolds with boundaries, which
suffice to construct a non-extended field theory describing the parity anomaly in Section~\ref{Sec: path integral parity}.
To capture the Hamiltonian perspective, we use the index theorem on manifolds with corners~\cite{LoyaMelrose} reviewed in Section~\ref{Sec: Index2} to construct an extended field theory describing the parity anomaly in Section~\ref{Sec: Extended Indext FQFT}. 
The Hamiltonian description of the parity anomaly has been largely unexplored in the literature with the exception of \cite{Chang:1986ri} and the recent study in~\cite{Lapa19}. One of the achievements of our approach is the computation
of the 2-cocycle twisting the projective representation of the symmetry group on the state space of any quantum field 
theory with parity anomaly in terms of geometric quantities, see Equation~\eqref{eq:2cocycleexplicit}.   
   
\section{Index theory part I: the APS-index theorem}\label{Sec: APS}
The description of the extended field theory encoding the parity anomaly relies on 
index theory for manifolds with corners. In this section we present the 
theory on manifolds with boundaries. 
This is already enough to formulate the invertible field theory in the 
non-extended setting. 

We set the stage by recalling a few standard definitions related to spin geometry 
mostly following~\cite{LMspin}.
\begin{definition}
Let $V$ be a vector space equipped with a quadratic form $q\colon V \longrightarrow \R$.
The \emph{Clifford algebra $\Cl(V,q)$} is the quotient of the tensor algebra
\begin{align}
T(V) = \bigoplus_{i=0}^\infty V^{\otimes i}
\end{align}
by the ideal generated by elements of the form $v\otimes v + q(v)1$ for $v\in V$. 
\end{definition}
Note that the $\N$-grading of $T(V)$ induces a natural $\Z_2$-grading
\begin{align}
\Cl(V,q) = \Cl^0(V,q) \bigoplus \Cl^1(V,q)
\end{align}
of the Clifford algebra $\Cl(V,q)$.
\begin{example}
Let $V$ be the vector space $\R^n$ and $q$ the quadratic form induced by
the canonical scalar product $\langle e_i , e_j \rangle=\delta_{ij}$ on $\R^n$. 
The corresponding Clifford algebra, which we denote by $\Cl_n$, is generated by 
the elements $\{e_i\}_{i=0,\dots ,n}$
modulo the relations \begin{align}
\{e_i, e_j \} \coloneqq e_i\cdot e_j + e_j\cdot e_i = -2 \delta_{ij} \ \ . 
\end{align} 
We can identify the vector space underlying the exterior algebra $\Lambda^* \R^n$
with $\Cl_n$ via the linear map 
\begin{align}
\begin{split}
\Lambda^* \R^n & \longrightarrow \Cl_n \\ 
\dd e_{i_1} \wedge \dots \wedge \dd e_{i_k} & \longmapsto e_{i_1} \cdot \dots \cdot e_{i_k} \ \ . 
\end{split}
\end{align}
However, this is not an isomorphism of algebras.
\end{example}

The \emph{Pin group} $\Pin_n$ is the subgroup of the unit group $\Cl_n^\times$ of $\Cl_n$ generated by 
all elements $v\in V$ with $q(v)=1$. The \emph{Spin group} $\Spin_n$ is 
the intersection of $\Pin_n$ with $\Cl^0_n$ in $\Cl_n$. The group $\Spin_n$
acts via conjugation on $\Cl_n$. This action preserves the subspace $\R^n \subset \Cl_n$,
its orientation and the canonical scalar product on it. For this reason, we get an 
induced group homomorphism $\Spin_n \longrightarrow SO(n)$. One can 
show~\cite[Theorem 2.9]{LMspin} that
this morphism fits into a short exact sequence 
\begin{align}
1 \longrightarrow \Z_2 \longrightarrow \Spin_n \longrightarrow SO(n) \longrightarrow 1 \ \ .
\end{align}  
Let $\C l_n= \Cl_n \otimes_\R \C$ be the complexified Clifford algebra. The representation
theory of $\C l_n$ is extremely well-behaved~\cite[Section I.5]{LMspin} due to the
following periodicity result. 
\begin{theorem}[Theorem 4.3 of \cite{LMspin}]
For all $n\geq 0$ there are isomorphisms 
\begin{align}
\C l_{n+2}\cong \C l_{n}\otimes_\C \C l_2 \ \ 
\end{align}
of algebras.
\end{theorem}
Together with $\C l_1=\C \oplus \C$ and $\C l_2= \C (2)$, where $\C (2)$ is the 
complex algebra of $2\times 2$-matrices with complex coefficients, this theorem 
completely classifies complexified Clifford algebras. This implies that Clifford
algebras are semisimple, for $n=2k$ there is
only one non-trivial irreducible representation of $\C l_n$ of dimension $2^k$ and
for $n=2k+1$ there are two non-trivial irreducible representations of $\C l_n$ of
dimension $2^k$.    
\begin{definition}
The \emph{complex spin representation} $\Delta_n\colon \Spin_n \longrightarrow \End_\C(S)$
is the restriction of a non-trivial irreducible representation $S$ of $\C l_n$ to 
$\Spin_n$. For odd $n$ this definition does not depend on the choice of representation. 
\end{definition}
\begin{example}
\begin{itemize}
\item 
The 1-dimensional Spin group is $\Spin_1= \Z_2= \{ \pm 1 \}$. The spin representation is 
the sign representation of $\Z_2$ on $\C$.   

\item 
The 2-dimensional Spin group is $\Spin_2= U(1)$. The map $\Spin_2\longrightarrow SO(2)=U(1)$
is given by sending $u\in U(1)$ to $u^2$. The spin representation is given by
multiplication $U(1)\times \C^2 \longrightarrow \C^2$. 

\item The 3-dimensional Spin group is $SU(2)$. The 2-dimensional spin representation
is the fundamental representation of $SU(2)$ on $\C^2$. 
\end{itemize}
\end{example}

\begin{figure}
\begin{center}
\begin{subfigure}[b]{0.3\textwidth}
\begin{overpic}[scale=1]{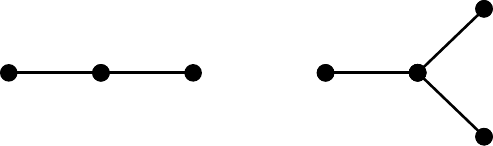}
			\put(44,14){ $\dots$}
		\end{overpic}
		\caption{$D_k$}
	\end{subfigure} \hspace{2cm}
	\begin{subfigure}[b]{0.3\textwidth}
		\begin{overpic}[scale=1]{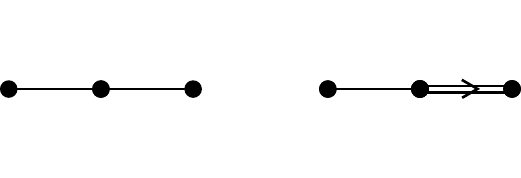}
			\put(44,16){ $\dots$}
		\end{overpic}
		\caption{$B_k$}
	\end{subfigure}
\end{center}
\caption{Dynkin diagram for the Lie algebra of $\Spin_n$. }
\label{Fig: DD Spin}
\end{figure}

\begin{remark}
Let $n$ be larger or equal to 3.
The Lie-algebra of $\Spin_n$ is described by the Dynkin diagram $D_k$ for $n=2k$ and
$B_k$ for $n=2k+1$, see Figure~\ref{Fig: DD Spin}. For $n=2k+1$ the spin representation is the
exponential of the fundamental representation attached to the right most vertex. 
For $n=2k$ the spin representation is the direct sum of the fundamental representations of the
two right most vertices~\cite{DeligneSpinor}.    
\end{remark}

The previous remark shows that for even $n>3$ the Spin representation is reducible. 
A way to distinguish the two irreducible subrepresentations is through the \emph{chirality operator} $\Gamma = 
\iu^{m}e_1\cdot \dots \cdot e_n \in \C l_n$ with $m=\tfrac{n}{2}$ or $m=\tfrac{n+1}{2}$
for $n$ even or odd, respectively. It is straightforward to check that $\Gamma^2=\id$,
$\Gamma v = -v\Gamma $ for even $n$ and $\Gamma v = v\Gamma $ for odd $n$. 
For $n$ even, the spin 
representation $S$ decomposes into the positive and negative eigenspace for the 
action of $\Gamma$: $S= S^+\oplus S^-$, where $S^+$ and $S^-$ are the two irreducible
representations of $\Spin_n$ building up $S$. 

Let $M$ be a compact oriented $n$-dimensional manifold with Riemannian metric $g\in \Gamma(\Sym^2(T^*M))$. The metric induces
a reduction of the frame bundle to an $SO(n)$-bundle $P_{SO(n)}$. 
A \emph{spin-structure} on $M$ is a principal $\Spin_n$ bundle 
$P_{\Spin_n}$ together with a map $P_{\Spin_n}\longrightarrow P_{SO(n)}$
compatible with the map $\Spin_n\longrightarrow SO(n)$. A manifold
together with the choice of a spin structure is called a \emph{spin manifold}. 
The \emph{spinor bundle $S_M$} on a spin manifold $M$ is the associated vector bundle to
$P_{\Spin_n}$ for the spin representation. For even dimensional spin manifolds
it decomposes as $S_M^+\oplus S_M^-$. 
Pulling back the Levi-Civita connection on $P_{SO(n)}$ along the map 
$P_{\Spin_n} \longrightarrow P_{SO(n)}$ induces a connection on $P_{\Spin_n}$
and $S_M$. We call this connection the \emph{spin-connection}.

The group $SO(n)$ acts on the Clifford algebra $\Cl_n$ and its complexification
$\C l_n$ allowing us to define 
for every oriented Riemannian manifold associated Clifford bundles $\Cl_M$ and $\C l_M$. 
If furthermore $M$ is a spin manifold the Clifford bundle $\C l_M$ acts naturally 
on the spinor bundle $S_M$, see e.g.~\cite[Section II.3]{LMspin} for details. 
  
Let $G$ be a compact Lie group
with Lie algebra $\mathfrak{g}$. We fix a unitary finite dimensional representation $\rho_G\colon G \longrightarrow \End(V)$ of $G$. Furthermore, let $M$ be an oriented $n$-dimensional spin manifold with Riemannian metric $g\in \Gamma(\Sym^2(T^*M))$ 
equipped with 
a principal $G$-bundle $P_M$ with connection $A_M\in \Omega^1(P_M; \mathfrak{g})$.
We denote by $\langle \cdot , \cdot \rangle \colon T^*M \longrightarrow TM $ the
isomorphism between the cotangent and tangent bundle induced by $g$. 
We define the \emph{twisted spinor bundle} to be $S^G_M= S\otimes V$, where $V$ is 
the associated vector bundle to $P$ and the representation $\rho_G$. The spin connection
and the connection on $P$ induce a connection $\nabla^{S^G_M}$ on $S^G_M$.   
This data defines a Dirac operator $\slashed{D }_{M}$ on $M$ as the following
composition 
\begin{align}
\Gamma \big(S^G_M\big) \overset{\nabla^{S^G_M}}{\longrightarrow} \Gamma\big(T^*M \otimes S^G_M\big)
\overset{\langle \cdot , \cdot \rangle}{\longrightarrow} \Gamma\big(TM \otimes S^G_M\big)
\overset{\cdot}{\longrightarrow} \Gamma \big(S^G_M\big)
\end{align} 
where $\cdot$ denotes the Clifford action of the tangent bundle on the spinor bundle, 
which is induced by considering the tangent bundle as a subbundle of the Clifford bundle
$\Cl_M$. We can also express $\slashed{D }_{M}$ in terms of a local orthonormal frames 
$\{ e_1,\dots e_n \}$ of $TM$ as 
\begin{align}
\slashed{D }_{M} = \sum_{i=1}^n e_i \cdot \nabla^{S^G_M}_{e_i}  \ \ .
\end{align} 
\begin{example} 
The spinor bundle on $\R^3$ is the trivial vector bundle $\R^3\times \C^2$. The 3-dimensional Dirac operator is 
\begin{align}
\D_{\R^3}= \iu \left(  \sigma_x \partial_x +\sigma_y \partial_y +\sigma_z \partial_z \right) \ \ ,
\end{align}
where $\sigma_i$ are the Pauli matrices 
\begin{align}
\sigma_x=  \bigg( \begin{matrix}
0 & 1 \\ 
1 & 0
\end{matrix} 
\bigg) \ , \quad
\sigma_y= \bigg(
\begin{matrix}
0 & -\,\iu\, \\ 
\,\iu\, & 0
\end{matrix} 
\bigg) \qquad \mbox{and} \qquad
\sigma_z=  \bigg(
\begin{matrix}
1 & 0 \\ 
0 & -1
\end{matrix} 
\bigg) \ . 
\end{align}
\end{example}

For even dimension $n$, the splitting of the spinor bundle $S_{M}= S_{M}^+\oplus S^-_{M}$ induces a splitting of $S_{M}^G= S_{M}^{G+}\oplus S^{G-}_{M}$ into positive and negative chirality spinors; the Dirac operator is odd with respect to this $\Z_2$-grading. 
On a closed manifold $M$, the chiral Dirac operator $\slashed{D
}^+_{M} \colon H^1(S^{G+} _{M})\longrightarrow L^2(S^{G-} _{M})$ is
a first order elliptic differential operator, where $H^1(S^{G+}
_{M})$ is the first Sobolev space of sections of $S^{G+} _{M}$,
i.e.\ spinors $\Psi$ whose image $\slashed{D }^+_{M}\Psi$ is square-integrable, and $L^2(S^{G-} _{M})$ is the Hilbert space of square-integrable sections of $S^{G-} _{M}$; the integration is with respect to the Hermitian structure on $S^G _{M}$ induced by the metric and the unitary representation $\rho_G$. Every elliptic operator acting on sections of a vector bundle of finite rank over a closed manifold $M$ is Fredholm. For
every Fredholm operator $D$ the index is defined by 
$$
\text{ind}(D)= \text{dim}\,\text{ker}(D)-\text{dim}\,\text{coker}(D) \ \ .
$$ 
The Atiyah-Singer index theorem~\cite{Atiyah-Singer} expresses the index of $\slashed{D }^+_{M}$ in
terms of local quantities:
\begin{theorem}
The index of $\slashed{D }^+_{M}$ is given by
\begin{align}
\ind(\slashed{D }^+_{M})= \int_M K_{\rm AS}
\end{align} where the Atiyah-Singer density 
$$
K_{\rm AS}= \ch\big(V \big)\wedge \widehat{A}\big(TM \big)\big|_n
$$ 
is the homogeneous differential form of top degree in $\Omega^n(M)$ occurring in the exterior product of the Chern character of the vector bundle $V$ associated to 
$P_M$ with the $\widehat{A}$-genus of the tangent bundle $TM$.
\end{theorem}
Even-though the concrete form of $K_{\rm AS}$ will not be important for the rest of this thesis (we only need the existence of such a differential form) we briefly 
define the characteristic classes involved. Let $W\longrightarrow M$ be a vector bundle over $M$ and $\nabla^W\colon 
\Gamma(W)\longrightarrow \Gamma(T^*M \otimes W )$ a 
connection on $W$. We denote by $K^{\nabla^W} \in \Omega^2(M; \End(W))$ the curvature of $\nabla^W$. 
The \emph{Chern character of $W$} $\ch(\nabla^W)\in \Omega^\bullet(M)$ is 
\begin{align}
\ch (W)\coloneqq \tr \left(\exp\left( \tfrac{K^{\nabla^W}}{2\pi \iu}\right)\right)  \in  \Omega^\bullet(M) \ \ ,
\end{align}
where the exponential is defined by the usual power series. Note that this is well-defined, since only finitely many terms
in the power series are non zero. Furthermore, the cohomology class $\ch (\nabla^W)$ is independent of the concrete choice of 
$\nabla^W$.    
The $\widehat{A}$-genus $\widehat{A}(W)\in \Omega^\bullet(M)$ is defined by the expression
\begin{align}
\widehat{A}(W) \coloneqq \det {}^{\tfrac{1}{2}} \left( \frac{K^{\nabla^W}/2}{\sinh(K^{\nabla^W}/2)} \right) \in \Omega^\bullet(M) \ \ .
\end{align} 
Again the expression has to be understood in terms of its Taylor expansion and its cohomology class is independent of the
choice of $\nabla^W$.

We now turn our attention to manifolds with boundary. 
Let $M$ be an $n$-dimensional oriented Riemannian manifold equipped with a spin structure and boundary $\partial M$. Recall that the spin structure consists of a double cover of the frame bundle $P_{SO(n)}(M)$ by a principal $\Spin_n$-bundle $P_{\Spin_n}(M) \longrightarrow M$. We can include the frame bundle $P_{SO(n-1)}(\partial M )$ into $P_{SO(n)}(M)$ by adding the inward pointing normal vector to an orthonormal frame of $\partial M$. The pullback of the double cover $P_{\Spin_n}(M)$ along this inclusion induces a spin structure on $\partial M$. 

We assume from now on that all structures are of product form on a fixed 
collar of the boundary $\partial M$. 
To describe the relation between the Dirac operator on the boundary and on the bulk manifold 
we use the embedding of Clifford bundles 
\begin{align}
\Cl_{\partial M} \longrightarrow \Cl_{M} \ , \quad 
T_x (\partial M) \ni v \longmapsto  v \, n_x \ ,  
\end{align}         
where $n$ is the inward pointing normal vector field corresponding to the collar. This gives $S_M\big|_{\partial M}$ the structure of a Clifford bundle over $\partial M$.
For the relation to the spinor bundle over the boundary we need to distinguish between even and odd dimensions.

If the dimension $n$ of $M$ is odd then we can identify $S_M\big|_{\partial M}$ with the spinor bundle over $\partial M$.
In this case the Dirac operator can be described in a neighbourhood of $\partial M$ by
\begin{align}
\label{Equation: Dirac operator on a product}
\D_M = n\cdot \big(\D_{\partial M} +\partial_n\big) \ .
\end{align} 
On the other hand, if the dimension $n$ of $M$ is even then the spinor bundle $S_M=S^+_M\oplus S^-_M$ decomposes into spinors of positive and negative chirality. The Clifford action of $\Cl_{\partial M}$ leaves this decomposition invariant and we can identify the spinor bundle over $\partial M$ with the pullback of the positive spinor bundle $S^+_M\big|_{\partial M}$. As the Clifford action of $\Cl_{\partial M}$ commutes with the chirality operator $\Gamma$, an identification with the negative spinor bundle is possible as well. Near the boundary the Dirac operator is given by
\begin{align}
\D_M = n\cdot\begin{pmatrix}
\D_{\partial M} +\partial_n & 0 \\ 
0 &  \Gamma|_{S^+_M} \, \D_{\partial M} \, \Gamma|_{S^-_M} +\partial_n
\end{pmatrix} \ .
\end{align}
When trying to extend the index of $\slashed{D }_{M}$ to manifolds with boundary one runs
into the problem that the Dirac operator on a manifold with boundary is never Fredholm. 
This can be solved by introducing suitable boundary condition. We will explain an alternative 
solution via attaching cylindrical ends.\footnote{This is equivalent to the introduction of
Atiyah-Patodi-Singer spectral boundary conditions on the spinors~\cite{APS}. We use the method of cylindrical ends, since it can be generalised to manifolds with corners and gives a natural cancellation of certain terms later on.}
We are mostly interested in the situation when $M$ is a cobordism, which
comes with a decomposition of its boundary $\partial M = \partial M_+ \sqcup 
\partial M_-$ and corresponding collars $[0,\epsilon) \times \Sigma_-  \longrightarrow M$
and  $(-\epsilon',0] \times \Sigma_+\longrightarrow M$.  
 We define
\begin{align}
\widehat{M}= M \sqcup_{\partial M} \big(   (-\infty,0] \times \Sigma_- \sqcup [0,\infty ) \times \Sigma_+   \big)  
\ ,
\end{align} 
where we use the collar and the fact that all structures on it are of product 
form to perform the gluing.
We extend the metric, spin structure and principal bundle with connection as products to $\widehat{M}$. The structure of the collars of the cobordism makes it natural to attach inward and outward pointing cylinders to the incoming and outgoing boundary, respectively, contrary to what is normally done in the index theory literature. It is further natural, again in contrast to what is normally done in index theory, to glue the cylinders along the identification of the collars with cylinders; this means that the gluing could ``twist'' bundles. Alternatively, we could first attach a mapping cylinder for the identification and then an infinite cylinder.

The Dirac operator $\slashed{D }^+_{\widehat{M}} \colon H^1(S^{G+}_{\widehat M})\longrightarrow L^2(S^{G-}_{\widehat M})$ is Fredholm if and only if the kernel of the induced Dirac operator on the boundary of $M$ is trivial~\cite{Loyaindex}. If the kernel is non-trivial, then we have to regularize the index in an appropriate way, which corresponds physically to introducing small masses for the massless fermions on $\widehat{M}$. This is done precisely by picking, for every connected component $\partial M_i$ of the boundary, a small number $\alpha_i$ with $0<\alpha_i  < \delta_i$, where $\delta_i$ is the smallest magnitude $|\lambda_i|$ of the non-zero eigenvalues $\lambda_i$ of the induced Dirac operator on $\partial M_i$. Now we can attach weights $\e^{\alpha_i \, s_i}$ to the integration measure on the cylindrical ends, where $s_i$ is the coordinate on the cylinder over $\partial M_i $. Denoting the corresponding weighted Sobolev spaces by $\e^{\alpha\cdot s}H^1(S^{G+}_{\widehat M})$ and $\e^{\alpha \cdot s}L^2(S^{G-}_{\widehat M})$, we then have\footnote{To be more precise, we have to first attach a mapping cylinder before we can apply~\cite[Theorem 5.60]{MelrosebGeo}.}
\begin{theorem}(\cite[Theorem 5.60]{MelrosebGeo})\label{Thm:D is Fredholm}
 \ The Dirac operator 
$$ 
\slashed{D }^+_{\widehat{M}} \colon \e^{\alpha\cdot s}H^1\big(S^{G+}_{\widehat M}\big)\longrightarrow \e^{\alpha \cdot s}L^2\big(S^{G-}_{\widehat M}\big)
$$ 
is Fredholm and its index is independent of the masses $\alpha_i$.
\end{theorem}
There is an extension of the index theorem to the case of manifolds with boundary.
One new ingredient is the $\eta $-invariant of the Dirac operator on a closed manifold $\Sigma$ of odd dimension equipped with appropriate geometric structures which calculates the number of positive eigenvalues minus the number of negative eigenvalues of $\slashed{D }_{\Sigma}$, and is defined by 
\begin{align}
\eta \big( \slashed{D }_{\Sigma}\big)= \lim_{s\rightarrow 0} \ \sum_{\substack{\lambda \in \text{spec}( \slashed{D }_{\Sigma}) \\ \lambda\neq0}} \, \frac{\text{sign}(\lambda)}{|\lambda|^s} \ . 
\end{align}
The limit here should be understood as the value of the analytic continuation of the meromorphic function $\sum_{\lambda\neq0} \, \tfrac{\text{sign}(\lambda)}{|\lambda|^s}$ at $s=0$; the regularity of this value is proven in \cite{APS}.
The $\eta $-invariant can be reformulated as an integral over the trace of the corresponding heat kernel operator as
\begin{align}
\label{Equation: Reformulation eta invariant}
\eta \big( \slashed{D }_{ \Sigma }\big) = \frac{1}{\sqrt{\pi}} \, \int_0^\infty\, t^{-{1}/{2}} \ \text{Tr}\Big(\slashed{D }_{ \Sigma} \, \e^{-t\, \slashed{D }{}_{\Sigma}^2} \Big) \ \diff t \ .
\end{align}
Now we can formulate the Atiyah-Patodi-Singer index theorem for manifolds where the cylindrical ends are attached along the
identity:  
\begin{theorem}[Atiyah-Patodi-Singer index theorem \cite{APS}]
The index of $\slashed{D }^+_{\widehat{M}}$ can be computed by the concrete formula: 
\begin{align}
\text{ind} \big(\slashed{D }^+_{\hat{M}^{d,1}} \big) =
  \int_{M}\,  K_{\rm AS}-\frac{1}{2} \, \Big( \eta \big(
  \slashed{D }_{\partial M}\big) +\dim \ker \big(\slashed{D
  }_{\partial_- M} \big)- \dim \ker \big(\slashed{D }_{\partial_+ M}\big) \Big) \ .
\label{APS Theorem}
\end{align} 
\end{theorem} 
\begin{remark}
The sign difference between the dimensions of the kernels in \eqref{APS Theorem} comes from the fact that we attach cylinders with opposite orientation to the incoming and outgoing boundary; this corresponds to a negative sign for the numbers $\alpha_i$ on the outgoing boundary $\partial_+M$ in the version of the Atiyah-Patodi-Singer index theorem given in \cite{MelrosebGeo}. 
\end{remark}

\section{Path integral description}\label{Sec: path integral parity}
In this section we will use the index theory for manifolds with boundaries
reviewed in the previous section to construct an invertible functorial quantum field
theory describing the parity-anomaly. 
For this we fix a compact Lie group $G$ and a unitary representation $\rho_G \colon G \longrightarrow 
\End(V)$ describing the matter content of the theory with anomaly.

The background fields consist of an orientation, a metric, a spin structure and a principal
$G$-bundle with connection described by the stack\footnote{This stack is not really of product form, since the spin structure depends on the choice of Riemannian metric.} 
\begin{align}
\mathscr{F}= \mathsf{Bun}_G^\nabla \, \times \, \mathsf{Met} \, \times \, \mathsf{Spin} \, \times \, \mathsf{Or} \ \ .
\end{align}
The stack $\F$-describes exactly the geometric structures required in Section~\ref{Sec: APS} to 
define the Dirac operator $\slashed{D}_M$. 
The field theory is defined on the category $\CobF$ constructed in 
Section~\ref{Sec: Non-extended field theories} where we make the additional assumption
that all structures are of product form on the collars; not just invariant. 
We make this assumption so that the index theory discussed in the previous 
section can be applied. By a slight abuse of notation we denote this 
category again by $\CobF$. The theory 
\begin{align}
\Aa_{\rm parity}^\zeta \colon \CobF \longrightarrow \fvs
\end{align}
depends on a complex parameter $\zeta \in \C^\times$. 
It assigns to every object $\Sigma \in \CobF$ the one dimensional 
vector space $\C$.  

To a regular morphism $M \colon \emptyset \longrightarrow \emptyset$, i.e.\ a 
closed manifold equipped with background fields we assign the linear map
\begin{align}
\Aa_{\rm parity}^\zeta(M)\colon \C \longrightarrow \C , \quad z\mapsto \zeta^{\text{ind}(\slashed{D }^+_{M})} \cdot z \ \ .
\end{align}
Let $M\colon \Sigma_- \longrightarrow \Sigma_+$ be regular morphisms in $\CobF$. Recall that $\widehat{M}$ is constructed from
$M$ by attached cylindrical ends of the form $(\infty,0]\times \Sigma_-$ 
and $[0,\infty)\times \Sigma_+$ to $M$ using the identification of $[0,\epsilon_-)\times \Sigma_-  \sqcup (-\epsilon_+,0] \times  \Sigma_+ $ with a collar of $\partial M$ which 
is part of the regular morphism $M$. Alternatively, we could first attach a mapping cylinder for the identification and then an infinite cylinder.

Having at hand the well-defined notion of an index for manifolds with boundaries in the form of Theorem~\ref{Thm:D is Fredholm}, we can now define
\begin{align}
\Aa_{\rm parity}^\zeta(M)\colon \C  \longrightarrow \C \ , \qquad z \longmapsto \zeta^{\text{ind}(\slashed{D }^+_{\widehat{M}})}\cdot z \ . 
\end{align}
We assign to a limit morphism $\phi$ the value of $\Aa_{\rm parity}^\zeta$ on a corresponding mapping cylinder; in order for $\Aa_{\rm parity}^\zeta$ to be well-defined, this construction must then be independent of the length of the mapping cylinder. We prove this as part of
\begin{theorem}\label{A is field theory}
 \ $\Aa_{\rm parity}^\zeta:\CobF \longrightarrow \fvs$ is an invertible quantum field theory.
\end{theorem}
\begin{proof}
The value of $\Aa_{\rm parity}^\zeta$ on a mapping cylinder is independent of its length, since the manifolds constructed by attaching cylindrical ends are $\mathscr{F}$-diffeomorphic. This proves that $\Aa_{\rm parity}^\zeta$ is well-defined on limit morphisms $\phi$.

If we cut a manifold $M$ along a hypersurface $H$ into two pieces $M_1$ and $M_2$, then from the Atiyah-Patodi-Singer index theorem \eqref{APS Theorem} we get\footnote{For this we need a cylindrical neighbourhood of $H$ on which all of the field content $\Fscr$ is of product form.} 
\begin{align}
\text{ind}\big(\slashed{D}^+_{\widehat M}\big)=\text{ind}\big(\slashed{D}^+_{\widehat M_1}\big)+\text{ind}\big(\slashed{D}^+_{\widehat M_2} \big) \ ,
\end{align} 
since the integration is additive and $\eta\big( \slashed{D }_{\Sigma}\big) = - \eta\big( \slashed{D }_{-\Sigma}\big)$, where $-\Sigma$ is the manifold $\Sigma$ with opposite orientation. The contributions from the boundary along which the cutting takes place cancel in (\ref{APS Theorem}) because of the opposite signs of the dimensions of the kernel of the boundary Dirac operator for incoming and outgoing boundaries.

For regular morphisms $M\colon \Sigma_1 \longrightarrow \Sigma_2$
and $M'\colon \Sigma_2 \longrightarrow \Sigma_3$, we can cut the
manifold $M'\circ M$ into $M$ and $M'$, with
the collar around $\Sigma_2$ removed and mapping cylinders
corresponding to the identification attached. This uses the
description of the gluing process in terms of mapping cylinders (see
Section~\ref{Sec: Non-extended field theories}), but as mentioned earlier, the index of
such pieces is the same as the index corresponding to a manifold where
the attachment is twisted by the identification of the collars with
cylinders. This implies 
\begin{align}
\Aa_{\rm parity}^\zeta\big(M'\circ M\big)= \Aa_{\rm parity}^\zeta\big(M'\big)\cdot \Aa_{\rm parity}^\zeta\big(M\big) \ .
\end{align}   
This proves that $\Aa_{\rm parity}^\zeta$ is a functor, which
is furthermore symmetric monoidal since all our constructions are
multiplicative under disjoint unions. The inverse functor (with respect to
the tensor product of field theories) is 
$\big(\Aa_{\rm parity}^\zeta\big)^{-1}=\Aa_{\rm parity}^{\zeta^{-1}}$.
\end{proof}
\begin{remark}
It may seem unnatural for $\Aa_{\rm parity}^\zeta$ to
assign the one-dimensional vector space $\C $ to every closed $n-1$-dimensional manifold $\Sigma$. 
Rather one would expect a complex line generated by all boundary conditions via an inverse limit construction as for example in~\cite{FreedQuinn,MonnierHamiltionianAnomalies}. Assigning $\C$ to every closed manifold is 
only possible due to the presence of canonical APS-boundary conditions related to the $L^2$-condition on the non-compact manifolds~$\widehat{M} $.
\end{remark}

\subsubsection*{Partition functions and symmetry-protected topological phases}

We turn our attention now to the partition function for a quantum field theory with parity anomaly.
According to Definition~\ref{Def: Anomalous partition function}, it is a
natural symmetric monoidal transformation
$Z_{\rm parity}^\zeta \colon \mbf1 \Longrightarrow \tr\Aa_{\rm parity}^\zeta$. This yields, for every closed $n-1$-dimensional manifold $\Sigma$ equipped
with background fields, a linear map \begin{align}
Z_{\rm parity}^\zeta(\Sigma)\colon \C \longrightarrow \Aa_{\rm parity}^\zeta(\Sigma)=\C \ \ .
\end{align}
 A linear map $Z_{\rm
  parity}^\zeta(\Sigma) \colon \C \longrightarrow \C$ can be canonically
identified with a complex number $Z_{\rm parity}^\zeta(\Sigma)\in \C$. There is no ambiguity in the
definition of the partition function as a complex number. The essence
of the parity anomaly, like most anomalies associated with the
breaking of a classical symmetry in quantum field theory, is the
lack of invariance of $Z_{\rm parity}^\zeta$ under limit morphisms $\phi$: the naturality of the partition function
implies that it transforms under gauge transformations $\phi $ by
multiplication with a 1-cocycle $\Aa_{\rm parity}^\zeta(\phi)\in\C^\times$; note that
in the present context `gauge transformations' also refer to
isometries and isomorphisms of the spinor bundle $S_{\Sigma}$. Since
$\Aa_{\rm parity}^\zeta$ depends only on topological data,
this multiplication is given by
\begin{eqnarray}
\Aa_{\rm parity}^\zeta(\phi) = \zeta^{{\rm ind}(\slashed{D}{}^+_{\mathfrak{M}(\Sigma,\phi)})}
\label{eq:Zparityphi}\end{eqnarray}
where $\mathfrak{M}(\Sigma, \phi)$ is the corresponding mapping torus constructed by identifying the boundary components of $ [0,1] \times \Sigma $ using $\phi$. 

We shall now illustrate how the functorial formalism of this section
connects with the more conventional treatments of the parity anomaly
in the physics literature, following~\cite{WittenFermionicPathInt} (see also~\cite{SeibergWitten});
indeed, what mathematicians call `invertible quantum field theories'
are known as `short-range entangled topological phases' to
physicists. A partition function with parity anomaly can be defined by
fixing its value on a representative for every gauge equivalence class of field configurations and applying \eqref{eq:Zparityphi} to determine all other values. 
Now consider the partition function with parity anomaly defined by 
\begin{align}
Z_{\rm parity}^{(-1)}\big(\Sigma \big)=\big|\text{det}\big(\slashed{D }_{\Sigma}\big)\big|  
\end{align} 
for an arbitrary chosen background $(A_{\Sigma},g_{\Sigma})$ in
every gauge equivalence class, where the definition of the determinant requires a suitable regularization.
Formally, this is the absolute value of the contribution to the path
integral measure from a massless Dirac fermion in $n-1$ dimensions coupled to a background $(A_{\Sigma},g_{\Sigma})$.
 There is an ambiguity in defining the phase of $Z_{\rm
   parity}^{(-1)}\big(\Sigma\big)$. Time-reversal (or space-reflection)
 symmetry forces $Z_{\rm parity}^{(-1)}\big(\Sigma\big)$ to be real. Here we
 chose the phase to make the partition function positive at the fixed representative.

From a physical perspective, having set the phase of the partition function at a fixed background $(A_{\Sigma},g_{\Sigma})$ we can calculate the phase at a gauge equivalent configuration $\phi(A_{\Sigma},g_{\Sigma})$, by following the path
\begin{equation}
(1-t)\, (A_{\Sigma},g_{\Sigma})+t\, \phi(A_{\Sigma},g_{\Sigma}) \ , \quad t\in[0,1]
\label{eq:spectralflow}\end{equation}
in the configuration space of the field theory, and changing the sign
every time an eigenvalue of the Dirac operator crosses through
zero. It is well-known that this spectral flow can be calculated by
the index of the Dirac operator on the corresponding mapping
cylinder~\cite{APS}. This physical intuition is formalised by the definition
above for $\zeta=-1$: The phase ambiguity is determined by requiring
the partition function to define a natural symmetric
monoidal transformation.

We can preserve gauge invariance by using Pauli-Villars regularization~\cite{WittenFermionicPathInt} leading to the gauge invariant partition function  
\begin{align}
Z_{\rm parity}\big(\Sigma\big)=\big|\text{det}\big(\slashed{D }_{\Sigma}\big)\big| \ (-1)^{\eta (\slashed{D }_{\Sigma})/2} \ .
\end{align}
The global parity anomaly is due to the fact that the fermion path integral is in general not a real number, whereas classical orientation-reversal (or `parity') symmetry, which acts by complex conjugation on path integrals, would imply that the path integral is real. 
Hence, the parity anomaly can be understood as the result that it is not possible to quantize 
the theory in such a way that gauge symmetry and parity symmetry are preserved. 

We can now apply the general framework from Section~\ref{Sec: Anomaly inflow} to cancel the parity anomaly: We combine bulk and boundary degrees of freedom by introducing for the bulk fields the action 
\begin{align}
\label{EQ: Action topological isulator bulk fields}
 S_{\rm bulk}\big(M\big) = \ii \pi \, \int_{M}\, K_{\rm AS} \ ,
\end{align} 
where $M$ is a regular morphism from $\Sigma$ to $\emptyset$, i.e.\ $\partial M=-\Sigma$. Then after integrating out the boundary fermion fields, the contribution to the path integral measure for the combined system is given by
\begin{align}
\begin{split} 
Z_{\rm bb}\big(M\big) &  = \Za^{(-1)}_{\rm parity}(M)[Z^{(-1)}_{\rm parity}(\Sigma)] 
\\[4pt] & = \big|\text{det}\big(\slashed{D }_{\Sigma}\big)\big| \ (-1)^{\text{ind}(\slashed{D }^+_{\hat M})}
\\[4pt] & = 
\big|\text{det}(\slashed{D }_{\Sigma})\big| \, \exp \Big(\,\frac{\ii\pi}2\, \eta \big(\slashed{D }_{\Sigma}\big) - \ii\pi\, \int_{M}\, K_{\rm AS}\, \Big)  \\
&= \big|\text{det}\big(\slashed{D }_{\Sigma}\big)\big| \ (-1)^{ \eta (\slashed{D }_{\Sigma})/2} \ \e^{-S_{\rm bulk}(M)}  \ ,
\end{split} 
\end{align}
where we used that $\big|\text{det}\big(\slashed{D }_{\Sigma}\big)\big|$ is zero 
as soon as the kernel of  $\slashed{D }_{\Sigma}$ is non-trivial and the Atiyah-Patodi-Singer index formula \eqref{APS Theorem}. This expression is real. Thus the combined bulk-boundary system is invariant under orientation-reversal and gauge transformations, since now its path integral is real, due to `anomaly inflow' from the bulk to the boundary. In particular, the non-anomalous partition function of the combined system
requires the full $n$-dimensional quantum field theory $\mathcal{A}^{(-1)}_{\rm parity}$, rather than just the truncation $\mathsf{tr}\mathcal{A}^{(-1)}_{\rm parity}$ in which the original partition function $Z^{(-1)}_{\rm parity}$ lives, to be well defined. 
Looking at this from a different perspective, we see that the existence of an effective long wavelength action \eqref{EQ: Action topological isulator bulk fields} for the bulk gauge and gravitational fields implies the existence of gapless charged boundary fermions with an anomaly cancelling the anomaly of the bulk quantum field theory under orientation-reversing transformations. 

This example provides a simple model for the general feature of some topological states of matter: Symmetry-protected topological phases in $n$ dimensions are related to global anomalies in $n-1$ dimensions. In the simplest case $n=2$, the quantum mechanical time-reversal anomaly on the $0+1$-dimensional boundary is encoded by the $1+1$-dimensional symmetry-protected topological phase in the bulk whose topological response action \eqref{EQ: Action topological isulator bulk fields} evaluates to $\ii\pi\,\Phi$, where $\Phi$ is the magnetic flux of the background gauge field through $M$. This sets the two-dimensional $\theta$-angle equal to $\pi$.

For the $n=4$ example of the time-reversal (or space-reflection)
invariant $3+1$-dimensional fermionic topological insulator with
$2+1$-dimensional boundary~\cite{WittenFermionicPathInt}, the integral
of the Atiyah-Singer index density $K_{\rm AS}$ in four dimensions
yields the sum of the instanton number $I$ of the background gauge
field and a gravitational contribution related to the signature
$\sigma$ of the four-manifold $M$~\cite{NashBook}. For the cancellation of the parity anomaly we had to introduce the term $\ii\pi\, I$ in the action, which is the anticipated statement that the $\theta$-angle parameterising the axionic response action is equal to $\pi$ inside a topological insulator. The bulk-boundary correspondence discussed above then resembles the well-known situation from three-dimensional Chern-Simons theory, to which the bulk theory reduces on $\partial M=\Sigma$~\cite{Niemi,AlvarezGaume}.

The present formalism generalises this perspective to systematically
construct quantum field theories with global parity symmetry that
characterise gapless charged fermionic boundary states of certain
symmetry-protected topological phases of matter in all higher even
dimensions $n\geq6$. Indeed, the anomaly of a quantum field theory in
$n=2k$ dimensions involving an action that integrates the
Atiyah-Singer index density $K_{\rm AS}$ reduces on the boundary $\partial M=\Sigma$ to coupled combinations of gauge and gravitational Chern-Simons type terms. The bulk action \eqref{EQ: Action topological isulator bulk fields} will now also involve couplings between gauge and gravitational degrees of freedom through intricate combinations of Chern and Pontryagin classes, such that the
bulk symmetry-protected topological phase completely captures the
parity anomaly of the boundary theory. Some examples of such mixed
gauge-gravity phases can be found e.g.\ in~\cite{Wang:2014pma}.

\section{Index theory part II: manifolds with corners}\label{Sec: Index2}
In the remainder of this chapter we extend the field theory $\mathcal{A}^{(-1)}_{\rm parity}$ to also capture 
the parity anomaly at the Hamiltonian level. It should not come as a surprise that this involves
index theory on manifolds with corners, which we review in this section.
We will present the index theory already tailored 
to the construction of an extended field theory in the next section, i.e.\ we formulate
the results for regular 2-morphisms in the bicategory $\ECobF$ again assuming
that all structures are of product form on the collars. Furthermore, we make the technical assumption
that the index of the Dirac operator on all corners vanishes.    
We use an index theorem for manifolds with corners based on 
b-geometry. Before stating the theorem we provide some background
on b-geometry.  

\subsection{b-Geometry}\label{Sec: b-geometry}
b-geometry (for `boundary geometry') is concerned with the study of geometric structures on manifolds with corners which can be singular at the boundary. 
We fix an $n$-dimensional $\langle 2\rangle$-manifold $M$ and an ordering of its connected faces $\{H_1, \dots, H_k\}$.
The central objects in b-geometry are b-vector fields. These are vector fields which are tangent to all boundary hypersurfaces. 
We denote by $\mathrm{Vect}_{\mathrm{b}}(M)$ the projective $C^\infty (M)$-module of b-vector fields. Then $\mathrm{Vect}_{\mathrm{b}}(M)$ is closed under the Lie bracket of vector fields.  
By the Serre-Swan theorem, the b-vector fields are naturally sections of the \emph{b-tangent bundle} with fibres
\begin{align}
^{\mathrm{b}}T_x M := \mathrm{Vect}_{\mathrm{b}}(M) \setminus \mathcal{I}_x(M) \cdot \mathrm{Vect}_{\mathrm{b}}(M) \ , 
\end{align}
where $\mathcal{I}_x(M)= \{f\in {C}^\infty(M) \mid f(x)=0 \}$ is the ideal of functions vanishing at $x\in M$. This allows us to define arbitrary b-tensors as in classical differential geometry. The inclusion $\mathrm{Vect}_{\mathrm{b}}(M) \hookrightarrow \mathrm{Vect}(M)$ induces a natural vector bundle map $\alpha_{\mathrm{b}} \colon ^{\mathrm{b}}TM \longrightarrow TM$.

The structures introduced so far can be summarized by saying that $(^{\mathrm{b}}TM,\alpha_{\mathrm{b}})$ is a boundary tangential Lie algebroid. The b-tangent bundle is isomorphic to the tangent bundle in the interior of $M$, and $(M, \mathrm{Vect}_{\mathrm{b}}(M))$ is an example of a manifold with Lie structure at infinity~\cite{LieManifolds}.
   
Using a set of boundary defining functions $\{x_i\}_{i=1\dots k}$, the Lie algebra $\mathrm{Vect}_{\mathrm{b}}(M)$ is locally spanned near a point $x\in H_i$ of index~$1$ by $\{ x_i\, \partial_{x_i}, \partial_{h_1}, \dots , \partial_{h_{n-1}}\}$, where $\{h_l\}_{l=1}^{n-1}$ is a local coordinate system for $H_i$. 
In a neighbourhood of $x\in H_i\cap H_j$, $i\neq j$, we can form a basis given by $\{x_i\, \partial_{x_i},x_j\, \partial_{x_j},\partial_{y_1}, \dots , \partial_{y_{n-2}} \}$, where $\{y_l\}_{l=1}^{n-2}$ is a local coordinate system on $Y_{ij}=H_i\cap H_j$. 
The dual basis for the b-cotangent bundle $^{\mathrm{b}}T^*M$ is denoted by $\big\{ \tfrac{\diff x_i}{x_i},\tfrac{\diff x_j}{x_j}, \diff y_1 ,\dots , \diff y_{d-2} \big\}$. 

A \emph{b-metric} $g$ is now simply a metric on the vector bundle $^{\mathrm{b}} TM$ over $M$. This defines an ordinary metric in the interior of $M$. The general expression in local coordinates near a corner point is
\begin{align}
g= \sum_{i,j=0,1}\, a_{ij} \ \frac{\diff x_i}{x_i} \otimes \frac{\diff x_j}{x_j} + 2 \, \sum_{i=0,1} \ \sum_{j=1}^{n-2}\, b_{ij} \ \frac{\diff x_i}{x_i} \otimes \diff y_j + \sum_{i,j=1}^{n-2}\, c_{ij} \ \diff y_i \otimes \diff y_j \ .
\end{align}
A b-metric $g$ is \emph{exact} if there exists a set of boundary defining functions $x_i$ such that it takes the form
\begin{align}
g= \begin{cases}
\displaystyle \frac{\diff x_i}{x_i} \otimes \frac{\diff x_i}{x_i} +h_{H_i} & \text{ near } H_i \ , \\[10pt]
\displaystyle  \frac{\diff x_i}{x_i} \otimes \frac{\diff x_i}{x_i}+ \frac{\diff x_j}{x_j} \otimes \frac{\diff x_j}{x_j} + h_{H_i\cap H_j} & \text{ near } H_i\cap H_j \ ,
\end{cases}
\end{align}
where $h_Y$ denotes a metric on $Y$.

A \emph{b-differential operator} is an element of the universal enveloping algebra of $\mathrm{Vect}_{\mathrm{b}}(M)$, the collection of which act naturally on $C^\infty (M)$. 
A \emph{b-differential operator $D\in \text{Diff}^{\,k}_{\mathrm{b}}(M,E_1,E_2)$ of order $k$} between two vector bundles $E_1$ and $E_2$ over $M$ is a smooth fibre-preserving map, which in any local trivialisations of $E_1$ and $E_2$ is given by a matrix of linear combinations of products of up to $k$ b-vector fields. 
Most concepts from differential geometry such as connections, symbols and characteristic classes can be generalized to the b-geometry setting. 

Since exact b-metrics are singular at the boundary it is necessary to define a renormalised b-integral. Heuristically, the problem stems from the fact that the integral $\int_0^1\, \tfrac{\diff x}{x}$ is divergent. The cure for this is to multiply with $x^z$ for $\text{Re}(z)> 0$.
\begin{lemma}(\cite[Lemma 4.1]{Loyaindex}) \ 
Let $M$ be a manifold with corners and an exact b-metric~$g$. Then for all $f\in {C}^\infty(M)$ and $z\in \C$ with $\operatorname{Re}(z)> 0$, the integral
\begin{align}
F(f,z):=\int_M\, x^z\, f\ \diff g
\end{align}
exists and extends to a meromorphic function $F(f,z)$ of $z\in\C$.
\end{lemma}

\begin{definition}
Let $M$ be a manifold with corners and an exact b-metric $g$. The \emph{b-integral} of a function $f\in C^\infty(M)$ is 
\begin{align}
\label{Definition b-integral}
{{}^{\mathrm{b}}}\!\!\int_M\, f\ \diff g = \operatorname{Reg}_{z=0} \ F(f,z)\ .
\end{align}
\end{definition}

This allows us to define the \emph{b-trace} of a pseudo-differential operator $D$ in terms of its kernel $D(x,y)$ as
\begin{align}
^{\mathrm{b}}\text{Tr}(D)= {{}^{\mathrm{b}}}\!\!\int_M\, \text{tr} \big(D(x,x)\big)\ \diff g(x) \ , 
\end{align}
where the trace $\mathrm{tr}$ is over the fibres of the vector bundle on which $D$ acts.

We conclude by describing the relation between $\langle 2\rangle$-manifolds with exact b-metrics and $\langle 2\rangle$-manifolds equipped
with $\F$-background fields. To a manifold $M$ equipped with a Riemannian metric we attach infinite cylindrical ends $H_i\times (-\infty,0]$ to the boundary hypersurfaces and $Y_{ij}\times (-\infty,0]^2$ to the corners. The coordinate transformation $x_i= \e^{t_i}$ for $t_i\in(-\infty,0]$ maps this non-compact manifold to the interior of a manifold $X$ with corners. The product metric on the cylindrical ends induces a b-metric on $X$, since $\diff t_i\otimes\diff t_i = \tfrac{\diff x_i}{x_i} \otimes \tfrac{\diff x_i}{x_i}$. For this reason one can view the study of manifolds with exact b-metrics as the study of manifolds with cylindrical ends. For regular 2-morphisms in $\ECobF$ it is again natural to attach the collars using the $\F$-diffeomorphisms 
which are part of the regular 2-morphism.

\subsection{Loya-Melrose index theorem for manifolds with corners}  
We have seen in Section~\ref{Sec: path integral parity} that it is helpful to attach mapping cylinders to a manifold encoding the data of the identification of the boundary components with lower-dimensional objects. In the extended case we also need mapping boxes at the corners. 
Let $Y_i$, $i=1,2,3, 4$ be four closed manifolds equipped with $\mathscr{F}$-fields of product form $\mathsf{f}_i \in \mathscr{F}\big((-\epsilon_1, \epsilon_1)^2 \times Y_i  \big)$, and a diagram of $\mathscr{F}$-diffeomorphisms $\varphi_{ij}$:
\begin{equation}
\begin{tikzcd}
Y_1 \ar[d,swap,"\varphi_{13}"] \ar[r,"\varphi_{12}"] & Y_2\ar[d,"\varphi_{24}"] \\
Y_3 \ar[r,swap,"\varphi_{34}"] & Y_4 
\end{tikzcd}
\end{equation}
Then the \emph{mapping box} $\mathfrak{M}(Y, \varphi)$ of length
$\epsilon$ corresponding to this data is constructed by gluing
$\big[0, \tfrac{3}{4}\, \epsilon\big)^2 \times Y_1 $, $\big(\tfrac{1}{4}\, \epsilon,\epsilon\big] \times \big[0,
\tfrac{3}{4}\, \epsilon\big) \times Y_2
$, $ \big[0, \tfrac{3}{4}\,
\epsilon\big)\times  \big(\tfrac{1}{4}\, \epsilon,\epsilon\big] \times Y_3$ and
$ \big( \tfrac{1}{4}\, \epsilon, \epsilon\big]^2 \times Y_4$ along
$\varphi_{ij}$. Using descent we can construct an element $\mathsf{f}
\in \mathscr{F}\big(\mathfrak{M}(Y, \varphi)\big)$.
   
Given a regular 2-morphism $M$ from $\Sigma_1\colon S_- \longrightarrow S_+$ to $\Sigma_2\colon S_-\longrightarrow S_+$ in $\ECobF$, by definition $M$ comes with collars $N_- \cong [0,\epsilon_1) \times \Sigma_1  $ and $N_+ \cong (-\epsilon_1,0] \times \Sigma_+$.
We first attach mapping cylinders of a fixed length $\epsilon \in \R_{>0}$ to $\Sigma_1$, $\Sigma_2$ and the $0$-boundary. In a second step we attach mapping boxes of length $\epsilon$ to the corners of $M$. We denote this new manifold by $M'$ (see Figure \ref{Fig: Extended Indextheorem}). For this to be well-defined we need compatibility of all collars involved.
The new manifold has four distinct boundaries which we denote by
$\Sigma'_1$, $\Sigma'_2$, $C\big(S_-\big)= \big[-\epsilon-\tfrac{1}{2}\, \epsilon_1, \epsilon+\tfrac{1}{2}\, \epsilon_1\big] \times S_- $ and $C\big(S_+\big)=  \big[-\epsilon-\tfrac{1}{2}\, \epsilon_1, \epsilon+\tfrac{1}{2}\, \epsilon_1\big] \times S_+ $. 
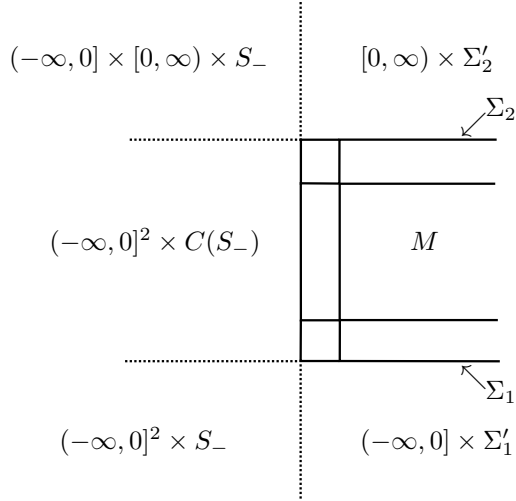
\begin{figure}
\footnotesize
\begin{center}
\begin{overpic}[width=6cm,
scale=1]{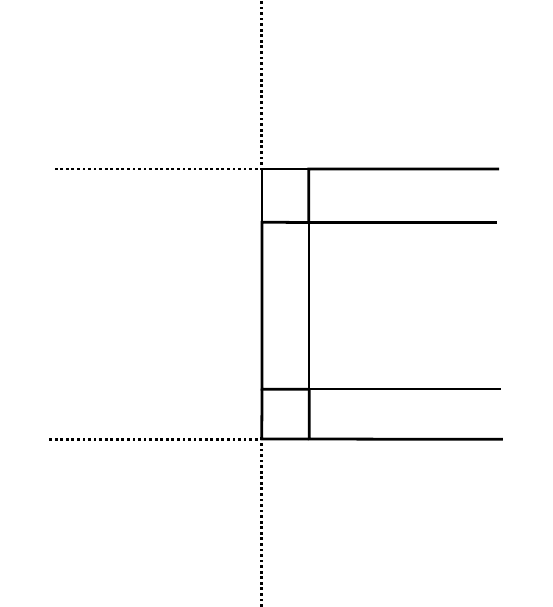}
\put(65,50){$M$}
\put(75,23){$\nwarrow$}
\put(80,20){$\Sigma_1$}
\put(75,74){$\swarrow$}
\put(80,77){$\Sigma_2$}
\put(-5,10){$ (-\infty,0]^2 \times S_-$}
\put(-7,50){$ (-\infty,0]^2 \times C(S_-)$}
\put(-15,87){$(-\infty,0]\times [0,\infty) \times S_- $}
\put(55,10){$ (-\infty,0] \times \Sigma'_1$}
\put(55,87){$[0, \infty) \times \Sigma'_2 $}
\end{overpic}
\end{center}
\caption{\small Illustration of the construction of $\hat{M}'$ near $\Sigma_-$.}
\label{Fig: Extended Indextheorem}
\end{figure} 
We can now attach cylindrical ends to $M'$. For this, we first define
\begin{align}
\begin{split} 
\hat{M}'{}^{\circ}= M'\ \sqcup_{\partial M'}\ & \big(  (-\infty, 0] \times  \Sigma'_1  \sqcup  [0, \infty ) \times  \Sigma'_2 \\
& \qquad \sqcup (-\infty,0] \times C(S_-)  \sqcup  [0,\infty) \times C(S_+) \big) \ ,
\end{split} 
\end{align} 
where we use the collars to glue the manifolds and extend all fields as products. Then $\hat{M}'{}^{\circ}$ is a non-compact manifold with corners. Further gluing (see Figure~\ref{Fig: Extended Indextheorem}) produces 
\begin{align}
\begin{split}
\hat{M}' = \hat{M}'{}^{\circ}\ &\sqcup_{\partial \hat{M}'{}^{\circ}}\ \big( (-\infty,0]^2 \times S_-  \sqcup (-\infty,0]\times [0,\infty)  \times S_-  \\ & \qquad \qquad \qquad \sqcup   [0 , \infty)^2 \times S_+ \sqcup (-\infty, 0] \times [0, \infty) \times  S_+  \big)
\end{split}  
\end{align}
with all structures extended as products. As in the case of manifolds with boundaries, the Dirac operator $\slashed{D }_{\hat{M}'}$ is not Fredholm in general, and one can prove analogously that $\slashed{D }_{\hat{M}'}$ is Fredholm if and only if the induced Dirac operators on the corners and boundaries are invertible~\cite{LoyaMelrose}. 

When the kernel of the corner Dirac operator is non-trivial, we have
to add again a mass perturbation~\cite{Loyaindex}. The induced twisted spinor bundle over $Y= S_+  \sqcup -S_- $ decomposes into spinors of positive and negative chirality. We pick a unitary self-adjoint isomorphism $T_i\colon \ker(\slashed{D }_{Y_i}) \longrightarrow \ker(\slashed{D }_{Y_i})$, for every connected component $Y_i$ of the corner $Y$, which is odd with respect to the $\Z_2$-grading of the spinor bundle; this is possible since the index of $\slashed{D}_Y$ is $0$ by assumption. We define
\begin{align}
 T_\pm= \bigoplus_{i=1}^n\, T_{\pm,i} \qquad \text{ and } \qquad T=T_-\oplus T_+ \ ,
\end{align} 
where $T_{\pm,i}:\ker\big(\slashed{D }_{S_{\pm,i}}\big)
\longrightarrow \ker\big(\slashed{D }_{S_{\pm ,i}}\big)$. Now the
operator $\slashed{D }_Y - T$ is invertible on $Y$. This suggests extending
$T$ to an operator $\hat{T}$ on $\hat{M}'$ such that the
massive Dirac operator $\slashed{D }_{\hat{M}'{}}-\hat{T}$ is Fredholm on weighted Sobolev spaces. A concrete construction of $\hat{T}$ can be found in \cite[Section 2.3]{LoyaMelrose}\footnote{$\hat{T}$ corresponds to $-S$ constructed in \cite[Section 2.3]{LoyaMelrose}, where we choose the same operators for the two remaining corners. Furthermore, the construction depends 
on an ordering of
the faces, which we chose such that 0-boundary is smaller than the 1-boundary.}, from which it is clear that $\hat{T}$ is independent of the length $\epsilon$ of the attached mapping cylinders and boxes up to a choice which is involved in the
construction and does not change the index theory. When we choose for every boundary component a small mass $\alpha_i$ as in Section~\ref{Sec: APS}, then 
\begin{align} 
\slashed{D }_{\hat{M}'{}}^+ -\hat{T}^+ \colon \e^{\alpha\cdot s}H^1\big(\hat S^{G+}_{\hat M'{}}\big)\longrightarrow \e^{\alpha\cdot s}L^2\big(\hat S^{G-}_{\hat M'{}}\big)
\end{align}  
is a Fredholm operator on weighted Sobolev spaces~\cite[Theorem
2.6]{LoyaMelrose}.
We restrict ourselves to a description of the corresponding index theorem on manifolds which are of the form $M'$ for a regular 2-morphism $M$ in $\ECobF$; the more general version can be found in~\cite[Theorem~6.13]{LoyaMelrose}.  

To define the $\eta$-invariant on a manifold $\Sigma$ with boundary
we proceed as in Section~\ref{Sec: APS} and define
$\hat{\Sigma}$ by attaching cylindrical ends to $\Sigma$. In
general, the Dirac operator $\slashed{D}_{\hat{\Sigma}}$ has a
continuous spectrum, so we have to use the expression \eqref{Equation: Reformulation eta invariant} to define the $\eta$-invariant as an integral
\begin{align}
{}^{\mathrm{b}}\eta \big(\slashed{D }_{\hat{\Sigma}}\big)= \frac{1}{\sqrt{\pi}} \, \int_0^\infty \, t^{-{1}/{2}} \ {}^{\mathrm{b}}\mathrm{Tr}\Big( \slashed{D} _{\hat{\Sigma}}\, \e^{-t\,\slashed{D }_{\hat{\Sigma}}^2} \Big) \ \diff t \ , 
\end{align}
where we have to replace the usual trace by the b-geometric trace (see
Section \ref{Sec: b-geometry}) because its argument is not a
trace-class operator on $\hat{\Sigma}$ in general. There are other
ways of defining $\eta $-invariants for manifolds with boundaries
using appropriate boundary conditions \cite{DaiFreed,
  LeschWojciechowski, MuellerEtaInvariant}. The
$^{\mathrm{b}}\eta$-invariant agrees with the canonical boundary
conditions on spinors induced by scattering Lagrangian subspaces, which we describe in Remark~\ref{Remark: Relation eta Invariants}.
 
There is a further contribution to the index theorem coming from the
corners. We define for $\Sigma'_i$, $i=1,2$ the scattering
Lagrangian subspace
\begin{align}
\begin{split} 
\Lambda_{C_i} = \Big\lbrace \, \lim_{s_- \rightarrow -\infty}& \, \Psi(y_-,s_-)\oplus  \lim_{s_+ \rightarrow \infty}\, \Psi(y_+,s_+) \ \Big| \\ 
&\ \Psi\in C^\infty\big(\hat{S}^G_{\hat{\Sigma}'_i}\big) \cap \ker\big(\slashed{D } _{\hat{\Sigma}'_i}\big) \ \text{bounded} \, \Big\rbrace \subseteq \ker\big(\slashed{D }_Y\big) 
\end{split} 
\end{align}
where $s_-\in(-\infty,0]$, $s_+\in[0,\infty)$ and $y_\pm\in S_\pm$.
The set $\Lambda_{C_i}\subset \ker \big(\slashed{D }_Y\big)$ is a Lagrangian subspace of $\ker \big(\slashed{D }_{\partial \Sigma_i}\big)$ with respect to the symplectic form
$\omega (\, \cdot\, , \, \cdot\, )= (\iu \Gamma \, \cdot\, ,\, \cdot\, )_{L^2}$~\cite{MuellerEtaInvariant}, where $\Gamma$ is the chirality operator
on the spinor bundle over the corners. 
We define an odd unitary self-adjoint isomorphism $C_i$ of $\ker
\big(\slashed{D }_{Y}\big)$, called the scattering matrix of
$\slashed{D } _{\hat{\Sigma}'_i}$, by setting $C_i=\id$ on
$\Lambda_{C_i}$ and $C_i=-\id$ on $\Lambda_{C_i}^\perp$. We denote by
$\Lambda_{T}\subset \ker \big(\slashed{D }_Y\big)$ the $+1$-eigenspace
of $T$. Following~\cite{LeschWojciechowski,Bunke}, we introduce the
`exterior angle' between Lagrangian subspaces by the spectral formula
\begin{align}
\label{EQ: Definition of m( , )}
\mu(\Lambda_{T}, \Lambda_{C_i})=-\frac{1}{\pi} \ \sum_{\substack{\e^{\ii\theta}\in \mathrm{spec}(-T^-\,C_i^+)\\ -\pi<\theta <\pi}} \, \theta \ ,
\end{align} 
where the grading is with respect to the $\Z_2$-grading of the twisted spinor bundle over the corners. 

With this notation, we can now formulate the index theorem for
manifolds of the form $M'$ as:
\begin{theorem}
\label{Theorem extended index theorem}
Let $M$ be a regular 2-morphism from $\Sigma_1\colon S_- \longrightarrow S_+$ to $\Sigma_2\colon S_-\longrightarrow S_+$ in $\ECobF$. Then
\begin{align}
\label{Index theorem corners}
\begin{split}
\mathrm{ind} \big(\slashed{D }_{\hat{M}'}^+ -\hat{T}^+
\big) = &\int_{M'} \, K_{\rm AS} -\frac{1}{2}\, \Big( -\,
^{\mathrm{b}}\eta \big(\slashed{D}_{\hat{\Sigma}'_1}
\big)+{}^{\mathrm{b}}\eta \big(\slashed{D }_{\hat{\Sigma}'_2}
\big) \\ & +\, \dim \ker \big(
\slashed{D}_{\hat{\Sigma}'_1} \big) -\dim \ker
\big(\slashed{D}_{\hat{\Sigma}'_2} \big) \\ & +\,
\dim(\Lambda_{T} \cap \Lambda_{C_1}) -\dim(\Lambda_T \cap
\Lambda_{C_2} ) + \mu(\Lambda_{T}, \Lambda_{C_1}) + \mu(\Lambda_{ T},
\Lambda_{C_2}) \, \Big)  .
\end{split}
\end{align}
\end{theorem}
\begin{remark}
The extra corner contributions in the last line of \eqref{Index theorem corners} to the usual
(b-geometric) Atiyah-Patodi-Singer formula \eqref{APS Theorem} can be
understood as follows. Let $\Sigma$ be a regular
1-morphism in the bicategory $\ECobF$. Then for every $T\in
\mathrm{End}_\C\big(\ker (\slashed{D }_{\partial \Sigma}) \big)$ as
above we can relate the spectral data of the massive Dirac operator on
$\hat \Sigma$ to their massless counterparts as
\begin{align}
{}^{\rm b}\eta\big(\slashed{D}_{\hat \Sigma} - \hat T \big)&= {}^{\rm b}\eta\big(\slashed{D}_{\hat \Sigma} \big)+\mu(\Lambda_{T},
  \Lambda_{C}) \ , \\[4pt]
\dim \ker \big(\slashed{D}_{\hat \Sigma} - \hat T \big) &= \dim
                                                             \ker
                                                             \big(\slashed{D}_{\hat
                                                             \Sigma}\big)+\dim(\Lambda_T\cap\Lambda_C)
                                                             \ ,
\end{align}
where $\Lambda_C$ is the scattering Lagrangian subspace for
$\Sigma$. 
\end{remark}
\begin{remark}\label{Remark: Relation eta Invariants}
We describe the relation between $^{\rm b}\eta $-invariants and
$\eta $-invariants with boundary conditions~\cite{RelationEtaInvariants}. We denote by $\Pi_+$ the
projection onto the space spanned by the positive eigenspinors of
$\slashed{D }_{\Sigma}$, and by $\Pi_T$ the projection onto the positive eigenspace $\Lambda_T$. This
allows us to define a Dirac operator $\D_T$, which coincides with $\slashed{D
}_{\partial \Sigma}$, on the domain
\begin{align}
\big\{
\Psi\in H^1\big(\hat S_{\hat \Sigma }\big) \ \big| \ (\Pi_+ + \Pi_T)\Psi\big|_{\partial \Sigma}=0
\big\} \ .
\end{align}
The operator $\D_T$ is self-adjoint and elliptic for all $T$. It is shown in \cite[Theorem~1.2]{RelationEtaInvariants} that
\begin{align}\label{EQ: b-eta= eta with boundary condition}
\eta(\D_T)= {}^{\rm b} \eta \big(\slashed{D}_{\hat \Sigma} \big)+\mu(\Lambda_{T},
  \Lambda_{C}) \ ,
\end{align}  
so that we can combine the ${}^{\rm
  b}\eta$-invariant and the exterior angle $\mu$ in \eqref{Index
  theorem corners} into an $\eta$-invariant for a Dirac operator with suitable boundary conditions induced by the Lagrangian subspace $\Lambda_T$.    
\end{remark}

\begin{proof}[Proof of Theorem \ref{Theorem extended index theorem}] 
From the general index theorem for manifolds with corners \cite[Theorem 6.13]{LoyaMelrose} we get 
\begin{align}
\begin{split} 
&\text{ind} \big(\slashed{D }_{\hat{M}'}^+ -\hat{T}^+
  \big) \\ &  = \int_{M'}\, K_{\rm AS} -\frac{1}{2}\, \Big( \, {}^{\rm
          b}\eta \big(\slashed{D}_{-\hat{\Sigma}'_1} \big)+{}^{\rm
          b}\eta \big(\slashed{D }_{\hat{\Sigma}'_2} \big) +
          {}^{\rm b}\eta \big(\slashed{D}_{\hat C(S_-)} \big)+
          {}^{\rm b}\eta \big(\slashed{D}_{\hat C(S_+)} \big)\\ 
&  + \dim \ker \big(\slashed{D}_{\hat{\Sigma}'_1}
  \big)+\dim(\Lambda_{T} \cap \Lambda_{C_1})-\dim \ker
  \big(\slashed{D}_{\hat{\Sigma}'_2} \big)-\dim(\Lambda_T \cap
  \Lambda_{C_2}) \\
&  +\dim \ker \big(\slashed{D}_{\hat{C}(S_-)} \big)
  +\dim(\Lambda_{\ii\Gamma\, T_-} \cap \Lambda_{C_-}) -\dim \ker
  \big(\slashed{D}_{\hat{C}(S_+)} \big) -\dim(\Lambda_{\ii\Gamma\,
  T_+} \cap \Lambda_{C_+}) \\
&  +\mu(\Lambda_{T}, \Lambda_{C_1})+ \mu(\Lambda_{T}, \Lambda_{C_2})
  +\mu(\Lambda_{\ii\Gamma\, T_+}, \Lambda_{C_+}) +
  \mu(\Lambda_{\iu \Gamma\, T_-}, \Lambda_{C_-}) \, \Big)
  \end{split}
\end{align}
where $\Lambda_{C_\pm}$ are the scattering Lagrangian subspaces for
$C(S_\pm)$, respectively. We can calculate the contributions
from the boundaries $C(S_\pm)$ explicitly and show that they all
vanish:
Attaching infinite cylindrical ends to $C(S_\pm)$ leads to the manifolds $ (-\infty,\infty) \times S_\pm $.
The Dirac operator on the manifold $(-\infty, \infty) \times S_\pm $ of odd dimension $n-1$ is given by \eqref{Equation: Dirac operator on a product}:
\begin{align}
\label{EQ: Proof index theorem}
\slashed{D}_{ {S_\pm} \times (-\infty, \infty)  }= \sigma_t\,
  \big(\slashed{D}_{S_\pm} + \partial_t \big) \ , 
\end{align}
where $t\in(-\infty,\infty)$. 
We are interested in the dimension of the space of harmonic
spinors $\Psi(y_\pm,t)$. By elliptic regularity there exists a basis of smooth
sections. Multiplying \eqref{EQ: Proof index theorem} with
$\sigma_t^{-1}$, we get
\begin{align}
\big(\slashed{D}_{S_\pm}+ \partial_t \big)\Psi(y_\pm,t)=0 \ .
\end{align} 
Using separation of variables $\Psi(y_\pm,t)=\psi(y_\pm)\, \alpha(t)$, this
equation reduces to a pair of equations
\begin{align}
\label{EQ: Dirac operator on cylinder Proof index theorem}
\slashed{D}_{S_\pm} \, \psi(y_\pm)= \lambda \, \psi(y_\pm)  \qquad
  \mbox{and} \qquad
\frac{\diff \alpha(t)}{\diff t}=-\lambda \, \alpha(t) \ , 
\end{align}
for an arbitrary constant $\lambda $ which must be real since $\slashed{D}_{S_\pm} $ is an elliptic operator. 
The second equation has solution (up to a constant) $\alpha(t)= \e^{-\lambda\, t}$, and
we finally see that there are no non-zero square-integrable spinors
$\Psi(y_\pm,t)$ with eigenvalue $0$.
Hence, the contributions from the terms $\dim \ker \big(\slashed{D}_{\hat{C}(S_\pm)}\big)$ are $0$.

A solution of \eqref{EQ: Dirac operator on cylinder Proof index
  theorem} is bounded if and only if $\lambda=0$, and so the
scattering Lagrangian subspace takes the form 
\begin{align}
\Lambda_{C_\pm}= \Delta\big(\ker (\slashed{D}_{S_\pm}) \big)=
  \big\{ \psi \oplus \psi \ \big| \ \psi \in \ker
  \big(\slashed{D}_{S_\pm} \big) \big\} \subset \ker
  \big(\slashed{D}_{S_\pm} \big) \oplus \ker
  \big(\slashed{D}_{-{S_\pm}} \big) \ .
\end{align}
This implies that 
\begin{align}
\dim(\Lambda_{\ii\Gamma \, T_\pm} \cap \Lambda_{C_\pm})= 0 \ ,
\end{align}
since the chirality operator $\Gamma$ on the outgoing and ingoing
boundaries differs by a sign while $T_\pm$ is the same over both boundaries.

Finally, by Remark~\ref{Remark: Relation eta Invariants}, ${}^{\rm
  b}\eta \big(\slashed{D}_{\hat{C}(S_\pm)}\big)+
\mu(\Lambda_{\ii\Gamma \, T_\pm}, \Lambda_{C_{\pm}})$ is the
$\eta$-invariant on a cylinder with identical boundary conditions at
both ends, which vanishes by~\cite[Theorem~2.1]{LeschWojciechowski}. 
\end{proof}

For later use, we derive here a formula for the index of a 2-morphism
under cutting. For this, we first have to study the behaviour of the
various quantities in the index formula \eqref{Index theorem corners} under orientation-reversal.
\begin{lemma}\label{Lemma: Properties terms in the index theorem}
Let $\Sigma$ be a regular 1-morphism in $\ECobF$ with fixed boundary condition
$T \in \mathrm{End}_\C \big(\ker(\slashed{D}_{\partial \Sigma}) \big)$ as
above. If we reverse the orientation of $\Sigma$, then $T$ still
defines a suitable boundary condition of $\slashed{D}_{\partial (-\Sigma)}$ and 
\begin{align}
\dim \ker\big(\slashed{D}_{\hat\Sigma} \big) &= \dim
                                                     \ker\big(\slashed{D}_{-\hat
                                                    \Sigma} \big)
                                                     \ \qquad
                                                     \mbox{and} \qquad
{}^{\rm b}\eta \big(\slashed{D}_{\hat{\Sigma}} \big)=-{}^{\rm
                                                     b}\eta
                                                     \big(\slashed{D}_{-\hat{\Sigma}}
                                                     \big) \ , \\[4pt] 
\dim(\Lambda_T \cap \Lambda_C)&=\dim(\Lambda_T \cap \Lambda_{-C})
                                \qquad \mbox{and} \qquad
\mu(\Lambda_{T}, \Lambda_{C})=-\mu(\Lambda_{T}, \Lambda_{-C}) \ ,
\end{align}
where $\Lambda_{-C}$ is the scattering Lagrangian subspace for $-\Sigma$ and in the last equation $\mu$ is computed on 
$Y$ and $-Y$, respectively. 
\end{lemma}

\begin{proof}
There is an equality $\slashed{D}_{\hat \Sigma}=-\slashed{D}_{-\hat \Sigma}$
of operators acting on sections of the underlying twisted
spinor bundle $\hat S_{\hat \Sigma}^G$, which implies the first two
equations. The Lagrangian subspaces $\Lambda_{C}$ and $\Lambda_{T}$
are independent of the orientation, which implies the third equation. 

We can interpret the exterior angle $\mu(\Lambda_{T}, \Lambda_{C})$ as
the $\eta$-invariant of a cylinder with boundary conditions induced by
$\Lambda_{T}$ and $\Lambda_{C}$~\cite{LeschWojciechowski}. Reversing
the orientation of this cylinder corresponds to $\mu(\Lambda_{T},
\Lambda_{-C})$. The last equation then follows from the fact that the
$\eta$-invariant changes sign under orientation-reversal.
\end{proof}

\begin{proposition}\label{Corollary: Index theorem for composition of regular 2-morphisms}
The index is additive under vertical composition of regular 2-morphisms in
$\ECobF$ if we choose identical boundary conditions on the corners. 
\end{proposition}

\begin{proof}
The contributions from the gluing boundary cancel each other by Lemma \ref{Lemma: Properties terms in the index theorem}.  We still have to show that 
\begin{align}
\int_{(M_2\circ M_1)'}\, K_{\rm AS}=\int_{M_1'}\,
  K_{\rm AS}+\int_{M_2'}\, K_{\rm AS} \ .
\end{align}
This is not completely obvious since the vertical composition also
involves deleting half of the collars of the gluing boundary. However,
from
\eqref{Index theorem corners} and the construction of
$\hat{M}'$ it is clear that $\int_{M'}\, K_{\rm AS}$
is independent of the length of the collars. Using the description of
gluing in terms of mapping cylinders we can cover $(M_2\circ
M_1)'$ by $\tilde{M}_1'$ and $\tilde{M}_2'$,
where $\tilde{M}_i'$ is the manifold $M_i'$ with $\tfrac{3}{4}$ of the collar corresponding to the gluing boundary removed. 
\end{proof}

\section{The extended index field theory and Hamiltonian description of the parity anomaly}\label{Sec: Extended Indext FQFT}
We shall now proceed to extend the quantum field theory
$\Aa_{\rm parity}^\zeta$ to an anomaly quantum field theory
$\Aa_{\rm parity}^\zeta\colon \ECobF\longrightarrow\Tvs$ describing the parity anomaly in $n-1$ dimensions. In contrast to \cite{parity} we will use coends instead of limits 
for the construction of the field theory, 
because this makes the relation to the
theories constructed in Section~\ref{Sec: Parallel transport} clearer. 
This is nothing more than a matter of taste. It would have been equally possible
to rewrite Chapter~\ref{Chapter: t Hooft} in terms of limits. 
Some of the constructions
might seem odd due to the fact that there are no boundary conditions
to chose on $n-1$ dimensional manifolds.   
For the convenience of the reader we review the definition and some basic 
properties of coends following~\cite{MacLane,Coend}. 
\begin{definition}
Let $F\colon \cC^{\opp}\times \cC \longrightarrow\cD$ be a functor. A \emph{wedge} for 
$F$ consists of an object $d\in \cD$ together with a family of morphisms
$\{ \alpha_c \colon d \longrightarrow F(c,c) \}_{c\in \cC}$  such that for all morphisms 
$f\colon c \longrightarrow c'$ in $\cC$
\begin{equation}
\begin{tikzcd}[row sep=1.25cm, column sep=1.25cm]
d \ar[r,"\alpha_{c'}"] \ar[d,"\alpha_{c}",swap] & F(c',c') \ar[d, "{F(f,\id_c')}"] \\
F(c,c) \ar[r, "{F(\id_c,f)}",swap] & F(c,c')
\end{tikzcd}
\end{equation}
commutes.
A \emph{morphism between wedges} $f\colon d \longrightarrow d'$ consists 
of a morphism $d\longrightarrow d'$ in $\cD$ such that 
\begin{equation}
\begin{tikzcd}[row sep=1.25cm, column sep=1.25cm]
d \ar[rd]\ar[rrd, bend left=20, "\alpha_{c'}"] \ar[ddr, bend right=20, "\alpha_c",swap] & & \\
& d' \ar[r,"\alpha'_{c'}"] \ar[d,"\alpha'_{c}",swap] & F(c',c') \ar[d, "{F(f,\id_c')}"] \\
& F(c,c) \ar[r, "{F(\id_c,f)}",swap] & F(c,c')
\end{tikzcd}
\end{equation}
commutes. 
\end{definition}
There is a dual notion of a \emph{cowedge} consisting of an object $d$ together with
morphisms $\alpha_c\colon F(c,c)\longrightarrow d$ such that the obvious diagram 
commutes. Morphisms of cowedges are defined as morphisms $d\longrightarrow d'$
such that the obvious diagrams commute. 
(Co)Ends are universal (co)wedges.
\begin{definition}
Let $F\colon \cC^{\opp}\times \cC \longrightarrow\cD$ be a functor. An \emph{end} of
$F$ written as $\int_{c\in \cC} F(c,c)$ is a terminal object in the category of wedges of $F$. 

An \emph{coend} of $F$ written as $\int^{c\in \cC} F(c,c)$ is an initial object in the category of 
cowedges of $F$. 
\end{definition} 
\begin{remark}\label{Rem: Coends are functorial}
Let us spell out the universal property of the coend $\int^{c\in \cC}F(c,c)$. Being a cowedge
it comes with morphisms $F(c,c)\longrightarrow \int^{c\in \cC}F(c,c)$ for all $c \in \cC$ such 
that for every other cowedge $F(c,c)\longrightarrow d$ there exists a unique morphism 
$\int^{c\in \cC}F(c,c)\longrightarrow d$ making 
\begin{equation}
\begin{tikzcd}
F(c,c') \ar[r] \ar[d] & F(c',c' ) \ar[d] \ar[ddr, bend left=20]& \\ 
F(c,c) \ar[rrd, bend right=20] \ar[r] & \int^{c\in \cC} F(c,c) \ar[dr] & \\ 
 & & d
\end{tikzcd}
\end{equation}
commute. This shows that coends are unique up to unique isomorphism. For this reason
we will speak of \underline{the} coend sometimes. The universal property also ensures
that coends are functorial, i.e.\ (assuming that all coends in $\cD$ exist) there is a functor
\begin{align}
\int^{c \in \cC} \colon [\cC^{\opp}\times \cC,\cD] \longrightarrow \cD \  \ .
\end{align}
\end{remark}
\begin{example}\label{Exa: limits as ends}
\begin{itemize}
\item
Let $F \colon \cC \longrightarrow \cD$ be a functor between categories. 
$F$ induces a functor $\widehat{F} \colon \cC^{\opp}\times \cC\xrightarrow{\pr_\cC} \cC \overset{F}{\longrightarrow} \cD$. In Chapter~\ref{Chapter: t Hooft} we will drop 
the  $\widehat{\phantom{a} }$ . Spelling out the definitions shows that the end $ \int_{c\in \cC} \widehat{F}(c,c)$ agrees with the limit of $F$ and the coend $\int^{c \in \cC} \widehat{F}(c,c)$ with the colimit.   

\item Let $\cC$ be a category equivalent to the category with one object 
and one morphism and $F\colon \cC^{\opp} \times \cC \longrightarrow D$. The (co)end 
is given by the value of $F(c,c)$ at an arbitrary element $c\in \cC$ with structure maps induced 
from $F$.

\item Let $F\colon \Vscr \longrightarrow \Vscr '$ be a linear functor between 2-vector spaces.
The value of $F$ at an object $V\in \Vscr$ can be computed by the coend 
\begin{align}
F(V) \cong \int^{V'\in \Vscr} \Hom(V',V)\otimes F(V') \ \ .
\end{align} 
Furthermore, this is natural in $V$ inducing a natural isomorphism 
\begin{align}
F(\cdot) \cong \int^{V'\in \Vscr} \Hom(V',\cdot)\otimes F(V') \ \ .
\end{align} 
Statements of this type are called \emph{generalized Yoneda lemmas}. The special
case $F=\id_\Vscr$ is sometimes called the (enriched) coYoneda lemma. We refer to
\cite[Section 2.3]{FSS16} for a proof in the more general context of finite tensor categories. 
\end{itemize}
\end{example}
(Co)ends can be expressed as (co)limits ensuring there existence in a lot of interesting 
examples. 
\begin{proposition}\label{Prop: End as equalizer}
Let $\cD$ be a complete and cocomplete category and $F\colon \cC^{\opp}\times \cC\longrightarrow \cD$ a functor. The end of $F$ exists and is given by the equalizer 
\begin{equation}
\int_{c\in \cC} F(c,c) \cong \operatorname{eq}\left( \prod_{c\in \cC} F(c,c)  \rightrightarrows \prod_{f\colon c \rightarrow c'} F(c,c') \right) \ \ . 
\end{equation}
Dually, the coend of $F$ exists and is given by the coequalizer 
\begin{align}
\int^{c \in \cC}F(c,c) \cong \operatorname{coeq} \left(   \coprod_{f\colon c \rightarrow c'} F(c,c') \rightrightarrows  \coprod_{c\in \cC} F(c,c) \right) \ \ .
\end{align} 
\end{proposition}
One of the advantages of the calculus of (co)ends is that iterated (co)ends are well behaved, as indicated by the integral notation.

\begin{theorem}[Fubini's theorem for (co)ends]\label{Thm: Fubini}
Let $F\colon \cC \times \cC^{\opp}\times \cE \times \cE^{\opp}\longrightarrow \cD$
be a functor. There are canonical natural isomorphisms 
\begin{align}
\int^{c\in \cC} \int^{e\in \cE} F(c,c,e,e) \cong \int^{(c,e)\in \cC\times  \cE} F(c,c,e,e) 
\cong \int^{e\in \cE} \int^{c\in \cC} F(c,c,e,e) \ \ .
\end{align}
The same statement holds for ends.
\end{theorem}   
After this short detour we come back to the index field theory. Following Section~\ref{Sec: path integral parity},
we would like to define something like $\Aa_{\rm parity}^\zeta(M) = \zeta^{\text{ind} (\slashed{D
  }_{\hat{M}'}^+ -\hat{T}^+ )}$ for a fixed
$\zeta\in\C^\times$ and every regular 2-morphism
$M$ of $\ECobF$.
The problem with this definition is that the index may depend on our
choice of $T \in \text{End}_\C\big(\ker (\slashed{D}_Y) \big)$. The
resolution is to include the data about the choice of $T$ into our
extended quantum field theory. 

We do this by combining, for each object $S$ of $\ECobF$, all possible boundary conditions $T$ into a
category $\Aa_{\rm parity}^\zeta(S)$ in the following way: Let $\mathsf{T}(S)$ 
be the category with one
object for every odd self-adjoint unitary $T\in \text{End}_\C\big(\ker (\slashed{D}_{S}) \big) $, which is local in the
sense that $T$ is the direct sum over the connected components $S_i$ of $S$ of odd self-adjoint unitary operators acting on the kernel  
$ \ker \big(\slashed{D}_{S_i}  \big)$.
There is exactly one 
morphism $I_{ij}\colon T_i \longrightarrow T_j$ between every pair of objects in $\mathsf{T}$. The
category $\Aa_{\rm parity}^\zeta(S)$ is then defined to be the finite completion 
of the $\C$-linearisation of the category $\mathsf{T}(S)$. 

We will frequently use the following concrete model for $\Aa_{\rm parity}^\zeta(S)$: 
\begin{itemize}
\item 
Objects of $\Aa_{\rm parity}^\zeta(S)$ are formal finite linear 
combinations $\bigoplus_{T_i\in \mathsf{T}(S)} V_i * T_i$, with $V_i \in \fvs$. We will sometimes 
write $T_i$ for $\C*T_i$.   

\item The space of morphism from $\C*T_i$ to $\C*T_j$ is the one dimensional vector 
space $\C[I_{ij}]$ generated by the morphism $I_{ij}$. In order to obtain the morphism spaces between all objects in $\Aa_{\rm parity}^\zeta(S)$, this definition 
has to be extended bilinearly, i.e.\
\begin{align}
\Hom_{\Aa_{\rm parity}^\zeta(S)} \left( \bigoplus_{i=1}^n V_i * T_i\ ,\ \bigoplus_{j=1}^m V_j * T_j \right) = \bigoplus_{i,j} \Hom(V_i,V_j) \otimes \C[I_{ij}]
\end{align} 
for all formal finite sums. We can see a morphism in $\Aa_{\rm parity}^\zeta(S)$  
as a matrix of morphisms between the the individual parts of the finite formal sums.
Composition is defined by matrix multiplication and composition in $\mathsf{T}(S)$.  
\end{itemize}
The linear category 
$\Aa_{\rm parity}^\zeta(S)$ carries a canonical $\fvs$-module structure
\begin{align}
\begin{split}
\fvs \boxtimes \Aa_{\rm parity}^\zeta(S) &\longrightarrow \Aa_{\rm parity}^\zeta(S) \\
V\boxtimes \bigoplus_{i=1}^n V_i * T_i & \longmapsto \bigoplus_{i=1}^n (V \otimes_\C V_i) * T_i\ \ .
\end{split} 
\end{align} 
We can construct a $\C$-linear functor $\Aa_{\rm parity}^\zeta (\Sigma):\Aa_{\rm parity}^\zeta(S_-)\longrightarrow\Aa_{\rm parity}^\zeta(S_+)$ for
a regular 1-morphism $ \Sigma \colon S_- \longrightarrow S_+$
in $\ECobF$ by using the corresponding boundary and corner contributions to the index formula \eqref{Index theorem corners}. For elements $T_\pm$ of $\mathsf{T}(S_\pm)$, we define 
\begin{align}
\begin{split}
I_{\Sigma}(T_-,T_{+})=&-\tfrac{1}{2} \, \big( \, {}^{\rm
                                               b}\eta(\slashed{D}_{\hat{\Sigma}'}
                                               )-\dim \ker
                                               (\slashed{D}_{\hat{\Sigma}'}
                                               ) \\ & \qquad \qquad - \dim(\Lambda_{T_-\oplus T_{+}} \cap \Lambda_C)+ \mu(\Lambda_{T_-\oplus T_{+}}, \Lambda_{C}) \, \big) \ .
\end{split} 
\end{align}
We will sometimes drop the subscript $\Sigma$ when it is clear from the context.
With this convention the index theorem for manifolds with corners Equation~\eqref{Index theorem corners} 
can be rewritten as 
\begin{align}\label{Eq: Index theorem with I}
\ind(\slashed{D}^+_{\hat{M}'}-\widehat{T_-\oplus T_{+}}{}^+) = \int_{M'} K_{\rm AS} + I_{\Sigma_2}(T_-,T_+) - I_{\Sigma_1}(T_-,T_+) \ \ .
\end{align}  
To construct $\mathcal{A}^\zeta_{\rm parity}$ on $\Sigma$ we define the functor 
\begin{align}
\begin{split}
\Sigma (\cdot , \cdot ) \colon \mathsf{T}(S_+)^{\opp} \times \mathsf{T}(S_-) &
\longrightarrow \fvs \\ 
T_+\times T_- & \longmapsto \C \\ 
I_{+ij}\times I_{-ij}\colon T_{+j}\times T_{-i}  \longrightarrow  T_{+i}\times T_{-j} 
& \longmapsto \id_\C \cdot \zeta^{I_\Sigma(T_{-j},T_{+i})-I_\Sigma(T_{-i},T_{+j})}  
\end{split} \ \ .
\end{align}
We fix a generator $T_-$ of $\Aa_{\rm parity}^\zeta (S_-)$ and define 
\begin{align}
\Aa_{\rm parity}^\zeta (\Sigma)[T_-] \coloneqq \int^{T_+\in \mathsf{T}(S_+)}
\Sigma(T_+,T_-)*T_+ \ \ .
\end{align}
This assignment extends linearly to a functor 
$\Aa_{\rm parity}^\zeta (\Sigma)\colon \Aa_{\rm parity}^\zeta (S_-) 
\longrightarrow \Aa_{\rm parity}^\zeta (S_+)$. 

To define $\Aa_{\rm parity}^\zeta$ on regular 2-morphisms we use the following
proposition:
\begin{proposition}
Let $M:\Sigma_1\Longrightarrow \Sigma_2$ be a regular 2-morphism. The 
collection of linear maps
\begin{align}
\xi(M)_{T_-,T_+} = \zeta^{\ind
  (\slashed{D}^+_{\hat{M}'}-\widehat{T_-\oplus T_{+}}{}^+)}\cdot
  \id_{\C} \colon \C \longrightarrow \C \ .
\end{align}
defines a natural transformation $\xi(M)\colon \Sigma_1(\cdot, \cdot) \Longrightarrow \Sigma_2(\cdot, \cdot)$. 
\end{proposition}
\begin{proof}
Naturality of $\xi(M) $ means that the diagram
\begin{equation}
\begin{tikzcd}[row sep= 1cm, column sep = 2cm]
\C \ar[rr,"{\zeta^{\ind(\slashed{D}^+_{\hat{M}'}-\widehat{T_-\oplus T_{+}}{}^+)}\cdot \id_{\C}}"] \ar[dd, "\zeta^{I_{\Sigma_1}(T'_{-},T_{+})-I_{\Sigma_1}(T_{-},T'_{+}) } \cdot \id_\C ", swap] & & \C \ar[dd, "\zeta^{I_{\Sigma_2}(T'_{-},T_{+})-I_{\Sigma_2}(T_{-},T'_{+})}\cdot \id_\C  ", ] \\ 
 & & \\
 \C  \ar[rr,"{\zeta^{\ind(\slashed{D}^+_{\hat{M}'}-\widehat{T'_-\oplus T'_{+}}{}^+)}\cdot \id_{\C}}"] & & \C 
\end{tikzcd}
\end{equation} 
has to commute.
This follows immediately from the index theorem for manifolds with corners rewritten as in Equation~\eqref{Eq: Index theorem with I}. 
\end{proof}
The natural transformation $\xi(M)$ induces by Remark~\ref{Rem: Coends are functorial} a natural transformation $\Aa_{\rm parity}^\zeta(M)\colon \Aa_{\rm parity}^\zeta(\Sigma_1)\Longrightarrow \Aa_{\rm parity}^\zeta(\Sigma_2)$.

We define the theory $\Aa_{\rm parity}^\zeta$ on limit 1-morphisms as the functor corresponding to a mapping cylinder. From the definition of $\Aa_{\rm parity}^\zeta(\Sigma)$ it is clear that this is independent of the length of the mapping cylinder since only the behaviour at infinity is important. On limit 2-morphisms we define the theory to be the value of $\Aa_{\rm parity}^\zeta$ on a mapping cylinder of length $\epsilon$. 
This completes the definition of $\Aa_{\rm parity}^\zeta$.

Before demonstrating that $\Aa_{\rm parity}^\zeta$ is an extended quantum field theory, we explicitly calculate the functors corresponding to limit 1-morphisms.

\begin{proposition}
\label{Proposition: A on limit 1-morphisms}
Let $\phi \colon S_- \longrightarrow S_+$ be a limit 1-morphism in $\ECobF$ with mapping cylinder $\mathfrak{M}(\phi)$, and let $T_\pm$ be fixed objects in $\sfT(S_\pm)$. Then 
\begin{align}
\label{EQ: A on limit 1-morphisms}
I_{\mathfrak{M}(\phi)}(T_-,T_+) = -\tfrac{1}{2} \, \big( \dim (\Lambda_{T_-} \cap \Lambda_{\phi^\ast T_+}) + \mu(\Lambda_{T_-} ,\Lambda_{\phi^\ast T_+}) \big) \ .
\end{align}
\end{proposition}

\begin{proof}
By Remark \ref{Remark: Relation eta Invariants} we get a term in $I_{\mathfrak{M}(\phi)}(T_-,T_+)$ corresponding to the $\eta$-invariant with boundary conditions induced by the Lagrangian subspaces $\Lambda_{T_\pm}$.
There is a diffeomorphism induced by $\phi^{-1}$ and $\text{id}$ from the mapping cylinder of $\phi$ with length $1$ to the cylinder $  [0,1] \times \Sigma_-$. The boundary conditions change to new boundary conditions induced by $\Lambda_{T_-}$ and $\Lambda_{\phi^\ast T_+}$. 
The $\eta$-invariant for this situation was calculated in~\cite[Theorem 2.1]{LeschWojciechowski}, from which we get
\begin{align}
\eta (\D_{T_-,\phi^\ast T_+})= \mu(\Lambda_{T_-},\Lambda_{\phi^\ast T_+}) \ .
\end{align}
We can extend the diffeomorphism induced by $\phi^{-1}$ and $\id$ above to manifolds with cylindrical ends attached. The expression (\ref{EQ: A on limit 1-morphisms}) then follows from similar arguments to those used in the proof of Theorem~\ref{Theorem extended index theorem}. 
\end{proof}

\begin{theorem}\label{Theorem: A is extended field theory}
 \ $\Aa_{\rm parity}^\zeta \colon \ECobF \longrightarrow \Tvs$ is an invertible extended quantum field theory.
\end{theorem}
\begin{proof}
We construct a family of natural isomorphisms $\Phi_{S} \colon \id \Longrightarrow \Aa_{\rm parity}^\zeta(\id_{S})$ for all objects $S\in \ECobF$. 
 
For this, it is enough to construct a natural isomorphism $\Phi'_{S}\colon \Hom(\cdot, 
\cdot) \Longrightarrow [0,1] \times S (\cdot, \cdot) $ by the enriched coYoneda lemma from Example~\ref{Exa: limits as ends} 
\begin{align}
\id_{\Aa_{\rm parity}^\zeta(S)}(T_-) \cong \int^{T_+\in \mathsf{T}(S)} \Hom (T_+, T_-)*T_+ \ \ .
\end{align}
Sending $I_{ij}\in 
\Hom (T_i, T_j) 
$ to the complex number $\zeta^{I_{[0,1]\times S}(T_i,T_j)}$ induces such an isomorphism.

The naturality follows immediately from the commuting diagram
\begin{equation}
\begin{tikzcd}
\Hom(T_i,T_j) \ar[dd,swap,"I_{ij}\mapsto I_{i'j'}"] \ar[rrr,"\zeta^{I(T_i,T_j)}"] & & &  \C \ar[dd,"{\zeta^{I(T_{i'},T_{j'})-I(T_{i},T_{j})}}"] \\
& & & \\
\Hom(T_{i'},T_{j'}) \ar[rrr,swap,"\zeta^{I(T_{i'},T_{j'})}"] & & & \C
\end{tikzcd}
\end{equation} 
Let $\Sigma_1 \colon S_-\longrightarrow S$ and $\Sigma_1 \colon S\longrightarrow S_+$
be two 1-morphisms in $\ECobF$.
For the composition we have to construct
natural $\C$-linear isomorphisms 
\begin{align}\label{eq:Phinatural}
\Phi_{\Sigma_1, \Sigma_2} \colon \Aa_{\rm parity}^\zeta\big(\Sigma_2\big) \circ \Aa_{\rm parity}^\zeta\big(\Sigma_1\big)\Longrightarrow \Aa_{\rm parity}^\zeta\big(\Sigma_2\circ \Sigma_1\big) \ .
\end{align}
Using the definition of $\Aa_{\rm parity}^\zeta$ on 1-morphisms 
we get 
\begin{align}
\begin{split}
\Aa_{\rm parity}^\zeta(\Sigma_2)\circ \Aa_{\rm parity}^\zeta(\Sigma_1)
[T_-] &=  \Aa_{\rm parity}^\zeta(\Sigma_2) \left[\int^{T\in \sfT(S)} \Sigma_1(T, T_-)*T \right ] \\
& \cong \int^{T\in \sfT(S)} \Sigma_1(T,T_-) \otimes \int^{T_+\in \sfT(S_+)} \Sigma_2(T_+,T) * T_+ \\ 
& \cong \int^{T_+\in \sfT(S_+)} \left( \int^{T\in \sfT(S)} \Sigma_2(T_+,T) \otimes \Sigma_1(T,T_-) \right) * T_+
\end{split} 
\end{align}
where we used the linearity of the coend, the symmetric monoidal structure
of $\fvs$ and Fubini's theorem~\ref{Thm: Fubini} for coends. 

On the other side we get  
\begin{align}
\Aa_{\rm parity}^\zeta(\Sigma_2\circ \Sigma_1)
[T_-] &=   \int^{T_+\in \sfT(S_+)} (\Sigma_2 \circ \Sigma_1)(T_+, T_-)*T_+ \ \ .
\end{align}
The natural isomorphism $\Phi_{\Sigma_1, \Sigma_2}$ is induced by the family
of isomorphisms 
\begin{align}
\begin{split}
\Sigma_2(T_+,T)\otimes \Sigma_1(T,T_-) & \longrightarrow (\Sigma_2 \circ \Sigma_1)(T_+,T_-) \\
1\otimes 1 & \longmapsto  \zeta^{I_{\Sigma_2 \circ \Sigma_1}(T_-,T_+)-I_{\Sigma_1}(T_-,T)-I_{\Sigma_2}(T,T_+)}
\end{split}
\end{align}
To show that this construction is natural, it is enough to observe that the diagram
\begin{equation}
\begin{tikzcd}
\C \ar[dd,swap,  "{\zeta^{I(T'_-,T')+I(T',T'_+)-I(T_-,T)-I(T,T_+)}}" ]\ar[rrrr, "{\zeta^{I(T_-,T_+)-I(T_-,T)-I(T,T_+)}}"] & & & &\C \ar[dd,  "{\zeta^{I(T'_-,T_+')-I(T_-,T_+)}}" ] \\
& & & & \\
\C \ar[rrrr, "\zeta^{I(T'_-,T'_+)-I(T'_-,T')-I(T',T'_+)}", swap]  & & &  & \C
\end{tikzcd}
\end{equation}
commutes.
This completes the construction of the natural isomorphism $\Phi$. 

In order for $\Phi$ to equip $\Aa_{\rm parity}^\zeta$ with the structure of a 2-functor, we need to check naturality with respect to 2-morphisms and associativity. We start with the compatibility with 2-morphisms. For this fix regular 
2-morphisms $M_1\colon \Sigma_- \longrightarrow \Sigma_+$ and $M_2\colon \Sigma'_- \longrightarrow \Sigma'_+$ between
regular 1-morphisms $\Sigma_-,\Sigma_+\colon S_-\longrightarrow S$ and  $\Sigma'_-,\Sigma'_+\colon S\longrightarrow S_+$. By the naturality of all our
constructions it is enough to check this for fixed corner conditions $T_-\in \mathsf{T}(S_-)$, $T\in \mathsf{T}(S)$ and $T_+\in \mathsf{T}(S_+)$. 
Using the index theorem, this follows from the calculation
\begin{align}
\begin{split}
&\log_{\zeta}\xi\big(M_2\bullet M_1\big)_{T_-,T_+} \\
& \qquad \qquad \quad = \int_{(M_2\bullet M_1)'} \, K_{\rm AS} +I_{ \Sigma'_+ \circ \Sigma_+}(T_-,T_+)-I_{\Sigma'_-\circ \Sigma_-}(T_-,T_+) 
\\[4pt]
& \qquad \qquad \quad = \int_{M'_1{}} \, K_{\rm AS}+\int_{M'_2{}} \, K_{\rm AS} + I_{ \Sigma'_+ \circ \Sigma_+}(T_-,T_+)-I_{\Sigma'_-\circ \Sigma_-}(T_-,T_+) 
\\& \qquad \qquad \qquad + \big(I_{\Sigma_+}(T_-,T)+I_{\Sigma'_+}(T,T_+) \big) - \big(I_{\Sigma_+}(T_-,T) +I_{\Sigma'_+}(T,T_+) \big)
\\& \qquad \qquad \qquad + \big(I_{\Sigma_-}(T_-,T)+I_{\Sigma'_-}(T,T_+) \big)- \big(I_{\Sigma_-}(T_-,T)+I_{\Sigma_-'}(T,T_+) \big) 
\\[4pt]
& \qquad \qquad \quad = \log_{\zeta}\xi \big(M_1\big)_{T_-,T_1}+ \log_{\zeta}\xi\big(M_2\big)_{T_1,T_+} \\ 
& \qquad \qquad \qquad +\big(I_{\Sigma_+'\circ \Sigma}(T_-,T_+)-I_{\Sigma_+}(T_-,T) -I_{\Sigma_+'}(T,T_+)\big) \\
& \qquad \qquad \qquad -\big(I_{\Sigma'_- \circ \Sigma_-}(T_-,T_+)-I_{\Sigma_-}(T_-,T) -I_{\Sigma_+}(T,T_+)\big) \\[4pt]
& \qquad \qquad \quad = \log_{\zeta} \xi \big(M_1\big)_{T_-,T_1}+ \log_{\zeta}\xi \big(M_2\big)_{T_1,T_+} \\ 
& \qquad \qquad \qquad +\log_{\zeta}\big(\Phi_{\Sigma_+,\Sigma_+'}\big)_{T_-,T_+}+\log_{\zeta}\big(\Phi^{-1}_{\Sigma_-,\Sigma_-'}\big)_{T_-,T_+} \ .
\end{split}
\end{align}
It remains to demonstrate compatibility with associativity: $\Phi \circ (\Phi \bullet\id) = \Phi \circ ( \id\bullet \Phi) $, i.e.\ the coherence condition \eqref{EQ1: Definition 2Functor}. For this, we fix three composable regular 1-morphisms $\Sigma_i$, $i=1,2,3$ from $S_{i-}$ to $S_{i+}$ in $\ECobF$. By the naturality of all constructions, it is enough to check the relation for fixed objects $T_-$ of $ \sfT(S_{1-})$, $T_1$ of $ \sfT(S_{2-})=\sfT(S_{1+})$, $T_2 $ of $ \sfT(S_{3-})= \sfT(S_{2+})$, and $T_+$ of $ \sfT(S_{3+})$. This follows immediately from the commutative diagram
\begin{equation}
\begin{tikzcd}
& & & & \C \ar[ddddrrrr, "{\zeta^{I(T_1,T_+)-I(T_2,T_1)-I(T_2,T_+)}}" description] \ar[ddddllll, "{\zeta^{I(T_-,T_2)-I(T_-,T_1)-I(T_1,T_2)}}" description] \ar[dddddddd, "{\zeta^{I(T_-,T_+)-I(T_2,T_1)-I(T_2,T_+)-I(T_-,T_1)}}" description] & & & & \\
 & & & &    & & & & \\
  & & & &    & & & & \\
   & & & &    & & & & \\
 \C \ar[ddddrrrr, "{\zeta^{I(T_-,T_+)-I(T_-,T_2)-I(T_2,T_+)}}" description] & & & & & & & & \C \ar[ddddllll, "{\zeta^{I(T_-,T_+)-I(T_-,T_1)-I(T_1,T_+)}}" description] \\
  & & & &    & & & & \\
   & & & &    & & & & \\
 & & & &    & & & & \\
 & & & & \C & & & &
\end{tikzcd} 
\end{equation}   
We finally have to check the coherence condition \eqref{EQ2: Definition 2Functor}.
We fix a regular 1-morphism $\Sigma \colon S_-\longrightarrow S_+$. We evaluate the resulting natural transformation at a fixed object $T_-$ of $\sfT(S_-)$. 
By naturality we can check the equation for the terms appearing in the coend. For this we fix 
an object $T_+$ of $\sfT(S_+)$. Then the composition gives
\begin{equation}
\label{EQ:Proof of C2}
\begin{small}
\begin{tikzcd}
\big( T_+ \ar[rr, "{\zeta^{I_{\id}(T_+,T_+)}}"] & & T_+ \ar[rrrrrrr,
  "{\zeta^{I_{ \mathfrak{M}_1(\id) \circ \Sigma}(T_-,T_+)-I_{\Sigma}(T_-,T_+)-I_{\id}(T_+,T_+)}}"]
  & & & & & & & T_+\big) = \big( T_+ \ar[r,"\id"]& T_+\big) \ ,
\end{tikzcd}
\end{small}
\end{equation}
where we used $I_{ \mathfrak{M}_1(\id) \circ \Sigma}(T_-,T_+)-I_{\Sigma}(T_-,T_+)=0$.
This proves the condition \eqref{EQ2: Definition 2Functor}. The coherence condition for $\Aa_{\rm parity}^\zeta(\id)\circ \Aa_{\rm parity}^\zeta(\Sigma)$ can be proven in the same way.

Next we come to the vertical composition of regular 2-morphisms. It is
enough to show that the composition is given by multiplication for
fixed objects $T_\pm $ of $\Aa_{\rm parity}^\zeta(S_\pm)$. This follows immediately from Proposition~\ref{Corollary: Index theorem for composition of regular 2-morphisms} by an argument similar to the one used in the proof of Theorem~\ref{A is field theory}.
The conditions for limit 1-morphisms and limit 2-morphisms follow now from their representations as mapping cylinders. 

Now we check compatibility with the monoidal structure. There are canonical $\C$-linear equivalences of categories given on objects by
\begin{align}
\chi^{-1}_{S, S'}\colon \Aa_{\rm parity}^\zeta(S \sqcup S') \longrightarrow \Aa_{\rm parity}^\zeta(S)\boxtimes \Aa_{\rm parity}^\zeta(S')
\end{align}
sending $(T, T')\in \sfT(S)\times \sfT(S') \cong \sfT(S\sqcup S')$ to $T \boxtimes T'$, and 
\begin{align}
\iota^{-1} \colon \Aa_{\rm parity}^\zeta(\emptyset) \longrightarrow \fvs
\end{align}
sending $0\in\{0\}=\text{End}_\C\big(\ker (\slashed{D}_{\emptyset}) \big)$ to $\C$. All further structures required for $\Aa_{\rm parity}^\zeta$ to be a symmetric monoidal 2-functor are trivial. It is straightforward if tedious to check that all diagrams in the definition of a symmetric monoidal 2-functor commute, but we shall not write them out explicitly.
Finally, it is straightforward to see that $\Aa_{\rm parity}^\zeta$ factors through the Picard 2-groupoid $\mathsf{Pic}_2(\Tvs)$, and hence $\Aa_{\rm parity}^\zeta$ is invertible. 
\end{proof}
\color{black}
\begin{remark}
The proof of Theorem \ref{Theorem: A is extended field theory} is more or less independent of the concrete form of $I_{\Sigma}(T_-,T_+)$ and the index theorem. It only uses additivity under vertical composition and the decomposition 
\begin{align}
\text{ind}\big(\slashed{D}^+_{\hat{M}'}-\widehat{T_-\oplus T_{+}}{}^+ \big) = \int_{M'{}} \, K_{\rm AS} +I_{\partial_+M}(T_-,T_+)-I_{\partial_-M}(T_-,T_+) \ ,
\end{align}
into a local part and a global part depending solely on boundary conditions. Hence, it should be possible to apply this or a similar construction to a large class of invariants depending on boundary conditions. Indeed, we apply the same construction in 
Section~\ref{Sec: Parallel transport} to build an extended functorial field theory
from the parallel transport on higher flat gerbes. This construction will be topological
in nature. 

A different example of particular interest would involve $\eta$-invariants on odd-dimensional manifolds with corners, which should be related to chiral anomalies in even dimensions and extend Dai-Freed theories~\cite{DaiFreed}. For steps towards constructing
such a field theory using different methods  
see~\cite{ChiralAnomaly, MonnierHamiltionianAnomalies}. 
\end{remark}

\subsubsection*{Field theories with parity anomaly and projective representations}\label{sec:projparity}

A quantum field theory with parity anomaly is now regarded as a theory relative to $\Aa_{\rm parity}^\zeta$ as described in Chapter~\ref{Sec: General Theory}, i.e.\ a natural
symmetric monoidal 2-transformation $Z^\zeta:\mbf1\Longrightarrow\mathsf{tr}\Aa_{\rm parity}^\zeta$. The concrete description of the extended quantum field theory
$\Aa_{\rm parity}^\zeta$ given in the proof of Theorem~\ref{Theorem:
  A is extended field theory} allows us to calculate the corresponding
groupoid 2-cocycle along the lines discussed in
Section~\ref{Sec: P anomaly actions}; this information about the parity anomaly is contained in the isomorphism \eqref{eq:Phinatural}. We choose a $\C$-linear equivalence of
categories $\chi \colon \Aa_{\rm parity}^\zeta(S) \longrightarrow
\fvs$ sending all objects $T$ of $\sfT(S)$ to $\C$ and all
morphisms $f$ of $\sfT(S)$ to $\id_\C$; a weak inverse is
given by picking a particular object $T_{S}$ in $\sfT(S)$ and
mapping $V\in \fvs$ to $V*T_S$. The functor $\Aa_{\rm parity}^\zeta(\phi)$ corresponding to a limit 1-morphism
$\phi \colon S_1 \longrightarrow S_2$ in the symmetry groupoid
$\Sym(\ECobF)$ is given by the coend
\begin{align}
\Aa_{\rm parity}^\zeta(\phi)(T_-)= \int^{T_+ \in \sfT(S)} \mathfrak{M}(\phi) (T_+,T_-)*T_+
\end{align}
where $\mathfrak{M}(\phi)$ is the mapping cylinder of $\phi $.
A choice of an object $T_{S} $ of $\sfT(S)$ defines an
isomorphism $\varphi_S \colon T_S \longrightarrow \Aa_{\rm parity}^\zeta(\phi)(T_S)  
$, which for simplicity we pick to be the same boundary mass
perturbation as chosen for the weak inverse above.  The groupoid
cocycle evaluated at $\phi_1 \colon S_1 \longrightarrow S_2$
and $\phi_2 \colon S_2 \longrightarrow S_3$ corresponding to
this choice is then given by
\begin{align}
\alpha^{\Aa_{\rm parity}^\zeta}_{\phi_1, \phi_2}= \zeta^{I_{\mathfrak{M}(\phi_2 \circ \phi_1)}(
  T_{S_1}, T_{S_3})-I_{\mathfrak{M}(\phi_1)}(
  T_{S_1}, T_{S_2})-I_{\mathfrak{M}(\phi_2)}(
  T_{S_2}, T_{S_3})} \ .
\end{align}
We can evaluate this expression explicitly by using \eqref{EQ: A on limit
  1-morphisms} to get
\begin{align}
\log_\zeta \alpha^{\Aa_{\rm parity}^\zeta}_{ \phi_1, \phi_2}=& -\frac{1}{2} \, \Big(\dim
                                      \big(\Lambda_{T_{S_1}}\cap \Lambda_{                                      
                                      \phi_1^\ast\,\phi_2^\ast T_{S_3}} \big) +
  \mu\big(\Lambda_{T_{S_1}},
  \Lambda_{\phi_1^\ast\, \phi_2^\ast T_{S_3}} \big) \nonumber \\ & \qquad \qquad 
  -\dim \big(\Lambda_{T_{S_1}
                                      }\cap \Lambda_{
                                      \phi_1^\ast T_{S_2}} \big)-\mu\big(\Lambda_{T_{S_1}}, \Lambda_{\phi_1^\ast
  T_{S_2}} \big) \nonumber \\ & \qquad \qquad 
  -\dim
                                      \big(\Lambda_{T_{S_2}}\cap
                                      \Lambda_{\phi_2^\ast
                                      T_{S_3}}\big)
                                       -\mu\big(\Lambda_{T_{S_2}},\Lambda_{ \phi_2^\ast T_{S_3}}
  \big) \Big) \ .
\label{eq:2cocycleexplicit}\end{align} 
To calculate the part of the 2-cocycle involving identity 1-morphisms we can use \eqref{EQ:Proof of C2} to get
\begin{align}
\alpha^{\Aa_{\rm parity}^\zeta}_{\phi, \id_{S}}= \alpha^{\Aa_{\rm parity}^\zeta}_{\id_{S},\phi} =
  \zeta^{-I_{ [0,1]\times S}(T_{S},T_{S})}=
  \zeta^{-\frac{1}{4} \dim\ker (\D_{S})} \ ,
\end{align}  
where the last equality follows from \eqref{EQ: A on limit
  1-morphisms}. From a physical point of view, it is natural to assume
this to be equal to $1$ since the identity limit morphism should
still be a non-anomalous symmetry of every quantum field theory. 
We can achieve this by normalising our anomaly quantum field theory
$\Aa_{\rm parity}^\zeta$ to the theory $\tilde{\mathcal{A}}_{\rm parity}^\zeta$ obtained by
redefining 
\begin{align}
\tilde I_{\Sigma}(T_-,T_+) = I_{\Sigma}(T_-,T_+)+
\tfrac{1}{8}\, \big( \dim \ker (\D_{S_-}) + \dim \ker
(\D_{S_+}) \big) \ .
\end{align}   
The proof of Theorem \ref{Theorem: A is extended field theory} then
carries through verbatum with $I_{\Sigma}$ replaced by $\tilde
  I_{\Sigma}$ everywhere.

\begin{example}
We conclude by illustrating how to extend
the partition function discussed in Section~\ref{Sec: path integral parity} to the
anomaly quantum field theory $\Aa_{\rm parity}^{(-1)}$, glossing
over many technical details.
To construct the second quantized Fock space of a quantum
field theory of fermions coupled to a background gauge field on a
Riemannian manifold
$S$, one needs a polarization
\begin{align}
H=H^+\oplus H^-
\end{align}  
of the one-particle Hilbert space $H$ of
wavefunctions, which we take to be the sections of the twisted spinor
bundle $S^{G}_{S}$. If the Dirac Hamiltonian $\D_{S}$ has no zero
  modes, then there exists a canonical polarization given by taking $H^+=H^{>0}$
  (resp. $H^-=H^{<0}$) to be the space spanned by the positive
  (resp. negative) energy eigenspinors. 
Given such a polarization we can define 
\begin{align}
Z_{\rm parity}^{(-1)}(S)= \mbox{$\bigwedge$} H^+\otimes \mbox{$\bigwedge$}
(H^-)^\ast \ ,
\end{align}  
where $\bigwedge H$ denotes the exterior algebra generated by the
vector space $H$. Now time-reversal (or orientation-reversal) symmetry acts
by interchanging $H^+$ and $H^-$, and there is no problem extending
this symmetry to the Fock space $Z_{\rm parity}^{(-1)}(M)$. 

In the case that $\ker(\D_{S})$ is non-trivial, as is the case
for fermionic gapped quantum phases of matter, one could try to declare all zero modes to belong to $H^{>0}$ or $H^{<0}$ and use the corresponding polarization to define a Fock space. We cannot apply this method of quantization, since it breaks orientation-reversal symmetry. 
Therefore, we are forced to use a different polarization compatible
with orientation-reversal symmetry. There is no canonical choice for
such a polarization, but rather a natural family parameterized by Lagrangian subspaces $\Lambda_T \subset \ker(\D_{S})$:
\begin{align}
H^+(\Lambda_T) = H^{>0}\oplus \Lambda_T \qquad \text{and} \qquad H^-(\Lambda_T) = H^{<0}\oplus \Gamma \Lambda_T \ . 
\end{align}  
Since orientation reversion acts proportionally to the chirality
operator $\Gamma$ on spinors, these polarizations are compatible with the symmetry. 
We then get a family of Fock spaces
\begin{align}
Z_{\rm parity}^{(-1)}(S,T)= \mbox{$\bigwedge$} H^+(\Lambda_T)\otimes \mbox{$\bigwedge$}
H^-(\Lambda_T)^\ast = \mbox{$\bigwedge$} H^{>0}\otimes
\mbox{$\bigwedge$} \big(H^{<0}\big)^{ \ast } \otimes F(S,T) \ ,
\end{align} 
where the essential part for our discussion is encoded in the
finite-dimensional vector space
\begin{align} 
F(S,T)=\mbox{$\bigwedge$} \Lambda_T \otimes \mbox{$\bigwedge$} (
\Gamma \Lambda_T)^\ast \ .
\end{align} 
These vector spaces define an element
\begin{align}
F(S)\coloneqq \int^{T\in \mathsf{T}(S)} F(S,T)*T 
\end{align}
of $\mathcal{A}^{(-1)}_{\text{parity}}(S)$, or equivalently a $\C$-linear functor
\begin{align}
Z_{\rm parity}^{(-1)}(S) : \fvs \longrightarrow \mathcal{A}^{(-1)}_{\rm parity}(S) \ .
\end{align}
To define the functor appearing in the coend on morphisms we
fix an ordered basis for every $\Lambda_T$ and assign to a morphism $T_1 \longrightarrow T_2$ 
the linear map induced by sending the fixed basis of $\Lambda_{T_1}$ to the basis of $\Lambda_{T_2}$.    

We sketch how to extend this to a natural symmetric
monoidal 2-transformation,
realising an anomalous quantum field theory $Z_{\rm parity}^{(-1)}$ with parity anomaly
according to Definition~\ref{Definition Anommalous field theory}.
For a 1-morphism $\Sigma \colon S_- \longrightarrow S_+$ we have to construct a natural transformation \begin{align}
 Z^{(-1)}_{\rm parity} (\Sigma)\colon \mathcal{A}^{(-1)}_{\rm parity}(\Sigma)\circ Z^{(-1)}_{\rm parity} (S_-) \Longrightarrow Z^{(-1)}_{\rm parity} (S_+)
\ \ .
\end{align} 
The left-hand side is given by the coend
\begin{align}
\int^{T_+ \in \mathsf{T}(S_+)} \int^{T_-\in \mathsf{T}(S_-)} \left(\Sigma(T_+,T_-)\otimes F(S,T_-)\right) *T_+ \ \ .
\end{align}
This implies that a natural transformation can be defined from a family of compatible linear maps $Z^{(-1)}_\text{\rm parity}(\Sigma)_{T_-,T_+}\colon \left(\Sigma(T_+,T_-)\otimes F(S,T_-)\right)\longrightarrow F(S_+,T_+)$. These should again be given by an appropriate regularization of path integrals, i.e.\ the determinant of the Dirac operator
with appropriate boundary conditions. As before we assume that these maps are well-defined up to a sign. 
To fix the sign, we have to consistently fix reference background
fields on all 1-morphisms. This is possible, for example, by using a connection on the universal bundle and pullbacks along classifying maps. 
Again we can fix the sign at these reference fields to be positive. Using a spectral flow similar to \eqref{eq:spectralflow} with boundary conditions $T_-$ and $T_+$, we can fix the sign for all other field configurations. 
Assuming that this spectral flow can be calculated by the index with
appropriate boundary conditions, we see that these sign ambiguities
satisfy the coherence conditions encoded by
$\mathcal{A}^{(-1)}_{\rm{parity}}$, i.e.\ they define a natural symmetric
monoidal 2-transformation. 
This demonstrates in which sense a field theory with parity anomaly takes values in $\mathcal{A}^{(-1)}_{\rm{parity}}$. 
\end{example}

\chapter{'t Hooft anomalies of discrete gauge theories}\label{Chapter: t Hooft}
This final chapter is concerned with the study of anomalies of 
Dijkgraaf-Witten theories. These are gauge theories with discrete gauge group. 
They provide a mathematically tractable toy modle for the topological aspects of 
quantum gauge theories, but also have applications in the description of symmetry protected
topological phases of matter~\cite{Wang:2014oya,PhysRevB.92.045101,Yoshida:2017xqa,He:2016xpi,Cong2017,Cong:2017ffh,Delcamp2017,Tiwari2017,Wen2017}. We start in Section~\ref{Sec: Parallel transport} by defining
Dijkgraaf-Witten theories as extended functorial field theories by orbifolding an 
invertible field theory constructed from higher flat gerbes. 
In Section~\ref{Sec: DW Sym} we study symmetries of discrete gauge theories and their 't Hooft anomalies. 
In Section~\ref{Sec: Boundary DW} we realize gauged versions of Dijkgraaf-Witten theories with 't Hooft anomaly as 
boundary states of higher dimensional invertible field theories. For this we develop a pushforward construction
for relative field theories generalizing previous work by Schweigert and Woike~\cite{OFK} from ordinary
to relative field theories.  

\section{Parallel transport of higher flat gerbes as an invertible homotopy quantum field theory}\label{Sec: Parallel transport}
In this section we define classical and quantum Dijkgraaf-Witten theories~\cite{DijkgraafWitten} as extended
functorial field theories. The non-extended case goes back to the work of Freed and Quinn~
\cite{FreedQuinn}.
The once extended case has been constructed in 3-dimensions by Morton~\cite{Morton}; 
see also~\cite{FLHT} for a discussion of the fully extended 3-dimensional version. 

We construct the classical Dijkgraaf-Witten theory as a special example of a more general construction.
Let $T$ be a manifold. 
For a line bundle with connection over $T$ we can compute the holonomy along closed
loops in $T$, describing the coupling of a point particle to an electromagnetic background
gauge field. For higher dimensional sigma models appearing for example in string theory one requires 
gauge fields which can be coupled to higher dimensional world volumes. Their global 
geometry is described by gerbes. For an $n$-gerbe we can compute its holonomy over 
an $n+1$-dimensional world volume.  

Gerbes are higher categorical generalizations of line bundles. They are classified 
by the Deligne hypercohomology group $H^{n+1}(T;\mathcal{D}(n+1))$~\cite{Deligne}, i.e.\ 
the hypercohomology of the complex $\mathcal{D}(n+1)$
\begin{align}
\underline{U(1)}_T \xrightarrow{\ \dd \log \ } \Omega^1(T) \overset{\dd }{\longrightarrow} \Omega^2(T)\overset{\dd }{\longrightarrow}  \dots \overset{\dd }{\longrightarrow}  \Omega^{n+2}(T)
\end{align} 
of sheaves on $T$, where $\underline{U(1)}_T$ is the sheaf of smooth $U(1)$-valued 
functions. The curvature of an $n$-gerbe is a $n+2$-form on $T$ and a gerbe is called 
\emph{flat} if its curvature vanishes. Flat $n$-gerbes on $T$ are classified by the 
ordinary (singular) cohomology group $H^{n+1}(T;U(1))$, see for example~\cite[Section 3.1]{Konrad}. For 
$\theta\in H^{n+1}(T;U(1))$ we can compute the holonomy of the gerbe corresponding to $\theta$ for 
an embedded $n+1$-dimensional compact oriented worldvolume 
$\iota \colon M\longrightarrow T$ by 
\begin{align}
\int_M \iota^* \theta = \langle \iota^*\theta, \sigma_M \rangle , 
\end{align}  
where we denote by $\sigma_M\in H_{n+1}(M)$ the fundamental class\footnote{
Recall that the choice of an orientation for a connected orientable compact $n$-dimensional manifold  
$M$ is equivalent to the choice of a generator $\sigma_M$ of $ H_n(M)\cong \Z$. The
generator $\sigma_M $ is called the \emph{fundamental class} of $M$.} of $M$ and by 
\begin{align}
\langle - ,- \rangle\colon H^*(M;U(1))\otimes H_*(M) \longrightarrow U(1) 
\end{align} 
the evaluation of cochains on chains. Hence, flat gerbes and their parallel transport can
be described purely using algebraic topology. For this reason, when working with 
flat gerbes, we can replace the manifold $T$ by an arbitrary topological space. In this section we 
construct for every $n$-cocycle $\theta$ in singular cohomologoy 
an invertible extended functorial field theory describing the parallel transport 
for the corresponding flat gerbe. 

Classical Dijkgraaf-Witten theory is obtained by setting $T=BG$ for a finite group $G$.
We construct the corresponding quantum theory using the 
extended orbifold construction of Schweigert and Woike~\cite{EOFK} and comment on its relation 
to representation theory. 
   
\subsection{Simplifications of the cobordism bicategory}
The background fields $\F$ relevant for classical Dijkgraaf-Witten theories 
are principal bundles with finite structure 
group $G$ and orientations. These background fields are topological in nature and hence lead to huge
simplifications in the description of the bicategory $\ECobF$ introduced in Section~\ref{Sec: Cob bicat}. Furthermore, the simplified version  
does not suffer from the shortcomings mentioned in Remark~\ref{Rem: Shortcomings bicategory}.
In this section we define this simplified bicategory denoted by $\EGCob$.
A convenient way to describe principal $G$-bundles are classifying maps to $BG$. The construction for the 
bordism category works with arbitrary target space $T$. We work in this generality as long as 
it does not cause any technical problems. Functorial field theories based on bordism categories
of this type are called \emph{homotopy quantum field theories} in \cite{turaev2010homotopy}.

The domain of definition for an extended homotopy quantum field theory with target space $T$ is the symmetric monoidal bordism bicategory $\ETCob$ of $T$-bordisms.

\begin{definition}\label{defmscbordcattarget}
For $n\ge 2$ and any non-empty topological space $T$, which we will refer to as the \emph{target space}, we define the bicategory $\ETCob$ as follows: 
	\begin{enumerate}
		\item[(0)] Objects are pairs $(S,\xi)$ consisting of an $n-2$-dimensional oriented closed manifold $S$ and a continuous map $\xi: S \longrightarrow T$.

		\item[(1)]
		A 1-morphism $(\Sigma,\varphi) : (S_0,\xi_0) \longrightarrow (S_1,\xi_1)$ is an oriented compact collared bordism $(\Sigma,\chi_-,\chi_+) : S_0 \longrightarrow S_1$ (by this we mean a compact oriented $n-1$-dimensional manifold $\Sigma$ with boundary equipped with orientation preserving diffeomorphisms $\chi_-\colon [0,1) \times S_0  \longrightarrow \Sigma_-$ and $\chi_+:  (-1,0]\times S_1  \longrightarrow \Sigma_+$ with $\Sigma_- \cup \Sigma_+$ being a collar of $\partial \Sigma$),
and a map $\varphi : \Sigma \longrightarrow T$ making the diagram
\begin{equation}
\begin{tikzcd}
\, & \Sigma \arrow{dd}{\varphi}  & \\
\{0\}\times S_0  \arrow{ru}{\chi_-} \arrow[swap]{dr}{\xi_0\circ \operatorname{pr}_{S_0}} & &  \{0\} \times S_1   \arrow[swap]{lu}{\chi_+} \arrow{dl}{\xi_1\circ \operatorname{pr}_{S_1}} \\
& T  &
\end{tikzcd}
\end{equation}
commute. No compatibility on the collars is assumed.  We define composition of 1-morphisms by gluing of bordisms along collars and maps. The collars are needed to define a smooth structure on the composition. The identities are  cylinders equipped with the trivial homotopy, where trivial means constant along the cylinder axis. 
		
\begin{figure}[hbt]\centering
\begin{overpic}[width=15.5cm]
	{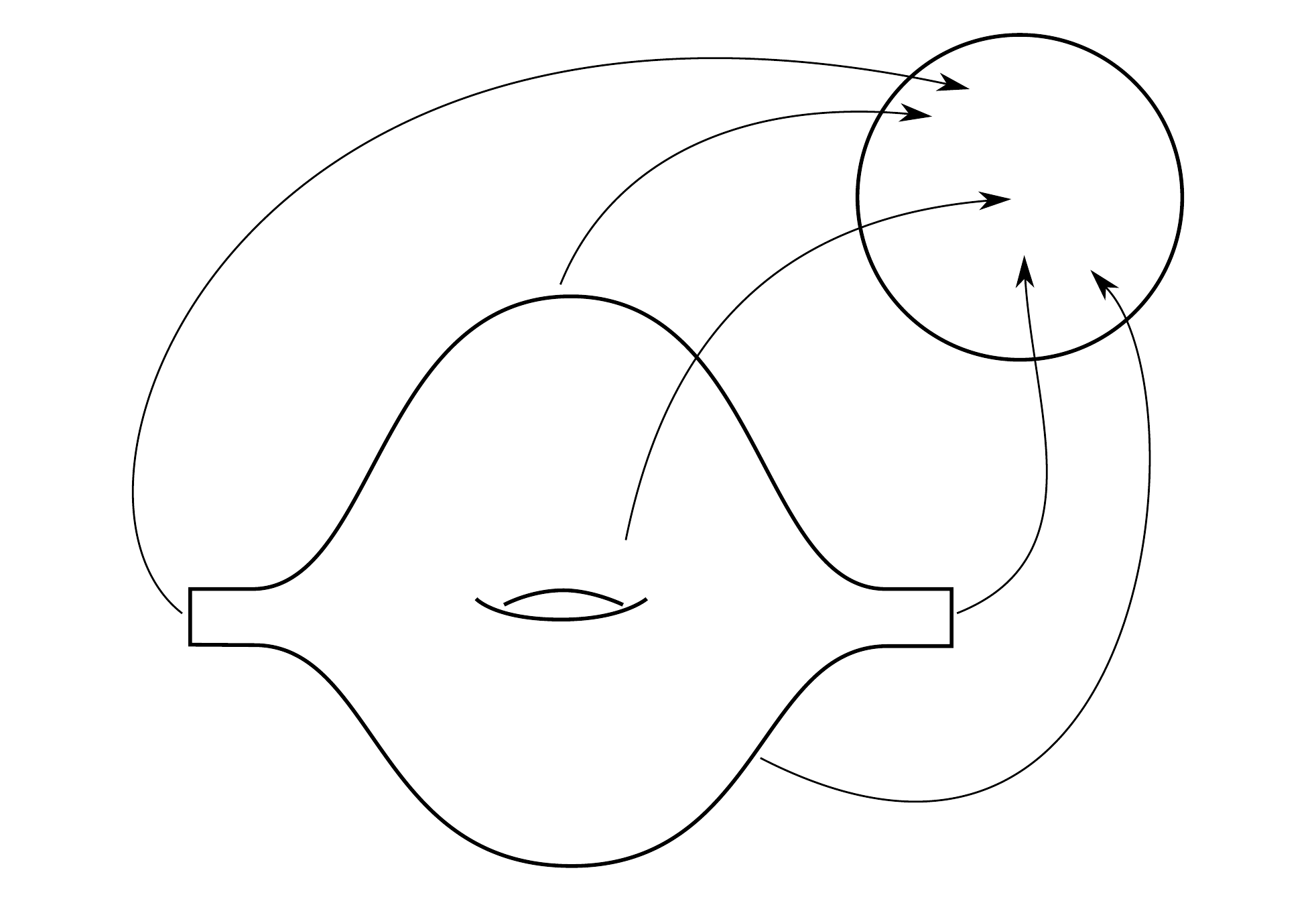}\put(79,53){$T$}\put(40,30){$M$}\put(40,0){$\Sigma_a$}
\put(34,45){$\Sigma_b$}
\put(80,10){$\varphi_a$}\put(40,51){$\varphi_b$}
\put(60,45){$\psi$}\put(-2,35){$ [0,1] \times \xi_0 $}
\put(68,30){$ [0,1] \times \xi_1$}
\end{overpic}
\caption{Sketch of a 2-morphism.}
\label{Fig:Sketch 2-Morphism}
\end{figure}
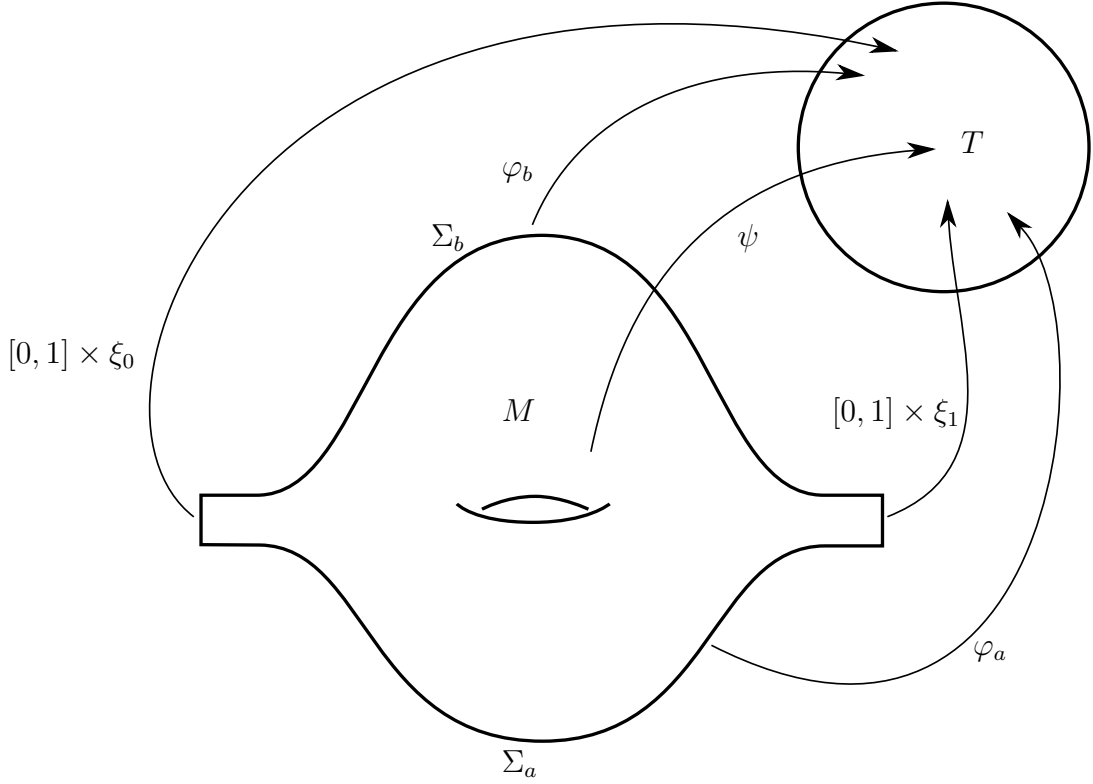

		\item[(2)]
		A 2-morphism $(\Sigma,\varphi) \Longrightarrow (\Sigma',\varphi')$ between 1-morphisms $(S_0,\xi_0) \longrightarrow (S_1,\xi_1)$ is defined to be an equivalence class of pairs $(M,\psi)$ consisting of an $n$-dimensional collared compact oriented bordism $M : \Sigma \longrightarrow \Sigma'$ with corners and a map $\psi : M \longrightarrow T$ (see Figure \ref{Fig:Sketch 2-Morphism}). Concretely, this consists of a compact oriented 
$\langle 2 \rangle$-manifold $M$ equipped with
\begin{itemize}
\item a decomposition of the 0-boundary $\partial_0 M = \partial_0 M_- \cup \partial_0 M_+$ and orientation preserving diffeomorphisms $\delta_- :  [0,1) \times \Sigma \longrightarrow M_-$ and $\delta_+ : (-1,0] \times\Sigma'   \longrightarrow M_+$ to collars of this decomposition,
			
			\item a decomposition of the 1-boundary $\partial_1 M = \partial_1 M_- \cup \partial_1 M_+$ and orientation preserving diffeomorphisms $\alpha_- :  [0,1] \times [0,1)  \times S_0  \longrightarrow M'_-$ and $\alpha_+ : [0,1]\times  (-1,0] \times S_1 \longrightarrow M'_+$ to collars of this decomposition
		making for some $ \varepsilon >0$ the diagrams
			\begin{equation}
			\label{Condition Collars 1}
			\begin{tikzcd}
			 { [0,\varepsilon) \times [0,1) \times S_0}  \ar{r}{\alpha_-}  \ar[swap]{rd}{\id \times \chi_- }& M &  { [0,\varepsilon) \times  (-1,0] \times S_1}  \ar[swap]{l}{\alpha_+} \ar{ld}{ \id \times \chi_+ } \\
			& {[0,\varepsilon) \times \Sigma}  \ar[swap]{u}{\delta_-}&
			\end{tikzcd},
			\end{equation}
			\begin{equation}
			\label{Condition Collars 2}
			\begin{tikzcd}
			\,          & { (-\varepsilon,0] \times \Sigma'}  \ar{d}{\delta_+}   & \\
			 { (1-\varepsilon,1] \times [0,1)  \times S_0 }  \ar{r}{\alpha_-} \ar{ru}{  \id-1 \times \chi'_- }& M &  (1-\varepsilon,1] \times (-1,0] \times S_1   \ar[swap]{l}{\alpha_+} \ar[swap]{lu}{\id-1 \times \chi'_+ }
			\end{tikzcd}
			\end{equation} and
			\begin{equation}
			\begin{tikzcd}
			\,        & M \ar{dd}{\psi} & \\
			{[0,1]\times S_0}  \sqcup \Sigma  \ar{ru}{\alpha_- \sqcup \delta_-} \ar[swap]{dr}{\xi_0\circ \text{pr}_{S_0} \sqcup \varphi} & & {[0,1] \times S_1  \sqcup \Sigma'} \ar[swap]{lu}{\alpha_+ \sqcup \delta_+} \ar{dl}{\xi_1\circ \text{pr}_{S_1} \sqcup \varphi'} \\
			& T &
			\end{tikzcd}
			\end{equation}
			 commute. Again no compatibility on the collars is assumed.
		\end{itemize}

		We define two pairs $(M,\psi)$ and $(\widetilde M,\widetilde \psi)$ to be equivalent if we can find an orientation-preserving diffeomorphism $\Phi : M \longrightarrow \widetilde M$ such that the diagram
		\begin{equation}
		\begin{tikzcd}
		\,        & M \ar{dd}{\Phi} & \\
		{[0,\epsilon)  \times  \Sigma} \ar{ru}{\delta_-} \ar[swap]{dr}{\widetilde \delta_-} & & {(-\epsilon,0] \times \Sigma' } \ar[swap]{lu}{\delta_+} \ar{dl}{\widetilde \delta_+} \\
		& \widetilde M &
		\end{tikzcd} 
		\end{equation}
		and a similar diagram involving the collars of the 1-boundary commute for small enough $\epsilon > 0$ and if moreover there exists a homotopy relative boundary from $\psi\colon M \longrightarrow T$ to $ \widetilde \psi \circ \Phi \colon M \longrightarrow T$.

\end{enumerate}
In order to define the vertical composition of 2-morphisms we fix a diffeomorphism $[0,2]\longrightarrow [0,1]$ which is equal to the identity on a neighborhood of $0$ and given by $x\longmapsto x-1$ in a neighborhood of $2$.
Now we define the vertical composition by gluing using the collars of 0-boundaries. We use our fixed diffeomorphism to rescale both the ingoing and outgoing 1-collars. 

We define horizontal composition of 2-morphisms by gluing manifolds and maps along 1-boundaries, where the new 0-collars arise from the old ones by restriction to $[0,\varepsilon)$ in a way \eqref{Condition Collars 1} and \eqref{Condition Collars 2} allow us to glue them along the boundary and then rescale the interval.
By disjoint union the structure of a symmetric monoidal bicategory with duals on $\ETCob$ is obtained.
\end{definition} 

\begin{remark}
In contrast to the bicategory constructed in Section~\ref{Sec: Cob bicat} the bicategory $\ETCob$ does 
not contain limit morphisms. The reason for this is that homotopies between maps to $T$ can be implemented by putting the homotopies on cylinders. Furthermore, orientation preserving diffeomorphisms can be implemented via the mapping cylinder construction. 
In general, one has to replace the symmetry groupoid introduced in Section~\ref{Sec: General Theory} by the 2-groupoid consisting of morphisms of this type since homotopies do not compose strictly. In the case of $T=BG$ it is possible to still work 
with a groupoid as we explain in Section~\ref{Sec: Classical DW}. In the present context the truncation $\tr \ETCob$ should as 
well be understood as the restriction to 2-morphisms of this type. This excludes certain invertible cobordisms
not of the form $[0,1]\times \Sigma$, which can appear in high dimensions, see e.g.~\cite[Warning 2.2.8]{Lurie2009a}. 
\end{remark}

\begin{definition}\label{Def: Homotopy QFT}
An \emph{extended homotopy quantum field theory} is a symmetric monoidal 2-functor 
\begin{align}
\ETCob \longrightarrow \Tvs \ \ .
\end{align}
\end{definition}

We denote the endomorphism category of $\emptyset$ in $\ETCob$ by $\TCob$. 
A \emph{homotopy quantum field theory} is a symmetric monoidal functor $\TCob \longrightarrow \fvs$. Via restriction every extended homotopy quantum field theory
induces a homotopy quantum field theory. 
\subsection{The field theory}
To a flat $n-1$-gerbe on a topological space $T$ represented by a singular cocycle $\theta \in Z^n(T;U(1))$ we associate
an extended homotopy quantum field theory
\begin{align} 
T_\theta : \ETCob \longrightarrow \Tvs 
\end{align} which can be understood as the parallel transport operator of $\theta$. The evaluation of $T_\theta$ on an $n$-dimensional closed oriented manifold $M$ together with a map $\psi : M \longrightarrow T$ yields an element in $U(1)$, namely the holonomy of $\theta$ with respect to $\psi$. This section is concerned with the concrete construction
of the field theory $T_\theta$, using similar methods to Section~\ref{Sec: Extended Indext FQFT}. Our construction
has been generalized to non-oriented extended homotopy quantum field theories in~\cite{Young2019}. 

\subsection*{Definition on objects}
Let $(S,\xi)$ be an object of $\ETCob$, i.e.\ a closed oriented $n-2$-dimensional manifold $S$ equipped with a continuous map $\xi : S \longrightarrow T$.
Denote by $\Fund(S)$ the groupoid of fundamental cycles of $S$. Its objects are fundamental cycles of $S$, i.e.\ those elements of $Z_{n-2}(S)$ representing the fundamental class of $S$ in $H_{n-2}(S)$. A morphism $\sigma \longrightarrow \sigma'$ between two fundamental cycles is an $n-1$-chain $\tau$ such that $\partial \tau = \sigma'-\sigma$. Composition is given by addition of $n-1$-chains. 

The extended homotopy quantum field theory $T_\theta$ assigns to $(S,\xi)$ a 2-vector space $T_\theta(S,\xi)$ defined as follows.
The objects in $T_\theta(S,\xi)$ are formal finite sums $\bigoplus_{i=1}^n V_i* \sigma_i$, where the $V_i$ are finite-dimensional complex vector spaces and the $\sigma_i$ are objects in $\Fund (S)$. 
We write $\sigma$ for $\C * \sigma$. 
The space of morphisms between $\sigma,\sigma' \in \Fund(S)$ seen as objects of $T_\theta(S,\xi)$ is given by
\begin{align} 
\Hom_{T_\theta(S,\xi)}(\sigma,\sigma') := \frac{\mathbb{C}[\Hom_{\Fund(S)}(\sigma,\sigma')]}{\sim} \ ,      \label{defmorphismon0cells}
\end{align} 
where $\mathbb{C}[\Hom_{\Fund(S)}(\sigma,\sigma')]$ is the free complex vector space on the set 
$\Hom_{\Fund(S)}(\sigma,\sigma')$, and for two morphisms $\tau,\widetilde \tau : \sigma \longrightarrow \sigma'$ we make 
the identification
\begin{align} 
\widetilde \tau \sim \langle \xi^* \theta,\lambda\rangle \tau\ ,  \label{defmorphismon0cells2} 
\end{align} 
whenever $\widetilde \tau-\tau=\partial \lambda$ for some $\lambda \in C_n(S)$.
Note that in \eqref{defmorphismon0cells2} the choice of $\lambda$ does not matter. 
In order to obtain the morphism spaces between all objects in $T_\theta(S,\xi)$, \eqref{defmorphismon0cells} has to be extended bilinearly, i.e.\
\begin{align}
\Hom_{T_\theta(S,\xi)} \left( \bigoplus_{i=1}^n V_i * \sigma_i\ ,\ \bigoplus_{j=1}^m V_j * \sigma_j \right) = \bigoplus_{i,j} \Hom(V_i,V_j) \otimes \Hom_{T_\theta(S,\xi)} (\sigma_i, \sigma_j)
\end{align} for all formal finite sums. The elements of $\Hom_{T_\theta(S,\xi)} \left( \bigoplus_{i=1}^n \sigma_i\ ,\ \bigoplus_{j=1}^m \sigma_j \right)$ can be
interpreted as matrices with equivalence classes of morphisms in $\Hom_{T_\theta(S,\xi)}(\sigma,\sigma')$ as entries.
Composition is defined by matrix multiplication and composition in $\Fund (S)$. 

The $\fvs$-module structure is given by 
\begin{align}
\begin{split} 
* \colon \fvs\times T_{\theta}(S,\xi) &\longrightarrow T_{\theta}(S,\xi) \\
 V\times \left(\bigoplus_{i=1}^n V_i * \sigma_i\right) & \longmapsto \left(\bigoplus_{i=1}^n (V\otimes V_i) * \sigma_i\right) \ .  
\end{split} 
\end{align}
This completes the definition of the 2-vector space $T_\theta(S,\xi)$. Since all the
objects $\sigma_i$ are isomorphic, $T_\theta(S,\xi)$ has one simple object up to 
isomorphism, i.e.\ it is a 2-line.

\subsection*{Definition on 1-morphisms}
Let $(\Sigma,\varphi) : (S_0,\xi_0) \longrightarrow (S_1,\xi_1)$ be a 1-morphism in $\ETCob$. Again, we denote by $\Fund(\Sigma)$ the groupoid of fundamental cycles of $\Sigma$, i.e.\ the groupoid of relative cycles in $C_{n-1}(\Sigma)$ representing the fundamental class of $\Sigma$ in $H_{n-1}(\Sigma,\partial \Sigma)$. For fundamental cycles $\sigma_0$ and $\sigma_1$ of $S_0$ and $S_1$, respectively, we denote by $\Fund_{\sigma_0}^{\sigma_1}(\Sigma)$ the subgroupoid of $\Fund(\Sigma)$ spanned by all fundamental cycles $\mu$ of $\Sigma$ with $\partial \mu = \sigma_1-\sigma_0$. Here we suppress the inclusion of the ingoing and outgoing boundary into $\Sigma$ in the notation. By \cite[VI.,~Lemma~9.1]{Bredon} the groupoid $\Fund_{\sigma_0}^{\sigma_1}(\Sigma)$ is non-empty and connected. 

In order to define the 2-linear map $T_\theta(\Sigma,\varphi) : T_\theta (S_0,\xi_0) \longrightarrow T_\theta(S_1,\xi_1)$ we define on the free vector space $\mathbb{C}[\Fund_{\sigma_0}^{\sigma_1}(\Sigma)]$ the equivalence relation
\begin{align}
\mu' \sim \langle\varphi^* \theta ,\nu\rangle \mu 
\end{align}
 for any $\nu \in C_n(\Sigma)$ such that $\partial \nu =\mu-\mu'$. We use the notation
\begin{align}
\Sigma^\varphi \spr{\sigma_1, \sigma_0} := \frac{\mathbb{C}[\Fund_{\sigma_0}^{\sigma_1}(\Sigma)]}{\sim}
\end{align} for the quotient and observe that $\Sigma^\varphi \spr{-,-}$ extends to a functor $\Fund^{\opp}(S_1)\times \Fund(S_0)\longrightarrow \fvs$, which is
defined on a morphism $\lambda : \sigma_1 \longrightarrow \sigma_2$ in $\Fund(S_1)$ by 
\begin{align}
\begin{split} 
\Sigma^\varphi \spr{-,\sigma} (\lambda):  \Sigma^\varphi \spr{\sigma_2,\sigma} & \longrightarrow \Sigma^\varphi \spr{\sigma_1,\sigma}   \\
\mu &\longmapsto \mu- \lambda \
\end{split}  
\end{align}
and on a morphism $\lambda : \sigma_1 \longrightarrow \sigma_2$ in $\Fund(S_0)$ by 
\begin{align}
\begin{split} 
\Sigma^\varphi \spr{\sigma,-} (\lambda):  \Sigma^\varphi \spr{\sigma,\sigma_1} & \longrightarrow \Sigma^\varphi \spr{\sigma,\sigma_2}   \\
\mu &\longmapsto \mu - \lambda \
\end{split}  
\end{align}
A straightforward calculation shows that this is well-defined.

Now $T_\theta(\Sigma,\varphi) : T_\theta (S_0,\xi_0) \longrightarrow T_\theta(S_1,\xi_1)$ is defined on objects by the coend
\begin{align} 
T_\theta(\Sigma,\varphi) \sigma_0 := \int^{\sigma_1 \in \Fund(S_1)} \Sigma^\varphi \spr{\sigma_1,\sigma_0} * \sigma_1 
\end{align} 
and linear extension.
Here, the coend can be replaced by an end, since it is taken over a groupoid and limits and colimits over essentially finite groupoids taken in a 2-vector space coincide.

\subsection*{Definition on 2-morphisms}
Let $(M,\psi): (\Sigma_a,\varphi_a) \Longrightarrow (\Sigma_b,\varphi_b)$ be a  2-morphism  
between 1-morphisms $(\Sigma_a,\varphi_a) $ and $ (\Sigma_b,\varphi_b)$ from $ (S_0,\xi_0)$ to $ (S_1,\xi_1)$ in $\ETCob$. We assign to $(M,\psi)$ the 2-morphism $T_\theta(\Sigma_a,\varphi_a) \Longrightarrow T_\theta(\Sigma_b,\varphi_b)$ between the 1-morphisms $T_\theta(\Sigma_a,\varphi_a), T_\theta(\Sigma_b,\varphi_b) : T_\theta(S_0,\xi_0) \longrightarrow T_\theta(S_1,\xi_1)$ consisting of the natural maps
\begin{align}
T_\theta(\Sigma_a,\varphi_a) \sigma_0 \longrightarrow T_\theta(\Sigma_b,\varphi_b) \sigma_0
\end{align} for $\sigma_0 \in \Fund(S_0)$ which are the maps between the respective coends induced by the natural transformation
\begin{align}
T_\theta(M)_{\sigma_1,\sigma_0} : \Sigma^{\varphi_a}_a(\sigma_1,\sigma_0)  \longrightarrow\Sigma^{\varphi_b}_b(\sigma_1,\sigma_0) \label{eqnmapbetweencoendsinduce}  
\end{align} defined as follows: For $\mu_a \in \Fund_{\sigma_0}^{\sigma_1}(\Sigma_a)$  we can find a fundamental cycle 
$\nu$ of $M$ with 
\begin{align}
 \partial \nu = \mu_b - \mu_a +( [0,1] \times \sigma_0- [0,1] \times \sigma_1 )
 \label{EQ: Condition Fundamental cycle on manifolds with corners}
\end{align} 
for some fundamental cycle $\mu_b \in \Fund_{\sigma_0}^{\sigma_1}(\Sigma_b)$.
Mapping $\mu_a$ to $\langle \psi^* \theta,\nu\rangle [\mu_b]$ yields a well-defined linear map
$\mathbb{C}[\Fund_{\sigma_0}^{\sigma_1}(\Sigma_a)] \longrightarrow \Sigma_b^{\varphi_b}(\sigma_1,\sigma_0)$, which descends to $\Sigma_a^{\varphi_a}(\sigma_1,\sigma_0)$ and gives us the needed map \eqref{eqnmapbetweencoendsinduce}.
\subsection*{Proof that $T_\theta$ is an extended homotopy quantum field theory} 
To complete the definition of $T_\theta$ we still have to specified
the necessary coherence isomorphisms. We will do this in the proof of the 
following theorem  

\begin{theorem}\label{mainthm}
	For any topological space $T$ and any cocycle $\theta \in Z^n(T;\U(1))$,
	\begin{align} 
	T_\theta : \ETCob \longrightarrow \Tvs 
	\end{align} is an invertible homotopy quantum field theory with target $T$.
	\end{theorem}

The proof of Theorem \ref{mainthm} will be split into different parts which will occupy a large part of the remainder of this section.
We start by proving a useful Gluing Lemma for $\langle 2\rangle$-manifolds:
\begin{lemma}[Gluing Lemma for $\langle 2\rangle$-manifolds]\label{lemmagluing2man}
Consider two $n$-dimensional $\langle 2 \rangle$-manifolds $M_1$ and $M_2$ with representatives for the fundamental class $\nu_1$ and $\nu_2$ such that $\partial \nu_i = \mu_{i,0} + \mu_{i,1}$ for $i=1,2$, where $\mu_{i,0}$ and $\mu_{i,1}$ are representatives for the fundamental class of the 0 and 1 boundary, respectively (note that this implies $\partial \mu_{i,0}= -\partial \mu_{i,1}$). 
Now assume we have an orientation reversing diffeomorphism from a connected component $\Sigma$ of, say, the 1 boundary of $M_1$ onto the 1 boundary of $M_2$ compatible with the fundamental cycles picked above, i.e.\ $\mu_{0,1}|_{\Sigma}= - \mu_{1,1}|_{\Sigma}$. Then $\nu_1+\nu_2$ is a representative of the fundamental class of the manifold obtained by gluing $M_1$ and $M_2$ (see Figure \ref{Fig:Sketch Composition}) along $\Sigma$. 
\end{lemma}

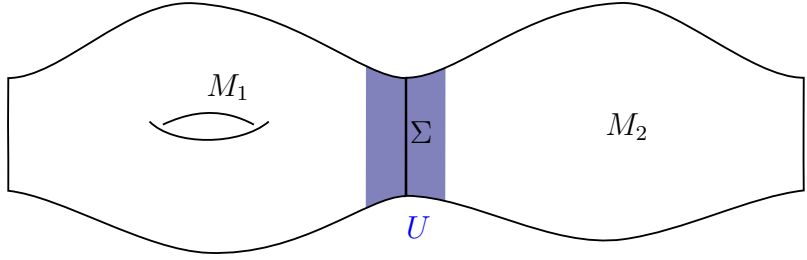
\begin{figure}[htb]\centering
\begin{overpic}[width=10.5cm
]{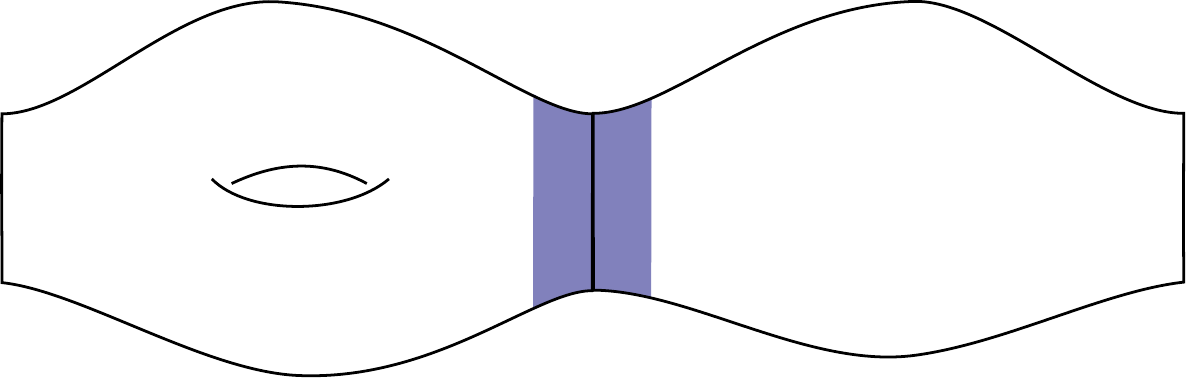}\put(25,20){$M_1$}\put(75,15){$M_2$}\put(50,2){${\color{blue}U}$}\put(50.5,14){$\Sigma$}\end{overpic}\caption{Sketch of the manifolds involved in lemma 2.8.}\label{Fig:Sketch Composition}\end{figure}

\begin{proof}
We denote the  composition of $M_1$ and $M_2$ by $M$.
Obviously, $\nu_1+\nu_2$ is a cycle relative $\partial M$. We have to show that it represents the fundamental class of $M$.
To this end, we use the long exact sequence 
\begin{small}
\begin{align}\label{eqnles1}
\dots \rightarrow H_n(\partial M \cup \Sigma, \partial M) \rightarrow H_{n}(M,\partial M)\rightarrow H_{n}(M, \partial M \cup \Sigma)\rightarrow H_{n-1}(\partial M \cup \Sigma, \partial M)\rightarrow \dots
\end{align}
\end{small}
in homology associated to the triple $\partial M \subset \partial M \cup \Sigma \subset M$ and compute the relevant terms occurring in it by means of a collar $U \cong  (-1,1)  \times \Sigma$ in $M$ (see Figure \ref{Fig:Sketch Composition} for a pictorial presentation):
\begin{itemize}
	\item For the computation of $H_*(\partial M \cup \Sigma, \partial M)$ we define $V:= \partial M \cap U$ and find by excision of the complement $W:= \partial M \setminus V$ of $V$ in $\partial M$
	\begin{align}
	H_*(\partial M \cup \Sigma, \partial M) \cong H_*((\partial M \cup \Sigma) \setminus W, \partial M \setminus W). 
	\end{align} Since the inclusion $(\Sigma,\partial \Sigma) \longrightarrow ((\partial M \cup \Sigma) \setminus W, \partial M \setminus W)$ is a homotopy equivalence, we arrive at
	\begin{align} H_*(\partial M \cup \Sigma, \partial M) \cong H_*(\Sigma,\partial \Sigma).\label{eqngluing1}\end{align} 
	\item For the computation of $H_*(M,\partial M \cup \Sigma)$ we use that the inclusion $(M,\partial M \cup \Sigma) \longrightarrow (M,\partial M \cup U)$ is a homotopy equivalence. After excising $\Sigma$ in $(M,\partial M \cup U)$ we find
	\begin{align}
	H_*(M,\partial M \cup \Sigma) \cong H_*(M_0 \setminus \Sigma,\partial M_0 \cup U_-) \oplus H_*(M_1 \setminus \Sigma,\partial M_1 \cup U_+),
	\end{align} where $U_-$ and $U_+$ is the image of $ (-1,0) \times \Sigma $ and $ (0,1) \times \Sigma $ in $U$, respectively. Since the inclusion $(M_0 \setminus \Sigma,\partial M_0 \cup U_-) \longrightarrow (M_0,\partial M_0 \cup U_- \cup \Sigma)$ induces an isomorphism in homology and since $\Sigma \longrightarrow U_-\cup \Sigma$ is a homotopy equivalence, we obtain $ H_*(M_0 \setminus \Sigma,\partial M_0 \cup U_-)\cong H_*(M_0,\partial M_0)$ and analogously $ H_*(M_1 \setminus \Sigma,\partial M_1 \cup U_+) \cong  H_*(M_1 ,\partial M_1 )$. Thus, we are left with
	\begin{align} H_*(M,\partial M \cup \Sigma) \cong H_*(M_0,\partial M_0) \oplus H_*(M_1 ,\partial M_1 ).\label{eqngluing2}\end{align}
 \end{itemize}
Using \eqref{eqngluing1} and \eqref{eqngluing2} we obtain from \eqref{eqnles1} the exact sequence
\begin{align}
 0 \longrightarrow H_{n}(M,\partial M)\longrightarrow H_{n}(M_0, \partial M_0) \oplus H_{n}(M_1, \partial M_1)\longrightarrow H_{n-1}(\Sigma, \partial \Sigma) \ \    , 
\end{align} where the morphism $H_{n}(M_0, \partial M_0) \oplus H_{n}(M_1, \partial M_1)\longrightarrow H_{n-1}(\Sigma, \partial \Sigma)$ takes $[\nu_1]\oplus [\nu_2]$ to $[\partial \nu_1 +\partial \nu_2]$. 
Evaluating the kernel of this morphism yields an isomorphism $H_{n}(M,\partial M)\cong \Z ([\nu_1]\oplus [\nu_2])$. This shows that $\nu_0+\nu_1$ is a generator of $H_{n}(M,\partial M)$, i.e.\ an orientation of $M$. This orientation agrees with the orientation of $M$ in a neighborhood of an arbitrary point away from the gluing boundary, hence they agree.
\end{proof}
As a corollary we get the following:

\begin{corollary}[Gluing Lemma for manifolds with boundary]\label{lemmamscgluingcyclescob}
	Let $\Sigma_a \colon S_0 \longrightarrow S_1$ and $\Sigma_b \colon S_1 \longrightarrow S_2$ be $n-1$-dimensional cobordism, and let $\nu \in C_{n-1}(\Sigma_a)$ and $\nu' \in C_{n-1}(\Sigma_b)$ be fundamental cycles with $\partial \nu =  \sigma_1 -  \sigma_0$ and $\partial \nu' =  \sigma_2 -  \sigma_1$ for fixed fundamental cycles $\sigma_j \in Z_{n-2}(S_j)$, $j=0,1,2$. Then by $\nu'\circ \nu:= \nu'+\nu \in  C_{n-1}(\Sigma_b \circ \Sigma_a) $ we get a fundamental cycle of $\Sigma_b \circ \Sigma_a$ satisfying $\partial (\nu'\circ \nu) =  \sigma_2 -  \sigma_0$.
\end{corollary}

Using Corollary~\ref{lemmamscgluingcyclescob} we show that $T_\theta$ respects the composition of 1-morphisms up to coherent natural isomorphism. This is a crucial part of the 2-functoriality of $T_\theta$. 

\begin{lemma}\label{lemmacompof1mor}
$T_\theta$ respects the composition of 1-morphisms up to coherent natural isomorphism.
\end{lemma}

\begin{proof}
The coherence isomorphisms consist of natural isomorphisms (see Definition~\ref{Definition Morphism Bicategory})
\begin{align}\label{eqnnatisocohuni}
\Phi_{(S,\xi)}\colon \id_{T_\theta(S,\xi)} \Longrightarrow T_\theta(\id_{S,\xi}) 
\end{align}
for all objects $(S,\xi)\in \ETCob$ and \begin{align}
\Phi_{(\Sigma_a,\varphi_a),(\Sigma_b,\varphi_b)}\colon T_\theta(\Sigma_b,\varphi_b) \circ T_\theta(\Sigma_a,\varphi_a) \Longrightarrow T_\theta((\Sigma_b,\varphi_b)\circ (\Sigma_a,\varphi_a))\label{eqnnatisocohasso}
\end{align}
for all composable 1-morphisms
\begin{align}
(\Sigma_a,\varphi_a) : (S_0,\xi_0) \longrightarrow (S_1,\xi_1) \text{ and } (\Sigma_b,\varphi_b) : (S_1,\xi_1) \longrightarrow (S_2,\xi_2)
\end{align}
in $\ETCob$.

Using the enriched co-Yoneda lemma from Example~\ref{Exa: limits as ends} we can write the identity as the coend
\begin{align}
\id_{T_\theta(S,\xi)}(-) \cong \int^{\sigma \in {T_\theta(S,\xi)}} \Hom_{T_\theta(S,\xi)} (\sigma , -)* \sigma .
\end{align} 
Without loss of generality, we can evaluate this at a generator $\sigma_0\in \Fund(S_1)$
\begin{align}
\sigma_0 \cong \int^{\sigma \in {T_\theta(S,\xi)}} \Hom _{T_\theta(S,\xi)} ( \sigma , \sigma_0)* \sigma \cong \int^{\sigma \in {\Fund(S)}} \Hom _{T_\theta(S,\xi)}( \sigma , \sigma_0)* \sigma.\label{cohid1eqn}
\end{align}  
On the other hand, we have
\begin{align}
T_\theta(\id_{S,\xi})(\sigma_0) = \int ^{\sigma \in {\Fund(S)}}  ([0,1] \times S)^{  [0,1] \times \xi }  (\sigma , \sigma_0)*   \sigma \ \ . \label{cohid2eqn}
\end{align}
There is a natural isomorphism
\begin{align}
\begin{split} 
( [0,1] \times S)^{  [0,1] \times \xi} (\sigma , \sigma_0) & \longrightarrow \Hom_{T_\theta(S,\xi)} ( \sigma , \sigma_0) \\
\mu & \longmapsto - {p_S}_* \mu
\end{split} 
\end{align}
using the projection $p_S: [0,1] \times S \longrightarrow S$. It induces an isomorphism  between the coends in \eqref{cohid1eqn} and \eqref{cohid2eqn} and gives us the desired isomorphism \eqref{eqnnatisocohuni}.

To specify the natural isomorphism \eqref{eqnnatisocohasso} for 1-morphisms $(\Sigma_a,\varphi_a) : (S_0,\xi_0) \longrightarrow (S_1,\xi_1)$ and $(\Sigma_b,\varphi_b) : (S_1,\xi_1) \longrightarrow (S_2,\xi_2)$ in $\ETCob$, we note that for a generator $\sigma_0 \in \Fund(S_0)$
\begin{align}
\begin{split} 
(T_\theta(\Sigma_b,\varphi_b) \circ & T_\theta(\Sigma_a,\varphi_a)) \sigma_0 = \int^{(\sigma_1,\sigma_2) \in \Fund (S_1)\times \Fund (S_2)}  \Sigma^{\varphi_a}_a\spr{ {\sigma_1},{\sigma_0} } \otimes \Sigma_b^{\varphi_b}\spr{ {\sigma_2},{\sigma_1}}  * \sigma_2 \  \\ &\cong \int^{\sigma_2 \in \Fund (S_2)} \left(\int^{\sigma_1 \in \Fund (S_1)}  \Sigma^{\varphi_a}_a \spr{ {\sigma_1},{\sigma_0} } \otimes \Sigma_b^{\varphi_b}\spr{ {\sigma_2},{\sigma_1}} \right) * \sigma_2,
\end{split}
\label{eqnrefinnercoend}
\end{align} where we have used Fubini's Theorem for coends (Theorem~\ref{Thm: Fubini}). 

To compute the inner coend in \eqref{eqnrefinnercoend}, we observe that 
\begin{align}
\begin{split} 
\Phi_{\sigma_1,\sigma_2} :  \Sigma_a^{\varphi_a} \spr{\sigma_0,\sigma_1}\otimes \Sigma_b^{\varphi_b}\spr{\sigma_1,\sigma_2} & \longrightarrow (\Sigma_b \circ \Sigma_a)^{\varphi_b \cup \varphi_a} \spr{ \sigma_2,\sigma_0}  \\
\mu_1\otimes \mu_2 & \longmapsto \mu_1 +\mu_2 \ .
\end{split} 
\end{align} is a canonical isomorphism by Lemma~\ref{lemmamscgluingcyclescob}. Here $\varphi_b \cup \varphi_a : \Sigma_b \circ \Sigma_a\longrightarrow T$ is the map obtained from gluing $\varphi_a$ and $\varphi_b$. 
We now obtain
\begin{align} 
\begin{split} 
\int^{\sigma_1 \in \Fund (S_1)} \Sigma_a^{\varphi_a} \spr{ {\sigma_1},{\sigma_0} } \otimes \Sigma_b^{\varphi_b}\spr{{\sigma_2},{\sigma_1}} &\cong \int^{\sigma_1 \in \Fund (S_1)} (\Sigma_b \circ \Sigma_a)^{\varphi_b \cup \varphi_a} \spr{ \sigma_2,\sigma_0}  \\&\cong  \int^{\sigma_1 \in \Fund (S_1)} (\Sigma_b \circ \Sigma_a)^{\varphi_b \cup \varphi_a} \spr{ \sigma_2,\sigma_0} \otimes \mathbb{C}  \\&\cong (\Sigma_b \circ \Sigma_a)^{\varphi_b \cup \varphi_a} \spr{ \sigma_2,\sigma_0},
\end{split}  
\end{align} 
where in the last step we used that $\Fund (S_1)$ is connected. Insertion into \eqref{eqnrefinnercoend} yields isomorphisms
\begin{align}
(T_\theta(\Sigma_b,\varphi_b) \circ T_\theta(\Sigma_a,\varphi_a)) \sigma_0 \cong (T_\theta( (\Sigma_b,\varphi_b) \circ (\Sigma_a,\varphi_a)) \sigma_0,
\end{align} which give us after linear extension the natural isomorphism \eqref{eqnnatisocohasso}.
\end{proof}

In the next lemma we show that $T_\theta$ is well defined on the equivalence classes
corresponding to a 2-morphism in $\ETCob$:

\begin{lemma}\label{lemmahomotopyinvariance}
$T_\theta$ is well defined on 2-morphisms. 
\end{lemma}

\begin{proof}
The non-trivial statement to show is the invariance under gauge 
transformation relative to the boundary.
Consider 2-morphisms $(M,\psi), (M,\psi') \colon  (\Sigma_a,\varphi_a) \Longrightarrow (\Sigma_b,\varphi_b)$ between 1-morphisms $(S_0,\xi_0) \longrightarrow (S_1,\xi_1)$ with $\psi \stackrel{h}{\simeq} \psi'$ relative $\partial M$. Let $\sigma_0$ and $\sigma_1$ be fundamental cycles of $S_0$ and $S_1$, respectively. Now for $\mu_a \in \Fund_{\sigma_0}^{\sigma_1} (\Sigma_a)$ and $\mu_b \in\Fund_{\sigma_0}^{\sigma_1} (\Sigma_b)$ we can 
find a fundamental cycle $\nu$ of $M$ adapted to $\mu_a$ and $\mu_b$ as in equation \eqref{EQ: Condition Fundamental cycle on manifolds with corners}. By definition of $T_\theta$ on 2-morphisms it suffices to show
\begin{align}
\langle \psi^*\theta,\nu\rangle = \langle {\psi'}^*\theta,\nu\rangle . 
\end{align} 
Indeed, if we see $h$ as a map defined on $ [0,1] \times M$, we find a chain homotopy between the chain maps $\psi_*$ and $\psi'_*$ given by
\begin{align} 
{\psi_*}_{p}- {\psi'_*}_{p} = \partial H_p + H_{p-1} \partial \quad \text{for all $p\in\mathbb{Z}$}, 
\end{align}
where $H_p := {h_*}_{p+1} D_p$ and 
\begin{align} 
D_p : S_p(M) \longrightarrow S_{p+1}([0,1]\times M), \quad c \longmapsto [0,1] \times c 
\end{align}
 is defined using the cross-product on singular chains, see \cite[IV.16]{Bredon}. Hence,
\begin{align}
\langle \psi^*\theta,\nu\rangle - \langle {\psi'}^*\theta,\nu\rangle = \langle  \theta, H_{n-1} \partial \nu  \rangle = \langle h^* \theta,  [0,1] \times \partial \nu  \rangle .\label{eqnhtpinvvanish}
\end{align}
The homotopy $h$ being stationary on the boundary entails
\begin{align} 
h|_{[0,1] \times \partial M }= \psi \circ p_{\partial M}
\end{align} 
with the projection $p_{\partial M} \colon \partial [0,1] \times M  \longrightarrow \partial M$. This yields
\begin{align}
\langle h^* \theta,  [0,1] \times \partial\nu  \rangle = \langle \psi^* \theta, {p_{\partial M}}_* ([0,1] \times \partial\nu ) \rangle . 	\end{align} 
We have 
\begin{align}
 \partial \left( {p_{\partial M}}_* ([0,1] \times \partial\nu )\right) = 0,
\end{align}
where we use that the boundaries corresponding to the $[0,1]$ part cancel under the projection, i.e.\ 
\begin{align}
{p_X}_* (\partial [0,1]  \times X)= 0,
\end{align} 	
for any space $X$ with projection $p_X : [0,1] \times  X \longrightarrow X$.
This shows that ${p_{\partial M}}_* ( [0,1] \times \partial\nu)$ is a cycle.
For dimensional reasons it must be a boundary as well. This shows that \eqref{eqnhtpinvvanish} vanishes and finishes the proof.
\end{proof}
Using the Gluing Lemma for $\langle 2\rangle$-manifolds~\ref{lemmagluing2man} we prove that $T_\theta$ strictly preserves the vertical composition of 2-morphisms:
\begin{lemma}\label{lemmaverticalcomp}
	$T_\theta$ preserves the vertical composition of 2-morphisms strictly.
\end{lemma}
\begin{proof}
Given two 2-morphisms 
\begin{align}
(M,\psi) :  (\Sigma_a,\varphi_a) \Longrightarrow (\Sigma_b,\varphi_b)\text{ and }(M',\psi') \colon (\Sigma_b,\varphi_b) \Longrightarrow (\Sigma_c,\varphi_c)
\end{align} 
between 1-morphisms $(S_0,\xi_0) \longrightarrow (S_1,\xi_1)$ it suffices to show that for fundamental cycles $\sigma_0$ and $\sigma_1$ of $S_0$ and $S_1$, respectively, the composition of linear maps
\begin{align}
	\Sigma^{\varphi_a}_a (\sigma_1,\sigma_0) \xrightarrow{T_\theta (M)_{\sigma_1,\sigma_0}} \Sigma^{\varphi_b}_b (\sigma_1,\sigma_0)   \xrightarrow{T_\theta (M')_{\sigma_1,\sigma_0}} \Sigma^{\varphi_c}_c (\sigma_1,\sigma_0) \end{align} as defined in \eqref{eqnmapbetweencoendsinduce} is equal to
\begin{align}
		\Sigma^{\varphi_a}_a (\sigma_1,\sigma_0) \xrightarrow{T_\theta (M'\circ M)_{\sigma_1,\sigma_0}}  \Sigma^{\varphi_c}_c (\sigma_1,\sigma_0) .
		\end{align} 
	Picking fundamental cycles $\nu$ and $\nu'$ for $M$ and $M'$ as in Lemma~\ref{lemmagluing2man} this follows from \begin{align}
	\langle \psi^* \theta ,\nu\rangle \cdot \langle {\psi'}^* \theta,\nu'\rangle  =\langle   (\psi'\cup \psi)^* \theta,\nu+\nu' \rangle,
	\end{align} where $\psi' \cup \psi : M'\circ M \longrightarrow T$ is the map obtained by gluing $\psi$ and $\psi'$.
\end{proof}
Now we can complete the proof of Theorem~\ref{mainthm}:
\begin{proof}[Proof of Theorem~\ref{mainthm}]
	Thanks to the Lemmata~\ref{lemmacompof1mor}, \ref{lemmahomotopyinvariance} and \ref{lemmaverticalcomp} it remains to prove the following: 
	\begin{itemize}

\item 
Horizontal composition:
Let 		
\begin{align}
(\Sigma_a,\varphi_a) : (S_0,\xi_0) \longrightarrow (S_1,\xi_1), \ \ 
(\Sigma_b,\varphi_b) : (S_1,\xi_1) \longrightarrow (S_2,\xi_2)
\end{align} 
be 1-morphisms and
\begin{align}
(M,\psi) : 	(\Sigma_a,\varphi_a) \Longrightarrow 	(\Sigma_b,\varphi_b), \ \
(M',\psi') : 	(\Sigma_b,\varphi_c) \Longrightarrow 	(\Sigma_c,\varphi_c)
\end{align} 
2-morphisms. 
We have to show that for fundamental cycles $\sigma_0$, $\sigma_1$ and $\sigma_2$ of $S_0$, $S_1$ and $S_2$, respectively, the square 
\begin{equation}
		\begin{tikzcd}
		\Sigma_a^{\varphi_a} (\sigma_1,\sigma_0)\otimes \Sigma_b^{\varphi_b} (\sigma_2,\sigma_1) \ar{rr}{\Phi} \ar[swap]{dd}{T_\theta (M)_{\sigma_1,\sigma_0}\otimes T_\theta (M')_{\sigma_2,\sigma_1}} & & (\Sigma_b \circ \Sigma_a)^{\varphi_b\cup \varphi_a} (\sigma_2,\sigma_1) \ar{dd}{T_\theta (M' \circ M)_{\sigma_2,\sigma_0}} \\
		& & \\
		\Sigma_b^{\varphi_b} (\sigma_1,\sigma_0)\otimes \Sigma_c^{\varphi_c} (\sigma_2,\sigma_1) \ar{rr}{\Phi} & & (\Sigma_c \circ \Sigma_b)^{\varphi_c\cup \varphi_b} (\sigma_2,\sigma_1)
		\end{tikzcd}
		\end{equation} featuring as the horizontal arrows the isomorphisms from the proof of Lemma~\ref{lemmacompof1mor} commutes.
This can be verified directly by picking representatives for the fundamental classes as in Lemma~\ref{lemmagluing2man}.
		
		\item Symmetric monoidal structure:
		there are natural equivalences of categories 
		\begin{align}
		\begin{split} 
		\iota_\theta : T_\theta (\emptyset) &\longrightarrow \fvs \\
		\sigma_\emptyset &\longmapsto \C
		\end{split}
		\end{align}	
		and 	  
		\begin{align}
		\begin{split}
		\chi_\theta((S_0, \xi_0),(S_1,\xi_1))\colon T_\theta (S_0,\xi_0)\boxtimes  T_\theta (S_1,\xi_1)  &\longrightarrow T_\theta(S_0\sqcup S_1, \xi_0 \sqcup \xi_1) \\
		\sigma_{S_0}\boxtimes \sigma_{S_1} & \longmapsto \sigma_{S_0} + \sigma_{S_1} ,
		\end{split} 
		\end{align}
		where we suppress the inclusion into the disjoint union and denote by $\sigma_\emptyset$ the unique fundamental cycle of the empty set (that it has by convention).
		The modifications which are part of the structure of a symmetric monoidal 2-functor (see Definition~\ref{Def: Symmetric monoidal 2-functor}) are trivial since the corresponding diagrams commute on generators.
		The simple form of the coherence isomorphism makes it straightforward to check that the corresponding diagrams commute.
	\end{itemize}
The field theory $T_\theta$ is obviously invertible.
\end{proof}

The following assertion shows that up to natural equivalence $T_\theta$ only depends on the cohomology
class of $\theta$. 
\begin{proposition}\label{Prop: Coboundary induces transformations}
	Let $\theta$ and $\theta'$ be $n$-cocycles on a topological space $T$ with values in $\U(1)$ and 
	$\Lambda$ an $n-1$-chain on $T$ satisfying $\operatorname{d}\Lambda= \theta'-\theta$. Then $\Lambda$ induces a 
	symmetric monoidal natural equivalence 
	\begin{align}
	T_\Lambda : T_\theta \longrightarrow T_{\theta'} \ \ .
	\end{align} 
\end{proposition}

\begin{proof}
	
	For all $(S,\xi)\in \ETCob$ we define linear functors 
	\begin{align}
	\begin{split} 
	T_\Lambda(S,\xi) : T_\theta(S,\xi) & \longrightarrow T_{\theta'}(S,\xi) \\
	\sigma & \longmapsto \sigma \\
	[\lambda] &\longmapsto [\langle \xi^*\Lambda, \lambda \rangle \cdot \lambda] \ \ . 
	\end{split} 
	\end{align}
	For a 1-morphism $(\Sigma,\varphi): (S_0,\xi_0)\longrightarrow (S_1,\xi_1)$ we get natural 
	linear maps between the vector spaces 
	\begin{align}
	\begin{split}
	\Sigma^\varphi_\theta(\sigma_1,\sigma_0)&\longrightarrow \Sigma^\varphi_{\theta'}(\sigma_1,\sigma_0) \\
	[\mu] &\longmapsto [\langle \varphi^*\Lambda, \mu \rangle \cdot \mu] \ \ , 
	\end{split} 
	\end{align}
	where we added the subscripts $\theta$ and $\theta'$ to indicate the respective cocycles that enter the definition of the vector spaces $\Sigma^\varphi(-,-)$. These maps 
	induce maps between the corresponding coends and combine into a natural transformations 
	\begin{align}
	T_\Lambda(\Sigma,\varphi): T_{\theta'}(\Sigma,\varphi)\circ T_\Lambda(S_0,\xi_0)\Longrightarrow  T_\Lambda(S_1,\xi_1)\circ T_{\theta}(\Sigma,\varphi)\ \ . 
	\end{align}
	A straightforward computation shows that this defines a natural transformation of 2-functors. 
	Furthermore, it is clear how to equip $T_\Lambda$ with the structure of a symmetric monoidal transformation.
	Finally, we observe that $T_\Lambda$ is even a symmetric monoidal equivalence, because $T_{-\Lambda}$ provides a weak inverse.       
\end{proof}
\begin{remark}
	In the same way an $n-2$-chain $\Omega$ satisfying $\operatorname{d}\Omega=\Lambda'-\Lambda$ induces symmetric monoidal 
	modifications between the natural transformations $T_\Lambda$ and $T_{\Lambda'}$. We do not spell 
	out the details here. 
\end{remark}

\begin{remark}\label{Rem: Non extended theory}
Restricting $T_\theta $ to the endomorphisms of the empty set induces a non-extended homotopy quantum field theory 
\begin{align}
T_\theta :  T\text{-}\mathbf{Cob}_{n} \longrightarrow \fvs \ ,
\end{align}
which admits the following concrete description
\begin{itemize}
\item 
To a closed $n-1$ dimensional manifold $\Sigma$ equipped with a map $\varphi : \Sigma \longrightarrow T$ it assigns the vector space $T_\theta(\Sigma, \varphi)=\Sigma^\varphi(\emptyset,\emptyset)= \C[\Fund(\Sigma)]/{\sim}$.

\item 
To a morphism $(M,\psi): (\Sigma_a,\varphi_a) \longrightarrow (\Sigma_b,\varphi_b)$ it assigns the linear map 
\begin{align}
\begin{split}
T_\theta(M,\psi): T_\theta(\Sigma_a,\varphi_a) & \longrightarrow T_\theta (\Sigma_b,\varphi_b) \\
[\sigma_{\Sigma_a}]& \longmapsto \langle \psi^*\theta, \sigma_M \rangle [\sigma_{\Sigma_b}] \ ,
\end{split}
\end{align}
with $\partial \sigma_M= \sigma_{\Sigma_b}- \sigma_{\Sigma_a}$.
\end{itemize}
This is the primitive homotopy quantum field theory constructed in \cite[I.2.1]{turaev2010homotopy}.
\end{remark}
\begin{remark}
For trivial $\theta=1$ we get a canonical equivalence $T_0 \longrightarrow \mbf1$
defined on objects by 
\begin{align}
 T_0(S,\xi) \ni \bigoplus_{i=1}^n V_i* \sigma_i  
\longmapsto \bigoplus_{i=1}^n V_i \in \fvs  \ \ .  
\end{align}
There is also a natural isomorphism  
in the other direction 
\begin{align}\label{Eq: Omega 0}
\begin{split} 
\Omega_0 \colon \mbf1 &\longrightarrow T_0 \\ 
 \mbf1(S,\xi) \ni \C &\longmapsto \lim_{\sigma\in \Fund(S)} \sigma \in T_0(S,\xi) \ \ .
\end{split}
\end{align}  
We will use the natural 
transformation constructed here in Section~\ref{Sec: State space}.
\end{remark}
\subsection{Classical Dijkgraaf-Witten theory}\label{Sec: Classical DW}
Dijkgraaf-Witten theories~\cite{DijkgraafWitten} are gauge theories with finite gauge group. 
We fix throughout this Section a finite group $G$. Every principal $G$-bundles carries
a unique flat connection. The stack of principal $G$ bundles is equivalent to the 
stack $[\cdot,BG]$ (see for example~\cite{OFK}) which sends a manifold $\Sigma$ to the fundamental groupoid 
of the mapping space $\Map(\Sigma,BG)$. The equivalence as stacks also includes the
statement that homotopy classes of homotopies between classifying maps can be identified
with gauge transformations. 
This is not true for arbitrary Lie groups as the following simple example shows:\footnote{This 
was pointed out to us by Lukas Woike.} the classifying space for $U(1)$ is the
Eilenberg-MacLane space $K(\Z,2)$. This implies $\Pi_1(\Map(\star, BU(1)))\cong \star \DS \star$, but $\Bun_{U(1)}(\star)\cong \star \DS U(1)$.    

The possible actions of $n$-dimensional topological gauge theories with finite gauge group $G$ are classified by the $n$-th cohomology group of the classifying space $BG$ with coefficients in $\R / \Z \cong U(1) $ \cite{DijkgraafWitten}. 
The singular homology of the topological space can be identified with the group 
cohomology~\cite{Weibel}
of $G$. We will switch freely between these
two perspectives. 
For a fixed representative $\theta'\in Z^n(BG;\R / \Z )$ of a cohomology class, the action for a $G$-bundle with 
classifying map $ \psi \colon M \longrightarrow BG$ on an oriented closed $n$-dimensional manifold $M$ is given by 
\begin{align}
S_{\theta'}(M,\psi)= \int_M \, \psi^* \theta' \ .
\end{align}
As a real number, this action is only well-defined modulo $\Z$ by definition. The quantity with physical relevance is the exponentiated action $\exp(2\pi \,\iu\, S_{\theta'}(M,\psi))$ which takes values in $U(1)$. 
For simplicity we work from now on with a cocycle $\theta \in Z^n(BG;U(1))$ and interpret its integration over the manifold as the exponentiated action. 

The corresponding classical field theory is a special case of the construction above
with $T=BG$. It is common to introduce a name for field theories with target $BG$:

\begin{definition}
Let $G$ be a finite group and $n$ a positive integer. We denote by $\EGCob := B\EGCob$ the \emph{symmetric monoidal category of $G$-bordisms} (there is a slight abuse of notation because $\EGCob$ could also describe bordisms with maps to the discrete space $G$; however that would not be interesting). We call an extended homotopy quantum field theory
\begin{align}
Z: \EGCob \longrightarrow \Tvs
\end{align} an \emph{extended $G$-equivariant topological quantum field theory}. 
\end{definition}
\begin{remark} In the non-extended case these appear as homotopy quantum field theories with aspherical targets in \cite{turaev2010homotopy}. Three-dimensional extended $G$-equivariant topological field theories are discussed in \cite{NMS} using the language of principal fiber bundles and with an emphasis on theories of Dijkgraaf-Witten type. A definition of extended $G$-equivariant topological field theories of arbitrary dimension and a detailed investigation of the three-dimensional case including a relation to equivariant modular categories is given in \cite{EOFK}.
\end{remark}
\begin{definition}
Let $G$ be a finite group and $\theta\in Z^n(BG;U(1))$ an $n$-cocycle on $G$ with values 
$U(1)$. The extended equivariant field theory $BG_\theta$ is called the 
\emph{classical Dijkgraaf-Witten theory with topological action $\theta$} and denoted
by \begin{align}
 E_\theta \colon \EGCob \longrightarrow \Tvs \ \ .
\end{align}
\end{definition} 

As explained in Section~\ref{Sec: P anomaly actions} an invertible field theory 
induces (higher) line bundles over the groupoid of background field configurations. 
We now describe the cocycles corresponding to these line bundles explicitly in terms of transgression. 
  
\subsubsection{Transgression}
Let us briefly recall the concept of transgression, see e.g.~\cite{TwistedDWandGerbs}: 
Let $M$ be an $\ell$-dimensional closed oriented manifold with fundamental class $\sigma$. For a topological space $T$ and a class in $H^k(T;\U(1))$ with $k\ge \ell$ represented by a cocycle $\theta$ we can define the class $\tau_M \theta \in H^{k-\ell} (T^M; \U(1))$ as being represented by the cocycle given by
\begin{align}
(\tau_M \theta) (\lambda) := (\text{ev}^* \theta) (\lambda \times \sigma) 
\end{align} for any $k-\ell$-simplex $\lambda : \Delta_{k-\ell} \longrightarrow T^M$, where $\text{ev} : T^M \times M \longrightarrow T$ is the evaluation map. Here $T^M$ is the space of maps $M\longrightarrow T$ equipped with the compact-open topology. 
This gives rise to a map 
\begin{align}
\tau_M : H^k(T;\U(1)) \longrightarrow H^{k-\ell} (T^M; \U(1)),
\end{align} the so-called \emph{transgression}. If $T$ is aspherical, then $T^M$ is equivalent to the groupoid $\Pi_1(M,T)$ of maps from $M \longrightarrow T$ with equivalence classes of homotopies as morphisms. In that case the transgression can be seen to take values in the group cohomology group $H^{k-\ell} (\Pi_1(M,T); \U(1))$. 

\subsubsection{Invariants of closed oriented manifolds equipped with bundles and transgression} 
Evaluating the field theory $E_\theta$ on closed manifolds of dimension $n$, $n-1$, and $n-2$ provides algebraic invariants 
for manifolds equipped with maps into $BG$.  
For a closed oriented $n$-dimensional manifold $M$, this invariant is a complex number given by the function
\begin{align}
E_\theta(M,-) : \Pi_1(M,BG) \longrightarrow\C, \quad \psi \longmapsto \langle \psi^*\theta,\sigma_M \rangle
\end{align} on the groupoid $\Pi_1(M,BG)$. This function is constant on isomorphism classes, i.e.\ it is a 0-cocycle in the cohomology of the groupoid $\Pi_1(M,BG)$. This cocycle is given by the transgression of $\theta$. More precisely, $E(M,-) \in H^0 ( \Pi_1(M,BG); \U(1))$ is the image of $\theta$ under the transgression map $\tau : H^n(M;\U(1)) \longrightarrow H^0 ( \Pi_1(M,BG); \U(1))$. 

We will now show that the invariant (higher line bundle) obtained from $E_\theta$ for manifolds of dimension $n-1$ and $n-2$ equipped with bundles can also be described by an appropriate transgression of $\theta$. 

Evaluating the primitive theory $E_\theta$ on a closed $n-1$-dimensional $\Sigma$ manifold gives a line bundle 
$E_\theta (\Sigma,-) : \Bun_G(\Sigma) \cong \Pi_1 (\Sigma , BG) \longrightarrow \fvs$. 
\begin{proposition}
Let $G$ be a finite group and $\theta \in Z^n(BG;\U(1))$. For any $n-1$-dimensional closed oriented manifold $\Sigma$ the class $\langle E_\theta(\Sigma,-)\rangle \in H^1 (\Pi_1(\Sigma,BG);\U(1))$ describing the line bundle $E_\theta (\Sigma,-) : \Pi _1(\Sigma , BG) \longrightarrow \fvs$ is given by
\begin{align} 
\langle E_\theta(\Sigma,-)\rangle = \tau_\Sigma \theta,
\end{align} i.e.\ by the transgression $\tau_\Sigma \theta$ of $\theta$ to $\Pi_1(\Sigma,BG)$. 
\end{proposition}
\begin{proof}
By abuse of notation we will denote the non-extended theory that $E_\theta$ gives rise to in the sense of Remark~\ref{Rem: Non extended theory} also by $E_\theta$. 
We fix a fundamental cycle $\sigma_\Sigma$ of $\Sigma$. This induces a linear isomorphism 
\begin{align}
\begin{split} 
E_\theta(\Sigma,\varphi)& \longrightarrow \C \\
[\sigma_\Sigma]  &\longmapsto 1
\end{split} 
\end{align}
for all $\varphi : \Sigma \longrightarrow BG$. Consider a morphism $h: \varphi_1\longrightarrow \varphi_2 $ in $\Pi_1(\Sigma,BG)$, i.e.\ a map $h\colon  [0,1] \times \Sigma \longrightarrow BG$. We can factor $h$ as 
\begin{equation}
\begin{tikzcd}
\, [0,1] \times \Sigma \ar[swap]{d}{h \times \id } \ar{r}{h} & BG \\
BG^\Sigma \times \Sigma  \ar[swap]{ru}{\text{ev}} &
\end{tikzcd} 
\end{equation} 
where we denote the image of $h$ under the adjunction $\Sigma \times - \dashv \ (-)^\Sigma $ again by $h$.
The cocycle ${Z}(\Sigma, -)$ evaluated on $h$ is given by 
\begin{align}
\begin{split} 
{Z_\theta}(\Sigma, -)(h)&= \langle h^* \theta , [0,1] \times \sigma_\Sigma   \rangle \\
 &= \langle \text{ev}^* \theta ,h \times \sigma_\Sigma \rangle \\
 &= \tau_\Sigma \theta \ . 
\end{split} 
\end{align}
\end{proof}

\noindent We obtain by evaluation of $E_\theta$ on an $n-2$-dimensional closed oriented manifold $S$ a representation $Z_\theta(S,-) : \Pi_2(S,BG) \longrightarrow \Tvs$ of the second fundamental groupoid of the mapping space $BG^S$ of maps from $S$ to $BG$. Since $BG$ is aspherical, this reduces to a 2-line bundle \begin{align} Z_\theta(S,-) : \Pi_1(S,BG) \longrightarrow \Tvs\end{align} over the groupoid of maps $S \longrightarrow BG$ with equivalence classes of homotopies as morphisms, i.e.\ the groupoid of $G$-bundles over $S$. This is accomplished by pulling back along a fixed equivalence
\begin{align}
\widehat{-} : \Pi_1(S,BG) \longrightarrow \Pi_2(S,BG) \end{align} being the identity on objects and sending a class $h$ of homotopies to an arbitrary, but fixed representative $\widehat{h}$. The coherence isomorphisms for $\widehat{-}$ are unique. 
\begin{theorem}\label{Theorem: Transgression}
	Let $G$ be a finite group and $\theta \in Z^n(BG;\U(1))$. Then for any $n-2$-dimensional closed oriented manifold $S$ the class $\langle E_\theta(S,-)\rangle \in H^2 (\Pi_1(S,BG);\U(1))$ describing the 2-line bundle $E_\theta (S,-) : \Pi_1 (S , BG) \longrightarrow \Tvs$ is given by
	\begin{align} 
	\langle E_\theta(S,-)\rangle = \tau_S \theta, 
	\end{align} i.e.\ by the transgression $\tau_S \theta$ of $\theta$ to $\Pi_1(S,BG)$. 
	\end{theorem}
\begin{proof}
As explained in Section~\ref{Sec: P anomaly actions} there is a 2-functor $\widetilde{E_\theta(S,-)} : \Pi_1 (S,BG) \longrightarrow \FinVect \DS \id \DS \mathbb{C}^\times$ such that 
\begin{equation}
\begin{tikzcd}
\Pi_1(S,BG)  \ar[rr,"{E_\theta(S,-)}"] \ar[rd, "{\widetilde{E_\theta(S,-)}}", swap] &  & \Tvs \\ 
 &  \FinVect \DS \id \DS \mathbb{C}^\times \ar[ru] & 
\end{tikzcd}
\end{equation}
commutes up to natural isomorphism. To compute $\widetilde{E_\theta(S,-)}$, we define for each $\xi \in \Pi_1(S,BG)$ the equivalence
\begin{align} 
\chi_ \xi : Z_\theta(S,\xi) \longrightarrow \FinVect, \quad V * \sigma \longmapsto V 
\end{align} 
and fix the choice of a fundamental cycle $\sigma_S$ of $S$ to obtain a weak inverse
\begin{align} 
\chi_\xi^{-1} : \FinVect \longrightarrow Z_\theta(S,\xi), \quad V \longmapsto V* \sigma_S.
\end{align}
Next for any morphism $h: \xi_0 \longrightarrow \xi_1$ in $\Pi_1(S,BG)$ we define the vector space $V_{\widehat{h}}$ by the (weak) commutativity of the square 
\begin{equation}
\begin{tikzcd}
Z_\theta(S,\xi_0) \ar{rr}{Z( [0,1] \times S,\hat{h})} \ar[swap]{dd}{\chi_{\xi_0}} & & Z_\theta(S,\xi_1)  \\
	& & \\
	\FinVect \ar{rr}{-\otimes V_{\widehat{h}}} & & \FinVect \ar[swap]{uu}{\chi_{\xi_1}^{-1}}
	\end{tikzcd}
\end{equation}
i.e.\
\begin{align}
V_{\widehat{h}} = \chi_{\xi_1} \int^{\sigma \in \Fund(S)} ( [0,1] \times S )^{\widehat h} (\sigma,\sigma_S) * \sigma \cong ( [0,1] \times S)^{\widehat h} (\sigma_S,\sigma_S).\label{eqnvsvh}
\end{align} 
Note that we have a canonical isomorphism
	\begin{align}
	V_{\widehat h} \longrightarrow \mathbb{C}, \ \  [0,1] \times \sigma_S  \longmapsto 1.\label{eqnisovhC}
	\end{align} For two composable morphisms $h$ and $h'$ in $\Pi_1(S,BG)$ we denote the composition by $h'h$ and obtain the 2-isomorphism
\begin{footnotesize}
	\begin{align} E_\theta( [0,1] \times S,\widehat{h'}) E_\theta([0,1] \times S ,\widehat{h})  \xrightarrow{\substack{\text{coherence} \\  \text{of $E_\theta$}}}  E_\theta([0,1] \times S ,\widehat{h'} \widehat{h}) 
	\xrightarrow{\substack{\text{evaluation of $E_\theta$ on} \\  \widehat{h'} \widehat{h} \simeq \widehat{h'h}}} 
	E_\theta([0,1] \times S ,\widehat{h'h}).
	\end{align}
\end{footnotesize}	
 By \eqref{eqnvsvh} this amounts to a map
	\begin{align}
	V_{\widehat h} \otimes V_{\widehat{h'}} \longrightarrow V_{\widehat{h'}\widehat{h}} \longrightarrow V_{\widehat{h'h}} ,\end{align} which by means of \eqref{eqnisovhC} can be seen as an automorphism of $\mathbb{C}$, i.e.\ an invertible complex number. By construction this is the number $\alpha_{\widetilde{E_\theta(S,-)}} (h,h')$, i.e.\ the evaluation of 
	$\alpha_{\widetilde{E_\theta(S,-)}} \in Z^2(\Pi_1(S,BG);\U(1))$ on the 2-simplex defined by the composable pair $(h,h')$. 
	
	By definition of $E_\theta$ we find
	\begin{align}
	\alpha_{\widetilde{E_\theta(S,-)}} (h,h') = \langle H^* \theta , \nu \rangle,
	\end{align} where
\begin{itemize}
\item $\nu$ is a fundamental cycle of $[0,1]^2\times S$ satisfying~\eqref{EQ: Condition Fundamental cycle on manifolds with corners}
		
\item and $H: \widehat{h'} \widehat{h} \longrightarrow \widehat{h'h}$ is a homotopy relative boundary. 
		
\end{itemize}
Using the triangulation of $[0,1]^2$ given by	
	
\begin{center}
	\begin{overpic}[width=8cm,
		scale=0.3]{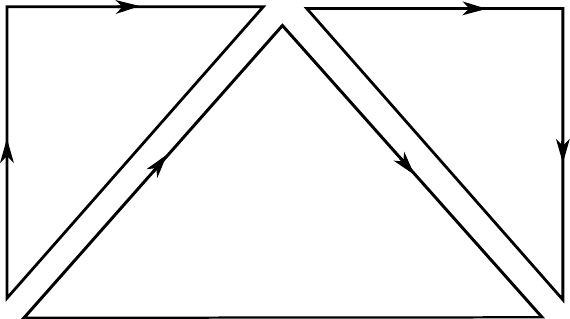}
\put(48,18){{$\nu_\Delta$}}
\put(10,40){{$\nu_-$}}
\put(85,40){{$\nu_+$}}
	\end{overpic}	 	
\end{center}
we get a fundamental cycle $\nu_\Box = \nu_- + v_\Delta + \nu_+$ of $[0,1]^2$.
Then $\nu :=  \nu_\Box \times \sigma_S $ satisfies \eqref{EQ: Condition Fundamental cycle on manifolds with corners}. To get a representative for $H$ we pick a 2-simplex $\widetilde H : \Delta_2 \longrightarrow BG^S$ such that $\partial_0\widetilde H = \widehat{h'}$, $\partial_1 \widetilde H = \widehat{h}$ and $\partial_2 \widetilde H = \widehat{h'h}$. This is possible since $BG^S$ is aspherical. Using the map $\Box : [0,1]^2 \longrightarrow BG^S$ sketched in Figure \ref{Fig:Sketch Box map} we define
\begin{align}
 H :  [0,1]^2  \times S \xrightarrow{\Box \times \id } BG^S \times S  \xrightarrow{\text{ev}} BG,
\end{align} where $\text{ev}: BG^S \times S \longrightarrow BG$ denotes the evaluation. 
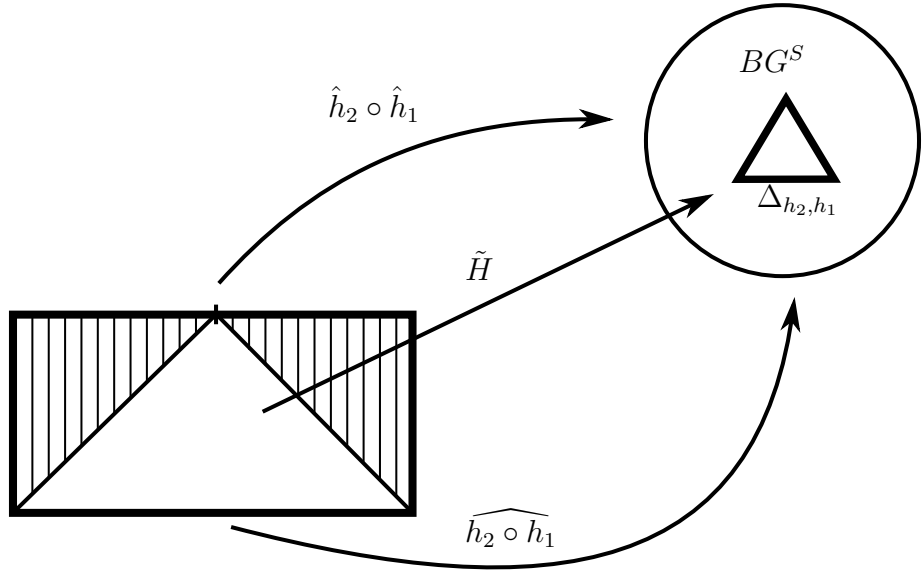
\begin{figure}[h]
	\centering
	\begin{overpic}[width=12cm,
		scale=0.3]{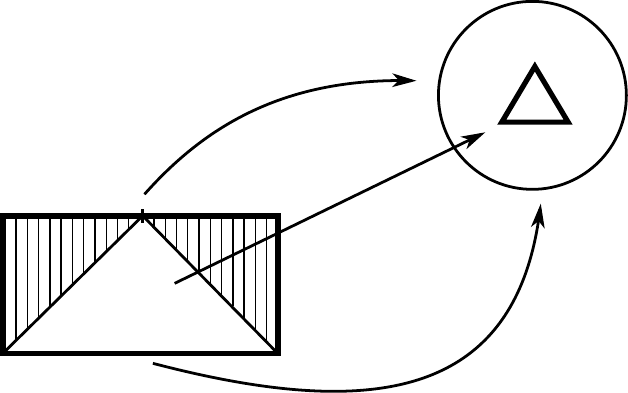}
		\put(35,50){{$\hat h_2 \circ \hat h_1$}}
		\put(50,3){{$\widehat{h_2 \circ  h_1}$}}
		\put(80,55){{$BG^S$}}
		\put(82,40){{$\Delta_{h_2, h_1}$}}
		\put(50,32){{$\tilde{H}$}}
	\end{overpic}
	\caption{Sketch for the definition of $\Box \colon [0,1^2]\longrightarrow BG^S$. The map is constant along the vertical lines.}
	\label{Fig:Sketch Box map}	
\end{figure} 

Now we find 
\begin{align}
\begin{split}
\alpha_{\widetilde{E_\theta(S,-)}} (h,h') &= \langle H^* \theta , \nu \rangle \\ &= \langle \theta,H_*( \nu_- \times \sigma_S )\rangle +  \langle \theta,H_*( \nu_\Delta \times \sigma_S)\rangle +\langle \theta,H_*( \nu_+ \times \sigma_S )\rangle\ \ .
\end{split} 	
\end{align} 
Without loss of generality we can work with a normalized representative for $\theta$ (see Proposition~\ref{Prop: Coboundary induces transformations}) which then vanishes on the degenerate simplices $H_*(\nu_-\times \sigma_S)$ and $H_*(\nu_+\times \sigma_S)$. Hence, we are left with
 \begin{align}\begin{split}
		\alpha_{\widetilde{E}_\theta(S,-)} (h,h') & = \langle \theta,H_*(\nu_\Delta \times \sigma_S)\rangle = \langle \text{ev}^* \theta, \Box_* \nu_\Delta \times \sigma_S \rangle \\
		 & =  \langle \text{ev}^* \theta, \widetilde H \times \sigma_S \rangle = \tau_S \theta (h,h')\ \ .
\end{split}
	\end{align}  This proves the assertion.
	\end{proof}

\subsection{Quantum Dijkgraaf-Witten theories}\label{Sec: DW via Orbifold}
The space of field configurations of a discrete gauge theory is the essentially finite 
groupoid of $G$-bundles. There exists a well-defined integration theory over such 
groupoids, which we review in Appendix~\ref{Sec: Integration over finite groupoids}.
The existence of such a measure makes the path integral quantization straightforward. 

The invariant assigned to an $n$-dimensional closed manifold $M$ is the path integral
\begin{align}
{\Za_{\text{DW}}}_\theta (M) \coloneqq \int_{\psi \in \Bun_G(M)}  E_\theta(M,\psi) \ \ .
\end{align}
We use the extended orbifold construction of Schweigert and Woike~\cite{EOFK} to define the quantum
theory as an extended topological field theory 
\begin{align}
{\Za_{\text{DW}}}_\theta \colon \ECob \longrightarrow \Tvs \ \ .
\end{align}  
The advantage of this approach is that it ensures that the result is a symmetric monoidal
2-functor. For a concrete construction of Dijkgraaf-Witten theories as extended field theories, see \cite{Morton} and \cite{FLHT} for the
fully extended field theory.  
One of the central results of \cite{EOFK} is the construction of an orbifoldization functor
$-/G$ from $n$-dimensional extended $G$-equivariant topological field theories to $n$-dimensional extended ordinary (i.e.\ non-equivariant) topological field theories. The orbifoldization combines a sum over twisted sectors with the computation of (homotopy) invariants by means of the parallel section functor developed in \cite{SWParallel}. 
Let $E\colon \EGCob \longrightarrow \Tvs$ be an extended quantum field theory.
We explain the orbifold construction on closed manifolds of dimension $n$, $n-1$ and 
$n-2$.  
The orbifold theory $E/G$ assigns to an $n$-dimensional manifold $M$ the  ``path integral"
\begin{align}
E/G(M) \coloneqq \int_{\psi \in \Bun_G(M)} E(M,\psi) \ \ . 
\end{align}
Restricting $E$ to the bundle groupoid over a closed $n-1$-dimensional manifold 
gives rise to a vector bundle $E(\Sigma, \cdot) \colon \Bun_G(\Sigma)\longrightarrow \fvs$.
The topological field theory $E/G$ assigns to $\Sigma$ the vector space of parallel sections
of this vector bundle. This space is the limit (in the categorical sense) of $E(\Sigma, \cdot)$.
We write this space suggestive as an end (see Example~\ref{Exa: limits as ends})
\begin{align}
E/G(\Sigma) \coloneqq \int_{\varphi \in \Bun_G(M)} E(\Sigma, \varphi) \ \ . 
\end{align}   
Similarly, the value on an $n-2$ dimensional manifold $S$ is the 2-vector space
of parallel sections of the 2-vector bundle $E(S,\cdot)\colon \Bun_G(S)\longrightarrow \Tvs$.
\begin{definition}\label{DW-theories}
For a finite group $G$ and $\theta\in H^n(BG; \U(1))$ the \emph{$n$-dimensional $\theta$-twisted Dijkgraaf-Witten theory 
\[
{{\Za_{\text{DW}}}} _\theta : \ECob \longrightarrow \Tvs
\]
with gauge group $G$} is defined to be the orbifold theory $E_\theta / G$ of the extended field theory $E_\theta : \EGCob \longrightarrow \Tvs$ associated to $\theta$.  
\end{definition}

We briefly explain the relation between Dijkgraaf-Witten theories and concepts from 
representation theory. We will highlight how topological arguments can be used
to prove results about algebraic objects. 
Let us first give the following formula for the number of simple objects in the category ${{\Za_{\text{DW}}}}_\theta(\mathbb{T}^{n-2})$ obtained by evaluation of the twisted Dijkgraaf-Witten theory on the $n-2$-dimensional torus $\mathbb{T}^{n-2}$. It follows as a special case from \cite[Theorem~4.21]{EOFK}: 
\begin{proposition}\label{satznumberofsimples}
	For a finite group $G$, $\theta\in H^n(BG; \U(1))$
	\begin{align}
	\# \{ \text{simple objects of ${{\Za_{\text{DW}}}}_\theta(\mathbb{T}^{n-2})$} \} = \frac{1}{|G|} \sum_{\substack{  g_1,\dots,g_n \in G \\ \text{mutually commuting}  }} \langle  \psi_{g_1,\dots,g_n}^* \theta, \sigma_{\mathbb{T}^n}   \rangle    \ , \label{eqnnumberofsimples}
	\end{align} where
	\begin{itemize}

		\item $\psi_{g_1,\dots,g_n}: \mathbb{T}^n \longrightarrow BG$ is a classifying map for the $G$-bundle $P$ over $\mathbb{T}^n$ specified by the holonomy values $g_1,\dots,g_n \in G$,
		
		\item $\sigma_{\mathbb{T}^n}$ is the fundamental class of the torus.
	
		\end{itemize}
	\end{proposition}

\begin{proof}
	For any extended topological quantum field theory $Z$, the number of simple objects of the 2-vector space $Z(\mathbb{T}^{n-2})$ is given by the number $Z(\mathbb{T}^n)$ assigned to the $n$-torus. This proves
	\begin{align}
		\# \{ \text{simple objects in ${{\Za_{\text{DW}}}}_\theta(\mathbb{T}^{n-2})$} \} = {{\Za_{\text{DW}}}}_\theta (\mathbb{T}^n).
	\end{align} Now the definition of $Z_\theta$ and the formula for the orbifold theory on closed oriented top-dimensional manifolds given in \cite[Proposition~3.47]{EOFK} yield the result.
	\end{proof}

\noindent We now focus on Dijkgraaf-Witten theories (see Definition~\ref{DW-theories}) in
2-dimensions:

\begin{proposition}\label{dwtheorydim2}
For any finite group $G$ and $\theta\in H^2(BG; \U(1))$ the evaluation of the topological quantum field theory ${{\Za_{\text{DW}}}} _\theta = E_\theta / G: \Cob_{2,1,0} \longrightarrow \Tvs$ on the point is given by the category of $\theta$-twisted projective representations of $G$.
\end{proposition}
 
\begin{proof}
The groupoid $\Pi_1(\star, BG)$ is equivalent to the groupoid $\star \DS G$ with one object and automorphism group $G$. By Theorem \ref{Theorem: Transgression} (note that the transgression is the identity in that case) we find that the 2-vector bundle $E_\theta(\star,-)$ is given by
\begin{align}
 \star \DS G \overset{\theta}{\longrightarrow} \fvs \DS \id \DS \C^\times \longrightarrow \Tvs \ ,
\end{align} 
where $\theta$ is understood as a 2-functor. 
According to the definition of the orbifold,
 the 2-vector space ${{\Za_{\text{DW}}}} _\theta(\star)$ is given by the category of 1-morphisms from the trivial line bundle over $ \star \DS G$ to ${{\Za_{\text{DW}}}} _\theta(\star,-)$, i.e.\ by the parallel sections of ${{\Za_{\text{DW}}}} _\theta(\star,-)$. 
Spelling this out we see that ${{\Za_{\text{DW}}}} _\theta(\star)$ is the category of projective representation twisted by $\theta$, see Section~\ref{Sec: P anomaly actions}. 
\end{proof}

Given the explicit description of ${{\Za_{\text{DW}}}} _\theta(\star)$ provided by Proposition~\ref{dwtheorydim2} in the two-dimensional case, we can compute the number of irreducible $\theta$-twisted representation of $G$ by using Proposition~\ref{satznumberofsimples}. The right hand side of \eqref{eqnnumberofsimples}, i.e.\ the value of $\theta$-twisted Dijkgraaf-Witten theory on the torus, already appears in \cite[Equation (6.40)]{DijkgraafWitten}, although we should note that the reasoning in the proof of Proposition~\ref{satznumberofsimples} is only valid because we have described twisted two-dimensional Dijkgraaf-Witten theory as an \emph{extended} quantum field theory. 
Now  \eqref{eqnnumberofsimples} reduces to
	\begin{align}
	\# \{ \text{irreducible $\theta$-twisted representation of $G$} \}  = \frac{1}{|G|}\sum_{gh=hg} \frac{\theta(h,g)}{\theta(g,h)}
	\end{align} and hence to the result found in \cite[Corollary 13]{TwistedDWandGerbs} by algebraic methods. 
\begin{example}\label{Ex: 2D coycles}
We give a few concrete examples of 2-cocycles:
\begin{itemize}
\item[(a)] The group cohomology $H^2(\Z_N\times \Z_N;U(1))$ is $\Z_N$.
If we write the cyclic group $\Z_N$ additively then the non-trivial 2-cocycle corresponding to 
$k\in \{0,1,\dots , N-1\}$ is 
\begin{align}\label{EQ:Def 2 cocycle}
\omega_k\big((a_1,b_1)\,,\,(a_2,b_2)\big)= \exp\Big(\frac{2\pi \,\iu\, k}{N}\,a_1\, b_2 \Big)
\end{align}  
with $(a_1,b_1),(a_2,b_2)\in \Z_N\times \Z_N$.
For $N=2$, the partition function $Z_{\omega_1}(\mathbb{T}^2)$ on $\mathbb{T}^2$ for the non-trivial $\Z_2\times \Z_2$ cocycle is $1$ corresponding to the fact that there exists only one $\omega_1$-twisted irreducible representation of $\Z_2\times \Z_2$~\cite{TwistedDWandGerbs}. 

\item[(b)] The degree~2 group cohomology of the dihedral group $D_8=\langle a,b \mid a^4=b^2=1 \
  , \ b\,a\,b^{-1}=a^{-1}\rangle$ with values in $U(1)$ is $\Z_2$. The non-trivial 2-cocycle is given by~\cite[Section 3.7]{ProjectiveRep}
\begin{align}
\omega\big(a^i\, b^j, a^{i'}\,b^{j'}\big)=
\begin{cases}
1 &, \quad \mbox{$j=0$} \ , \\
\exp\big( \frac{2\pi \,\iu\, }{4} \, i' \big) &, \quad \mbox{$j=1$} \ .
\end{cases}
\end{align} 
\end{itemize}
\end{example}

Moving up to 3-dimensions one finds 
\begin{theorem}\label{Thm: 3D DW}
For any finite group $G$ and $\theta\in H^3(BG; \U(1))$ the evaluation of ${{\Za_{\text{DW}}}} _\theta = E_\theta / G$ on the circle is given by the representation category of the $\theta$-twisted Drinfeld double of $G$. 
\end{theorem}

\begin{proof}
The 2-vector space associated to $S^1$ is the space of parallel sections of the 
2-line bundle $\tau \theta \colon G \DS G \longrightarrow \star \DS \star \DS \C^\times \longrightarrow \Tvs$. This 2-vector space is given, as explained in detail in
Section~\ref{Sec: P anomaly actions} by the category of $\tau_{\mathbb{S}^1}\theta$-twisted representations of the action groupoid $G\DS G$. By \cite[Proposition~8 and Theorem~17]{TwistedDWandGerbs} this category is the representation category of the twisted Drinfeld double \cite{DPR}. 
\end{proof}
\noindent This result implies for instance that we can get from Proposition~\ref{satznumberofsimples} an easy topological proof of the formula for the number of irreducible representations of the $\theta$-twisted Drinfeld double given in \cite[Theorem 21]{TwistedDWandGerbs}. 

Moreover, we see that the $\theta$-twisted Dijkgraaf-Witten theory of Definition
~\ref{DW-theories} generalizes the $\theta$-twisted Dijkgraaf-Witten theory constructed in~\cite{Morton} for the 3-2-1-dimensional case to arbitrary dimension because they yield the same modular category upon evaluation on the circle, which is sufficient by the classification result of \cite{bartlett2015modular}. 
\begin{example} 
The cohomology group $H^3(\Z_N;U(1))$ is $\Z_N$. The 3-cocycles have the
concrete form~\cite[Proposition~2.3]{Huang2014}  
\begin{align}\label{Eq: Zn 3-cocycle}
\omega_k(a,b,c)= \exp\Big(\frac{2\pi \,\iu\, k}{N} \, a \, \Big\lfloor\frac{b+c}{N}\Big\rfloor \Big) 
\end{align}  
for $a,b,c,k\in \Z_N=\lbrace 0,1,\dots, N-1 \rbrace$, where $\lfloor
r\rfloor$ denotes the integer part of the real number $r\in \R$,
i.e.~the largest integer less than or equal to $r$. These theories are
studied in~\cite{Kapustin:2014gua}. They have been extended to a product
of an arbitrary number of cyclic groups $\Z_{N_i}$ (i.e.\ a generic
finite abelian group) in~\cite{Chen:2011pg,Wang:2014tia}.
\end{example} 

\subsection{Equivariant Dijkgraaf-Witten theories}\label{Sec: Equivariant DW}
In Section~\ref{Sec: Gauging} we will use an equivariant version of the Dijkgraaf-Witten theories
generalizing work of \cite{NMS}:
as a generalization of the orbifold construction, we get for any morphism $\lambda : H \longrightarrow J$ of finite groups a \emph{pushforward map} $\lambda_*$ from $H$-equivariant to $J$-equivariant topological field theories, see \cite[Section~6]{OFK} for the non-extended case and \cite[Section~3.3]{EOFK} for the extended case needed here. The orbifold
construction is recovered as the push along the group homomorphism
$G\longrightarrow 1$ to the group $1$ with one element.

\begin{definition}\label{defgendwtheory}
Let $\lambda : H \longrightarrow J$ be a morphism of finite groups and $\theta\in H^n(BH;\U(1))$.
The \emph{$\theta$-twisted $J$-equivariant Dijkgraaf-Witten theory} ${{\Za_{\text{DW}}}}^\lambda_\theta := \lambda_* E_\theta : J\text{-}\Cob_{n,n-1,n-2} \longrightarrow \Tvs$ is defined to be the pushforward of $E_\theta$ along $\lambda$. 
\end{definition}

\begin{remark}
This construction can be generalized as follows: Given a sequence of finite groups 
\begin{align}
G_0 \overset{\lambda_1}{\longrightarrow} G_1 \overset{\lambda_2}{\longrightarrow} \dots \xrightarrow{\lambda_{n}} G_n
\end{align}
and 3-cocycles $\theta_j \in H^3(BG_j;\U(1))$ for $0\le j\le n$ we can construct the $G_n$-equivariant topological quantum field theory
\begin{align}
{\lambda_{n}}_*(\cdots ({\lambda_2}_*({\lambda_1}_* Z_{\theta_0}\otimes Z_{\theta_1})\otimes \cdots )\otimes Z_{\theta_{n-1}}) \otimes Z_{\theta_n} \ .
\end{align}
Corresponding to this theory there exists a potentially interesting $G_n$-multimodular tensor category.  
\end{remark}

We will now provide an explicit description of the corresponding 
non-extended field theory 
\begin{align}
{{\Za_{\text{DW}}}}^\lambda_\theta \colon \GCob_n \longrightarrow \fvs
\end{align}
for a surjective group homomorphism $\lambda \colon \widehat{G}\longrightarrow G$,
which will be crucial for the constructions 
in Section~\ref{Sec: Gauging}. 
Let $\Sigma$ be an $n{-}1$-dimensional closed manifold.
The group homomorphism $\lambda$ induces an extension functor 
\begin{align}
\lambda_* \colon \mathsf{Bun}_{\widehat{G}}(\Sigma) \longrightarrow \mathsf{Bun}_G(\Sigma) \ . 
\end{align}  
This functor acts on classifying maps by post-composition with the map $B\widehat{G}\longrightarrow B G$ induced by $\lambda \colon \widehat{G}\longrightarrow G$, which by a slight abuse of notation we denote again by $\lambda$. 
For a bundle $\varphi \in \BunG(\Sigma)$ we denote by $\lambda_*^{-1}[\varphi]$ the homotopy fibre (see Appendix~\ref{Sec: Homotop Groupoid})
\begin{equation}
\begin{tikzcd}
\lambda_*^{-1}[\varphi] \ar[r] \ar[d] & \mathsf{Bun}_{\widehat{G}}(\Sigma) \ar[d, "\lambda_*"] \\
* \ar[r,swap, "\varphi"] & \BunG(\Sigma)
\end{tikzcd}
\end{equation}
Concretely, objects of $\lambda_*^{-1}[\varphi]$ are pairs $(\widehat{\varphi},h)$ of a $\widehat{G}$-bundle $\widehat{\varphi}$ and a gauge transformation $h \colon \lambda_* \widehat{\varphi} \longrightarrow \varphi $. Morphisms are gauge transformations $\widehat{h} \colon \widehat{\varphi} \longrightarrow \widehat{\varphi}\,'$ such that the diagram
\begin{equation}
\begin{tikzcd}
\lambda_* \widehat{\varphi} \ar[rr, "\lambda_* \widehat{h}"] \ar[rd, " h ", swap] & & \lambda_* \widehat{\varphi}\,' \ar[ld, " h' "] \\
 & \varphi &
\end{tikzcd}
\end{equation}
commutes.
The theory ${{\Za_{\text{DW}}}}^\lambda_\theta$ is defined on an object $(\Sigma, \varphi \colon \Sigma \longrightarrow BG)$ as
\begin{align}
\label{EQ:EDW on Objects}
{{\Za_{\text{DW}}}}^\lambda_\theta(\Sigma, \varphi)= \int_{\lambda_*^{-1}[\varphi]} \, \Big( \lambda_*^{-1}[\varphi] \longrightarrow \mathsf{Bun}_{\widehat{G}}(\Sigma) \xrightarrow{L_{\widehat{\omega}}} \fvs \Big) \ .
\end{align}
This should be regarded as a quantization of the $D=\ker(\lambda)$-gauge fields
while leaving the $G$-sector classical. 
This limit can be realized as a vector space of parallel sections. In this case a parallel section $f\in {{\Za_{\text{DW}}}}^\lambda_\theta(\Sigma, \varphi)$ consists of an element $f(\widehat{\varphi},h)\in E_{\theta}(\Sigma,\widehat{\varphi})$ for all $(\widehat{\varphi},h) \in \lambda_*^{-1}[\varphi] $
satisfying 
\begin{align}
f(\widehat{\varphi}', h')= E_\theta ([0,1]\times \Sigma, \widehat{h}) f(\varphi, h)
\end{align} 
for all morphisms $\widehat{h}\colon (\varphi, h)\longrightarrow (\varphi', h')$.
Let $(M, \psi)\colon (\Sigma_a, \varphi_a)\longrightarrow (\Sigma_b,\varphi_b)$ be a morphism in $G\text{-}\Cob$.
To define the pushforward on a parallel section $f(\,\cdot\,)\in {{\Za_{\text{DW}}}}^\lambda_\theta(\Sigma_a, \varphi_a)$ we fix fundamental cycles $\sigma_a$ and $\sigma_b$ of $\Sigma_a$ and $\Sigma_b$, respectively, and write $f$ as $f(\,\cdot\,)=\mathbf{f}(\,\cdot\,)\,[\sigma_a]$. 
We define
\begin{align}
\label{EQ:EDW on morphisms}
\begin{split}
 {{\Za_{\text{DW}}}}^\lambda_\theta&(M, \psi)[f](\widehat{\varphi}_b,h_b) \\ 
& = \Big(\int_{(\widehat{\psi},h,\widehat{h}\,)\in \lambda_*^{-1}[\psi]|_{(\widehat{\varphi}_b,h_b)}} \, \big\langle \, \widehat{h}^* \widehat{\theta},[0,1] \times \sigma_b \big \rangle \, \big\langle  \widehat{\psi}^* \widehat{\theta}, \sigma_M \big \rangle \, \mathbf{f}\big(\widehat{\psi}|_{\Sigma_a}, h|_{\Sigma_a}\big)\Big) \ [\sigma_b]
\end{split} 
\end{align}
with $\sigma_M \in \Fund^{\sigma_b}_{\sigma_a}(M)$; here the homotopy pullback
$\lambda_*^{-1}[\psi]|_{(\widehat{\varphi}_b,h_b)}$ is the groupoid
with objects $(\widehat{\psi},h,\widehat{h}\,)$ where $
\widehat{\psi} \colon M \longrightarrow B\widehat{G}$ is a
$\widehat{G}$-bundle, $h \colon \lambda_* \widehat{\psi}
\longrightarrow \psi$ is a gauge transformation, and
$\widehat{h}\colon \widehat{\psi}\,|_{\Sigma_b} \longrightarrow
\widehat{\varphi}_b$ is a gauge transformation
such that the diagram
\begin{equation}
\begin{tikzcd}
\lambda_*\widehat{\psi}\,\big|_{\Sigma_b} \ar[rr, "\lambda_*\widehat{h}"] \ar[rd,"h|_{\Sigma_b}",swap] & & \lambda_*\widehat{\varphi}_b \ar[ld, "h_b"] \\
 & \varphi_b &
\end{tikzcd}
\end{equation}
commutes. 
\begin{remark}
For concrete computations it is sometimes helpful to note that for surjective $\lambda$ 
the induced map $\Bun_{\widehat{G}}(M)\longrightarrow \Bun_{G}(M)$ is a fibration of 
groupoids (see Section~\ref{Sec: Homotop Groupoid}) in most models for the bundle groupoid. This allows us to replace homotopy
fibres with ordinary fibres. For example to describe the ordinary Dijkgraaf-Witten theory
corresponding to the group homomorphism $G\longrightarrow 1$ we 
can describe the groupoid of principal $1$-bundles on
a manifold $\Sigma$
by the terminal groupoid with one object and one morphism. Every map into 
the terminal groupoid is a fibration. 
The description of the vector space ${{\Za_{\text{DW}}}}_\theta(\Sigma)$  by
parallel sections reduces to: a {parallel
  section} $f$ consists of an element
$f(\varphi)\in E_{\theta}(\Sigma,\varphi)$ for all $\varphi \in
\BunG(\Sigma)$ such that $E_\theta ([0,1]\times \Sigma, h)
\big(f(\varphi) \big)= f(\varphi')$ for all gauge transformations
$h\colon \varphi \longrightarrow \varphi'$. The space of parallel
sections can be regarded as the space of gauge-invariant functions on
the set of classical gauge field configurations.
For this reason the definition can be interpreted as an implementation of the Gauss Law in quantum gauge theory, which requires that physical states must be gauge-invariant.

Now consider a cobordism $M\colon \Sigma_1 \longrightarrow \Sigma_2$.  We fix representatives $\sigma_1$ and $\sigma_2$ of the fundamental classes of $\Sigma_1$ and $\Sigma_2$, respectively. This allows us to express the value of a parallel section $f\in {{\Za_{\text{DW}}}}_\theta(\Sigma_1)$ on a principal $G$-bundle $\varphi_1\in \BunG(\Sigma_1)$ as $f(\varphi_1)=\mathbf{f}(\varphi_1)\, [\sigma_1]$ with $\mathbf{f}(\varphi_1) \in \C$. The definition
of ${{\Za_{\text{DW}}}}_\theta$ in \eqref{EQ:EDW on morphisms} reduces to
\begin{align}
\label{Definition on Morphisms}
{{\Za_{\text{DW}}}}_\theta (M)(f)(\varphi_2)= \Big( \int_{(\psi,h)\in
  \BunG(M)|_{\varphi_2}} \, \langle  h^* \theta, [0,1] \times
  \sigma_{2} \rangle \, \langle \psi^* \theta , \sigma_M  \rangle \,
  \mathbf{f}(\psi|_{\Sigma_1}) \Big) \ [\sigma_2] \ ,
\end{align}  
where $\sigma_M$ is a representative for the fundamental class of $M$ satisfying $\partial \sigma_M = \sigma_2-\sigma_1$ and we consider the gauge transformation $h$ as a homotopy $h \colon [0,1]\times \Sigma_2 \longrightarrow BG$. This definition is independent of all choices involved.
\end{remark}

\section{Discrete symmetries and 't Hooft anomalies}\label{Sec: DW Sym}
In this section we study actions of a finite group $G$ as symmetries of a 
Dijkgraaf-Witten theory with gauge group $D$ and topological action 
$\omega \in Z^n(BD;(1))$. These actions are closely related to non-abelian group 
cohomology, which we review in Section~\ref{Sec: Non abelian group cohomology}.
Afterwards we study the gauging of the symmetry group $G$ in Section~\ref{Sec: Gauging}.
It is not always possible to gauge a given symmetry. A 't Hooft anomaly is an 
obstruction to the gauging of symmetries. The corresponding obstruction theory, which 
is the content of Section~\ref{Sec: Obstruction} is 
encoded by the Lyndon-Hochschild-Serre spectral sequence.

\subsection{Discrete symmetries of Dijkgraaf-Witten theories}\label{Sec: Finite symmetries}
So far all symmetries considered in this thesis where encoded as limit morphisms 
in the bicategory of cobordisms with background fields $\ECobF$. Symmetries of this
type correspond to the natural notion of isomorphisms of $\F$-background 
fields and are present in any quantum field theory defined on $\ECobF$.   

In addition, a particular theory can be invariant under additional transformations of 
background gauge fields. For example a field theory containing a complex valued scalar
field might be invariant under complex conjugation of the field; or might not. 
It is this type of symmetry we study in
this section for the theory ${E}_\omega$ and 
${{\Za_{\text{DW}}}}_\omega$ . 
 
Before asking if ${E}_\omega$ is invariant under a transformation
of background fields we first need to implement the action.
For the symmetry to be compatible with cutting and gluing of manifolds,
we describe it as an endofunctor of $\DCob$ acting by pullback along
the inverse on a field theory. There is a natural way to construct
endofunctors of $\DCob$ from homeomorphisms of $BD$ which is described by a 2-functor 
\begin{align}\label{End(BG)=>End(D-Cob)}
 \mathcal{R} \colon * \DS \Pi_1 [BD,BD] \longrightarrow * \DS  \End (\DCob) \ , 
\end{align}
where $\Pi_1[BD,BD]$ is the category with continuous maps $BD\longrightarrow BD$ as objects and equivalence classes of homotopies as morphisms. 
Concretely, $\mathcal{R}$ sends a continuous map $\chi \colon BD \longrightarrow BD$ to the endofunctor
\begin{align}
\begin{split}
\mathcal{R}(\chi) \colon \DCob &\longrightarrow \DCob \\
(\Sigma,\varphi \colon \Sigma \longrightarrow BD) & \longmapsto (\Sigma ,\chi \circ \varphi\colon \Sigma \longrightarrow BD) \\
\big( (M,\psi) \colon (\Sigma_1,\varphi_1) \longrightarrow 
  (\Sigma_2,\varphi_2) \big) & \longmapsto \big( (M,\chi \circ\psi) \colon (\Sigma_1,\chi \circ\varphi_1)\longrightarrow (\Sigma_2,\chi \circ \varphi_2) \big)
\end{split}
\end{align}
and a homotopy $h \colon \chi_1 \longrightarrow \chi_2$ to the natural transformation $\mathcal{R}(h)\colon \mathcal{R}(\chi_1)\Longrightarrow \mathcal{R}(\chi_2)$ with components
\begin{align}
\mathcal{R}(h)_{(\Sigma,\varphi)} = \left( [0,1] \times \Sigma , h \circ \id_\varphi \right)
\end{align}
where $h \circ \id_\varphi$ denotes the horizontal composition of homotopies. The naturality of $\mathcal{R}(h)$ follows from Lemma~\ref{Lem: non-constant morphism} adjusted 
to $\DCob$. By the bicategorical Yoneda Lemma, automorphisms of $BD$ correspond to automorphisms of the stack of principal $D$-bundles (which is represented by $BD$). Hence a symmetry corresponding to a homeomorphism of $BD$ acts on the space of field configurations.\footnote{
More generally, automorphisms of a stack $\F$ induce endofunctors of the 
category $\CobF$.}  

To define a symmetry, the group $G$ only has to act up to gauge transformations. 
Recall that for finite groups, gauge transformations and homotopies between classifying maps are in one-to-one correspondence. 
For this reason we expect a symmetry for every action of $G$ on $BD$ up 
to a `homotopy' which preserves $\omega$. Since $BD$ is a homotopy 1-type, 
we can work with the following concrete description.
\begin{definition}\label{Def: homotopy Coherent action}
An {action of $G$ on $BD$ up to (coherent) homotopy} is a 2-functor 
\begin{align}
\alpha \colon * \DS G \longrightarrow * \DS \Pi_1[BD,BD] \ ,
\end{align}
where $* \DS G$ is considered as a 2-category with one object, the group $G$ as 1-morphisms and only identity 2-morphisms. 
\end{definition}

\begin{remark}
To unpack this compact definition note that the category   
$\Pi_1[BD,BD]$ is equivalent to the action groupoid
\begin{align}
[\pi_1(BD),\pi_1(BD)]\DS D = \End_{\Grp}(D)\DS D \ ,
\end{align}
where the action of $D$ on a group homomorphism is by conjugation. Every action of $G$ up to homotopy takes values in the full subgroupoid $\Aut_{\Grp}(D)\DS D$ of automorphisms of $D$.
An arbitrary 2-functor ${* \DS G} \longrightarrow *\DS (\Aut_{\Grp}(D)\DS D)$ is called a \emph{non-abelian group cocycle} \cite{Baez2010, Blanco2005}. Hence, homotopy coherent actions on $BD$ are classified by non-abelian group cocycles. 
Non-abelian cocycles also appear in the construction of equivariant Dijkgraaf-Witten theories~\cite{NMS} under the name weak 2-cocycles.
We discuss non-abelian group cohomology in more detail in Section~\ref{Sec: Non abelian group cohomology}.  

If $D$ is abelian there are no morphisms between different objects in $\Aut_{\Grp}(D)\DS D$. This implies that an action up to homotopy of $G$ on $BD$ is given by a proper action of $G$ on $D$ and a group 2-cocycle in $H^2(BG;D)$ describing the coherence isomorphisms of the corresponding 2-functor. This agrees with the physical description in \cite{Kapustin:Symmetries}. 
\end{remark}
For every action $\alpha \colon {* \DS G} \longrightarrow *\DS \Pi_1[BD,BD]$ up to homotopy the 2-functor \eqref{End(BG)=>End(D-Cob)} induces via pullbacks a 2-functor
\begin{align}\label{EQ: Action of kinematic symmetries on TFTs}
\begin{split}
\rho \colon {* \DS G}&\longrightarrow \DTFT\DS \End_{\Cat}(\DTFT)\hookrightarrow \Cat \\
g &\longmapsto \mathcal{R}\big(\alpha(g^{-1})\big)^* \ ,
\end{split}
\end{align}  
where we denote by $\Cat$ the 2-category of categories and by $\DTFT$ the category 
of n-dimensional $D$-equivariant topological field theories.

The (exponentiated) action of a gauge theory can be considered as a gauge-invariant map from the space of field configurations on an $n$-dimensional manifold $M$, in our case $\BunD(M)$, to $U(1)$. An action of $G$ on the space of field configurations induces an action via pullbacks on the set of gauge-invariant functions from the space of field configurations to $U(1)$.
A theory admits the symmetry $G$ if its (exponentiated) action is invariant under this action, i.e.\ it is a fixed point. 
By categorification we arrive at the following description.

\begin{definition}\label{Def: Field theory with k. symmetry}  
A \emph{$D$-equivariant field theory with kinematical symmetry}\footnote{The name is taken from \cite{BrauerGroup} where similar symmetries of 3-dimensional Dijkgraaf-Witten theories are studied.} described by  
\begin{align}
\rho\colon \underline{* \DS G}\longrightarrow \DTFT \DS \End_{\Cat}(\DTFT) \ ,
\end{align}
as in~\eqref{EQ: Action of kinematic symmetries on TFTs}, is a homotopy fixed point of $\rho$, i.e.\ a natural 2-transformation $Z\colon 1 \Longrightarrow \rho$, 
where $1$ is the unique 2-functor sending $*$ to the category with one object and only identity morphisms. 
\end{definition}
\begin{remark}
Unpacking the definition, a $D$-equivariant field theory with kinematical symmetry consists of
\begin{itemize}
\item[(a)]
A functor $Z \colon 1 \longrightarrow \DTFT$; and

\item[(b)]
Natural transformations $\Upsilon_g \colon \rho(g) [Z] \Longrightarrow Z$ for all $g\in G$;
\end{itemize} 
satisfying natural coherence conditions. Since $1$ represents the identity 2-functor on $\Cat$ this is the same as a field theory $Z \in \DTFT$, together with coherent natural symmetric monoidal transformations $\Upsilon_g \colon  \mathcal{R}(\alpha(g^{-1}))^* Z \Longrightarrow Z$ for $g\in G$.
\end{remark}

An arbitrary Dijkgraaf-Witten theory with topological action $\omega\in Z^n(BD;U(1))$ does not admit a kinematical symmetry in general. On the other hand, there may be different ways to equip a given field theory with the structure of a homotopy fixed point. We give a sufficient condition for a kinematical symmetry to exist.
For this we need to introduce the following notion.

\begin{definition}\label{Def: Preserved}
An $n$-cocycle $\omega\in Z^n(BD;U(1))$ is \emph{preserved by the action $\alpha$} if it can be equipped with the structure of a homotopy fixed point for the induced action of $G$ via the pullback along $\alpha(g^{-1}) $ on the category $Z^n(BD;U(1))$ whose morphisms are $n{-}1$-cochains up to coboundaries.  
\end{definition} 
In general there are non-isomorphic choices for the fixed point structure. A necessary condition for such a fixed point to exist is $\alpha(g)^*[\omega]=[\omega]$ for all $g\in G$. 

\begin{remark}\label{Rem: Preserved}
Concretely, the additional structure consists of an equivalence class of cochains $\Phi_{g}\in C^{n-1}(BD;U(1))$ up to coboundary satisfying\footnote{Throughout we switch freely between the additive and multiplicative notation for $U(1)$-valued cocycles.} $\delta \Phi_{g} = \omega- \alpha(g^{-1})^* \omega$. These cochains have to satisfy the coherence relations
\begin{align}
\Phi_{g_1}+\alpha(g_1^{-1})^*\Phi_{g_2} = \Phi_{g_1\,g_2}+\sigma_{g_1,g_2}[\omega]  \ , 
\end{align}
up to coboundary terms, 
where $\sigma_{g_1,g_2}[\omega]$ is the $n{-}1$-cochain induced by the homotopy $\sigma_{g_1,g_2}\colon \alpha(g_2^{-1})\circ \alpha(g_1^{-1}) \longrightarrow \alpha(g_2^{-1}\,g_1^{-1})$. The difference between two homotopy fixed point structures can be described by a group homomorphism $G \longrightarrow H^{n-1}(BD;U(1))$.  
\end{remark}

\begin{proposition}\label{Prop: Classical symmetry}
Let $\omega \in Z^n(BD;U(1))$ be a topological action and $\alpha\colon {* \DS G} \longrightarrow * \DS \Pi_1[BD,BD]$ a homotopy coherent action of $G$ on $BD$. If $\alpha$ preserves $\omega$, then the classical Dijkgraaf-Witten theory $E_\omega \colon \DCob \longrightarrow \fvs$ admits a kinematical symmetry described by $\alpha$.
\end{proposition}
\begin{proof}
We set $Z=E_{\omega}$ and define natural transformations $\Upsilon_g \colon  E_{\alpha(g^{-1})^* \omega} \Longrightarrow E_{\omega}$ by
\begin{align}
\begin{split} 
{\Upsilon_g}_{\,(\Sigma, \varphi)}\colon E_{\alpha(g^{-1})^* \omega}(\Sigma, \varphi)&\longrightarrow E_{ \omega}(\Sigma, \varphi) \\
 [\sigma_\Sigma] &\longmapsto \left\langle \varphi^* \Phi_g, \sigma_\Sigma \right\rangle \, [\sigma_\Sigma] \ , 
 \end{split}
\end{align}  
where $\Phi_g$ is the $n{-}1$-cochain of Remark~\ref{Rem: Preserved} satisfying $\delta \Phi_{g} = \omega-\alpha(g^{-1})^* \omega $. That this defines a natural isomorphism follows from Proposition~\ref{Prop: Coboundary induces transformations}.
The coherence conditions follow from the fact that the collection $\Phi_g$ corresponds to a homotopy fixed point structure.
\end{proof}

Now we study the incarnation of kinematic symmetries in the quantum theory ${{\Za_{\text{DW}}}}_\omega$.
We describe the symmetries of quantum Dijkgraaf-Witten theory following~\cite[Section~2.4]{Freed:2014eja} by the following notion.
\begin{definition}\label{Def: Internal symmetry}
Let $G$ be a finite group and denote by $G\text{-}\mathsf{Rep}$ the category of finite-dimensional $G$-representations.
Let $Z\colon \Cob_n \longrightarrow \fvs$ be a topological field theory. An \emph{internal $G$-symmetry} of $Z$ is a lift 
\begin{equation}
\begin{tikzcd}
 & G\text{-}\mathsf{Rep} \ar[dr] & \\
\Cob_n \ar[rr,swap,"Z"] \ar[ur,"Z_G"] & & \fvs
\end{tikzcd}
\end{equation}
of $Z$, where $G\text{-}\mathsf{Rep} \longrightarrow \fvs$ is the
forgetful functor. 
\end{definition}
\begin{remark}
This definition is equivalent to fixing a group homomorphism $G\longrightarrow \Aut_\otimes(Z)$ to the group of symmetric monoidal natural automorphisms of $Z$. 
\end{remark}

Kinematical symmetries of classical Dijkgraaf-Witten theories extend to the quantum theory: 
For a fixed manifold $\Sigma$, $\Upsilon_g$ induces a natural isomorphism $E_\omega \circ \mathcal{R}(\alpha(g^{-1}))|_{\BunD(\Sigma)}\longrightarrow E_{\omega}|_{\BunD(\Sigma)}$, which induces a linear map 
\begin{align}
 \int_{\BunD(\Sigma)} \, E_{\omega} \circ
 \mathcal{R}\big(\alpha(g^{-1})\big) \longrightarrow
 \int_{\BunD(\Sigma)} \, E_{\omega} = Z_\omega (\Sigma)\ .
\end{align} 
The equivalence $\mathcal{R}(\alpha(g^{-1}))$ induces a morphism 
\begin{align}
{{\Za_{\text{DW}}}}_\omega(\Sigma) = \int_{\BunD(\Sigma)} \, E_{\omega} \longrightarrow
  \int_{\BunD(\Sigma)} \, E_{\omega} \circ \mathcal{R}\big(\alpha(g^{-1})\big) \ .
\end{align}
The action of $G$ consisting of $\varphi_g \colon {{\Za_{\text{DW}}}}_\omega(\Sigma) \longrightarrow {{\Za_{\text{DW}}}}_\omega(\Sigma) $ is defined as the composition of these two maps. 
If the fixed point structure is the one constructed in Proposition~\ref{Prop: Classical symmetry} the action has the following concrete description:
fixing a fundamental class $\sigma_\Sigma$ of $\Sigma$ as in \eqref{Definition on Morphisms}, and a parallel section $f(\,\cdot\,)= \mathbf{f}(\,\cdot\,)\,[\sigma_\Sigma] \in {{\Za_{\text{DW}}}}_\omega $, the action of $G$ on ${{\Za_{\text{DW}}}}_\omega(\Sigma)$ takes the concrete form  
\begin{align}\label{Eq: Concrete form of the quantum symmetry}
g\triangleright \mathbf{f}(\Sigma,
  \varphi)=\mathbf{f}\big(\mathcal{R}(\alpha(g^{-1}))[\Sigma,
  \varphi]\big) \, \langle \varphi^* \Phi_{g} , \sigma_\Sigma \rangle
\end{align} 
with $\delta \Phi_{g}= \omega-\alpha(g^{-1})^* \omega$ as in the proof of Proposition~\ref{Prop: Classical symmetry}. 
\begin{proposition}\label{Prop: symmetry extends to quantum}
The collection $\varphi_g$ defines a representation of $G$ on the state spaces ${{\Za_{\text{DW}}}}_\omega (\Sigma)$ such that ${{\Za_{\text{DW}}}}_\omega$ is a functor into the category $G\text{-}\mathsf{Rep}$ of finite-dimensional $G$-representations. 
\end{proposition}
\begin{proof}
This is a direct consequence of the functoriality of the orbifold construction \cite[Remark~3.43]{OFK} and the coherence conditions for the homotopy fixed point. 
\end{proof}

\begin{example}
The trivial action of $G$ on $BD$ is always an internal $G$-symmetry. 
Any action of $G$ is an internal $G$-symmetry for a theory with trivial topological Lagrangian. 
We will provide some more profound examples in Sections~\ref{Sec: Non abelian group cohomology} and~\ref{Sec: Gauging}.
\end{example}

\subsection{Non-abelian group cohomology}\label{Sec: Non abelian group cohomology}

Following \cite{Blanco2005} we review non-abelian group 2-cocycles and show how they classify extensions. For simplicity we only discuss groups, which is enough for the study of anomalies in Dijkgraaf-Witten theories. The generalisation to groupoids is straightforward. Let $G$ and $D$ be finite groups. Recall from Section~\ref{Sec: Finite symmetries} that a non-abelian 2-cocycle on $G$ with coefficients in $D$ is a 2-functor $\alpha \colon {* \DS G} \longrightarrow * \DS (\Aut_{\Grp}(D)\DS D) \subset \Grpd$, where $\Grpd$ is the 2-category of groupoids. 
The 2-category $* \DS (\Aut_{\Grp}(D)\DS D)$ can be considered as a full sub-2-category of $\Grpd$ by sending the only object to the groupoid $* \DS D$. We assume without loss of generality that $\alpha$ preserves identities strictly.  
Spelling out the definition, we see that $\alpha$ consists of  maps of sets $\alpha \colon G \longrightarrow \Aut_{\Grp}(D)$ and $\sigma_\alpha \colon G\times G \longrightarrow D$ satisfying
\begin{align}
\alpha(1)&=\id_D \ , \\[4pt]
\sigma_\alpha (1,1)&=1 \ , \\[4pt]
\alpha(g_1\,g_2)[d] &= \sigma_\alpha(g_1,g_2)^{-1} \, \alpha (g_1)
                      \big[\alpha(g_2) [d]\big] \,
                      \sigma_\alpha(g_1,g_2) \ , \\[4pt]
\sigma_\alpha(g_1,g_2) \, \sigma_\alpha(g_1\, g_2,g_3) & = \alpha
                                                         (g_1)
                                                         \big[\sigma_\alpha(g_2,g_3)\big]
                                                         \,
                                                         \sigma_\alpha(g_1,
                                                         g_2 \, g_3) \ . \label{Eq: Non abelian cocycle} 
\end{align}
Using the Grothendieck construction, the 2-functor $\alpha \colon {*\DS G}\longrightarrow \Grpd$ induces a fibration of groupoids
\begin{align}
\int \alpha \longrightarrow *\DS G \ 
\end{align}  
having the following concrete description: There is only one object which we denote by $*$, endomorphisms are given by pairs $(g,d)\in G\times D$ and composition is defined by 
\begin{align}
(g_2,d_2)\,(g_1,d_1)= \big(g_2\,g_1\,,\, d_2 \, \alpha (g_2)[d_1] \, \sigma_\alpha(g_2,g_1)\big) \ . 
\end{align} 
The fibration corresponds to an extension of $G$ by $D$, i.e.\ an exact sequence 
\begin{align} \label{Eq: Extension}
1\longrightarrow D \longrightarrow \widehat G \longrightarrow G
  \longrightarrow 1 \ ,
\end{align} 
with $\widehat{G}=\End_{\int \alpha}(*)$. 
It is a classical result that extensions of $G$ by $D$ are
classified by non-abelian 2-cocycles on $G$ with coefficients in
$D$~\cite{Schreier}. A proof using the Grothendieck construction can 
be found in \cite{Blanco2005}. 

We give a few explicit examples.
\begin{example}
For every pair of groups $(G,D)$ there is a trivial non-abelian 2-cocycle corresponding to the constant 2-functor ${* \DS G}\longrightarrow * \DS (\Aut_{\Grp}( D) \DS D)$. The corresponding extension is
\begin{align}
1\longrightarrow D \longrightarrow D\times G \longrightarrow G \longrightarrow 1 \ .
\end{align}
\end{example}
\begin{example}\label{Ex: Abelian cohomology}
If $D$ is abelian and $\alpha \colon G\longrightarrow \Aut_{\Grp}(D)$ is trivial, then a non-abelian 2-cocycle reduces to an ordinary 2-cocycle $\sigma\in H^2(G;D)$ and the corresponding extensions are the usual central extensions classified by the abelian 2-cocycle. 
From a physical point of view such a 2-cocycle can appear if the $G$-action on matter fields only closes up to a $D$-gauge transformation.
We give two concrete examples for later use. Let $N$ and $M$ be positive integers.
Identifying the cyclic groups $\Z_N$ and $\Z_M$ with $\{0,1, \dots , N-1 \} $ and $\{0,1, \dots , M-1 \} $ we define the 2-cocycle 
\begin{align}
\begin{split}
\sigma \colon \Z_M\times \Z_M &\longrightarrow \Z_N \\
(a, b)& \longmapsto \Big\lfloor \frac{a+b}{M} \Big\rfloor \text{ mod }N \ .
\end{split}
\end{align}  
The corresponding central extension is 
\begin{align}
0 \longrightarrow \Z_N \xrightarrow{M\,\cdot}  \Z_{N\,M}
  \longrightarrow \Z_M \longrightarrow 0 \ , 
\end{align}
where the first map is multiplication by $M$ and the second map is
reduction modulo $N$.
This example can be adapted to an arbitrary number of copies of $\Z_N$ and $\Z_M$. An example is the abelian 2-cocycle corresponding to
\begin{align}
(0,0) \longrightarrow \Z_N\times \Z_N \xrightarrow{(M ,M)\,
  \cdot}\Z_{N\,M} \times \Z_{N\,M} \longrightarrow \Z_M \times \Z_{M}
  \longrightarrow (0,0)
\end{align}
which is given by 
\begin{align}
\begin{split}
 (\Z_M\times \Z_M)^2 &\longrightarrow \Z_N\times \Z_N \\
\big( (a_1, b_1)\,,\,(a_2,b_2)\big)& \longmapsto \Big(\Big\lfloor
                                     \frac{a_1+a_2}{M} \Big\rfloor
                                     \text{ mod }N \,,\, \Big\lfloor
                                     \frac{b_1+b_2}{M} \Big\rfloor
                                     \text{ mod }N \Big) \ .
\end{split}
\end{align}
\end{example}
\begin{example}
Given a group homomorphism $\alpha\colon G\longrightarrow
\Aut_{\Grp}(D)$, we can consider $\alpha$ as a non-abelian 2-cocycle
with trivial map $\sigma_\alpha$. The corresponding extension is the semi-direct product
\begin{align}
1 \longrightarrow D \longrightarrow G \ltimes_\alpha D \longrightarrow G \longrightarrow 1 \ .
\end{align}
\end{example}

We give a second point of view on non-abelian group cohomology using the $(\infty, 1)$-topos of infinity groupoids or alternatively the 3-category of 2-groupoids following \cite{nlabNonAbelianGroupCohomology}.
A non-abelian 2-cocycle is a 1-morphism $\alpha \colon BG \longrightarrow B (\Aut(D) \Ds D)$ where $\Aut(D)\Ds D$ is the 2-group with automorphisms of $D$ as objects and conjugation by elements of $D$ as morphisms. The homotopy pullback of the diagram
\begin{equation}
\begin{tikzcd}
 & * \ar[d] \\
 BG \ar[r,"\alpha"] &  B (\Aut(D) \Ds D)
\end{tikzcd}
\end{equation} 
is given by 
\begin{equation}
\begin{tikzcd}
 \Aut(D)\Ds \widehat{G} \ar[r] \ar[d]& * \ar[d] \\
 BG \ar[r,"\alpha"] &  B (\Aut(D) \Ds D)
\end{tikzcd}
\end{equation} 
with $\widehat{G}$ as above. We can see this as follows:  The homotopy pullback can be calculated as an ordinary pullback of the universal $\Aut(D)\Ds D$ principal 2-bundle $E(\Aut(D)\Ds D)$ over $ B (\Aut(D) \Ds D)$, which can be explicitly construced \cite[Section 5]{3AutGroupof2Group} as follows

\begin{itemize}
\item Objects are elements of $\Aut(D)$;

\item A 1-morphism $f_1 \longrightarrow f_2$ is a triangle
\begin{equation}
\begin{tikzcd}
 & \ \ar[dd, Rightarrow, shorten <= 15, shorten >= 15, "d"] & \bullet \ar[dd, "f(d)"] \\
\bullet \ar[rru,bend left, "f_1"] \ar[rrd,bend right, "f_2",swap] & & \\
 & \ & \bullet
\end{tikzcd}
\end{equation}
in $\Aut(D)\Ds D$ with $d\in D$ and $f(d)\in \Aut(D)$. Horizontal composition of 1-morphisms is defined by
\begin{equation}
\begin{tikzcd}
 &  \ \ \ \ \ar[d, Rightarrow, shorten <= 5, shorten >= 5, "d_1"] & \bullet \ar[d, "f(d_1)"] \\
\bullet \ar[rru,bend left] \ar[rrd,bend right,swap] \ar[rr] & \ \ \ \ \ar[d, Rightarrow, shorten <= 5, shorten >= 5, "d_2"] &  \bullet \ar[d, "f(d_2)"] \\
 & \ & \bullet
\end{tikzcd}
=
\begin{tikzcd}
 & \ \ar[dd, Rightarrow, shorten <= 15, shorten >= 15, "{d_2f(d_2)[d_1]}"] & \bullet \ar[dd, "f(d_2)f(d_1)", bend left] \\
\bullet \ar[rru,bend left, "f_1"] \ar[rrd,bend right, "f_2",swap] & & \\
 & \ & \bullet
\end{tikzcd} 
\end{equation}
\item A 2 morphism $(d_1,f(d_1)) \longrightarrow (d_2,f(d_2))$ is an element $d \in D$ such that 
\begin{equation}
\begin{tikzcd}
 & \ \ar[dd, Rightarrow, shorten <= 15, shorten >= 15, "d_1",swap] & \bullet \ar[dd, "f(d_1)", bend right, swap] \ar[dd, "f(d_2)", bend left] \\
\bullet \ar[rru,bend left, "f_1"] \ar[rrd,bend right, "f_2",swap] & & \overset{d}{\longrightarrow}  \\
 & \ & \bullet
\end{tikzcd} 
=
\begin{tikzcd}
 & \ \ar[dd, Rightarrow, shorten <= 15, shorten >= 15, "d_2"] & \bullet \ar[dd, "f(d_2)"] \\
\bullet \ar[rru,bend left, "f_1"] \ar[rrd,bend right, "f_2",swap] & & \\
 & \ & \bullet
\end{tikzcd} 
\end{equation}
Note that $d=d_1^{-1}d_2$. Horizontal composition is given $d_2 \circ d_1 = d_2f_2[d_1]$ and vertical composition is defined by multiplication in $D$. 
\end{itemize}    
There is a natural fibration 
\begin{align}
\begin{split}
\pi \colon E(\Aut(D)\Ds D) &\longrightarrow B(\Aut(D)\Ds D) \\
f &\longmapsto * \\
(d,f(d)) &\longmapsto f(d) \\
d &\longmapsto d  \ .
\end{split}
\end{align}    
This allows us to calculate the homotopy pullback as a regular pullback:
\begin{itemize}
\item 
Objects are elements of $\Aut(D)$;

\item
Morphisms are pairs of morphisms 
\begin{equation}
\left(g,
\begin{tikzcd}
 & \ \ar[dd, Rightarrow, shorten <= 15, shorten >= 15, "d"] & \bullet \ar[dd, "f(d)"] \\
\bullet \ar[rru,bend left, "f_1"] \ar[rrd,bend right, "f_2",swap] & & \\
 & \ & \bullet
\end{tikzcd} \right)
\end{equation} 
such that $f(d)= \lambda (g)$;

\item
There are only identity 2-morphisms. 
\end{itemize}
Composition is defined by 
\begin{equation}
\left(g_2,
\begin{tikzcd}
 & \ \ar[dd, Rightarrow, shorten <= 15, shorten >= 15, "d_2"] & \bullet \ar[dd, "\lambda(g_2)"] \\
\bullet \ar[rru,bend left, "f_2"] \ar[rrd,bend right, "f_3",swap] & & \\
 & \ & \bullet
\end{tikzcd} \right) 
\circ 
\left(g_1,
\begin{tikzcd}
 & \ \ar[dd, Rightarrow, shorten <= 15, shorten >= 15, "d_1"] & \bullet \ar[dd, "\lambda(g_1)"] \\
\bullet \ar[rru,bend left, "f_1"] \ar[rrd,bend right, "f_2",swap] & & \\
 & \ & \bullet
\end{tikzcd} \right) 
\end{equation}
\begin{equation}
= 
\left(g_2g_1,
\begin{tikzcd}
 &  \ \ \ \ \ar[d, Rightarrow, shorten <= 5, shorten >= 5, "d_1"] & \bullet \ar[d, "\lambda(g_1)", swap] \ar[dd, "\lambda (g_2g_1)", bend left=50, in=80, out=100] & \\
\bullet \ar[rru,bend left] \ar[rrd,bend right,swap] \ar[rr] & \ \ \ \ \ar[d, Rightarrow, shorten <= 5, shorten >= 5, "d_2"] &  \bullet \ar[d, "\lambda(g_2)",swap] \ar[r, Leftarrow, shorten <= 2, shorten >= 11] & \ \\
 & \ & \bullet & 
\end{tikzcd} \right) 
\end{equation}
where the unlabelled 2-morphisms are part of the definition of $\alpha$. This reproduces the composition law in $\widehat{G}$. This gives a graphical way to think about the 
multiplication in $\widehat{G}$, which we found helpful for concrete computations. 

Taking repetitively homotopy fibres of the homotopy fibre computed with respect to the base point gives the fibre sequence to the left
\begin{equation}
\begin{tikzcd}
G \ar[r]\ar[d, "\Omega \alpha"] & * \ar[d] & \\
\mybox{$\Aut (D) \Ds D$} \ar[d] \ar[r] & \mybox{$\Aut(D)\Ds \widehat{G}$} \ar[r] \ar[d]& * \ar[d] \\
* \ar[r] & \mybox{$BG$} \ar[r,"\alpha"] &  B (\Aut(D) \Ds D)
\end{tikzcd}
\end{equation} 
where  $\Omega \alpha$ is the looping of $\alpha$. 
The new homotopy pullbacks in the fibre sequence can be calculated by the pasting property and looping.
There are natural functors $\iota_{\widehat{G}}  \colon B\widehat{G}\longrightarrow  \Aut(D)\Ds \widehat{G}$ and $\iota_{D}\colon BD \longrightarrow \Aut(D)\Ds D $. 
We get from the boxed elements the following commuting diagram
\begin{equation}
\label{Diagram: Fiber sequence implies exact sequence}
\begin{tikzcd}
BD \ar[r] \ar[d,"\iota_D",swap] & B\widehat{G} \ar[d,"\iota_{\widehat{G}}",swap] \ar[rd] & \\
\Aut(D)\Ds D \ar[r]  & \Aut(D)\Ds \widehat{G} \ar[r] & BG
\end{tikzcd} 
\end{equation}  
encoding the exact sequence of groups.

\subsection{Gauging discrete symmetries}\label{Sec: Gauging}
Now we define what it means to gauge the symmetries from 
Section~\ref{Sec: Finite symmetries}. 
There is an inclusion of categories $i\colon \Cob_n \hookrightarrow G\text{-}\Cob_n$ for every group $G$ by equipping every manifold with the trivial $G$-bundle. The pullback $i^* Z_G \colon \Cob_n \longrightarrow \fvs$ of a $G$-equivariant field theory $Z_G:G\text{-}\Cob_n\longrightarrow\fvs$ carries additional structure: 
By evaluating $Z_G$ on gauge transformations of the trivial $G$-bundle
on an $n{-}1$-dimensional manifold $\Sigma$ we get a representation of
$G$ on $i^* Z_G(\Sigma)$ which is compatible with the definition on
cobordisms. Hence $i^* Z_G$ is a quantum field theory with internal
$G$-symmetry in the sense of Definition~\ref{Def: Internal symmetry}, i.e.\ a symmetric monoidal functor
\begin{align}
i^* Z_G \colon \Cob_n \longrightarrow G\text{-}\mathsf{Rep} \ .
\end{align}   
Recall that we considered a $G$-equivariant field theory as a field theory coupled to classical $G$-gauge fields. Given a field theory $Z \colon \Cob_n \longrightarrow G\text{-}\mathsf{Rep}$ with internal $G$-symmetry we can ask if the symmetry can be gauged.
\begin{definition}\label{Def: Gauging}
Let $Z\colon \Cob_n \longrightarrow G\text{-}\mathsf{Rep}$ be a topological quantum field theory with internal $G$-symmetry. A $G$-equivariant field theory $Z_G\colon \GCob \longrightarrow \fvs$ \emph{gauges} the internal $G$-symmetry if $i^*Z_G=Z$ as functors $\Cob_n \longrightarrow G\text{-}\mathsf{Rep} $.
\end{definition}
In general it may be impossible to gauge a given symmetry due to cohomological obstructions. In this case we say that the symmetry has a 't Hooft anomaly. In the following we will study under which conditions the symmetries discussed in Section \ref{Sec: Finite symmetries} have 't Hooft anomalies.

\begin{remark}
In three dimensions the question of whether a given field theory can be gauged is related to an interesting algebraic problem~\cite{GaugingMC}.
A three-dimensional extended topological quantum field theory is described by a modular tensor category $\mathsf{M}$~\cite{bartlett2015modular}. An internal $G$-symmetry corresponds to a homotopy coherent action of $G$ on $\mathsf{M}$ via braided autoequivalences. The group of braided autoequivalences up to natural isomorphism is known as the Brauer-Picard group.   
The modular tensor category corresponding to the Dijkgraaf-Witten theory with gauge group $D$ and topological action $\omega\in Z^3(BD;U(1))$ is the category of finite-dimensional modules over the $\omega$-twisted Drinfeld double of the group algebra $\C[D]$, see 
Theorem~\ref{Thm: 3D DW}.
The corresponding Brauer-Picard group for $\omega=0$ is studied in detail in~\cite{Lentner:2015pla}. The more general case of the representation category of Hopf algebras which includes the case of non-trivial $\omega$ is studied in~\cite{2017arXiv170205133L}. The kinematical symmetries studied in this paper correspond to the subgroup of classical symmetries in~\cite{2017arXiv170205133L}.       
Three-dimensional $G$-equivariant extended field theories correspond to $G$-modular categories~\cite{TV14,EOFK}. 
The symmetry corresponding to a homotopy coherent action of $G$ on a modular tensor category $\mathsf{M}$ can be gauged if there exists a $G$-(multi) modular category $\mathsf{M}_G= \bigoplus_{g\in G}\, \mathsf{M}_g$ such that $\mathsf{M}_1=\mathsf{M}$ in a compatible way. The question of under which conditions such an extension exists is answered in~\cite{2009arXiv0909.3140E}, whereby the case relevant for Dijkgraaf-Witten theories is discussed in their appendix. 

In this algebraic framework the gauging of more complicated symmetries
of arbitrary three-dimensional extended topological field theories can
be addressed using the cobordism hypothesis and representation
theoretic techniques, see also~\cite{GaugingMC}. A detailed study of this would be interesting. However, we refrain from doing so in this thesis and focus instead on a largely dimension-independent discussion.   
\end{remark}

The non-abelian group 2-cocycle describing the action of $G$ on a Dijkgraaf-Witten theory with gauge group $D$ and topological action $\omega \in Z^n(BD;U(1))$ determines an (not necessarily central) extension 
\begin{align}
1\longrightarrow D \overset{\iota}{\longrightarrow} \widehat{G} \overset{\lambda }{\longrightarrow} G \longrightarrow 1\ .
\label{Group extension}
\end{align} 
The short exact sequence should be understood as a way to combine $D$- and $G$-gauge fields into a single $\widehat{G}$-gauge field. 
If there exists $\widehat{\omega} \in H^n(B\widehat{G};U(1))$ such that $\iota^* \widehat{\omega}= \omega$ we say that the symmetry $G$ is {anomaly-free}. 
In particular, the existence of $\widehat{\omega}$ ensures that $\omega$ is preserved by
the non-abelian 2-cocycle as we will now explain:

We define an $n{-}1$-cochain $\Phi_g$ on $D$ as follows: 
Let $\chi \colon \Delta^{n-1}\longrightarrow BD$ be an $n{-}1$-chain which we can include into $B\widehat{G}$ along $\iota$. Putting $\widehat{g}$ on the interval we get a map $[0,1]\times \Delta^{n-1} \longrightarrow B\widehat{G}$. Integration of the pullback of $\widehat{\omega}$ over $[0,1]\times \Delta^{n-1}$ gives the inverse of the value of $\Phi_g$ evaluated on the $n{-}1$-simplex. 
The value of $-\delta \Phi_g$ on an $n$-simplex $(d_1,\dots ,d_n) \colon \Delta^n \longrightarrow BD$ is given by 
\begin{align}
\begin{split}
\big\langle [\widehat{g}\times(d_1,\dots , d_n)]^*\widehat{\omega} & \,,\, [0,1]\times \partial \Delta^n \big\rangle \\ & = \big\langle [\widehat{g}\times(d_1,\dots ,d_n)]^*\widehat{\omega} \,,\, \partial[0,1]\times \Delta^n - \partial([0,1]\times \Delta^n) \big\rangle \\[4pt]
&= \big\langle [\widehat{g}\times(d_1,\dots ,d_n)]^*\widehat{\omega} \,,\, (\{1\}-\{0\})\times \Delta^n \big\rangle \\[4pt]
&= \alpha(g^{-1})^*\omega(d_1, \dots ,d_n)- \omega(d_1, \dots , d_n) \ .
\end{split}
\end{align}  
Hence the $\Phi_g$ provide a homotopy fixed point structure. 
In this case the equivariant Dijkgraaf-Witten theory ${{\Za_{\text{DW}}}}^\lambda_{\widehat{\omega}}$~\ref{Sec: Equivariant DW} can be used to gauge the symmetry.
   
\begin{theorem}\label{Theorem: Gauging}
Let ${{\Za_{\text{DW}}}}_{{\omega}}$ be a discrete gauge theory with topological action $\omega\in Z^n(BD;U(1))$ and kinematical $G$-symmetry described by an extension
\begin{align}
1\longrightarrow D \overset{\iota}{\longrightarrow} \widehat{G} \overset{\lambda }{\longrightarrow} G \longrightarrow 1
\end{align}
such that there exists $\widehat{\omega} \in Z^n(B\widehat{G};U(1))$ satisfying $\omega =\iota^* \widehat{\omega}$ and the 
fixed point structure is induced by $\widehat{\omega}$. Then the $G$-equivariant Dijkgraaf-Witten theory 
\begin{align}
{{\Za_{\text{DW}}}}^\lambda_{\widehat{\omega}} \colon \GCob \longrightarrow \fvs 
\end{align} 
gauges this symmetry. 
\end{theorem}
\begin{proof}
First we show that the trivial sector, i.e.\ the evaluation on trivial bundles, of 
${{\Za_{\text{DW}}}}^\lambda_{\widehat{\omega}}$ is ${{\Za_{\text{DW}}}}_{{\omega}}$. 
From the exact sequence of groups we get a fibration
\begin{align}
BD \overset{\iota}{\longrightarrow} B\widehat{G} \overset{\lambda}{\longrightarrow} BG
\end{align}
of classifying spaces. Let $N$ be a manifold of dimension $n$ or $n-1$. We have to evaluate the homotopy fibre $\lambda_*^{-1}[\star]$ of the trivial bundle $\star \colon N \longrightarrow BG$. 
Using that $\lambda_*$ is a fibration we can replace the homotopy fibre with the 
regular fibre $\Pi_1[N, BD]$.
This shows that we can replace limits and integration over the homotopy fibre of the trivial bundle with limits and integration over $\BunD(N)$ for every manifold $N$. For this reason \eqref{EQ:EDW on Objects} and \eqref{EQ:EDW on morphisms} reduce the
corresponding formulas for ${{\Za_{\text{DW}}}}_{{\omega}}$, since $\widehat{\omega}$ pulls back to $\omega$.

Next we show that this gauges the symmetry, see~\eqref{Eq: Concrete form of the quantum symmetry}:
\begin{align}
{{\Za_{\text{DW}}}}^\lambda_{\widehat{\omega}}(\Sigma,g)f(\iota_* \varphi_D,\id)=f\big(\iota_* \mathcal{R}(\alpha(g^{-1}))\varphi_D, \id\big) \, \langle \varphi_D^*\Phi_g , \sigma_\Sigma \rangle
\end{align}  
for all closed $n{-}1$-dimensional manifolds $\Sigma$, $f(\,\cdot\,) \in {{\Za_{\text{DW}}}}^\lambda_{\widehat{\omega}}(\Sigma, \star \colon \Sigma \longrightarrow BG)$ and $\varphi_D \colon \Sigma \longrightarrow BD$, where we interpret $g\in G$ as a homotopy from the constant map $\star$ to itself. 
By~\cite[Proposition~4.2~(b)]{OFK} we have 
\begin{align}
{{\Za_{\text{DW}}}}^\lambda_{\widehat{\omega}}(\Sigma,g)f(\iota_* \varphi_D,\id)= f(\iota_* \varphi_D,g^{-1}) \ .
\end{align}  
We have to calculate a lift for the homotopy $g^{-1}$, i.e.\ a gauge transformation $\widehat{g}^{\,-1} \colon \iota_* \varphi_D \longrightarrow \iota_* \varphi'_D$ such that $\lambda (\widehat{g}^{\,-1}) = g^{-1}$. We use the concrete description of $\widehat{G}$-bundles as elements of the functor category $[\Pi_1(\Sigma),* \Ds \widehat G\,]$, where $\Pi_1(\Sigma)$ is the fundamental groupoid of $\Sigma$. A lift of the gauge transformation is then given by conjugation with $(g^{-1},1)\in \widehat G$. We calculate its action on the image $d\in D \subset \widehat{G}$ of a path in $\Sigma$. The inverse is given by \cite{Blanco2005}
\begin{equation}
\left( g \,,\, \sigma_\alpha(g,g^{-1})^{-1} \right ) \ .
\end{equation} 
Then
\begin{align}
\begin{split}
\big(g^{-1},1\big)\,\big(1,d\big)\,\big( g\,,\, \sigma_\alpha(g,g^{-1})^{-1}\big) &= \big(g^{-1},\alpha(g^{-1})[d]\big)\,\big( g\,,\, \sigma_\alpha(g,g^{-1})^{-1}\big) \\[4pt]
&= \big( 1\,,\,\alpha(g^{-1})[d]\,\alpha(g^{-1})[\sigma_\alpha(g,g^{-1})^{-1}]\,\sigma_\alpha(g^{-1},g)\big) \\[4pt]
&= \big( 1\,,\,\alpha(g^{-1})[d]\, \sigma_\alpha (g^{-1},g)^{-1} \, \sigma_\alpha (g^{-1},g) \big) \\[4pt]
&= \big( 1,\alpha(g^{-1})[d]\big) \ ,
\end{split}
\end{align}
where in the third equality we used $\sigma_\alpha(g^{-1},g)^{-1}= \alpha(g^{-1})[\sigma_\alpha(g,g^{-1})^{-1}]$, which follows from \eqref{Eq: Non abelian cocycle} with $g_1=g^{-1}$, $g_2=g$ and $g_3= g^{-1}$ using $\sigma_\alpha({1,g})=\sigma_\alpha({g,1})=1$ for all $g\in G$. This shows that $\varphi_D'= \mathcal{R}(\alpha(g^{-1}))\varphi_D$. 
That $f(\,\cdot\,)$ is a parallel section implies
\begin{align}
f(\iota_* \varphi_D,g^{-1})= E_{\widehat{\omega}}^{-1}(\Sigma,\widehat{g})f\big(\mathcal{R}(\alpha(g^{-1}))\varphi_D , \id \big) \ . \label{Eq: Proof Main theorem}
\end{align}  
By definition $E^{-1}_\omega(\Sigma,\widehat{g})= \langle \varphi_D^* \Phi_g, \sigma_\Sigma \rangle$ with $\Phi_g$ induced by $\widehat{\omega}$ as above. Inserting this into \eqref{Eq: Proof Main theorem} gives \eqref{Eq: Concrete form of the quantum symmetry} where $\Phi_g$ provide the homotopy fixed point structure.
\end{proof}

\begin{remark}
Theorem~\ref{Theorem: Gauging} provides a general mechanism for the gauging of
symmetries. However, we cannot show that it is impossible to gauge the
symmetry when the conditions of Theorem~\ref{Theorem: Gauging} are not satisfied, i.e.\
when no such $\widehat{\omega}$ exists. Also the homotopy fixed point 
structure which is gauge by this construction is induced by $\widehat{\omega}$. It 
is not clear to us whether all homotopy fixed point structures arise in this way. In general, we 
expect this not to be the case. \end{remark}

\begin{example}\label{Example: Gauging1}
We describe a discrete two-dimensional gauge theory with gauge group $D=\Z_N\times \Z_N$ and topological action $\omega_k \in H^2(\Z_N\times \Z_N; U(1))$ as defined in~\eqref{EQ:Def 2 cocycle}. The action of the symmetry group $G$ on $D$ can be encoded in a short exact sequence
\begin{align}
1\longrightarrow D \longrightarrow \widehat{G} \longrightarrow G \longrightarrow 1 \ .
\end{align}  
Set $G= \Z_M \times \Z_M$ and consider the extension
\begin{align}
(0,0)\longrightarrow \Z_N\times \Z_N \xrightarrow{(M , M)\,\cdot}
  \Z_{N\,M}\times \Z_{N\,M}\longrightarrow \Z_M \times \Z_M
  \longrightarrow (0,0) \ .
\end{align} 
In this case we can gauge the symmetry in the manner of Theorem~\ref{Theorem: Gauging} for the topological action
$\omega_k$ with $k\in \lbrace 0,1,\dots, N-1 \rbrace$ if and only if
$k$ is divisible by $M$ modulo $N$, i.e.\ there exists $k'\in \Z$ such that $k'\,M = k \text{ mod}\ N$. Concretely, $\widehat{\omega}\in H^2(\Z_{N\,M}\times \Z_{N\,M}; U(1))$ is given by $\omega_{k'}\in Z^2(\Z_{N\,M}\times \Z_{N\,M};U(1))$. 
This simple example already shows that we cannot gauge every symmetry using Theorem~\ref{Theorem: Gauging};
it is discussed in~\cite{Kapustin:2014gua,Gaiotto2014} in the context of
$0$-form and $1$-form global symmetries. We will discuss obstructions
to finding an appropriate lift $\widehat{\omega}$ in more detail and generality in Section~\ref{Sec: Obstruction}. 
\end{example}

\begin{example}\label{Ex: D8 Z2}
The cyclic group $\Z_2$ acts on the dihedral group $D_8$ by conjugation with the generator $a$. Since this is an action via inner automorphisms it preserves the non-trivial 2-cocycle $\omega\in H^2(BD_8;U(1))$ from Example~\ref{Ex: 2D coycles}.b. This action defines a non-abelian 2-cocycle with trivial map $\sigma$. The corresponding extension is
\begin{align}
1 \longrightarrow D_8 \longrightarrow D_8 \rtimes \Z_2 \longrightarrow \Z_2 \longrightarrow 1 \ . 
\end{align} 
The Pauli group is the subgroup 
\begin{align}
P_1 = \lbrace \pm\, \mathds{1}_2 , \pm \,\iu\, \mathds{1}_2 , \pm\, \sigma_x , \pm \,\iu\, \sigma_x , \pm\, \sigma_y , \pm \,\iu\, \sigma_y , \pm\, \sigma_z , \pm \,\iu\, \sigma_z \rbrace
\end{align}
of the unitary group $U(2)$ with the Pauli spin matrices
\begin{align}
\sigma_x=  \bigg( \begin{matrix}
0 & 1 \\ 
1 & 0
\end{matrix} 
\bigg) \ , \quad
\sigma_y= \bigg(
\begin{matrix}
0 & -\,\iu\, \\ 
\,\iu\, & 0
\end{matrix} 
\bigg) \qquad \mbox{and} \qquad
\sigma_z=  \bigg(
\begin{matrix}
1 & 0 \\ 
0 & -1
\end{matrix} 
\bigg) \ . 
\end{align}
There is an equivalence of extensions 
\begin{equation}
\begin{tikzcd}
 &  & D_8\rtimes \Z_2 \ar[dd, "\vartheta"] \ar[rd] & & \\
1 \ar[r] & D_8 \ar[rd] \ar[ru] & & \Z_2 \ar[r] & 1\\
 &  & P_1 \ar[ru] & &
\end{tikzcd}
\end{equation}
given by $\vartheta (a^i\,b^j,k)= (\,\iu\, \sigma_x)^i\,\sigma_y^j\,\sigma_x^k$, showing that this extension is non-trivial even though it comes from an inner automorphism. The intuitive reason for this is that conjugation by $a^2$ is the identity even though $a^2$ itself is not. 
We will show in Example~\ref{Ex: D8 Z2 v2} that this symmetry cannot be gauged in the manner of Theorem~\ref{Theorem: Gauging}.
\end{example}

\begin{example}\label{Example: Gauging2}
In three dimensions we can look at the extension 
\begin{align}
0\longrightarrow \Z_N \xrightarrow{M\,\cdot } \Z_{N\,M}
  \longrightarrow \Z_M \longrightarrow 0 \ .
\end{align}
The 3-cocycle $\omega_k$ defined in~\eqref{Eq: Zn 3-cocycle} can
always be gauged by the 3-cocycle $\widehat{\omega}_k \in
H^3(B\Z_{N\,M};U(1))$ corresponding to the same value of $k$. 
\end{example}

\subsection{Obstructions to gauging}\label{Sec: Obstruction}
We start by recalling some basic definitions related to spectral sequences
focusing on first quadrant spectral sequences for simplicity. 
A spectral sequence is a tool for the computation of (co)homology groups. 
They usually arise in contexts where the homology groups 
of different chain complexes are related, e.g.\ for a topological fibre bundle 
$F\longrightarrow E \longrightarrow B$ there exists a spectral sequence relating 
the cohomology of $B$ and $F$ to the cohomology $E$. 

\begin{definition}[See e.g.\ Definition 5.2.3 of~\cite{Weibel}]
A \emph{(cohomological) spectral sequence (starting at $E_2$)} consists of a family of abelian groups\footnote{More generally, one can define a spectral sequence with values in an arbitrary abelian 
category.}
$\{ E_r^{p,q} \}$ for all integers $p,q\in \Z$ and $r\geq 2$ together with differentials
\begin{align}
d_r^{p,q} \colon E_r^{p,q} \longrightarrow E_r^{p+r,q-r+1}
\end{align} 
and isomorphisms $E_{r+1}^{p,q}\cong \ker (d_r^{p,q})/ \im (d_r^{p-r,q+r-1})$.

A \emph{first quadrant spectral sequence} is a spectral sequence which satisfies 
$E_r^{p,q}= 0$ for $p< 0$ and or $q< 0$.
\end{definition}  
       
The collection of all terms for a fixed value of $r$ is called the \emph{$E_r$-page} of 
the spectral sequence. For example the $E_2$-page of a general first quadrant spectral sequence looks as follows
\begin{center}
\begin{tikzpicture}
  \matrix (m) [matrix of math nodes,
    nodes in empty cells,nodes={minimum width=5ex,
    minimum height=5ex,outer sep=-5pt},
    column sep=1ex,row sep=1ex]{
                & \vdots        &    \vdots      &     \vdots    &     \vdots      & \ddots \\
          2     &     E_2^{0,2} &   E_2^{1,2}   &   E_2^{2,2}  &    E_2^{3,2}    & \dots  \\
          1     &  E_2^{0,1} &   E_2^{1,1} &  E_2^{2,1} &  E_2^{3,1}       & \dots  \\
          0     &  E_2^{0,0}  & E_2^{1,0} &  E_2^{2,0} &  E_2^{3,0}     &  \dots \\
    \quad\strut &   0  &  1  &  2  &   3   & \strut \\};
  \draw[-stealth] (m-2-2.south east) -- (m-3-4.north west);
  \draw[-stealth] (m-3-3.south east) -- (m-4-5.north west);
\draw[thick] (m-1-1.east) -- (m-5-1.east) ;
\draw[thick] (m-5-1.north) -- (m-5-6.north) ;
\end{tikzpicture}
\end{center}
where we have only drawn two differentials.

For a first quadrant spectral sequence $( E_r^{p,q}, d_r^{p,q}) $ and fixed $p,q$ 
there exists an $r_0\in \N$ such that for all $r\geq r_0$ we have $q-r+1<0$ and $p-r<0$ implying 
$\ker (d_r^{p,q})=E_r^{p,q}$, $\im (d_r^{p-r,q+r-1})=0$ and hence $E_{r_0}^{p,q}\cong 
E_{r_0+1}^{p,q}\cong E_{r_0+2}^{p,q} \cong \dots $ . We set $E_\infty^{p,q}= \lim_{r\geq r_0} E_r \cong E_{r_0}$.      

\begin{definition}
Let $\{H^n\}_{n\geq 0}$ be a family of abelian groups. We say a 
first quadrant spectral sequence 
$( E_r^{p,q}, d_r^{p,q})$ \emph{converges to $H^n$} if there are filtrations 
\begin{align}
0=F^{n+1}H^n \subset F^{n}H^n \subset \dots \subset F^{0}H^n= H^n  
\end{align} 
for $n\geq 0$ and isomorphisms $E_\infty^{pq}\cong F^p H^{p+q} / F^{p+1} H^{p+q}$. For a convergent 
spectral sequence we write 
\begin{align}
E_2^{p,q}\Longrightarrow H^{p+q} \ \ .
\end{align}
\end{definition} 
\begin{example}
Let $f \colon E \longrightarrow B$ be a Serre fibration with fibre $F$. There
is a convergent spectral sequence of cohomology group~\cite[Section 5.3]{Weibel}. 
\begin{align}
E_2^{p,q}=H^p(B; H^q(F))\Longrightarrow H^{p+q}(E) \ \ .
\end{align} 
\end{example}
\begin{remark}
Let $E_2^{p,q}\Longrightarrow H^{p,q}$ be a convergent spectral sequence of vector 
space. The splitting of every exact sequence of vector spaces implies 
\begin{align}
H^n \cong \bigoplus_{p+q=n} E_\infty^{p,q} \ \ . 
\end{align}
In general the spectral sequence determines $H^n$ only up to the solution of iterated 
extensions problems. 
\end{remark}

Let $E_2^{p,q}\Longrightarrow H^{p+q}$ be a convergent first quadrant spectral sequence.
The terms $E_r^{0,n}$ and $E_r^{n,0}$ are called \emph{edge terms}. 
Note that all differentials reaching $E_r^{0,n}$ are zero implying 
$E_\infty^{0,n} \subset \dots \subset E_r^{0,n}\subset \dots \subset E_2^{0,n}$. 
We get induced \emph{edge maps}
$H^n \longrightarrow H^n/F^1H^n \cong E_\infty^{0,n}\longrightarrow E_2^{0,n}$. 
Dually, all the differentials  $d_r^{n,0}$ are zero and hence $E_{r+1}^{n,0}$ is 
a quotient of $E_{r}^{n,0}$ and there are natural edge maps 
$E_2^{n,0}\longrightarrow E_\infty^{n,0}= F^nH^n\subset H^n$.  

After this brief introduction to spectral sequences we can apply them to the 
problem at hand.
In this section we work with the group cohomology $H^n(G;U(1))$ instead of 
the cohomology of $BG$ with coefficients in $U(1)$.
Let 
\begin{align}
1 \longrightarrow D \overset{\iota}{\longrightarrow} \widehat{G} \overset{\lambda}{\longrightarrow} G \longrightarrow 1
\end{align}
be an exact sequence of finite groups and $\omega$ an $n$-cocycle on $D$. 
For the application of Theorem~\ref{Theorem: Gauging} the existence of an $n$-coclycle
$\widehat{\omega}$ on $\widehat{G}$ is required. 
There are obstructions for $\widehat{\omega}$ to exist:\footnote{For a physical perspective on these obstructions and the corresponding spectral sequence, see~\cite{Thorngren2015}.}
 there is an action of $G$ on $H^n(D;U(1))$ induced by conjugation in $\widehat G$. Every cocycle on $\widehat{G}$ is invariant under conjugation and hence the first obstruction for $\widehat \omega$ to exist is 
\begin{align}
\omega \in H^n\big(D;U(1)\big)^G \ .
\end{align}
By definition, the obstruction is always satisfied if the extension
corresponds to a kinematical symmetry. There are further obstructions
encoded by first quadrant
Lyndon-Hochschild-Serre spectral sequence corresponding to the exact
sequence of groups~\eqref{Group extension} which takes the form  
\begin{align}
E^{p,q}_2=H^p\big(G;H^q(D;U(1))\big)\Longrightarrow H^{p+q}\big(\widehat G ; U(1)\big) 
\end{align}
with edge maps $H^{n}(\widehat G ; U(1)) \twoheadrightarrow E^{ 0,n}_\infty = E^{0,n}_{n+2} \hookrightarrow H^n(D;U(1))^G $ given by the restriction to $D$ (see e.g.~\cite[Section~6.8]{Weibel}). Hence, we see that $\omega \in \im(\iota^*)=E^{0,n}_{n+2}$ if and only if 
\begin{align}\label{Eq: Obstruction}
{\rm d}^{0,n}_i \omega = 0 \ \in \ E^{i,n+1-i}_i 
\end{align}
for all $i \in \{ 2, \dots , n+1 \}$. Note that $\dd^{0,n}_i \omega$ is only well-defined if $\dd^{0,n}_{i-1}\omega = 0$ and $E^{i,n+1-i}_i$ is a sub-quotient of $H^i\big(G;H^{n+1-i}(D;U(1))\big)$. 

To understand these obstructions in more detail we introduce the algebraic model for the spectral sequence~\cite[Section 2]{HS53}. The group cohomology of $\widehat{G}$ can be computed from the normalised cochain complex $C^\bullet(\widehat{G};U(1))$:
\begin{align}
0 \longrightarrow C^0\big(\widehat{G};U(1)\big) \longrightarrow C^{1}\big(\widehat{G};U(1)\big) \longrightarrow \cdots \ .
\end{align}
We introduce a filtration 
\begin{align}
 C^\bullet\big(\widehat{G};U(1)\big)=F^0C^\bullet\big(\widehat{G};U(1)\big) \supseteq F^1 C^\bullet\big(\widehat{G};U(1)\big)\supseteq F^2 C^\bullet\big(\widehat{G};U(1)\big) \supseteq \cdots 
\end{align}
where $F^i C^n(\widehat{G};U(1))$ is $0$ for $i>n$ and otherwise consists of all normalized $n$-cochains which are $0$ as soon as $n-i+1$ entries are in the image of $D$.
This filtration is compatible with the coboundary operator $\delta$ and hence induces a spectral sequence, which is the Lyndon-Hochschild-Serre spectral sequence. Concretely we set
\begin{align}
Z^{p,q}_r & \coloneqq \text{ker} \Big(F^p
            C^{p+q}\big(\widehat{G};U(1)\big)\overset{\delta}{\longrightarrow}
            C^{p+q+1}\big(\widehat{G};U(1)\big)\big/F^{p+r}
            C^{p+q+1}\big(\widehat{G};U(1)\big) \Big) \ , \\[4pt]
B^{p,q}_r & \coloneqq \delta\Big(
            F^{p-r+1}C^{p+q-1}\big(\widehat{G};U(1)\big)\Big)\cap F^p
            C^{p+q}\big(\widehat{G};U(1)\big) \ , \\[4pt]
E^{p,q}_r &\coloneqq Z^{p,q}_r \big/ \big(B^{p,q}_r + Z_{r-1}^{p+1,q-1}\big) \ .
\end{align} 
The differential $\delta\colon  C^{p+q}(\widehat{G};U(1)) \longrightarrow C^{(p+r)+(q-r+1)}(\widehat{G};U(1))$ induces the corresponding differentials
\begin{align}
\dd^{p,q}_r \colon E^{p,q}_r \longrightarrow E^{p+r,q-r+1}_r 
\end{align}
in the spectral sequence.

We consider the two-dimensional case as a warm-up. We fix $\omega \in H^2(D;U(1))$. 
The corresponding element in $E^{0,2}_2$ is the 2-cochain 
\begin{align}
\begin{split} 
\tilde{\omega} \colon \widehat{G}\times \widehat{G} &\longrightarrow U(1) \\
\big((d,g)\,,\, (d',g')\big) & \longmapsto \omega(d,d') \ .
\end{split} 
\end{align}  
This is not a cocycle in general since the multiplication in $\widehat{G}$ is twisted by the corresponding non-abelian 2-cocycle. This cochain obviously pulls back to $\omega$. The ensuing calculation can be understood as trying to find a 2-cochain on $\widehat{G}$ which is $0$ when pulled back to $D$ such that its sum with $\tilde{\omega}$ is closed. 

The first obstruction $\dd^{0,2}_2 \tilde{\omega} = 0$ is equivalent to $\delta \tilde{\omega} \in B^{2,1}_2 + Z^{3,0}_1$. This implies that there exists $\gamma_1 \in F^1C^2(\widehat{G};U(1))$ such that 
\begin{align}
\delta \gamma_1 \in F^2 C^3\big(\widehat{G};U(1)\big) \qquad \mbox{and} \qquad \delta (\tilde{\omega}-\gamma_1) \in Z^{3,0}_1  \ . 
\end{align} 
This means that we can consider $\tilde{\omega}$ as an element of
$E^{0,2}_3 \cong \ker\,\dd^{0,2}_2$. 
Note that the identification is not the identity, rather we have to map $\tilde{\omega}$ to $\tilde{\omega}-\gamma_1$. 
We have thus shown that if the first obstruction vanishes, then there exists $\theta\in Z^{3,0}_1 = Z^{3}(G;U(1))$ and a cochain $\omega'=\tilde{\omega}-\gamma_1$ such that $\delta \omega' = \lambda^* \theta$ and $\iota^* \omega'= \omega$. 
 
The next obstruction is $\dd^{0,2}_3 \tilde{\omega}= 0$. This is equivalent to $\delta (\tilde{\omega}-\gamma_1) \in B^{3,0}_3$, hence there exists $\gamma_2\in F^1C^2(\widehat{G};U(1))$ such that $\delta \gamma_2= \delta(\tilde{\omega}-\gamma_1)\in F^3C^2(\widehat{G};U(1))$. This implies $\delta(\tilde{\omega}-\gamma_1 -\gamma_2)= 0$ and $\iota^*(\tilde{\omega}-\gamma_1 -\gamma_2) = \omega$ since $\gamma_1$ and $\gamma_2$ are elements in $F^1C^2(\widehat{G};U(1))$. This gives the desired 2-cocycle $\widehat{\omega}=\tilde{\omega}-\gamma_1 -\gamma_2$.
 
The discussion above readily generalises to arbitrary dimension $n$. If the first obstruction vanishes then there exists $\gamma_1 \in  F^1C^n(\widehat{G};U(1))$ such that
\begin{align}
\delta(\tilde{\omega}-\gamma_1)\in F^3C^{n+1}\big(\widehat{G};U(1)\big) \ .
\end{align}
More generally if the first $m\leq n$ obstructions vanish, there are elements $\gamma_1,\dots , \gamma_m \in F^1C^n(\widehat{G};U(1)) $ such that 
\begin{align}
\delta \gamma_i \ &\in \ F^i C^{n+1}\big(\widehat{G};U(1)\big) \ , \\[4pt]
\delta\Big(\tilde{\omega} -\sum_{i=1}^k\, \gamma_i\Big) \  &\in \ F^{k+2} C^{n+1}\big(\widehat{G};U(1)\big) \ , 
\end{align} 
for all $i,k=1,\dots,m$.
In particular, if all obstructions vanish then
\begin{align}
\delta\Big(\tilde{\omega} -\sum_{i=1}^n\, \gamma_i\Big) = 0 
\end{align}
and 
\begin{align}
\iota^* \Big(\tilde{\omega} -\sum_{i=1}^n\, \gamma_i\Big) = \omega \ \in \ H^n\big(D;U(1)\big) \ . 
\end{align}
We are mostly interested in the case when all obstructions except the last one vanish. In this case 
\begin{align}
\delta\Big(\tilde{\omega} -\sum_{i=1}^{n-1}\, \gamma_i\Big) = \lambda^* \theta  
\end{align}
with $\theta \in Z^{n+1}(G;U(1))$, since closed elements of $F^{n+1}C^{n+1}(\widehat{G};U(1))$ are in one-to-one correspondence with $Z^{n+1}(G;U(1))$. 
We summarize the present discussion in 

\begin{proposition}\label{Prop: Obstructions}
Let 
\begin{align}
1 \longrightarrow D \overset{\iota}{\longrightarrow} \widehat{G}
  \overset{\lambda}{\longrightarrow} G \longrightarrow 1
\end{align}
be a short exact sequence of groups, $n$ a natural number and $\omega$ an $n$-cocycle on $D$ with values in $U(1)$. 
\begin{itemize}
\item[{\rm (a)}]
When all obstructions in \eqref{Eq: Obstruction} vanish, then there exists $\widehat{\omega}\in Z^n(\widehat{G};U(1))$ satisfying $\iota^* \widehat{\omega}= \omega$. 

\item[{\rm (b)}]
When the first $n-1$ obstructions in \eqref{Eq: Obstruction} vanish, then there exist $\omega'\in C^n(\widehat{G};U(1))$ and $\theta\in Z^{n+1}(G;U(1))$ satisfying $\iota^* \omega'= \omega$ and $\delta \omega'=\lambda^*\theta$. 
\end{itemize} 
\end{proposition}

\begin{remark}
If the first $n-1$ obstructions vanish we can realize the anomalous
field theory as a boundary state of a classical
$n{+}1$-dimensional Dijkgraaf-Witten theory with topological action
$\theta$. In Section~\ref{Sec: Boundary DW} we will
explain this point in more detail. 
\end{remark}

\begin{example}
We have seen in Example~\ref{Example: Gauging2} that for the extension 
\begin{align}
0\longrightarrow \Z_N \longrightarrow \Z_{N\,M}\longrightarrow \Z_M \longrightarrow 0
\end{align}
all 3-cocycles on $\Z_N$ arise as pullbacks of 3-cocycles on $\Z_{N\,M}$, hence all obstructions vanish in this case.  
The $E_2$-page of the corresponding spectral sequence is 
\begin{center}
	\begin{tikzpicture}
	\matrix (m) [matrix of math nodes,
	nodes in empty cells,nodes={minimum width=5ex,
		minimum height=5ex,outer sep=-5pt},
	column sep=1ex,row sep=1ex]{
		& \vdots        &    \vdots      &     \vdots    &     \vdots      & \ddots \\
		3     &     \Z_M & \Z_{\gcf(M,N)} &  \Z_{\gcf(M,N)} & \Z_{\gcf(M,N)}   & \dots  \\
		2     &     0 &   0   &   0  &   0    & \dots  \\
		1     &  \Z_m &   \Z_{\gcf(M,N)} &  \Z_{\gcf(M,N)} & \Z_{\gcf(M,N)}       & \dots  \\
		0     &  U(1) & \Z_N &  0 &  \Z_N      &  \dots \\
		\quad\strut &   0  &  1  &  2  &   3   & \strut \\};
	\draw[-stealth] (m-2-2.south east) -- (m-3-4.north west);
	\draw[-stealth] (m-4-3.south east) -- (m-5-5.north west);
	\draw[thick] (m-1-1.east) -- (m-6-1.east) ;
	\draw[thick] (m-6-1.north) -- (m-6-6.north) ;
	\end{tikzpicture}
\end{center}
where $\gcf(M,N)$ is the greatest common factor of $M$ and $N$. For $\gcf(M,N)=1$ this implies directly that 
all obstruction vanish. To conclude the existence of $\widetilde{\omega}$ for $\gcf(M,N)>1$ from the spectral 
sequence a more detailed analysis is needed. 
\end{example}

\begin{example}\label{Ex: D8 Z2 v2}
Following up on Example~\ref{Ex: D8 Z2} we show that for the symmetry described by 
\begin{align}\label{Eq: Exact sequence D8 Z2}
1 \longrightarrow D_8 \longrightarrow P_1 \longrightarrow \Z_2 \longrightarrow 1
\end{align} 
the non-trivial 2-cocycle $\omega\in H^2(D_8;U(1))$ cannot be
gauged. The cohomology groups of the Pauli group $P_1$ can be computed using a
computer algebra package such as {\tt GAP}~\cite{Joyner08aprimer} and
the universal coefficient theorem to get
\begin{align}
\begin{split}
H^0\big(P_1; U(1)\big)&= U(1) \ , \\[4pt]
H^1\big(P_1; U(1)\big)&= \Z_2 \times \Z_2 \times \Z_2 \ , \\[4pt]
H^2\big(P_1; U(1)\big)&= \Z_2\times \Z_2 \ , \\[4pt]
H^3\big(P_1; U(1)\big)&= \Z_2\times \Z_2\times \Z_8  \ .
\end{split}
\end{align}  
The $E_2$ page of the corresponding spectral sequence is 
\begin{center}
\begin{tikzpicture}
  \matrix (m) [matrix of math nodes,
    nodes in empty cells,nodes={minimum width=5ex,
    minimum height=5ex,outer sep=-5pt},
    column sep=1ex,row sep=1ex]{
                &         &          &         &           & \\
          2     &     \Z_2 &  \Z_2   &   \Z_2  &   \Z_2    &  \\
          1     &  \Z_2\times \Z_2 &   \Z_2\times \Z_2  &  \Z_2\times \Z_2 &  \Z_2\times \Z_2       &  \\
          0     &  U(1)  & \Z_2 &  0  &  \Z_2     &  \\
    \quad\strut &   0  &  1  &  2  &   3   & \strut \\};
  \draw[-stealth] (m-2-2.south east) -- (m-3-4.north west);
  \draw[-stealth] (m-3-3.south east) -- (m-4-5.north west);
\draw[thick] (m-1-1.east) -- (m-5-1.east) ;
\draw[thick] (m-5-1.north) -- (m-5-6.north) ;
\end{tikzpicture}
\end{center}
The two differentials drawn are $0$ as can be checked by using the concrete description of the differentials in~\cite{HUEBSCHMANN1981296} and the fact that \eqref{Eq: Exact sequence D8 Z2} is the extension of $\Z_2$ by $D_8$ corresponding to the inner automorphism of $D_8$ given by conjugation with $a\in D_8$.  
Hence the $E_3$ page is given by 
\begin{center}
\begin{tikzpicture}
  \matrix (m) [matrix of math nodes,
    nodes in empty cells,nodes={minimum width=5ex,
    minimum height=5ex,outer sep=-5pt},
    column sep=1ex,row sep=1ex]{
                &         &          &         &           & \\
          2     &     \Z_2 &     &     &       &  \\
          1     &  \Z_2\times \Z_2 &   \Z_2\times \Z_2  &   &   &  \\
          0     &  U(1)  & \Z_2 &  0  &  \Z_2     &  \\
    \quad\strut &   0  &  1  &  2  &   3   & \strut \\};
\draw[thick] (m-1-1.east) -- (m-5-1.east) ;
\draw[thick] (m-5-1.north) -- (m-5-6.north) ;
\end{tikzpicture}
\end{center}
From $E_3^{1,1}=\Z_2\times \Z_2=E_\infty^{1,1}$ and $H^2(P_1;U(1))=
\Z_2\times \Z_2$ we deduce that the differential $\dd_3^{0,2}\colon
\Z_2 \longrightarrow \Z_2$ is an isomorphism. 
This implies that the symmetry corresponding to \eqref{Eq: Exact sequence D8 Z2} of the non-trivial topological action $\omega \in H^2(D_8;U(1))$ cannot be gauged using Theorem~\ref{Theorem: Gauging}, since the second obstruction corresponding to $\dd^{0,2}_3$ does not vanish. However, since the first obstruction vanishes we can gauge the symmetry using the relative field theory constructed in Section~\ref{Sec: Boundary DW}.   
\end{example}

\begin{example}
We have seen in Example~\ref{Example: Gauging1} that for the extension 
\begin{align}
(0,0) \longrightarrow \Z_2\times \Z_2 \longrightarrow \Z_4\times \Z_4
  \longrightarrow \Z_2\times \Z_2 \longrightarrow (0,0) 
\end{align}
the 2-cocycle $\omega_1 \in H^2(\Z_2\times \Z_2;U(1))$ cannot be obtained as the pullback of a 2-cocycle on $\Z_4\times \Z_4$. 
The corresponding 2-cochain is given by 
\begin{align}
\begin{split}
\tilde{\omega}_1 \colon (\Z_4\times \Z_4)^2 & \longrightarrow U(1)\\
\big((a_1,b_1)\,,\, (a_2,b_2)\big) & \longmapsto \exp\Big(\pi \,\iu\,
                                     \Big\lfloor\frac{a_1}{2}\Big\rfloor\,
                                     \Big\lfloor\frac{b_2}{2}\Big\rfloor\Big)
                                     \ .
\end{split}
\end{align}
To find the corresponding obstructions we calculate using $\big\lfloor
\frac{a+b}{2} \big\rfloor = a\,b +\big\lfloor \frac{a}{2}
\big\rfloor+\big\lfloor \frac{b}{2} \big\rfloor$ mod~$2$ to get
\begin{align}
\delta \tilde{\omega}_1 \big((a_1,b_1)\,,\,
  (a_2,b_2)\,,\,(a_3,b_3)\big) = \exp\bigg(\pi \,\iu\, \Big(
  a_1\,a_2\,\Big\lfloor\frac{b_3}{2}\Big\rfloor
  +\Big\lfloor\frac{a_1}{2}\Big\rfloor \, b_2\,b_3 \Big) \bigg) \ .
\end{align}
Using the computer algebra program Maple we verified 
by checking all possibilities that there are no solutions to the equation
\begin{align}
\delta (\tilde{\omega}_1-\gamma_1)=\lambda^*\theta
\end{align}
with $\gamma_1 \in F^1C^2(\Z_4\times\Z_4;U(1))$ and $\theta\in Z^3(\Z_2\times\Z_2;U(1))$. Hence the first obstruction $\dd^{0,2}_2\omega_1$ does not vanish. 
\end{example}

\subsection{Fully extended TQFT's and defects}\label{Sec: Fully extended}
We conclude with a few general remarks on different interpretations of the obstructions discussed in the previous section. This part will be informal and conjectural. In particular, we will use $\infty$-categories in an intuitive way without choosing a model or giving 
any details.
 
In an extended field theory manifolds with boundaries can be decomposed into manifolds with corners, but we cannot
decompose complicated $n-2$-dimensional manifolds into simpler pieces. For this and other reasons it is desirable to 
include manifolds with corners of arbitrary codimension. This leads to the introduction of an $n$-category of 
$n$-dimensional cobordisms with structure. Defining these $n$-categories is a notoriously hard problem, which to 
the best of our knowledge has only been solved in the case of manifolds with a particular type of topological structures, called $G$-structures, using $(\infty,n)$-categories~\cite{Lurie2009a, CS}.
\begin{definition}
	Let $G$ be a Lie group equipped with a smooth group homomorphism $G\longrightarrow GL(n)$ and $M$ an $n$-dimensional 
	manifold. A \emph{$G$-structure on $M$} is a lift up to homotopy 
	\begin{equation}
	\begin{tikzcd}
	& BG \ar[d] \\
	M \ar[ru] \ar[r, "\tau_M", swap] & BGL(n)
	\end{tikzcd}	
	\end{equation} 
	of the map $\tau_M$ classifying the frame bundle of $M$.
\end{definition} 
In more geometric terms a $G$-structure is a reduction of the structure group of the frame bundle to $G$. However, the
homotopy theoretical description is important to define morphisms of $G$-structures in a way suitable for the application to topological
field theories.\footnote{Actually, manifolds with $G$-structure form an $\infty$-category, see e.g.~\cite{CS} for details.}  

\begin{example}\label{Exa: G-structure}
	\begin{itemize}
		\item
		For $G=*$ the definition of a $G$-structure on $M$ reduces to the definition of a framing of $M$, i.e.\ the choice of a trivialisation of 
		the tangent bundle. Framed manifolds play a central role in the classification of fully extended topological field 
		theories as we will explain below.
		
		\item For $G=SO(n)$ considered as a subgroup of $GL(n)$ a $G$-structure corresponds to the choice of an orientation. 
		A geometric inclined reader would probably have expected an $SO(n)$-structure to also involve the choice of a metric. 
		This is not the case, because of the homotopical definition of morphisms between $G$-structures.     
		
		\item 
		Let $D$ be a finite group, $G=D\times SO(n)$ and $\rho \colon D\times SO(n) \longrightarrow SO(n) \longrightarrow
		Gl(n)$ the composition of the projection onto $SO(n)$ composed with the inclusion of $SO(n)$ into $Gl(n)$ and 
		$M$ an $n$-dimensional manifold. 
		To describe $G$-structures we note that $B(D \times SO(n) )\cong BD \times BSO(n)$ and hence a $G$-structure can be described 
		by a pair of continuous maps $\psi_D\colon M \longrightarrow BD$ and $\psi_{SO(n)}\colon M \longrightarrow BSO(n)$ 
		together with a homotopy filling the diagram 
		\begin{equation}
		\begin{tikzcd}
		& BSO(n) \ar[d] \\
		M \ar[ru, "\psi_{SO(n)}"] \ar[r, "\tau_M", swap] & BGL(n)
		\end{tikzcd} 	
		\end{equation}    
		In summary a $G$-structure on $M$ consists of a principal $D$-bundle with classifying map $\psi_D$ and an orientation. 
 $D\times SO(n)$-structures are a topological way of describing the background fields considered in this chapter.   
	\end{itemize}
\end{example}
In~\cite{CS} a symmetric monoidal $(\infty,n)$-category $\Cob_{G,n}^\infty$ of cobordisms equipped with $G$-structures 
is constructed
based on a proposal by Lurie~\cite{Lurie2009a}. We will write $\DCob_n^\infty$ for $\Cob_{D\times SO(n),n}^\infty$. 
Let $\cC$ be a symmetric monoidal $(\infty,n)$-category. A \emph{fully extended topological field theory} with values in 
$\cC$ is a symmetric monoidal functor 
\begin{align}
Z \colon \Cob_{G,n}^\infty \longrightarrow \cC \ \ .
\end{align}  
There exist a conjectural classification of fully extended field theories, called the cobordism hypothesis going back 
to the work of~\cite{CH}. Lurie has given a detailed sketch of a proof~\cite{Lurie2009a}. There is another proposal for a proof 
using factorization homology due to Ayala and Francis~\cite{AFCobordism}. 
According to the cobordism hypothesis fully extended framed topological field theories $\Cob_{*,n}^\infty\longrightarrow \cC$
are classified by the $(\infty,0)$-category (space) $\cC^{f.d.}$ of fully dualizable objects in $\cC$. 
The precise definition will not play any importance in what follows and hence we refer the reader to~\cite{Lurie2009a} for a
definition (see also \cite[III.5]{NotesFH} for an informal discussion). 

Let $G$ be a group equipped with a group homomorphism $G\longrightarrow Gl(n)$. This induces an action of $G$ on 
$\Cob_{*,n}^\infty$ via rotation of the framing, which induces an action of $G$ on the space of framed topological 
field theories $\cC^{f.d.}$. The coboridms hypothesis conjectures that fully extended topological field theories 
$\Cob_{G,n}^\infty \longrightarrow \cC$ are classified by the space $\cC^{f.d.,G}$ of homotopy fixed points of this action~\cite{Lurie2009a}.

\begin{definition}		
	Let $X$ be a topological space with right $G$-action. A \emph{homotopy fixed point} of the $G$-action is a $G$-equivariant
	map $EG \longrightarrow X$. 
\end{definition}    

A corollary of the cobordism hypothesis is a classification of fully extended equivariant field theories
\begin{align}		
Z\colon \DCob_n^\infty \longrightarrow \cC \ \ .
\end{align} 

\begin{theorem}\label{Theo: Classification G-TFT}
Assuming that the cobordism hypothesis is true the following holds:
the $(\infty,0)$-category of $D$-equivariant fully extended field theories is equivalent to the $(\infty,0)$-category 
	of fully extended oriented topological field theories with $D$ action.  
\end{theorem}
\begin{proof}
	The cobordism hypothesis implies that the space of $D$-equivariant fully extended field theories is $\Hom_{D\times SO(n)}(ED\times ESO(n), \cC^{f.d.})$. Via adjunction this space can be described as $\Hom_{D}(ED,\Hom_{SO(n)}(ESO(n),\cC^{f.d.}))$. Again using the cobordism hypothesis the space 
	$\Hom_{SO(n)}(ESO(n),\cC^{f.d.})$ is the space of oriented fully extended topological field theories. We recall
	from Example~\ref{Exa: G-structure} that the $D$ action on the framing is trivial. In this case an equivariant map is the same as 
	a continuous map $BD=ED/D \longrightarrow \Hom_{SO(n)}(ESO(n),\cC^{f.d.})$, i.e.\ an oriented field theory with $D$-action.  
\end{proof}

\begin{example}
	One-dimensional oriented topological field theories with values in vector spaces are classified by finite-dimensional 
	vector spaces. Theorem~\ref{Theo: Classification G-TFT} implies now that 1-dimensional $D$-equivariant topological field
	theories are classified by finite dimensional representations of $D$. 
	This approach should also be helpful to study fully extended topological field theories in higher dimensions. 
	In 2-dimensions one would reproduce the results of~\cite{2DEHQFT}. 
	Furthermore, Theorem~\ref{Theo: Classification G-TFT} relates the results from~\cite{2009arXiv0909.3140E} directly to equivariant field 
	theories. We hope to develop this approach to the classification of equivariant topological field theories in more detail in the future.  
\end{example}

Theorem~\ref{Theo: Classification G-TFT} provides a different point of view on the absence of 't Hooft anomalies, 
since it implies
that any symmetry which can be extended to the point can be gauged in a canonical way. Dijkgraaf-Witten theories are believed
to be fully extended, see~\cite{FLHT} for arguments into this direction. This makes the 
following interpretation of the $n+1$ obstructions from Section~\ref{Sec: Obstruction} natural: the first obstruction 
correspondence to extending the symmetry from the partition function to the state space. 
Indeed, recall from Proposition~\ref{Prop: Classical symmetry} that a homotopy coherent action $\alpha \colon BG
\longrightarrow B\Aut(D)\DS D $ satisfying $[\alpha(g)^*\omega]=[\omega]$ can be extended to the state space if 
there are $\Phi_g \in C^{n-1}(G;U(1))$ such that $\delta \Phi_g = \omega -\alpha(g^{-1})^*\omega$ and 
\begin{align}\label{Eq: extension}
\Phi_{g_1}+\alpha(g_1^{-1})^*\Phi_{g_2} = \Phi_{g_1\,g_2}+\sigma_{g_1,g_2}[\omega]  \ , 
\end{align} 
up to coboundary terms. To see whether this equation admits a solution we pick a collection of $\Phi'_g$ 
satisfying the first condition\footnote{This is possible because $\alpha$ preserves the cohomology class of $\omega$.} 
and consider the 2-cocycle  
\begin{align}
U(g_1,g_2) \coloneqq \Phi'_{g_1^{-1}}+\alpha(g_1)^*\Phi'_{g_2^{-1}}-\Phi'_{(g_1\,g_2)^{-1}}+\sigma_{g_1^{-1},g_2^{-1}}[\omega]
\end{align}
on $G$ with values in $H^{n-1}(D;U(1))$. If this 2-cocycle is trivial in cohomology we can chose $\Phi_g$ satisfying
Equation~\eqref{Eq: extension}. In~\cite[Section 11.8.1]{2009arXiv0909.3140E} it is explained that this 2-cocycle is
the differential $d_2^{0,n}\omega \in H^2(G;H^{n-1}(D;U(1)))$ for $n=3$. We expect this relation to hold for arbitrary
$n$.  
Summarizing the discussion we have seen that the structure ensuring that we can extend the symmetry to the state space
exists if and only if the first obstruction from the spectral sequence vanishes.  
  
We conjecture that this correspondence continues, i.e.\ that the vanishing of the second obstruction in the spectral sequence 
allows us to extend the symmetry to the extended Dijkgraaf-Witten from Section~\ref{Sec: DW via Orbifold}, the third 
obstruction to the extension of a symmetry for the twice extended Dijkgraaf-Witten theory, and so on. Furthermore, 
we expect this to be related to equipping the cocycle $\omega$ with higher homotopy fixed point structures
analogues to Definition~\ref{Def: Preserved}.  

We now sketch an interpretation of 't Hooft anomalies in terms of defects. 
Defects can be considered as extended observables, consisting of a submanifold of the spacetime manifold together 
with a ``label" specifying the type of defect. In general defects can meet on lower dimensional submanifolds. 
Defect networks can be described by labelled stratified manifolds, see Figure~\ref{Fig: Defect} for an example. 
There are versions of the cobordism category $\CobF$ including defects~\cite{DefectTFT}. Every quantum field theory 
admits a trivial defect $1$ which does not change the partition function.      

\begin{figure}[h]
\centering
\includegraphics[scale=0.7]{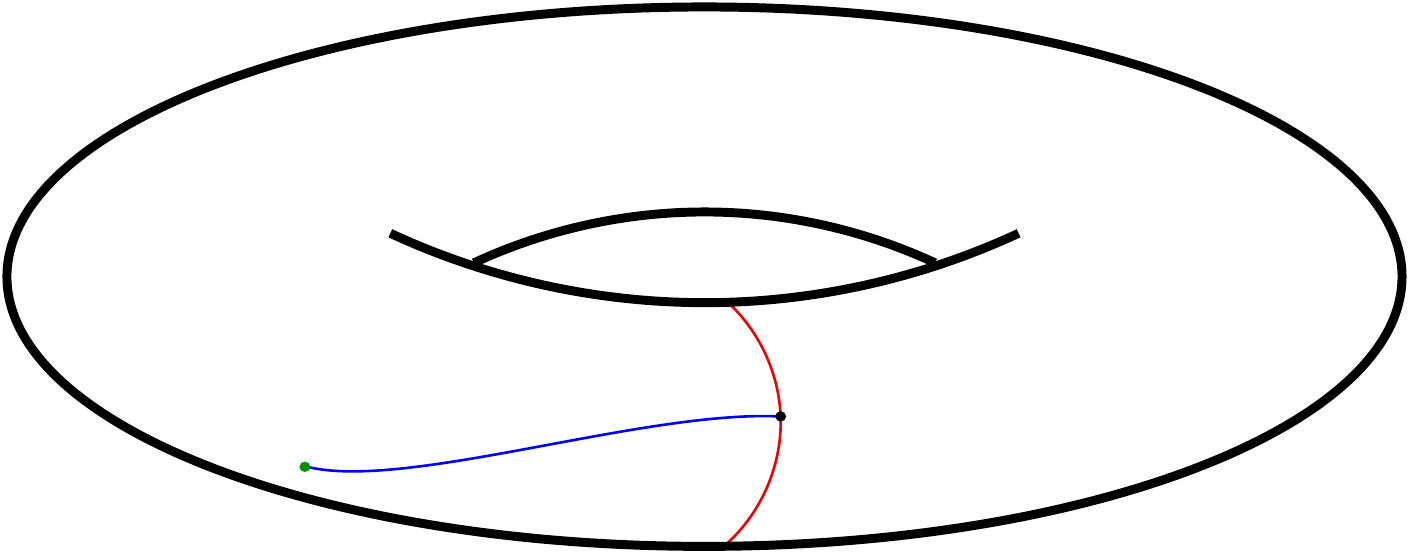}	
\caption{Sketch of a defect network on the 2-dimensional torus with two codimension one defects in blue and red and two 
point defects in black and green. The green point defect is between the trivial defect $1$, which we have not drawn, and the blue defect }
\label{Fig: Defect}	
\end{figure} 
A defect is called \emph{topological} 
if its expectation value is invariant under isotopies of the defect network. For topological defects there is a ``fusion
operation" which corresponds to bringing two defects close together. A topological defect $D$ is called \emph{invertible}
if there exists a defect $D^{-1}$ such that their fusion product $D\otimes D^{-1}$ is the trivial defect.  
The symmetries of a quantum field theory,
are closely related to the invertible codimension one topological defects in the theory: for every symmetry there should
be a corresponding invertible topological defect which acts on a field passing through the defect by applying the symmetry 
(See Figure~\ref{Fig:Sym_Defect} for a sketch). For more details on this approach we refer to~\cite{SymmetriesD} and     
for an application of defects to the description of symmetries of 2-dimensional Yang-Mills theory to~\cite{2DYM}. 
It is believed that (topological) defects of $n$-dimensional field theories assemble into an $n$-category with 
objects $n$-dimensional quantum field theories, 1-morphisms domain walls (codimension 1 defects) between quantum field 
theories, 2-morphisms codimension 2 defects between domain walls and so on up to $n$-morphisms which are given by point 
defects. This believe has been made prices in 2 and 3 dimensional topological field theories~\cite{2CatDefect,3CatDefect}.

\begin{figure}[htb]
\small
\begin{center}
\begin{overpic}[scale=1]
{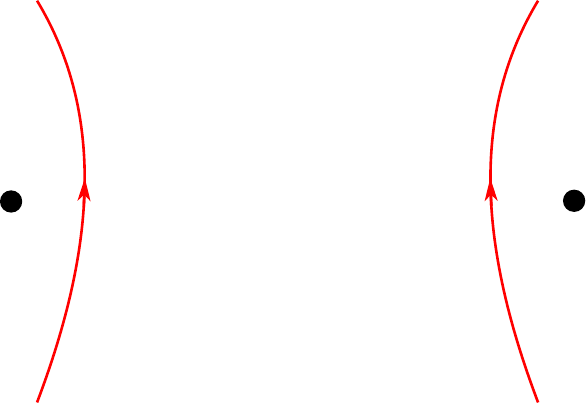}
\put(-1,38){$\Psi$}
\put(90,38){$\varphi\cdot\Psi$}
\put(40,30){{\Large$\longleftrightarrow$}}
\end{overpic}
\end{center}
\caption{\small A defect corresponding to the symmetry $\varphi$,
indicated by the directed lines. A field insertion $\Psi$ moving through 
the defect corresponds to the action of $\varphi$ on $\Psi$.}
\label{Fig:Sym_Defect}
\normalsize
\end{figure}

In the context of state sum models for 2-dimensional topological field theories the defect bicategory $\mathcal{B}$ has the following
algebraic description~\cite{2CatDefect, LecturenotesDefects}:
\begin{itemize}
	\item 
	\textbf{Objects:} The objects of $\mathcal{B}$ are separable symmetric Frobenius algebras. 
	
	\item 
	\textbf{1-Morphisms:} Let $A$ and $B$ be separable symmetric Frobenius-algebras. A 1-morphism $M\colon A \longrightarrow B$
	is a finite dimensional $A,B$-bimodule. Composition is given by the relative tensor product of bimodules. 
	
	\item 
	\textbf{2-Morphisms:} The 2-morphisms in $\mathcal{B}$ are given by bimodule maps.  	
\end{itemize}    

We now relate the discussion of 't Hooft anomalies presented above specialised to 2-dimensions to topological defects.
Let $D$ be a finite group and $\omega \in Z^2(D;U(1))$ a $D$ 2-cocycle with values in $U(1)$. The separable symmetric 
Frobenius algebra describing the corresponding 2-dimensional Dijkgraaf-Witten theory is the twisted group algebra
$\C[D]^{\omega}$ with basis $\{ \delta_d \}_{d\in D}$ and multiplication 
\begin{align}
\delta_d\cdot \delta_{d'} \coloneqq \omega(d,d') \delta_{d\cdot d'} \ \ .
\end{align} 
A gauge transformation corresponding to an element $d\in D$ acts on this algebra 
via 
\begin{align}
\delta_{d'} \longmapsto \theta_d(d')\delta_{d d'd^{-1}}, 
\end{align}
where $\theta_d = \omega(d,d (-) d^{-1})\cdot \omega(-,d)^{-1}$ is the canonical 1-chain satisfying 
$\delta \theta_d= \omega(-,-)-\omega(d(-)d^{-1},d(-)d^{-1})$. This is nothing else than the algebraic description of the singular 1-chain
constructed by crossing with an interval labelled by $d$ similar to the construction before Theorem~\ref{Theorem: Gauging} to
the algebraic setting.  

Let $G$ be a finite group and $\alpha \colon {*\DS G} \longrightarrow \Aut(D)\DS D $ a non-abelian 2-cocycle 
such that $[\alpha(g)^*\omega]=[\omega]$. Every choice of $\Phi_g\in C^{1}(D;U(1))$ such that $\delta \Phi_g = \omega -
\alpha(g^{-1})^*\omega$ induces an automorphism 
\begin{align}
\begin{split}
\varphi_g \colon \C[D]^\omega &\longrightarrow \C[D]^{\omega} \\
\delta_d & \longmapsto \Phi_g(d) \delta_{\alpha(g^{-1})(d)}  
\end{split}
\end{align}  
of $\C[D]^{\omega}$ implementing the symmetry. Every automorphism of $\C[D]^{\omega}$ induces an invertible bimodule with $\C[D]^{\omega}$ as 
underlying vector space and action given by
\begin{align}
\C[D]^{\omega} \otimes \C[D]^{\omega} \otimes \C[D]^{\omega} \xrightarrow{\varphi_{g}\otimes \id \otimes \id}  
\C[D]^{\omega} \otimes \C[D]^{\omega} \otimes \C[D]^{\omega} \longrightarrow \C[D]^{\omega} , 
\end{align} 
where the last arrow is the multiplication in $\C[D]^{\omega}$. We denote this bimodule by $\C[D]^\omega_{\varphi_g}$. 
This invertible defect implements the symmetry 
corresponding to an element $g\in G$ and the non-abelian 2-cocycle $\alpha$. The relative tensor product 
$\C[D]^\omega_{\varphi_{g_2}} \otimes_{\C[D]^\omega} \C[D]^\omega_{\varphi_{g_1}}$ can be identified with the bimodule
corresponding to the automorphism $\alpha(g_2g_1)^{-1}$ of $D$ and the 1-chain $\Psi= \alpha(g_1^{-1})^*\Psi_{g_2} + \Psi_{g_1} $. 
Note that in general $\Psi$ differs from $\Psi_{g_2g_1}$. To implement the composition law described by the non-abelian 
2-cocycle $\alpha$, the terms should agree up to a gauge transformation by the element $\alpha(g_2,g_1)\in D$.     
Concretely, this implies that we need to choose the $\Psi_{g}$ such that 
\begin{align}
\alpha(g_1^{-1})^*\Psi_{g_2} + \Psi_{g_1}= \Psi_{g_2,g_1} + \theta_{\alpha(g_2,g_1)}
\end{align}
holds. As explained above this equation only admits a solution if the first obstruction in the spectral sequence vanishes. 
In this situation we can implement the action of $G$ via a homotopy coherent action in the defect bicategory $\mathcal{B}$.

To gauge a symmetry realized by a homotopy coherent $G$ action in $\mathcal{B}$ one needs to construct so called orbifold data from the defects~\cite{DefectTFT, GOrb}. This is in general not possible and the obstruction is an element in $H^3(G;U(1))$. For 
Dijkgraaf-Witten theories we expect this cocycle to agree with the obstruction coming from the spectral sequence.    
In case all obstruction vanish we can gauge the theory by introducing defect networks 
implementing the $G$-bundle as for example in~\cite[Section 2.2]{2DYM} and computing
the partition function for this defect network. The result should agree with the one 
derived from the push construction. This is a discrete analogue of the constructions in 
\cite{2DYM}.

\section{Anomalous Dijkgraaf-Witten theories as boundary states}\label{Sec: Boundary DW}
Let ${\Za_{\text{DW}}}_\omega \colon \Cob \longrightarrow \fvs$ be a topological gauge theory with gauge group 
$D$, topological action $\omega\in Z^{n-1}(D;U(1))$ and a kinematic symmetry described 
by an exact sequence of finite groups
 \begin{align}
1\longrightarrow D \overset{\iota}{\longrightarrow} \widehat{G} \overset{\lambda }{\longrightarrow} G \longrightarrow 1\ .
\end{align} 
Assume that we are in the situation of Proposition~\ref{Prop: Obstructions}(b), i.e.\ 
that there exists $\omega'\in C^{n-1}(\widehat{G};U(1))$ and $\theta \in Z^n(G;U(1))$
such that $\iota^*\omega'=\omega$ and $\delta \omega'= \lambda^* \theta$.
In this section we explicitly construct the gauged theory as a relative field theory  
$Z_{\omega'}\colon 1 \Longrightarrow \tr E_\theta$. 
In more physical terms we realise the anomalous field theory as the boundary state 
of a symmetry protected topological phase described by $\theta$.
We start by constructing 
the theory at the level of partition functions in Section~\ref{Sec: partition function DW}. This construction is similar to the one 
in~\cite[Section~3.3]{Witten:2016cio}. However, we use the language of functorial field theories and homotopy fibres to describe the construction. The way boundary field theories 
appear here is to some extent reversed to the way they appear in 
\cite{Witten:2016cio}, where anomalous boundary field theories are
constructed starting from a bulk
Dijkgraaf-Witten theory. 
Instead we start from a field theory with anomaly and show how 
to realize this theory as a boundary field theory.

Afterwards, we construct the theory at the level of state spaces in 
Section~\ref{Sec: State space} using a general pushforward construction for relative field theories.

\subsection{Partition function}\label{Sec: partition function DW}

In this section we construct a natural transformation $Z_{\omega'} \colon 1 \Longrightarrow \text{tr}\, E_\theta$ (see Definition~\ref{Def: Anomalous partition function}) of 
non-extended field theories.

Following the general theory outlined in Section~\ref{Sec: General Theory} we have to specify an element $Z_{\omega'}(\Sigma, \varphi\colon M \longrightarrow BG)$ of $E_\theta(\Sigma,\varphi)$ for all objects $(\Sigma,\varphi)\in \GCob$. Let $\sigma_\Sigma$ be a representative for the fundamental class of $\Sigma$.
We set
\begin{align}\label{Eq: Partition function relative field theory}
Z_{\omega'}(\Sigma,\varphi)= \Big(\int_{(\widehat{\varphi},h)\in \lambda_*^{-1}[\varphi]} \, \langle {\widehat{\varphi}}^{\,*} \omega', \sigma_\Sigma \rangle \, \langle  h^* \theta, [0,1] \times \sigma_\Sigma \rangle\Big) \ [\sigma_\Sigma] \ \in \ E_\theta(\Sigma,\varphi) \ .
\end{align}
\begin{proposition}
$Z_{\omega'}$ is a partition function with anomaly
$E_\theta:\GCob\longrightarrow \fvs$ (see Definition~\ref{Def: Anomalous partition function}).
\end{proposition}     
\begin{proof}
We have to show that $Z_{\omega'}$ is a well-defined natural transformation. This is an immediate consequence of the construction in Section~\ref{Sec: State space}. To get a feeling on how to work concretely with the constructions involved, we present here a part of the proof.
We start by showing that $\langle  {\widehat{\varphi}}\,^* \omega', \sigma_\Sigma \rangle \, \langle  h^* \theta, [0,1] \times \sigma_\Sigma \rangle$ is well-defined on isomorphism classes of $\lambda_*^{-1}[\varphi]$.
Let $\widehat{h}\colon (\widehat{\varphi}_1,h_1) \longrightarrow (\widehat{\varphi}_2,h_2)$ be a morphism in $\lambda_*^{-1}[\varphi]$,
i.e.~a homotopy $\widehat{h}\colon \widehat{\varphi}_1 \longrightarrow \widehat{\varphi}_2$ such that the diagram
\begin{equation}
\begin{tikzcd}
\lambda_* \widehat{\varphi}_1 \ar[rd,"h_1", swap] \ar[rr,"\lambda_* \widehat{h}"] & & \lambda_*\widehat{\varphi}_2 \ar[ld,"h_2"] \\
 & \varphi & 
\end{tikzcd}
\end{equation}
commutes.  The homotopy induces a chain homotopy $H\colon \widehat{\varphi}_{1*} \longrightarrow \widehat{\varphi}_{2*}$ between the induced maps on singular chains given by $H(c)= \widehat{h}_* ( [0,1] \times c)$ for all chains $c \in C_\bullet (\Sigma)$. Hence, writing $U(1)=\R/\Z$ additively for the calculation, we find
\begin{align}
\begin{split}
\big\langle {\widehat{\varphi}_2}^{\,*} \omega', \sigma_\Sigma \big\rangle - \big\langle  {\widehat{\varphi}_1}^{\,*} \omega', \sigma_\Sigma \big\rangle &= \big\langle \omega', \partial H(\sigma_\Sigma) - H (\partial \sigma_\Sigma) \big\rangle\\[4pt]
& = \big\langle \omega' , \partial H (\sigma_\Sigma) \rangle\\[4pt]
& = \big\langle \, \widehat{h}^* \lambda^* \theta, [0,1]  \times \sigma_\Sigma \big\rangle \\[4pt]
&= \big\langle  h_1^* \theta - h_2^* \theta, [0,1] \times \sigma_\Sigma \big\rangle \ .
\end{split}
\end{align} 
This shows that the integration in \eqref{Eq: Partition function relative field theory} is well-defined. 

Let $\sigma_\Sigma'$ be a different representative for the fundamental class of $\Sigma$ and $\chi$ an $n$-chain satisfying $\partial \chi = \sigma_\Sigma' - \sigma_\Sigma $.
To show that \eqref{Eq: Partition function relative field theory} is an element of $E_\theta(\Sigma, \varphi)$ we calculate 
\begin{align}
\begin{split}
\langle {\widehat{\varphi}}^{\,*} \omega',\sigma_\Sigma' - \sigma_\Sigma \rangle \, \langle {h}^* \theta, [0,1] \times (\sigma_\Sigma' - \sigma_\Sigma) \rangle &= \langle  {\widehat{\varphi}}^{\,*} \omega',\partial \chi \rangle \, \langle {h}^* \theta, [0,1] \times \partial \chi \rangle \\[4pt]
& = \langle {\widehat{\varphi}}^{\,*} \lambda^* \theta, \chi  \rangle \, \langle {h}^* \theta, -\lbrace 0 \rbrace \times  \chi + \lbrace 1 \rbrace \times \chi   \rangle \\[4pt]
&= \langle  \varphi^* \theta, \chi    \rangle \ .
\end{split}
\end{align}
This is exactly the required transformation behaviour. We leave the verification of naturality to the reader. 
\end{proof}

\begin{remark}
Before extending the field theory we give the precise form of the composite partition function \eqref{Def: Combined partition function}. We fix an $n$-dimensional manifold $M$ with boundary $\partial M = -\Sigma$ and a principal $G$-bundle $\psi \colon M \longrightarrow BG$. Evaluating $E_\theta$ on $(M,\psi)$ gives a linear map $E_\theta(M,\psi)\colon E_\theta(\Sigma,\psi|_\Sigma)\longrightarrow \C$. The composite partition function is then
\begin{align}
\begin{split} 
Z_{\omega'\,{\rm bb}}(M,\psi, \Sigma)&= E_\theta(M,
                                       \psi)[Z_{\omega'}(\Sigma,\psi|_\Sigma)]
  \\[4pt] &= \Big(\int_{(\widehat{\varphi},h)\in
            \lambda_*^{-1}[\psi|_\Sigma]} \, \langle
            {\widehat{\varphi}}^{\,*} \omega', \partial \sigma_M
            \rangle^{-1} \, \langle  h^* \theta, [0,1] \times \partial
            \sigma_M \rangle^{-1}\Big) \ \langle \psi^* \theta ,
            \sigma_M \rangle \ ,
\end{split} 
\end{align}  
which is gauge-invariant according to the 
general theory outlined in Section~\ref{Sec: Anomaly inflow}. 
\end{remark}

\subsection{State space from a pushforward construction for relative field theories}\label{Sec: State space}
In this section we extend $Z_{\omega'}$ to an anomalous field theory 
$Z_{\omega'}\colon 1 \Longrightarrow \tr E_\theta$ using a general pushforward construction for 
relative field theories. This construction seems to be new to the best of our 
knowledge.\footnote{We thank the referee of \cite{tHooft} for suggesting this
way of interpreting our results.} Let $\lambda \colon \widehat{G}\longrightarrow G$
be a group homomorphism between finite groups and $\lambda \colon \EWGCob \longrightarrow 
\EGCob$ the induced functor between bicategories. Furthermore, let $\Za_1,\Za_2\colon 
\EGCob \longrightarrow \Tvs$ be extended functorial field theories and $Z \colon \tr \lambda^* \Za_1 \Longrightarrow \tr \lambda^*\Za_2$ a relative field theory. 
The relative pushforward construction produces a relative field theory $\lambda_*Z
\colon \tr \Za_1 \longrightarrow \tr \Za_2$. 

Let $(S,\xi\colon S \longrightarrow BG)$ be an object of $\EGCob$. 
Recall that the homotopy fibre $\lambda^{-1}[\xi]$ has pairs $(\widehat{\xi} \colon S \longrightarrow B\widehat{G}, h\colon
\lambda \widehat{\xi}  \longrightarrow \xi)$ as objects. 
From $Z$ we construct a diagram 
\begin{align}
\begin{split}
\widehat{Z}(S,\xi)\colon \lambda^{-1}[\xi] & \longrightarrow [\Za_1(S,\xi), \Za_2(S, \xi)] \\ 
(\widehat{\xi},h) & \longmapsto\left( \Za_1(S,\xi)\xrightarrow{\Za_1(h^{-1})}\lambda^*\Za_1(S,\widehat{\xi})\xrightarrow{Z(S,\widehat{\xi})} \lambda^*\Za_2(S,\widehat{\xi})\xrightarrow{\Za_2(h)} \Za_2(S, \xi) \right) \ \ 
\end{split}
\end{align}
in the functor category $ [\Za_1(S,\xi), \Za_2(S, \xi)]$. Concretely, the action on a morphism $\widehat{h}\colon (\widehat{\xi}, h)\longrightarrow (\widehat{\xi}', h')$ is 
\begin{equation}\label{Eq: Concrete description on morphisms}
\begin{tikzcd}[row sep= 1cm, column sep=1.6cm]
\Za_1(S,\xi)  \ar[rd, "{\Za_1({h'}^{-1})}",swap] \ar[r,"\Za_1(h^{-1})"] & \lambda^* \Za_1(S,\widehat{\xi}) \ar[ld, Rightarrow,shorten >= 5em]  
\ar[r, "{Z(S, \widehat{\xi})}"] \ar[d, "{\lambda^*\Za_1(\widehat{h})}"] & \lambda^* \Za_2(S,\widehat{\xi})
\ar[r, "\Za_2(h)"]\ar[d, "{\lambda^*\Za_2(\widehat{h})}"] \ar[ld, "{ Z(\widehat{h})}", Rightarrow, swap] \ar[rd, Rightarrow,shorten >= 5em]  
& \Za_2(S,\xi) 
\\
\  & \lambda^*\Za_1(S, \widehat{\xi}') \ar[r, "{Z(S,\widehat{\xi}')}", swap] &  \lambda^* \Za_2(S, \widehat{\xi}') \ar[ru, "\Za_2(h')", swap]   & \
\end{tikzcd}
\end{equation}
where the unlabelled natural transformations are part of the coherence isomorphisms for $\Za_1$ and $\Za_2$.  

We set 
\begin{align}
\lambda_*Z(S,\xi) \coloneqq \lim_{\lambda^{-1}[\xi]} \widehat{Z}(S,\xi) =  \int\displaylimits_{\lambda^{-1}[\xi]} \widehat{Z}(S,\xi) \in  [\Za_1(S,\xi), \Za_2(S, \xi)] 
\end{align}
writing the limit as an end in the second equality.  

Let $(\Sigma, \varphi\colon \Sigma \longrightarrow BG)\colon (S_-,\xi_-)\longrightarrow 
(S_+,\xi_+)$ be a 1-morphism in $\EGCob$. 
$(\Sigma, \varphi)$ induces a span of groupoids 
\begin{equation}
\begin{tikzcd}
 & \lambda^{-1}[\varphi] \ar[rd, "{r_+}"] \ar[ld, "{r_-}", swap] & \\ 
 \lambda^{-1}[\xi_-] & & \lambda^{-1}[\xi_+]
\end{tikzcd} \ \ 
\end{equation}
We can use the relative field theory
$Z$ to construct a natural transformation 
\begin{equation}\label{Eq: natural transformation for push}
\begin{tikzcd}
 & \lambda^{-1}[\varphi] \ar[rd] \ar[ld] & \\ 
 \lambda^{-1}[\xi_-] \ar[rd,"{\Za_2(\Sigma, \varphi )\circ \widehat{Z}(S_-,\xi_-)}",swap] \ar[rr,Rightarrow,"{\widehat{Z}(\Sigma, \varphi)}",shorten <= 2em, shorten >= 2em] & & \lambda^{-1}[\xi_+] \ar[ld, "{ \widehat{Z}(S_+,\xi_+) \circ \Za_1(\Sigma,\varphi )}"] \\ 
  & {[\Za_1(S_-,\xi_-), \Za_2(S_+, \xi_+)]} &
\end{tikzcd} \ \ 
\end{equation}   
with component 
\begin{equation}
\begin{scriptsize}
\begin{tikzcd}
\Za_1(S_-,\xi_-) \ar[d, swap, "{\Za_1(\Sigma, \varphi)}"] \ar[r,"\Za_1(h|^{-1}_{S_-})"] & \lambda^* \Za_1(S_-,\widehat{\varphi}|_{S_-}) \ar[ld, Rightarrow]  \ar[r, "{Z(S, \widehat{\varphi}|_{S_-})}"] \ar[d, "{\lambda^*\Za_1(\Sigma, \widehat{\varphi})}"] & \lambda^* \Za_2(S_-,\widehat{\varphi}|_{S_-})
\ar[r, "\Za_2(h|_{S_-})"] \ar[d, Rightarrow, shorten <= 0em, shorten >= 0.3em, swap, "{Z(\Sigma, \widehat{\varphi})}"] \ar[rd, "{\lambda^* \Za_2(\Sigma, \widehat{\varphi})}", swap] & \Za_2(S_-,\xi) \ar[r, "{\Za_2(\Sigma, \varphi)}"] \ar[d, Rightarrow] & \Za_2(S_+,\xi_+)  
\\
\Za_1(S_+,\xi_+) \ar[r,swap, "{\Za_1(h|^{-1}_{S_+})}"] & \lambda^*\Za_1(S_+, \widehat{\varphi}|_{S_+}) \ar[rr, "{Z(S_+,\widehat{\varphi}|_{S_+})}", swap] & \  & \lambda^* \Za_2(S_+, \widehat{\varphi }|_{S_+}) \ar[ru, "\Za_2(h|_{S_+})", swap] &  
\end{tikzcd}
\end{scriptsize}
\end{equation}
at an object $(\widehat{\varphi},h)\in \lambda^{-1}[\varphi]$, where the outer natural
transformations are construct from Lemma~\ref{Lem: non-constant morphism} adjusted to the situation at hand.
There is a general push and pull construction which can be applied in this situation~\cite{SWParallel} inducing a morphism 
\begin{align}
\lambda_*Z(\Sigma,\varphi) \colon \int\displaylimits_{\lambda^{-1}[\xi_-]} \Za_2(\Sigma, \varphi) \circ \widehat{Z}(S_-,\xi_-) \longrightarrow \int\displaylimits_{\lambda^{-1}[\xi_+]} \widehat{Z}(S_+,\xi_+)  \circ \Za_1(\Sigma, \varphi)    \ \ .
\end{align} 
Concretely, the morphism $\lambda_*Z(\Sigma,\varphi)$ is given by the composition of
the three natural morphisms
\begin{itemize}
\item 
The natural map \begin{equation}
 r_-^* \left(\int\displaylimits_{\lambda^{-1}[\xi_-]} \Za_2(\Sigma, \varphi)\circ \widehat{Z}(S_-,\xi_-) \right) \longrightarrow \int\displaylimits_{\lambda^{-1}[\varphi]} \Za_2(\Sigma, \varphi)\circ \widehat{Z}(S_-,\xi_-)\circ r_- 
\end{equation}
which is induced by noticing that $r_-^* \left(\int\displaylimits_{\lambda^{-1}[\xi_-]} \Za_2(\Sigma, \varphi)\circ \widehat{Z}(S_-,\xi_-) \right)$ is a cone over $\Za_2(\Sigma, \varphi)\circ \widehat{Z}(S_-,\xi_-)\circ r_- $.
 
\item The map 
\begin{equation}
\int\displaylimits_{\lambda^{-1}[\varphi]} \Za_2(\Sigma, \varphi)\circ \widehat{Z}(S_-,\xi_-)\circ r_-  \longrightarrow \int\displaylimits_{\lambda^{-1}[\varphi]} \widehat{Z}(S_+,\xi_+) \circ  \Za_1(\Sigma, \varphi) \circ r_+ 
\end{equation} 
induced by the natural transformation $\widehat{Z}(\Sigma, \varphi)$.

\item
The last map 
\begin{equation}
 \int\displaylimits_{\lambda^{-1}[\varphi]} \widehat{Z}(S_+,\xi_+) \circ  \Za_1(\Sigma, \varphi) \circ r_+ \longrightarrow \int\displaylimits_{\lambda^{-1}[\xi_+]} \widehat{Z}(S_+,\xi_+)  \circ \Za_1(\Sigma, \varphi)
\end{equation} 
is a bit more involved and only works because all groupoids involved are 
essentially finite. 
Let $(\widehat{\xi}_+, h\colon \lambda \widehat{\xi}_+ \longrightarrow \xi_+)$ be an element
of $\lambda^{-1}[\xi_+]$. The homotopy fibre $r_+^{-1}[(\widehat{\xi}_+,h)]$ has objects consisting of triples of a map
\begin{align}
\widehat{\varphi} \colon \Sigma \longrightarrow B\widehat{G} \ ,
\end{align}
a gauge transformation (homotopy)
\begin{align}
g \colon \lambda \widehat{\varphi} \longrightarrow \varphi \ ,
\end{align}
and a gauge transformation
\begin{align}
\widehat{h}\colon \widehat{\varphi}\,|_{S_+} \longrightarrow \widehat{\xi}_+
\end{align}
such that the diagram
\begin{equation}
\begin{tikzcd}
\lambda_*\widehat{\varphi}\,|_{S_+} \ar[rr, "\lambda_*\widehat{h}"] \ar[rd,"g|_{S_+}",swap] & & \lambda_*\widehat{\xi}_+ \ar[ld, "h"] \\
 & \xi_+ &
\end{tikzcd}
\end{equation}
commutes. An object $(\widehat{\varphi},g,\widehat{h})$ induce a morphism 
\begin{align}
\begin{split}
\nu_{\widehat{\varphi},g,\widehat{h}} \colon & \int\displaylimits_{\lambda^{-1}[\varphi]} \widehat{Z}(S_+,\xi_+) \circ  \Za_1(\Sigma, \varphi) \circ r_+ \longrightarrow 
\widehat{Z}(S_+,\xi_+) \circ  \Za_1(\Sigma, \varphi)[\widehat{\varphi}|_{S_+},g|_{S_+}]\\ 
&\xrightarrow{\widehat{Z}(S_+,\xi_+) \circ  \Za_1(\Sigma, \varphi)[\widehat{h}]} 
\widehat{Z}(S_+,\xi_+)  \circ \Za_1(\Sigma, \varphi)[\widehat{\xi}|_{+},h]
\end{split}
\end{align}
where the first morphism is part of the universal cone for the limit. A simple
 calculation~\cite{SWParallel} shows that the resulting morphism is well defined on isomorphism 
classes of  $r_+^{-1}[(\widehat{\xi}_+,h)]$ allowing us to define the morphism 
\begin{align}
\nu_{\widehat{\xi}_+,h}= \int\displaylimits_{(\widehat{\varphi},g,\widehat{h})\in r_+^{-1}[(\widehat{\xi}_+,h)]} \nu_{\widehat{\varphi},g,\widehat{h}} \coloneqq 
\sum_{(\widehat{\varphi},g,\widehat{h})\in \pi_0\left(r_+^{-1}[(\widehat{\xi}_+,h)]\right)} \frac{\nu_{\widehat{\varphi},g,\widehat{h}}}{|\Aut(\widehat{\varphi},g,\widehat{h})|}
\end{align}
from $\int\displaylimits_{\lambda^{-1}[\varphi]} \widehat{Z}(S_+,\xi_+) \circ  \Za_1(\Sigma, \varphi) \circ r_+$ to $ \widehat{Z}(S_+,\xi_+)  \circ \Za_1(\Sigma, \varphi)[\widehat{\xi}|_{+},h]$. It is straightforward to verify that these morphism form a cone for the diagram 
$\widehat{Z}(S_+,\xi_+)  \circ \Za_1(\Sigma, \varphi)$~\cite{SWParallel} inducing the desired morphism.

\end{itemize} 
The functors $\Za_1(\Sigma, \varphi)$ and $\Za_2(\Sigma, \varphi)$ are continuous 
and hence $\lambda_*Z(\Sigma,\varphi)$ can be seen as a morphism 
$ \Za_2(\Sigma, \varphi) \circ \lambda_*Z(S_-,\xi_-) \longrightarrow 
\lambda_*(S_+,\xi_+)  \circ \Za_1(\Sigma, \varphi) $.

We can now state the main theorem of this section 
\begin{theorem}\label{Thm: push relative}
Let $\lambda \colon \widehat{G}\longrightarrow G$
be a group homomorphism between finite groups and $\lambda \colon \EWGCob \longrightarrow 
\EGCob$ the induced functor between bicategories. Furthermore, let $\Za_1,\Za_2\colon 
\EGCob \longrightarrow \Tvs$ be extended functorial field theories. 
There is a relative pushforward construction $\lambda_* \colon [\tr \lambda^*\Za_1, \tr
 \lambda^* \Za_2] \longrightarrow [\tr \Za_1, \tr
  \Za_2]$ generalizing the pushforward construction of \cite{OFK}. 
\end{theorem}
\begin{proof}
The proof of this theorem consists of 4 largely independent parts.

\textbf{1. Construction of the missing natural isomorphisms:}
We start by providing the additional data required for an relative field theory,
see Definition~\ref{Def: relative field theory} and~\ref{Def:2sym transformation}. 
First we construct a natural isomorphism
\begin{equation}
\begin{tikzcd}
\fvs \ar[rr] \ar[rd] & \ar[d, Rightarrow, "M_{\lambda_*Z}"] & \Za_2(\emptyset) \\ 
 & \Za_1(\emptyset) \ar[ru,"\lambda_*Z(\emptyset)",swap]
\end{tikzcd} 
\end{equation}
Note that the homotopy fibre of the empty bundle over the empty set only contains 
one element and one morphism. For this reason we can set $M_{\lambda_*Z}=M_Z$, 
where $M_Z$ is the natural transformation which is part of the relative field 
theory $Z$. 

The final data missing to complete the construction of $\lambda_* Z$ are 
natural isomorphisms
\begin{equation}
\begin{tikzcd}
\Za_1(S_1,\xi_1)\boxtimes \Za_1(S_2,\xi_2) \ar[rr] \ar[dd, "{Z(S_1,\xi_1)\boxtimes Z(S_2,\xi_2)}",swap] & \ \ar[dd, "\ \Pi_{\lambda_* Z}", Rightarrow,shorten <= 0.3em, shorten >= 0.3em] & \Za_1(S_1\sqcup S_2, \xi_1 \sqcup \xi_2) \ar[dd,"{Z(S_1 \sqcup S_2,\xi_1 \sqcup \xi_2)}"] \\
 & & \\
\Za_2(S_1,\xi_1)\boxtimes \Za_2(S_2,\xi_2) \ar[rr] & \ & \Za_2(S_1\sqcup S_2, \xi_1 \sqcup \xi_2)  
\end{tikzcd} 
\end{equation}
To construct this map we use that 
\begin{align}
\int\displaylimits_{\lambda^{-1}[\xi_1]}\widehat{Z}(S_1,\xi_1)\boxtimes \int\displaylimits_{\lambda^{-1}[\xi_2]} \widehat{Z} (S_2,\xi_2) \cong \int\displaylimits_{\lambda^{-1}[\xi_1]\times \lambda^{-1}[\xi_2]}  \widehat{Z} (S_1,\xi_1) 
\boxtimes  \widehat{Z} (S_2,\xi_2) \ \ ,
\end{align}
and
$\lambda^{-1}[\xi_1]\times \lambda^{-1}[\xi_2] \cong \lambda^{-1}[\xi_1\sqcup \xi_2]$.
The continuity of all functors involved allows us to compute the limits in the 
functor category $[\Za_1(S_1,\xi_1)\boxtimes \Za_1(S_2,\xi_2), \Za_2(S_1\sqcup S_2, \xi_1 \sqcup \xi_2)  ]$. The natural isomorphism $\Pi_{Z}$ induces a natural transformation 
between the two diagrams under consideration which induces the map $\Pi_{\lambda_* Z}$
between the limits. The compatibility conditions for $M_Z$ and $\Pi_Z$ ensure that 
$M_{\lambda_* Z}$ and $\Pi_{\lambda_* Z}$ satisfy the appropriate compatibility 
conditions (Definition~\ref{Def:2sym transformation}).  

\textbf{2. Gauge invariance of $\lambda_*Z(\Sigma,\varphi)$:}
We now turn our attention to the definition of $\lambda_* Z$ on morphisms. 
Let $(\Sigma,\varphi_1 \colon \Sigma \longrightarrow BG)$ and 
$(\Sigma,\varphi_2 \colon \Sigma \longrightarrow BG)$ be morphisms in 
$\EGCob$ such that there exists a homotopy relative boundary $h\colon \varphi_1 
\longrightarrow \varphi_2$. The homotopy induces an equivalence 
\begin{align}
\begin{split} 
\lambda^{-1}[\varphi_1]& \longrightarrow \lambda^{-1}[\varphi_2] \\ 
 (\widehat{\varphi}, H \colon \lambda \widehat{\varphi}\longrightarrow \varphi_1) & 
 \longmapsto (\widehat{\varphi}, h \circ H\colon \lambda \widehat{\varphi}\longrightarrow \varphi_2)
\end{split} 
\end{align} 
such that 
\begin{equation}
\begin{tikzcd}
 & \lambda^{-1}[\varphi_1] \ar[rd, "r_{+1}"] \ar[ld,"r_{-1}",swap] \ar[dd] & \\ 
 \lambda^{-1}[\xi_-] & & \lambda^{-1}[\xi_+] \\ 
 & \lambda^{-1}[\varphi_2] \ar[ru, "r_{+2}",swap] \ar[lu,"r_{-2}"]  &
\end{tikzcd}
\end{equation}
commutes. Furthermore, the equivalence is compatible with the natural transformation from 
Equation~\ref{Eq: natural transformation for push} (this follows from the 
invariance of $Z$ under homotopies relative boundary) and hence the induced natural 
transformations $\lambda_*Z(\Sigma,\varphi_1)$ and $\lambda_*Z(\Sigma,\varphi_2)$
agree.

\textbf{3. Compatibility with composition:}
Let $(\Sigma_1,\varphi_1)\colon (S_-,\xi_-)\longrightarrow (S,\xi)$ and 
$(\Sigma_2,\varphi_2)\colon (S,\xi)\longrightarrow (S_+,\xi_+)$ be 
1-morphisms in $\EGCob$.  
The most involved part of the proof is to show that
\begin{equation}
\begin{footnotesize}
\begin{tikzcd}
  \Za_2(\Sigma_2,\varphi_2)\circ \Za_2(\Sigma_1,\varphi_1) \circ \lambda_* Z(S_-,\xi_-) \ar[rd, bend left=12, "{\lambda_* Z(\Sigma_1,\varphi_1)}"] \ar[d, "{\Phi_{\Za_2}}", swap]  & \\ 
 \Za_2(\Sigma_2\circ \Sigma_1, \varphi_2\circ \varphi_1) \circ \lambda_*Z(S_-,\xi_-) 
 \ar[d, "{\lambda_* Z(\Sigma_2\circ \Sigma_1, \varphi_2\circ \varphi_1)}", swap]& \Za_2(\Sigma_2,\varphi_2)\circ \lambda_* Z(S,\xi) \circ \Za_1(\Sigma_1,\varphi_1) 
 \ar[d, "{\lambda_* Z(\Sigma_2,\varphi_2)}"] \\
\lambda_* Z(S_+,\xi_+) \circ \Za_1(\Sigma_2\circ \Sigma_1, \varphi_2\circ \varphi_1)  & \lambda_* Z(S_+,\xi_+) \circ \Za_2(\Sigma_2,\varphi_2)\circ  \Za_1(\Sigma_1,\varphi_1)  
\ar[l, "{\Phi_{\Za_{1}}}"]
\end{tikzcd}
\end{footnotesize}
\end{equation}
commutes. Using continuity we see that the left composition is constructed from the span 
  \begin{equation}\label{Eq: 1 proof relative OFK}
\begin{tikzcd}
 & \lambda^{-1}[\varphi_2 \circ \varphi_1] \ar[dd] \ar[rd] \ar[ld] & \\ 
 \lambda^{-1}[\xi_-] \ar[rd,"{ \Za_2(\Sigma_2,\varphi_2)\circ \Za_2(\Sigma_1,\varphi_1) \circ \widehat{Z}(S_-,\xi_-)}",swap]\ar[r,Rightarrow,"{\Phi \bullet \id_{ \widehat{Z}}}",shorten <= 2em, shorten >= 2em]  & \ \ar[r,Rightarrow,"{\widehat{Z}(\Sigma, \varphi)}",shorten <= 2em, shorten >= 2em] & \lambda^{-1}[\xi_+] \ar[ld, "{ \widehat{Z}(S_+,\xi) \circ \Za_1(\Sigma_2\circ \Sigma_1, \varphi_2\circ \varphi_1) }"] \\ 
  & {[\Za_1(S_-,\xi_-), \Za_2(S_+, \xi_+)]} &
\end{tikzcd} 
\end{equation}
On the other hand the right composition is constructed via push and pull operations
from the diagram 
\begin{equation}
\begin{footnotesize}
\begin{tikzcd}[column sep=1cm]
 & \lambda^{-1}[ \varphi_1]  \ar[rd] \ar[ld] & & \lambda^{-1}[ \varphi_2]  \ar[rd] \ar[ld] & \\ 
 \lambda^{-1}[\xi_-] \ar[rrd,"",swap, bend right=15]\ar[rr,Rightarrow,"{\id_{\Za_2(\Sigma_2)}\bullet \widehat{Z}(\Sigma_1,\varphi_1)}",shorten <= 2em, shorten >= 2em]  &  & \lambda^{-1}[\xi]  \ar[d, ""] \ar[rr,Rightarrow,"{\id \bullet \Phi_{\Za_1} \circ \widehat{Z}(\Sigma_2, \varphi_2) \bullet \id }",shorten <= 2em, shorten >= 2em] & & \lambda^{-1}[\xi_+] \ar[lld, bend left=15] \\ 
  &  & {[\Za_1(S_-,\xi_-), \Za_2(S_+, \xi_+)]} &   &
\end{tikzcd} 
\end{footnotesize}
\end{equation}   
By the equivariant Beck-Chevalley condition~\cite[Proposition 2.3]{SWParallel}
this composition agrees with the linearisation of the span 
\begin{equation}\label{Eq: 2 proof relative OFK}
\begin{footnotesize}
\begin{tikzcd}[column sep=1cm] 
& & \lambda^{-1}[\varphi_1]\times_{\lambda^{-1}[\xi]}\lambda^{-1}[\varphi_2]\ar[rd] \ar[ld] \\ 
 & \lambda^{-1}[ \varphi_1] \ar[rr, Rightarrow, "\zeta", shorten <= 2em, shorten >= 2em] \ar[rd] \ar[ld] & & \lambda^{-1}[ \varphi_2]  \ar[rd] \ar[ld] & \\ 
 \lambda^{-1}[\xi_-] \ar[rrd,"",swap, bend right=15]\ar[rr,Rightarrow,"{\id_{\Za_2(\Sigma_2)}\bullet \widehat{Z}(\Sigma_1,\varphi_1)}",shorten <= 2em, shorten >= 2em]  &  & \lambda^{-1}[\xi]  \ar[d, ""] \ar[rr,Rightarrow,"{\id \bullet \Phi_{\Za_1} \circ \widehat{Z}(\Sigma_2, \varphi_2) \bullet \id }",shorten <= 2em, shorten >= 2em] & & \lambda^{-1}[\xi_+] \ar[lld, bend left=15] \\ 
  &  & {[\Za_1(S_-,\xi_-), \Za_2(S_+, \xi_+)]} &   &
\end{tikzcd} 
\end{footnotesize}
\end{equation}
where $\lambda^{-1}[\varphi_1]\times_{\lambda^{-1}[\xi]}\lambda^{-1}[\varphi_2]$ is the
homotopy pullback and $\zeta$ the canonical natural isomorphism corresponding to it, see Appendix~\ref{Sec: Homotop Groupoid}. 
There is an equivalence of groupoids 
\begin{align}
\begin{split}
\Xi \colon \lambda^{-1}[\varphi_1\circ \varphi_2] & \longrightarrow \lambda^{-1}[\varphi_1]\times_{\lambda^{-1}[\xi]}\lambda^{-1}[\varphi_2] \\ 
(\widehat{\varphi}_{1,2}, h\colon \lambda \widehat{\varphi}_{1,2} \longrightarrow \varphi_1\circ \varphi_2 ) & \longmapsto (\widehat{\varphi}_{1,2}|_{\Sigma_1},h|_{\Sigma_1}, \widehat{\varphi}_{1,2}|_{\Sigma_2},h|_{\Sigma_2} , \id_{\widehat{\varphi}_{1,2}|_{S}})
\end{split} \ \ .
\end{align}  
The statement now follows from checking that the pullback along $\Xi$ of 
the natural transformation in Equation~\ref{Eq: 2 proof relative OFK} agrees with
the natural transformation in Equation~\ref{Eq: 1 proof relative OFK}.
We check this explicitly by evaluation the natural transformation at an 
object $(\widehat{\varphi}_{1,2}, h\colon \lambda \widehat{\varphi}_{1,2} \longrightarrow \varphi_1\circ \varphi_2 ) $: 
\begin{landscape}
\begin{equation*}
\begin{small}
\begin{tikzcd}[column sep= 1.6cm, row sep=1.4cm]
\Za_1(S_-,\xi_-) \ar{dd}[rotate=-90, pos=0.5, swap, description] {\Za_1(\Sigma_2 \circ \Sigma_1, \varphi_2 \circ \varphi_1)}  \ar[r,"{\Za_1(h|^{-1}_{S_-})}"] \ar[rd,"{\Za_1(\Sigma_1,\varphi_{1})}",swap]& \lambda^* \Za_1(S_-,\widehat{\varphi}_{1,2}|_{S_-}) \ar[rd] \ar[d, Rightarrow] \ar[r, "{Z(S_-,\widehat{\varphi}_{1,2}|_{S_-})}"] & \lambda^* \Za_2(S_-,\widehat{\varphi}_{1,2}|_{S_-})  \ar[r,"{\Za_2(h|_{S_-})}"] \ar[rd] \ar[d,Rightarrow, "{Z(\Sigma_1, \widehat{\varphi}_{1,2}|_{\Sigma_1})}"] &  \Za_2(S_-,\xi_-)  \ar[d,Rightarrow] \ar[r,"{\Za_2(\Sigma_1,\varphi_{1})}"] & \Za_2(S,\xi)  \ar[r,"{\Za_2(\Sigma_2,\varphi_{2})}"] \ar[d, Rightarrow] & \Za_2(S_+,\xi_+)  \\ 
& \Za_1(S,\xi)  \ar[r,"{\Za_1(h|^{-1}_{S})}",swap] \ar[rd,"{\Za_1(\Sigma_2,\varphi_{2})}",swap] & \lambda^* \Za_1(S,\widehat{\varphi}_{1,2}|_{S}) \ar[d, Rightarrow] \ar[rd] \ar[r, "{Z(S,\widehat{\varphi}_{1,2}|_{S})}",swap]  \ & \lambda^* \Za_2(S,\widehat{\varphi}_{1,2}|_{S}) \ar[ru, "{\Za_2(h|_S)}", swap] \ar[r] \ar[d,Rightarrow, "{Z(\Sigma_2, \widehat{\varphi}_{1,2}|_{\Sigma_2})}"] & \lambda^* \Za_2(S_+,\widehat{\varphi}_{1,2}|_{S_+}) \ar[ru, "{\Za_2(h|_S)}", swap]  & \\
 \Za_1(S_+,\xi_+) \ar[rr, "{\id}", swap] \ar[ru, Leftarrow, "\Phi_{\Za_1}"] & &\Za_1(S_+,\xi_+) \ar[r,"{\Za_1(h|^{-1}_{S_+})}",swap] & \lambda^* \Za_1(S_+,\widehat{\varphi}_{1,2}|_{S_+}) \ar[ru, "{Z(S_+,\widehat{\varphi}_{1,2}|_{S_+})}",swap]  \ & 
\end{tikzcd}
\end{small}
\end{equation*} 
\begin{equation*}
\parallel \ \ \text{(Coherence of $\Za_1$ and $\Za_2$)}
\end{equation*}
\vspace{-1cm}
\begin{equation*}
\begin{small}
\begin{tikzcd}[column sep= 1.5cm, row sep=1.4cm]
 & & & & \Za_2(S,\xi) \ar[d, Rightarrow, "{\Phi_{\Za_2}}", shorten >= 1em] \ar[rd] & \\
\Za_1(S_-,\xi_-) \ar[r,"{\Za_1(h|^{-1}_{S_-})}"] \ar{ddr}[rotate=-45, pos=0.2] {\Za_1(\Sigma_2 \circ \Sigma_1, \varphi_2 \circ \varphi_1)} & \lambda^* \Za_1(S_-,\widehat{\varphi}_{1,2}|_{S_-}) \ar{ddr}[rotate=-40, pos=0.8, swap] {\lambda^*\Za_1(\Sigma_2 \circ \Sigma_1, \widehat{\varphi}_{1,2})}  \ar[rd] \ar[dd, Rightarrow] \ar[r, "{Z(S_-,\widehat{\varphi}_{1,2}|_{S_-})}"]  & \lambda^* \Za_2(S_-,\widehat{\varphi}_{1,2}|_{S_-})  \ar[r,"{\Za_2(h|_{S_-})}"] \ar[rd] \ar[rrd] \ar[d,Rightarrow, "{Z(\Sigma_1, \widehat{\varphi}_{1,2}|_{\Sigma_1})}"] &  \Za_2(S_-,\xi_-)  \ar[ru]  \ar[rr,"{\Za_2(\Sigma_2 \circ \Sigma_1, \varphi_2 \circ \varphi_{1})}"] & \ \ar[d, Rightarrow] & \Za_2(S_+,\xi_+)  \\ 
 & \ & \lambda^* \Za_1(S,\widehat{\varphi}_{1,2}|_{S}) \ar[d, Rightarrow, "{\Phi_{\Za_1}}"] \ar[rd] \ar[r, "{Z(S,\widehat{\varphi}_{1,2}|_{S})}",swap]  \ & \lambda^* \Za_2(S,\widehat{\varphi}_{1,2}|_{S}) \ar[rd] \ar[d,Rightarrow, "{Z(\Sigma_2, \widehat{\varphi}_{1,2}|_{\Sigma_2})}"] & \lambda^* \Za_2(S_+,\widehat{\varphi}_{1,2}|_{S_+}) \ar[ru, "{\Za_2(h|_S)}", swap] \ar[l, Rightarrow] &  \\
 & \Za_1(S_+,\xi_+) \ar[r,"{\Za_1(h|^{-1}_{S_+})}",swap] &\Za_1(S_+,\xi_+) \ar[r,"{\id}",swap] & \lambda^* \Za_1(S_+,\widehat{\varphi}_{1,2}|_{S_+}) \ar[r, "{Z(S_+,\widehat{\varphi}_{1,2}|_{S_+})}",swap]  \ &  \lambda^* \Za_2(S_+,\widehat{\varphi}_{1,2}|_{S_+}) \ar[u, "{\id}",swap]
\end{tikzcd}
\end{small}
\end{equation*} 
\begin{equation*}
\parallel \ \ \text{($Z$ is a relative field theory)}
\end{equation*}
\vspace{-1cm}
\begin{equation}
\begin{small}
\begin{tikzcd}[column sep= 1.5cm, row sep=1.8cm]
 & & & & \Za_2(S,\xi) \ar[d, Rightarrow, "{\Phi_{\Za_2}}", shorten >= 1em] \ar[rd] & \\
\Za_1(S_-,\xi_-) \ar[r,"{\Za_1(h|^{-1}_{S_-})}"] \ar{dr}[rotate=-30, pos=0.2] {\Za_1(\Sigma_2 \circ \Sigma_1, \varphi_2 \circ \varphi_1)} & \lambda^* \Za_1(S_-,\widehat{\varphi}_{1,2}|_{S_-}) \ar{dr}[rotate=-25, pos=0.8, swap] {\lambda^*\Za_1(\Sigma_2 \circ \Sigma_1, \widehat{\varphi}_{1,2})}  \ar[d, Rightarrow] \ar[r, "{Z(S_-,\widehat{\varphi}_{1,2}|_{S_-})}"]  & \lambda^* \Za_2(S_-,\widehat{\varphi}_{1,2}|_{S_-})  \ar[r,"{\Za_2(h|_{S_-})}"]  \ar[rd] \ar[d,Rightarrow, "{Z(\Sigma_1, \widehat{\varphi}_{1,2}|_{\Sigma_1})}"] &  \Za_2(S_-,\xi_-)  \ar[ru]  \ar[rr,"{\Za_2(\Sigma_2 \circ \Sigma_1, \varphi_2 \circ \varphi_{1})}"] \ar[d, Rightarrow] & \  & \Za_2(S_+,\xi_+)  \\ 
 & \Za_1(S_+,\xi_+) \ar[r,"{\Za_1(h|^{-1}_{S_+})}",swap] &\Za_1(S_+,\xi_+) \ar[r, "{Z(S_+,\widehat{\varphi}_{1,2}|_{S_+})}",swap]   \ &  \lambda^* \Za_2(S_+,\widehat{\varphi}_{1,2}|_{S_+})  \ar[rru]& \
\end{tikzcd}
\end{small}
\end{equation} 
\end{landscape}
This shows that the natural transformation agree and hence finishes this part of the proof.  

\textbf{4. Compatibility with identity morphisms:}
Finally, we have to show that 
\begin{equation}
\begin{tikzcd}
 & \lambda_* Z(S, \xi) \ar[rd, "{\id_{\lambda_* Z(S, \xi)} \bullet \Phi_{\Za_1}(S,\xi)}"] \ar[ld, "{\Phi_{\Za_2}(S,\xi)}  \bullet \id_{\lambda_* Z(S, \xi)} ", swap] & \\ 
 \Za_2(\id_{S,\xi}) \circ \lambda_* Z(S,\xi) \ar[rr, "\lambda_*Z(\id_{S,\xi})",swap] & & \lambda_* Z(S,\xi) \circ \Za_1(\id_{S,\xi}) 
\end{tikzcd}
\end{equation}
commutes. Using continuity the morphism $ \lambda_*Z(\id_{S,\xi}) \circ  \left( \Phi_{\Za_2}(S,\xi)  \bullet \id_{\lambda_* Z(S, \xi)} \right)$ is induced by the natural transformation
\begin{equation}
\begin{tikzcd}
 & {\lambda^{-1}[\xi] } \ar[d, "{[0,1] \times \cdot }"] \ar[rdd, bend left] \ar[ldd, bend right] & \\
 & {\lambda^{-1}[[0,1]\times \xi]} \ar[rd] \ar[ld]& \\
{\lambda^{-1}[\xi] } \ar[rd, swap, "{\widehat{Z}(S,\xi)}"] \ar[rr, Rightarrow, shorten <= 2em, shorten >= 2em, "{\widehat{Z}(\id_{S,\xi})\circ \Phi_{\Za_2}(S,\xi) }"] & &  {\lambda^{-1}[\xi] } \ar[ld, "{\widehat{Z}(S,\xi)\circ \Za_1(\id_{S,\xi})}"] \\ 
 & {[\Za_1(S,\xi),\Za_2(S,\xi)]} & 
\end{tikzcd}
\end{equation}
where $[0,1] \times \cdot \colon {\lambda^{-1}[\xi] } \longrightarrow {\lambda^{-1}[[0,1]\times \xi]} $ is an equivalence and
hence we can also use the upper span to compute the morphism. Since all maps in the span are equivalences the induced map via push and pull agrees with the map induced by the natural transformation. 
Hence we have to evaluate the natural transformation  
\begin{equation}
\begin{footnotesize}
\begin{tikzcd} 
 & & & & \ar[d, Rightarrow, shorten <= 0.3em, shorten >= 0.9em, "{\Phi_{\Za_2}}"] & \\
\Za_1(S,\xi) \ar[d, swap, "{\Za_1(\id_{S,\xi})}"] \ar[r,"\Za_1(h^{-1})"] & \lambda^* \Za_1(S,\widehat{\xi}) \ar[ld, Rightarrow]  \ar[r, "{Z(S,\widehat{\xi})}"] \ar[d] & \lambda^* \Za_2(S,\widehat{\xi})
\ar[r, "\Za_2(h)"] \ar[d, Rightarrow, shorten <= 0em, shorten >= 0.3em, swap, "{Z(\id_{S,\widehat{\xi}})}"] \ar[rd] & \Za_2(S,\xi) \ar[rr, "{\Za_2(\id_{S,\xi})}"] \ar[d, Rightarrow] \ar[rr, "{\id}", bend left =50] & \ & \Za_2(S,\xi)  
\\
\Za_1(S,\xi) \ar[r,swap, "{\Za_1(h^{-1})}"] & \lambda^*\Za_1(S, \widehat{\xi}) \ar[rr, "{Z(S,\widehat{\xi})}", swap] & \  & \lambda^* \Za_2(S, \widehat{\xi}) \ar[rru, "\Za_2(h)", swap] &  &
\end{tikzcd}
\end{footnotesize}
\end{equation}
Using the coherence of $\Za_2$ we can rewrite this as 
\begin{equation}
\begin{footnotesize}
\begin{tikzcd}[row sep = 1.2cm, column sep=1cm] 
\Za_1(S,\xi) \ar[d, swap, "{\Za_1(\id_{S,\xi})}"] \ar[r,"\Za_1(h^{-1})"] & \lambda^* \Za_1(S,\widehat{\xi}) \ar[ld, Rightarrow]  \ar[r, "{Z(S,\widehat{\xi})}"] \ar[d] & \lambda^* \Za_2(S,\widehat{\xi})
\ar[r, "\Za_2(h)"] \ar[dl, Rightarrow, shorten <= 0em, shorten >= 0.3em, swap, "{Z(\id_{S,\widehat{\xi}})}"] \ar[d, bend left,  ""{name=U}, "{\id}", pos=0.4] \ar[d, bend right, "{ \lambda^* \Za_2(\id)}",swap,  ""{name=U'},pos=0.7] & \Za_2(S,\xi)
\\
\Za_1(S,\xi) \ar[r,swap, "{\Za_1(h^{-1})}"] & \lambda^*\Za_1(S, \widehat{\xi}) \ar[r, "{Z(S,\widehat{\xi})}", swap] & \   \lambda^* \Za_2(S, \widehat{\xi}) \ar[ru, "\Za_2(h)", swap] &  
\arrow[Rightarrow, from=U, to=U', shorten <= 0.3em, shorten >= 0.3em, "{\Phi_{\Za_2}}", swap, pos=0.01]
\end{tikzcd}
\end{footnotesize}
\end{equation}
Using that $Z$ satisfies the condition we want to check for $\lambda_* Z$ we get 
\begin{equation}
\begin{footnotesize}
\begin{tikzcd}[row sep = 1.2cm, column sep=1cm] 
\Za_1(S,\xi) \ar[d, swap, "{\Za_1(\id_{S,\xi})}"] \ar[r,"\Za_1(h^{-1})"] & \lambda^* \Za_1(S,\widehat{\xi}) \ar[ld, Rightarrow]  \ar[r, "{Z(S,\widehat{\xi})}"] \ar[d, bend left,  ""{name=U}, "{\id}", pos=0.4] \ar[d, bend right, "{ \lambda^* \Za_1(\id)}",swap,  ""{name=U'},pos=0.7]  & \lambda^* \Za_2(S,\widehat{\xi})
\ar[r, "\Za_2(h^{-1})"] & \Za_2(S,\xi)
\\
\Za_1(S,\xi) \ar[r,swap, "{\Za_1(h)}"] & \lambda^*\Za_1(S, \widehat{\xi}) \ar[ru, "{Z(S,\widehat{\xi})}", swap] &  &  
\arrow[Rightarrow, from=U, to=U', shorten <= 0.3em, shorten >= 0.3em, "{\Phi_{\Za_1}}", swap, pos=0.01]
\end{tikzcd}
\end{footnotesize}
\end{equation}
Finally, we use the coherence of $\Za_1$ to simplify the expression to 
\begin{equation}
\begin{footnotesize}
\begin{tikzcd}[row sep = 1.2cm, column sep=1cm] 
\Za_1(S,\xi) \ar[r,bend left, "{\id}", ""{name=U}]  \ar[r,bend right, "{\Za_1(\id)}", ""{name=U'}] & \Za_1(S,\xi) \ar[r,"\Za_1(h^{-1})"] & \lambda^* \Za_1(S,\widehat{\xi})   \ar[r, "{Z(S,\widehat{\xi})}"]  & \lambda^* \Za_2(S,\widehat{\xi})
\ar[r, "\Za_2(h^{-1})"] & \Za_2(S,\xi)
\arrow[Rightarrow, from=U, to=U', shorten <= 0.3em, shorten >= 0.3em, "{\Phi_{\Za_1}}", swap,]
\end{tikzcd}
\end{footnotesize}
\end{equation}
Now the statement follows from the continuity of $\Za_1(\id_{S,\xi})$.  
\end{proof}

\begin{remark}
If both extended field theories $\Za_1$ and $\Za_2$ are trivial our construction 
reduces to the pushforward construction of~\cite{OFK} by applying Corollary~\ref{Coro: Relative for 1 induces ordinary}. 
\end{remark}

Using the pushforward construction of relative field theories we can now define the
anomalous quantum field theory $Z_{\omega'}\colon \mbf1 \Longrightarrow \tr E_\theta$: 
recall from Proposition~\ref{Prop: Coboundary induces transformations} that $\omega'$ induces a natural isomorphism 
$ \tr E_0 \Longrightarrow \tr E_{\delta \omega'} = \tr \lambda^* E_\theta$ composition
with the natural transformation $\Omega_0$ from Equation~\eqref{Eq: Omega 0} gives
the relative field theory $\widehat{Z}_{\omega'} \colon \mbf1 \Longrightarrow \tr E_{\delta \omega'} = \tr \lambda^* E_\theta$. 
We set $Z_{\omega'}\coloneqq \lambda_* \widehat{Z}_{\omega'}$.

Similarly to the proof of Theorem~\ref{Theorem: Gauging}, one can show that the relative field theory $Z_{\omega'}$ gauges the $G$-symmetry. 
Let us explain in more detail what this means: the pullback
$i^*E_\theta$ along the inclusion $i\colon \ECob \longrightarrow \EGCob$ is naturally isomorphic to
the trivial theory $\mbf1 \colon \ECob\longrightarrow \Tvs$. The pullback $i^* Z_{\omega'}\colon \mbf1
\Longrightarrow \tr i^*E_\theta\cong \mbf1$ is a field theory
relative to the trivial theory. From Corollary~\ref{Coro: Relative for 1 induces ordinary} it follows that $i^*Z_{\omega'}$ 
is an $n-1$ dimensional topological quantum field theory. This field theory comes with an internal
$G$-symmetry from the evaluation of $Z_{\omega'}$ on gauge transformations of the trivial bundle. 
Gauging the symmetry means that the field theory $i^*Z_{\omega'}$ recovers the Dijkgraaf-Witten theory
$Z_\omega$ together with its internal symmetry induced by the group extensions \begin{align}
1\longrightarrow D \overset{\iota}{\longrightarrow} \widehat{G}\overset{\lambda}{\longrightarrow} G \longrightarrow 1 \ .
\end{align}  
Next we will spell out explicitly the pushforward construction for $Z_\omega$ and recover
the formulas from \cite{tHooft}:
Let $(S,\xi)$ be an object in $\EGCob$. The pushforward construction of 
relative field theories evaluated at $(S,\xi)$ is
\begin{align}\label{Eq: Definition state space}
\begin{split}
Z_{\omega'}(S,\xi) \colon \fvs &\longrightarrow E_\theta \\ 
\C & \longmapsto \int_{(\widehat{\xi},h)\in \lambda^{-1}[\xi]}  \int^{\sigma_+\in \Fund(S)}
\int_{\sigma\in \Fund(S)}   ([0,1]\times S)^h(\sigma_+,\sigma)*\sigma_+
\end{split} \ \ .
\end{align} 
The limit over $\sigma$ can be computed by evaluation at $\sigma_+$. The limit over
$\lambda^{-1}[\xi]$ can be computed in vector spaces. We can identify $([0,1]\times S)^h(\sigma_+,\sigma_+)$
with $\C$ using the canonical basis element $[0,1]\times \sigma_+$ of $([0,1]\times S)^h(\sigma_+,\sigma_+)$. Under this
identification the natural transformation (see Equation~\eqref{Eq: Concrete description on morphisms}) corresponds to multiplication
with $\langle\,\widehat{h}^* \omega', [0,1] \times \sigma_+ \rangle \cdot \tau_S \theta (\lambda_* \widehat{h}, h')^{-1}$, 
where we used Theorem \ref{Theorem: Transgression} to identify the coherence isomorphism of $E_\theta$ with the 
transgression of $\theta$. 

The whole construction assembles into a functor
\begin{align}
L_{\xi, \omega'} \colon \Fund_\theta(S,\xi)^{\text{op}} \longrightarrow \big[\lambda_*^{-1}[\xi], \fvs\big]\ \ . 
\end{align} 
Let $\sigma_+$ and $\sigma_+'$ be representatives for the fundamental class of $S$ and $\Lambda \in C_{n-1}(S) $ an $n{-}1$-chain satisfying $\partial \Lambda = \sigma_+'-\sigma_+$, i.e.\ a morphism in $\Fund_\theta(S,\xi)$. 
The natural transformation corresponding to $\Lambda$ under the trivialisation picked above is 
\begin{align}
\label{Def: L natural}
\begin{split}
L_{\xi, \omega'}(\Lambda) \colon L_{\xi, \omega'}(\sigma_+') &\Longrightarrow L_{\xi, \omega'}(\sigma_+) \\
{L_{\xi, \omega'}}(\Lambda)_{ (\widehat\xi,h)} \colon \C
                                                             &\longrightarrow
                                                               \C \ ,
                                                               \quad 1
                                                               \longmapsto
                                                               \big\langle
                                                               {\widehat{\xi}}\,^*
                                                               \omega',\Lambda
                                                               \big\rangle^{-1}
                                                               \,
                                                               \big\langle
                                                               h^*
                                                               \theta
                                                               , [0,1]
                                                               \times
                                                               \Lambda
                                                               \big\rangle^{-1}
                                                               \ .
\end{split}
\end{align}
We denote by $\tilde{Z}^{(S,\xi)}\colon \Fund(S)^{\opp}\longrightarrow \fvs$ the limit
of $L_{\xi, \omega'}$ over $\lambda^{-1}[\xi]$. Using this notation Equation~\eqref{Eq: Definition state space} reduces to
\begin{align}
\int^{\sigma\in \Fund(S)} \tilde{Z}^{(S,\xi)}(\sigma) *\sigma \in E_\theta(S,\xi) \ \ .
\end{align} 
This reproduces the definition from \cite{tHooft} on objects\footnote{To be more precise in \cite{tHooft} the term $\tau_S \theta (\lambda_* \widehat{h}, h')^{-1}$ is not present. However, the term is trivial on automorphisms $\widehat{h}$ in $\lambda^{-1}[\xi]$ because they need to satisfy $\lambda_* \widehat{h}=1$ and hence its inclusion does not change the limit.
Note that this kind of argument works also in general}.

To evaluate the theory $Z_{\omega'}$ on a 1-morphisms $(\Sigma, \varphi)\colon 
(S_-,\xi_-)\longrightarrow (S_+,\xi_+) $ in $\EGCob$, we describe the functors in
Equation~\eqref{Eq: natural transformation for push} evaluated at $\C \in \fvs$. The source 
of the natural transformation $ \widehat{Z}_{\omega'}$ is 
\begin{align}
\begin{split}
\lambda^{-1}[\varphi] &\longrightarrow E_\theta(S_+,\xi_+) \\
(\widehat{\varphi},h) &\longmapsto  \int\displaylimits^{\sigma_+\in \Fund(S_+)} \int\displaylimits^{\sigma_-\in \Fund(S_-)}
\hspace{-0.8cm}\int\displaylimits_{\sigma\in \Fund(S_-)}   ([0,1]\times S_-)^{h|_{S_-}}(\sigma_-,\sigma) \otimes \Sigma^\varphi(\sigma_+,\sigma_-)* \sigma_+ \\ 
& \phantom{\longmapsto} \ \ \ \  \cong  \int\displaylimits^{\sigma_+\in \Fund(S_+)} \int\displaylimits^{\sigma_-\in \Fund(S_-)}  \Sigma^\varphi(\sigma_+,\sigma_-) \otimes L_{\xi_-,\omega'}(\sigma_-, \widehat{\varphi}|_{S_-},h|_{S_-}) * \sigma_+
\end{split}  .
\end{align}
The target is 
\begin{align}
\begin{split}
\lambda^{-1}[\varphi] &\longrightarrow E_\theta(S_+,\xi_+) \\
(\widehat{\varphi},h) &\longmapsto  \int\displaylimits^{\sigma_+\in \Fund(S_+)}
  L_{\xi_+,\omega'}(\sigma_+,\widehat{\varphi}|{S_+}, h|_{S_+})*\sigma_+
 \end{split} \  \  .
\end{align}
The natural transformation is induced by the linear maps 
\begin{align}
\begin{split}
\Sigma^\varphi(\sigma_+,\sigma_-) \otimes L_{\xi_-,\omega'}(\sigma_-, \widehat{\varphi}|_{S_-},h|_{S_-}) & \longrightarrow  L_{\xi_+,\omega'}(\sigma_+,\widehat{\varphi}|{S_+}, h|_{S_+}) \\ 
\Lambda \otimes 1 & \longmapsto \langle \widehat{\varphi}^*\omega', \Lambda \rangle \cdot \langle h^*\theta, [0,1]\times \Lambda \rangle
\end{split}
\end{align}
For a concrete realisation of the map induced by this natural transformation in terms of 
parallel sections we refer to~\cite{tHooft}. There we also show directly by long and explicit 
computations that the concrete formulas define a relative field theory. Here it is ensured 
abstractly by Theorem~\ref{Thm: push relative}.

\begin{remark}
Let $S$ be a closed oriented $n{-}2$-dimensional manifold and $\sigma_S$ a representative of its fundamental class.
The general theory outlined in Section~\ref{Sec: P anomaly actions}
implies that the vector spaces
$\tilde{Z}_\omega^{(S,\,\cdot\,)}(\sigma_S)$ form a projective
representation of $\BunG(S)$. The 2-cocycle $\alpha$ twisting the
projective representation is completely described by the coherence
isomorphisms for $E_\theta$. Theorem~\ref{Theorem: Transgression} shows that the 2-cocycle twisting this representation is given by the transgression of
$\theta\in Z^n(BG;U(1))$ to the groupoid $\BunG(S)$. This generalizes the low-dimensional descriptions of anomalies
and projective representations on state spaces discussed
in~\cite[Section~2.1]{Tachikawa2016}: In the simplest
$n=1$ case, with $S=\{\ast\}$ the 2-cocycles $\alpha$ and $\theta$ may be identified,
and describe the same 2-cocycle specifying both the two-dimensional bulk $G$-symmetry protected
phase and the class of the
projective $G$-representation on the one-dimensional boundary state,
whereas for $n=3$ with $S=\mathbb{S}^1$ transgression induces a homomorphism
$H^3(BG;U(1))\longrightarrow H^2(B(G\DS G);U(1))$ specifying the
two-dimensional $G$-symmetry protected phase on the
boundary of the three-dimensional $G$-symmetry protected phase.

In a more geometric language this means that the state spaces of the
gauged theory form a section of the transgression 2-line bundle of the
flat $n{-}1$-gerbe on the classifying space $BG$ described by $\theta$, as the classical
gauge theory corresponding to $\theta$ describes the parallel
transport for the $n{-}1$-gerbe. This 2-line bundle is trivial if and only if the corresponding 2-cocycle is a boundary. Hence the obstruction for the projective representation to form an honest representation is the non-triviality of the transgression 2-line bundle. 
\end{remark}

\begin{remark}
Let $(\Sigma,\varphi)\colon (S,\xi)\longrightarrow \varnothing$ be a 1-morphism in $\EGCob$. According to \eqref{Eq: Def combined state space} the state space of the composite system is given by
\begin{align}
Z_{\omega'\,{\rm bb}}(\Sigma,\varphi, S)= E_\theta(\Sigma,\varphi
  )[Z_{\omega'}(S,\varphi|_S)] \cong \Sigma^\varphi (\varnothing,
  \sigma_S)\otimes_\C \tilde{Z}^{(S,\varphi|_S)}_{\omega'}(\sigma_S) \ .
\end{align}  
It is independent of the choice of $\sigma_S$ up to unique isomorphism corresponding to the choice of a representative of the coend. The composite state space carries an honest representation of the gauge group $G$ described in~\eqref{Eq: Def action combined state space}. 
\end{remark}

\appendix 
\chapter{Symmetric monoidal bicategories}\label{Chap: Bicat}
In this appendix we provide detailed definitions related to symmetric monoidal bicategories, following~\cite{leinster:1998,schommer2011classification} for the most part.

\section{Basic definitions}
We list the basic definitions for bicategories following \cite{leinster:1998}.

\begin{definition}
\label{DefinitionBicategory}
A \emph{bicategory} $\mathscr{B}$ consists of the following data:
\begin{itemize}
\item[(a)] A class $\Obj(\mathscr{B})$ of objects.

\item[(b)] A category $\mathsf{Hom}_\Bscr(A,B)$ for all $A,B \in \Obj(\mathscr{B})$, whose objects $f:A \longrightarrow B$ we call 1-morphisms and whose morphisms $f\Longrightarrow g$ we call 2-morphisms.

\item[(c)] Composition functors \[\circ_{ABC} \colon \mathsf{Hom}_\Bscr (B,C)\times \mathsf{Hom}_\Bscr(A,B)\longrightarrow \mathsf{Hom}_\Bscr(A,C)\] for all $A,B,C\in \Obj(\mathscr{B})$.

\item[(d)] Identity functors \[\mathsf{Id}_A \colon \mathsf{1} = \star \, \big/\!\!\big/ \, \{\id_\star\}\longrightarrow \mathsf{Hom}_\Bscr(A,A)\] for all $A\in \Obj(\mathscr{B})$.

\item[(e)] Natural associator isomorphisms \[\mathsf{a}_{A,B,C,D} \colon \circ_{ACD} \circ \big(\id_{\mathsf{Hom}_\Bscr(C,D)} \times \circ_{ABC} \big) \Longrightarrow \circ_{ABD} \circ \big(\circ_{BCD} \times \id_{\mathsf{Hom}_\Bscr(A,B)}\big)\] for all $A,B,C,D\in \Obj(\mathscr{B})$, expressing associativity of the composition. 

\item[(f)] Natural right and left unitor isomorphisms 
\[ \mathsf{r}_A \colon \circ_{AAB} \circ \big(\id_{\mathsf{Hom}_\Bscr(A,B)}\times \mathsf{Id}_A\big)\Longrightarrow \id_{\mathsf{Hom}_\Bscr(A,B)}\] and 
\[\mathsf{l}_A \colon \circ_{AAB} \circ \big(\mathsf{Id}_B \times \id_{\mathsf{Hom}_\Bscr(A,B)}\big)\Longrightarrow \id_{\mathsf{Hom}_\Bscr(A,B)} \]
for all $A,B\in \Obj(\mathscr{B})$.
\end{itemize}  
These data are required to satisfy the following coherence axioms:
\begin{itemize}
\item[(C1)] The pentagon diagram
\small
\begin{equation}
\begin{tikzcd}[row sep=scriptsize, column sep=tiny]
& \big((k\!\circ\! h)\!\circ\! g\big)\!\circ\! f\arrow[dl, swap, Rightarrow, "\mathsf{a}"] \arrow[rr,Rightarrow,"\mathsf{a}\bullet \id"]& &\big(k\!\circ\!(h\!\circ\! g)\big)\!\circ\! f \arrow[rd,Rightarrow,"\mathsf{a}"] & \\
(k\!\circ\! h)\!\circ\!(g\!\circ\! f)\arrow[rrd,swap,Rightarrow,"\mathsf{a}"]& & & &\arrow[lld,Rightarrow,"\id\bullet \mathsf{a}"]k\!\circ\!\big((h\!\circ\! g)\!\circ\! f\big) \\
& &k\!\circ\!\big(h\!\circ\!(g\!\circ\! f)\big) & &
\end{tikzcd} 
\end{equation}
\normalsize
commutes for all composable 1-morphisms $k$, $h$, $g$ and $f$, where $\bullet$ denotes the horizontal composition of natural transformations.

\item[(C2)] The triangle diagram 
\begin{equation}
\begin{tikzcd}
(g\circ \mathsf{Id})\circ f\arrow[rr,Rightarrow,"\mathsf{a}"] \arrow[rd,swap,Rightarrow,"\mathsf{r}\bullet\id"]& &g\circ (\mathsf{Id}\circ f)\arrow[dl,Rightarrow,"\id\bullet\mathsf{l}"]\\
& g\circ f &
\end{tikzcd} 
\end{equation}
commutes for all composable 1-morphisms $f$ and $g$.
\end{itemize}
\end{definition}  

There are different definitions for functors between bicategories corresponding to different levels of strictness. We use the following definition.
\begin{definition}
\label{Definition Morphism Bicategory}
A \emph{2-functor} $\mathcal{F} \colon \mathscr{B}\longrightarrow  \mathscr{B}'$ between two bicategories $\Bscr$ and $\Bscr'$ consists of the following data:
\begin{itemize}
\item[(a)] A map $\mathcal{F} \colon \Obj(\mathscr{B})\longrightarrow  \Obj(\mathscr{B}'\,)$.

\item[(b)] A functor $\mathcal{F}_{AB} \colon \mathsf{Hom}_\Bscr(A,B)\longrightarrow  \mathsf{Hom}_{\Bscr'}\big(\Fa(A),\Fa(B)\big)$ for all $A,B\in \Obj(\mathscr{B})$.

\item[(c)] A natural isomorphism $\Phi_{ABC}$ given by
\begin{equation}
 \begin{tikzcd} 
\mathsf{Hom}_\Bscr(B,C)\times \mathsf{Hom}_\Bscr(A,B) \arrow{r}{\circ }\arrow[d,swap,"{\mathcal{F}_{BC}\times \mathcal{F}_{AB}}"] & \mathsf{Hom}_\Bscr(A,C) \arrow{d}{\mathcal{F}_{AC}}\\
\mathsf{Hom}_{\Bscr'}\big(\mathcal{F}(B),\mathcal{F}(C)\big)\times \mathsf{Hom}_{\Bscr'}\big(\mathcal{F}(A),\mathcal{F}(B)\big)\ar[ru, Rightarrow, shorten <= 2ex, shorten >= 2ex, "\Phi_{ABC}"] \arrow[r,swap,"\circ'"] & \mathsf{Hom}_{\Bscr'}\big(\mathcal{F}(A),\mathcal{F}(C)\big)
\end{tikzcd} 
\end{equation}
for all $A,B,C\in \Obj(\mathscr{B})$.

\item[(d)] A natural isomorphism $\Phi_A$ given by
\begin{equation}
 \begin{tikzcd} \mathsf{1} \arrow[d,swap,"\id"] \arrow{rr}{\mathsf{Id}_A} & & \mathsf{Hom}_{\Bscr}(A,A)\arrow{d}{\mathcal{F}_{AA}} \\ \mathsf{1} \arrow[rr,swap, "\mathsf{Id}'_{\mathcal{F}(A)}"] \ar[rru, shorten <= 2ex, shorten >= 2ex, Rightarrow, "\Phi_A", pos=0.7] & & \mathsf{Hom}_{\Bscr'}\big(\mathcal{F}(A),\mathcal{F}(A)\big)
\end{tikzcd}
\end{equation} for all $A\in \Obj(\mathscr{B})$.
\end{itemize} 
These data are required to satisfy the following coherence axioms:
\begin{itemize}
\item[(C1)] The diagram
\begin{equation}
\label{EQ1: Definition 2Functor}
\begin{tikzcd} 
\big(\mathcal{F}(h) \circ' \mathcal{F}(g)\big)\circ' \mathcal{F}(f) \arrow[d,swap,Rightarrow,"\mathsf{a}'"] \arrow[r,Rightarrow,"\Phi \bullet' \id"] & \mathcal{F}(h\circ g) \circ' \mathcal{F}(f) \arrow[r,Rightarrow,"\Phi"] & \mathcal{F}\big((h\circ g)\circ f\big)\arrow[d,Rightarrow,"\mathcal{F}(\mathsf{a})"] \\ \mathcal{F}(h)\circ' \big(\Fa(g)\circ' \Fa(f)\big) \arrow[r,swap,Rightarrow,"\id \bullet' \Phi"] & \Fa(h) \circ' \Fa(g\circ f) \arrow[r,,swap, Rightarrow,"\Phi"] & \Fa\big(h\circ(g\circ f)\big)
\end{tikzcd}
\end{equation}
commutes for all composable 1-morphisms.

\item[(C2)] The diagram
\begin{equation}
\label{EQ2: Definition 2Functor}
\begin{tikzcd}
\mathcal{F}(f) \circ' \mathsf{Id}'_{\mathcal{F}(A)} \arrow[r,Rightarrow,"\id\bullet' \Phi"] \arrow[dr,swap,Rightarrow,"\mathsf{r}'"] & \mathcal{F}(f) \circ' \mathcal{F}(\mathsf{Id}_A)\arrow[r,Rightarrow,"\Phi"] & \mathcal{F}(f\circ \mathsf{Id}_A) \arrow[dl,Rightarrow,"\mathcal{F} (\mathsf{r})"] \\ & \mathcal{F}(f) &
\end{tikzcd} 
\end{equation} 
commutes for all composable 1-morphisms.

\item[(C3)] A diagram analogous to \eqref{EQ2: Definition 2Functor} for the left unitors $\mathsf{l}$ and $\mathsf{l}'$ commutes.
\end{itemize}
\end{definition} 

Again there are different ways to define natural transformations between 2-functors. The following definition is suitable for our purposes.
\begin{definition}
\label{Definition transformation Bicategory}
Given two 2-functors $\mathcal{F},\mathcal{G}  \colon \mathscr{B}\longrightarrow  \mathscr{B'} $, a \emph{natural 2-transformation} $\sigma \colon \mathcal{F}\Longrightarrow  \mathcal{G}$ consists of the following data: 
\begin{itemize}
\item[(a)] A 1-morphism $\sigma_A \colon \mathcal{F}(A)\longrightarrow  \mathcal{G}(A)$ for all $A \in \mathrm{Obj}(\mathscr{B})$.

\item[(b)] A natural transformation $\sigma_{AB}$ given by\footnote{Here we use $\ast$ to denote pullbacks and pushforwards in the usual way.}
\begin{equation}
 \begin{tikzcd} 
\Hom_{\mathscr{B}}(A,B) \arrow{r}{\mathcal{F}_{AB} }\arrow[d,swap,"\mathcal{G}_{AB}"] & \Hom_{\mathscr{B}'}\big(\mathcal{F}(A),\mathcal{F}(B)\big)\arrow{d}{\sigma_{B\ast}}\\
\Hom_{\mathscr{B}'}\big(\mathcal{G}(A),\mathcal{G}(B)\big)\arrow[Rightarrow]{ru}{\sigma_{AB}} \arrow[r,swap,"\sigma_A^\ast"] & \Hom_{\mathscr{B}'}\big(\mathcal{F}(A),\mathcal{G}(B)\big)
\end{tikzcd}
\end{equation}
for all $A,B\in \mathrm{Obj}(\mathscr{B})$. In particular, these natural transformations comprise families of 2-morphisms $\sigma_f \colon \mathcal{G}_{AB}(f) \circ' \sigma_A \Longrightarrow  \sigma_{B}\circ' \mathcal{F}_{AB}(f)$ for all 1-morphisms $f:A\longrightarrow  B$ in $\Bscr$.
\end{itemize}
These data are required to satisfy the following coherence axioms:
\begin{itemize}
\item[(C1)] The diagram 
\begin{equation}
\label{Equation1: Definition Transformation}
\begin{tikzcd}
\big(\mathcal{G}(g) \circ' \mathcal{G}(f)\big)\circ' \sigma_A \arrow[r,Rightarrow,"\mathsf{a}'"] \arrow[d,swap,Rightarrow,"\Phi_\Ga\bullet'\id"] & \mathcal{G}(g) \circ' \big(\mathcal{G}(f)\circ' \sigma_A\big)\arrow[r,Rightarrow, "\id\bullet' \sigma_f"] & \mathcal{G}(g)\circ' \big(\sigma_{B}\circ' \mathcal{F}(f)\big)\arrow[d,Rightarrow,"\mathsf{a}'"] \\
\mathcal{G}(g\circ f)\circ' \sigma_A \arrow[d,swap,Rightarrow,"\sigma_{g\circ f}"] & & \big(\mathcal{G}(g)\circ' \sigma_{B}\big)\circ' \mathcal{F}(f) \arrow[d,Rightarrow,"\sigma_g\bullet'\id"] \\
\sigma_{C}\circ' \mathcal{F}(g\circ f) \arrow[r,swap,Leftarrow,"\id\bullet'\Phi_\Fa"] & \sigma_{C}\circ' \big(\mathcal{F}(g)\circ' \mathcal{F}(f)\big) \arrow[r,swap,Leftarrow,"\mathsf{a}'"] & \big(\sigma_{C}\circ' \mathcal{F}(g)\big)\circ' \mathcal{F}(f)
\end{tikzcd}
\end{equation}
commutes for all 1-morphisms $f \colon A\longrightarrow  B$ and $g \colon  B\longrightarrow  C$ in $\mathscr{B}$.

\item[(C2)] The diagram
\begin{equation}
\label{Equation2: Definition Transformation}
\begin{tikzcd}
\mathsf{Id}'_{\mathcal{G}(A)}\circ' \sigma_A\arrow[r,Rightarrow,"\mathsf{l}'"] \arrow[d,swap,Rightarrow,"\Phi_\Ga\bullet'\id"]&\sigma_A\arrow[r,Rightarrow,
"\mathsf{r}'{}^{-1}"] & \sigma_A \circ' \mathsf{Id}'_{\Fa(A)}\arrow[d,Rightarrow,"\id\bullet'\Phi_\Fa"] \\
\mathcal{G}(\mathsf{Id}_A)\circ' \sigma_A \arrow[rr,swap,Rightarrow,"\sigma_{\mathsf{Id}_A}"] & & \sigma_A \circ' \mathcal{F}(\mathsf{Id}_A) 
\end{tikzcd}
\end{equation}
commutes for all $A\in \mathrm{Obj}(\mathscr{B})$.
\end{itemize}
\end{definition}

\begin{remark}
We do not require the natural transformation $\sigma_{AB}$ to be invertible. There is a different
definition of natural transformation where $\sigma_{AB}$ goes in the other direction. In the 
literature these two versions are called lax and op-lax~\cite{Johnson-Freyd:2017ykw}. In this thesis we never use
op-lax transformations and hence refrain from introducing the term lax.    
\end{remark}  
In the context of bicategories there exist a natural way to compare natural transformations. 
\begin{definition}
\label{Definition Modification Bicategory}
Given two natural 2-transformations $\sigma,\tau \colon \mathcal{F}\Longrightarrow  \mathcal{G}$, a \emph{modification} $\mit\Gamma \colon \sigma \Rrightarrow \tau$ consists of a 2-morphism ${\mit\Gamma}\!_A \colon \sigma_A \Longrightarrow  \tau_A$ for each $A\in \mathrm{Obj}(\mathscr{B})$ such that the diagram
\begin{equation}
\label{EQ: Modification}
\begin{tikzcd}
\mathcal{G}(f)\circ' \sigma_A \arrow[rr,Rightarrow,"\id\bullet' {\mit\Gamma}\!_A"] \arrow[d,swap,Rightarrow,"\sigma_f"] & & \mathcal{G}(f)\circ' \tau_A \arrow[d,Rightarrow,"\tau_f"] \\
\sigma_B\circ' \mathcal{F}(f)\arrow[rr,swap,Rightarrow,"{\mit\Gamma}\!_{B}\bullet'\id"] & & \tau_B\circ' \mathcal{F}(f)
\end{tikzcd}
\end{equation} 
commutes for all 1-morphisms $f \colon A\longrightarrow  B$ in $\Bscr$.
\end{definition}
The collection of all bicategories, 2-functors, natural transformations and modifications 
forms a tricategory $\BiCat$~\cite{GPS}.
\section{Symmetric monoidal bicategories}

In this Section we define symmetric monoidal structures on bicategories 
following~\cite{schommer2011classification}.

\begin{definition}
A \emph{symmetric monoidal bicategory} consists of a bicategory $\mathscr{B}$ together with the following data:
\begin{itemize}
\item[(a)] A monoidal unit $1\in \mathrm{Obj}(\mathscr{B})$.

\item[(b)] A 2-functor $\otimes \colon \mathscr{B} \times \mathscr{B} \longrightarrow  \mathscr{B}$.

\item[(c)] Equivalence natural 2-transformations\footnote{Here `equivalence' means the natural 2-transformations in question have weak inverses.} $\alpha \colon \otimes \circ (\id \times \otimes) \Longrightarrow  \otimes \circ (\otimes \times \id),$ $\lambda \colon 1\otimes \, \cdot \, \Longrightarrow  \id$ and $\rho \colon \id \Longrightarrow  \, \cdot \, \otimes 1$.
We pick adjoint inverses which are part of the data and denoted them
by $^\star$, leaving the adjunction data implicit.

\item[(d)] An equivalence natural 2-transformation $\beta \colon a\otimes b \Longrightarrow  b\otimes a$.

\item[(e)] The four invertible modifications
\begin{equation}
\begin{tikzcd}[column sep=small]
 &  \otimes \circ (\otimes \times \otimes)\arrow[dr,Rightarrow,"\alpha"] & \\
\otimes \circ (\otimes \times \id) \circ (\otimes \times \id \times \id) \arrow[ur,Rightarrow,"\alpha"] \arrow[d,swap,Rightarrow,"\alpha\otimes\id"] &  & \otimes \circ (\id \times \otimes) \circ (\id \times \id \times \otimes)  \\
 \otimes \circ (\otimes \times \id) \circ (\id \times \otimes \times \id) \arrow[rr,swap,Rightarrow,"\alpha"] & 
\tarrow[swap,shorten <=15pt,shorten >=5pt," \ \mit\Xi "]{uu}  
  & \otimes \circ (\id \times \otimes) \circ (\id \times \otimes \times \id) \ar[u,swap,Rightarrow,"\id\otimes\alpha"]
\end{tikzcd}
\end{equation}
\begin{equation}
\begin{tikzcd}
\otimes \circ \big(\id \times (1 \otimes \, \cdot \, ) \big) \ar[rr,Rightarrow,"\alpha"] & \tarrow[shorten <=2pt,shorten >=5pt, "\ \mit\Theta"]{d} &\otimes \circ \big(( \, \cdot \, \otimes 1 ) \times \id\big) \ar[d,Leftarrow,"\rho\otimes\id"] 
\\
\otimes \ar[u,Leftarrow,"\id\otimes \lambda"] \ar[rr,swap, Rightarrow, "\id"] & \ & \otimes 
\end{tikzcd}
\end{equation}
\begin{equation}
\begin{tikzcd}
\otimes \circ \big((1 \otimes \, \cdot \, )\times \id\big) \ar[rr,Rightarrow,"\lambda\otimes\id"] \ar[rd,swap,Rightarrow,"\alpha"] & \tarrow["\ \mit\Lambda",shorten <=2pt,shorten >=2pt]{d} &\otimes \\
& (1 \otimes \, \cdot \, ) \circ (\id \times \otimes) \ar[ru,swap,Rightarrow,"\lambda"] & 
\end{tikzcd}
\end{equation}
and
\begin{equation}
\begin{tikzcd}
\otimes \ar[rr,Rightarrow,"\id\otimes \rho"] \ar[rd,swap,Rightarrow,"\rho"] & \tarrow[ "\ \mit\Psi ",shorten <=2pt,shorten >=2pt]{d} &\otimes \circ \big(\id \times ( \, \cdot \, \otimes 1) \big) \\
& ( \, \cdot \, \otimes 1) \circ (\id \times \otimes) \ar[ru,swap,Rightarrow,"\alpha"] & 
\end{tikzcd}
\end{equation}
\item[(f)] Further invertible modifications
\begin{equation}
\begin{tikzcd}
 & a\otimes (b\otimes c) \ar[r, Rightarrow, "\beta" {name=ar1}] & (b \otimes c) \otimes a \arrow[rd,Rightarrow,"\alpha"] & \\ 
(a\otimes b)\otimes c \arrow[ru,Rightarrow,"\alpha"] \arrow[rd,swap, Rightarrow,"\beta\otimes\id"] & & &b\otimes(c \otimes a) \\
 & (b\otimes a) \otimes c\ar[r,swap,Rightarrow, "\alpha" {name=ar2}] &b \otimes (a \otimes c) \arrow[ru,swap,Rightarrow,"\id\otimes\beta"] & 
 \tarrow[shorten <=10pt,shorten >=10pt, from=ar1, to=ar2, "\ R"]{}
\end{tikzcd} \end{equation}
and
\begin{equation}\begin{tikzcd}
 & (a\otimes b)\otimes c \ar[r,Rightarrow, "\beta"{name=ar1}] & c \otimes (a\otimes b) \arrow[rd,Rightarrow,"\alpha"] & \\ 
a\otimes (b\otimes c) \arrow[ru,Rightarrow,"\alpha"] \arrow[rd,swap,Rightarrow,"\alpha\circ(\beta\otimes\id)\circ\alpha"] & & &(c\otimes a) \otimes b  \\
 & b\otimes (a \otimes   c)\ar[r,swap,Rightarrow, "\beta"{name=ar2}] &(a\otimes c) \otimes b \arrow[ru,swap,Rightarrow,"\beta\otimes\id"] & 
\tarrow[shorten <=10pt,shorten >=10pt, from=ar1, to=ar2, "\ S"]{}
\end{tikzcd} \end{equation}
\item[(g)] An invertible modification
\begin{equation}
\begin{tikzcd}
a \otimes b  \ar[rr, Rightarrow, "\id"] \ar[rd,swap,Rightarrow,"\beta"] & \tarrow["\ \mit\Sigma ",shorten <=1pt,shorten >=1pt]{d} &a \otimes b  \\
& b \otimes a \ar[ru,swap,Rightarrow,"\beta"] & 
\end{tikzcd}
\end{equation}
\end{itemize}
These data are required to satisfy a long list of coherence diagrams, see~\cite[Appendix~C]{schommer2011classification} for details.
\end{definition}
\begin{definition}\label{Def: Symmetric monoidal 2-functor}
A \emph{symmetric monoidal 2-functor} between two symmetric monoidal bicategories $\mathscr{B}$ and $\mathscr{B}'$ consists of a 2-functor $\Ha \colon \mathscr{B} \longrightarrow  \mathscr{B}' $ of the underlying bicategories together with the following data:
\begin{itemize}
\item[(a)] Equivalence natural 2-transformations\footnote{We fix again adjoint inverses and the adjunction data.} $\chi \colon \otimes ' \circ \big(\Ha(\, \cdot\, )\times \Ha(\, \cdot\, )\big)\Longrightarrow  \Ha \circ \otimes $ and $\iota \colon 1' \Longrightarrow  \Ha(1)$, where here we consider $1$ as a 2-functor from the bicategory with one object, one 1-morphism and one 2-morphism to $\mathscr{B}$.

\item[(b)] The three invertible modifications
\end{itemize}
\scriptsize
\begin{equation}
\begin{tikzcd}
 &\Ha(a)\otimes' \big(\Ha(b)\otimes' \Ha(c)\big) \ar[r,Rightarrow,"\id\otimes'\chi\ "{name=ar1}] & \Ha(a)\otimes' \Ha(b\otimes c)\ar[rd,Rightarrow,"\chi"] & \\
\big(\Ha(a)\otimes' \Ha(b)\big)\otimes' \Ha(c) \ar[ur,Rightarrow,"\alpha'"] \ar[rd,swap,Rightarrow,"\chi\otimes'\id"]& & & \Ha\big(a\otimes (b\otimes c)\big) \\
 & \Ha(a\otimes b)\otimes' \Ha(c) \ar[r,swap,Rightarrow,"\chi"{name=ar2}] & \Ha\big((a\otimes b)\otimes c\big)\ar[ru,swap,Rightarrow,"\Ha(\alpha)"]
\tarrow[shorten <=10pt,shorten >=10pt, from=ar1, to=ar2, "\ \mit\Omega"]{}
\end{tikzcd}
\end{equation}
\normalsize

\begin{center}
$
\begin{tikzcd}
\Ha(1)\otimes' \Ha(a)\ar[r,Rightarrow,"\chi"{name=ar1}] & \Ha(1\otimes a) \ar[d,Rightarrow,"\Ha(\lambda)"]\\
1' \otimes' \Ha(a) \ar[u,Rightarrow,"\iota\otimes'\id"] \ar[r,swap,Rightarrow, "\lambda'"{name=ar2}] & \Ha(a)
\tarrow[shorten <=5pt,shorten >=5pt, from=ar1, to=ar2, "\ \mit\Gamma"]{}
\end{tikzcd}
$ \ \ \ \ \ and \ \ \ \ \ 
$
\begin{tikzcd}
\Ha(a)\otimes' 1'\ar[r,Rightarrow,"\id\otimes'\iota\ "{name=ar1}] & \Ha(a)\otimes' \Ha(1) \ar[d,Rightarrow,"\chi"]\\
\Ha(a) \ar[u,Rightarrow,"\rho'"]\ar[r,swap,Rightarrow,"\Ha(\rho)"{name=ar2}] & \Ha(a\otimes 1)
\tarrow[swap,shorten <=5pt,shorten >=5pt, from=ar2, to=ar1, "\ \mit\Delta "]{}
\end{tikzcd}
$
\end{center}

\begin{itemize}
\item[(c)] An invertible modification
\begin{equation}
\begin{tikzcd}
& \Ha(b \otimes a) \ar[rd,Rightarrow,"\Ha(\beta)"] \tarrow[shorten <=5pt,shorten >=5pt, " \ \mit\Upsilon"]{dd} & \\
\Ha(b)\otimes' \Ha(a)\ar[ru,Rightarrow,"\chi"] \ar[rd,swap,Rightarrow,"\beta'"] & & \Ha(a\otimes b) \\
 &\Ha(a)\otimes' \Ha(b)\ar[ru,swap,Rightarrow,"\chi"]  &
\end{tikzcd}
\end{equation}
\end{itemize} 
These data are required to satisfy a long list of coherence conditions, see \cite{schommer2011classification} and references therein for details.
\end{definition}  
Our definition of symmetric monoidal transformations differs slightly from~\cite{schommer2011classification}. The definition we give is tailored to the application 
in functorial field theories. 
In contrast to the definition given in~\cite{schommer2011classification}, we require the appearing modifications to be invertible. However, the 2-morphisms corresponding to the underlying natural transformations are not invertible in our definition, so our definition is also weaker than the definition given in~\cite{schommer2011classification}.  

\begin{definition}
\label{Def:2sym transformation}
A \emph{natural symmetric monoidal 2-transformation} between symmetric monoidal 2-functors $\Ha ,\Ka \colon \mathscr{B} \longrightarrow  \mathscr{B}'$ consists of a natural 2-transformation $\theta \colon \Ha  \Longrightarrow  \Ka $ of the underlying 2-functors together with invertible modifications
\begin{equation}
\begin{tikzcd}
 & \Ha (a\otimes b) \ar[rd,Rightarrow,"\theta"] & \\
\Ha (a)\otimes' \Ha (b) \ar[ru,Rightarrow,"\chi_\Ha"] \ar[d,swap,Rightarrow, "\theta\otimes'\id"] & & \Ka(a\otimes b) \\
 \Ka(a) \otimes' \Ha (b) \ar[rr,swap,Rightarrow,"\id\otimes'\theta"] & \tarrow[shorten <=5pt,shorten >=5pt,swap, " \ \mit\Pi"]{uu} & \Ka(a) \otimes' \Ka (b) \ar[u,swap,Rightarrow,"\chi_\Ka"]
\end{tikzcd}
\end{equation}
and
\begin{equation}
\begin{tikzcd}
1' \ar[rr,Rightarrow,"\iota_\Ka"] \ar[dr,swap,Rightarrow,"\iota_\Ha"] & \tarrow[shorten <=2pt,shorten >=2pt, " \ M"]{d} & \Ka(1) \\
 & \Ha (1) \ar[ur,swap, Rightarrow, "\theta"] &  
\end{tikzcd}
\end{equation}
which satisfy the following coherence conditions expressed as equalities between 2-morphisms (omitting tensor product symbols on objects and 1-morphisms to streamline the notation):
\footnotesize
\begin{equation*}
\begin{tikzcd}[column sep=small]
 &\Ka (a)\big(\Ka (b)\Ha (c)\big) \ar[r,"\theta"] & \Ka(a)\big(\Ka( b)\Ka( c)\big) \ar[r,"\alpha'"] & \big((\Ka( a)\Ka (b)\big)\Ka (c) \ar[dr,"\chi_\Ka"] & \\
\big(\Ka (a)\Ka (b)\big)\Ha (c) \ar[ur,"\alpha'"] \ar[urrr,bend right=10,"\theta"] \ar[rd,swap,"\chi_\Ka"] & \ar[u, Leftarrow,swap, "\;\; ",shorten <=2pt,shorten >=2pt] & & & \Ka (ab)\Ka (c)\ar[dd,"\chi_\Ka"]\\
 & \Ka (ab)\Ha (c) \ar[rrru, bend right=10,"\theta"] & \ar[ddd, Rightarrow, "\ \mit\Pi",shorten <=15pt,shorten >=15pt] \ar[uu, Leftarrow,swap, "\;\; ",shorten <=5pt,shorten >=30pt, pos=0.3] & & \\
\big(\Ka (a)\Ha( b)\big)\Ha (c) \ar[uu,"\theta"]  &\ar[l, Leftarrow,shorten <=10pt,shorten >=10pt, " \mit\Pi \otimes' \id"] & & & \Ka \big((ab) c\big) \ar[dd,"\Ka(\alpha)"] \\
 & \Ha (ab)\Ha (c) \ar[uu,swap,"\theta"]\ar[rd,"\chi_\Ha"] \ar[dd, Rightarrow, "\ \mit\Omega\!_\Ha",shorten <=10pt,shorten >=10pt] & &\ar[dd, Rightarrow, "\ \theta_{\alpha}",shorten <=10pt,shorten >=10pt] & \\
\big(\Ha (a) \Ha (b)\big)\Ha (c)\ar[uu,"\theta"] \ar[ur,"\chi_\Ha"] \ar[dr,swap,"\alpha'"] & & \Ha \big((ab) c\big)\ar[rd,"\Ha(\alpha)"]\ar[rruu,"\theta"] & & \Ka \big(a(bc)\big) \\
 &\Ha (a) \big(\Ha (b)\Ha (c)\big)\ar[r,swap,"\chi_\Ha"] & \Ha (c)\Ha (bc) \ar[r,swap,"\chi_\Ha"] & \Ha (a(bc))\ar[ru,swap,"\theta"] &
\end{tikzcd}
\end{equation*}
\normalsize
\LARGE
\[
\parallel
\]
\normalsize
\begin{equation}
\label{EQ:1 Definition s.m. transformation}
\begin{footnotesize}
\begin{tikzcd}[column sep=small]
 &\Ka (a)\big(\Ka (b)\Ha (c)\big) \ar[r,"\theta"] & \Ka (a)\big(\Ka (b)\Ka (c)\big) \ar[ddddd, Rightarrow,shorten <=20pt,shorten >=20pt,"\ \id\otimes'\mit\Pi"] \ar[r,"\alpha'"] \ar[rddd, bend right=15,swap,"\chi_\Ka"] & \big(\Ka( a)\Ka (b)\big)\Ka (c) \ar[dr,"\chi_\Ka"] \ar[d,"\alpha'"] & \\
\big(\Ka (a)\Ka (b)\big)\Ha (c) \ar[ur,"\alpha'"] &  &\ar[r,Leftarrow, swap,shorten <=30pt,shorten >=5pt] &\Ka (a)\big(\Ka( b)\Ka (c)\big)  \ar[dd,"\chi_\Ka"] & \Ka (ab)\Ka (c)\ar[dd,"\chi_\Ka"]\ar[ddl, Rightarrow,shorten <=10pt,shorten >=10pt,"\ \mit\Omega\!_\Ka"]\\
 & \ar[ld, Leftarrow,shorten <=10pt,shorten >=10pt,swap, "\alpha'{}^\star"] & &  & \\
\big(\Ka (a)\Ha (b)\big)\Ha (c) \ar[uu,"\theta"] \ar[rd,"\alpha'"]  & & &\Ka (a)\Ka (bc) \ar[rdd,"\chi_\Ka"] \ar[ddd,Rightarrow,shorten <=10pt, shorten >=10pt, "\ \mit\Pi"] & \Ka \big((ab)c\big) \ar[dd,"\Ka(\alpha)"] \\
 & \Ka (a)\big(\Ha (b)\Ha (c)\big)\ar[uuuu,"\theta"] \ar[rdd, Rightarrow,shorten <=20pt,shorten >=20pt,"\Phi_{\otimes'}\ ",swap] \ar[rd,"\chi_\Ha"]& &  & \\
\big(\Ha (a)\Ha (b)\big)\Ha (c) \ar[uu,"\theta"] \ar[dr,swap,"\alpha'"] \ar[ru, Rightarrow,shorten <=10pt,shorten >=10pt, "\alpha'{}^\star",pos=0.6] & &\Ka (a)\Ha (bc)\ar[uur,"\theta"] & & \Ka \big(a(bc)\big) \\
 &\Ha (a)\big(\Ha (b)\Ha (c)\big)\ar[r,swap,"\chi_\Ha"]\ar[uu,"\theta"] & \Ha (a)\Ha (bc)\ar[u,swap,"\theta"] \ar[r,swap,"\chi_\Ha"] & \Ha \big(a(bc)\big)\ar[ru,swap,"\theta"] & 
\end{tikzcd}
\end{footnotesize}
\end{equation}
\begin{equation*}
\begin{tikzcd}
& \Ka (1)\Ha (a) \ar[rr,"\theta"] & \ar[d, Rightarrow, "\ \mit\Pi"] &\Ka (1)\Ka (a) \ar[rd,"\chi_\Ka"]  & \\
\Ha (1)\Ha (a)\ar[ur,"\theta"]\ar[rr, "\chi_\Ha"{name=ar1}]& &\Ha (1a)\ar[rr, "\theta"{name=ar2}] \ar[d,"\Ha(\lambda)"]  & & \Ka (1a)\ar[d,"\Ka(\lambda)"] \\
1'\Ha (a) \ar[u,"\iota_\Ha"]\ar[rr,swap, "\lambda'"{name=ar3}] \ar[rrd,swap,"\theta"]& & \Ha (a) \ar[rr, swap, "\theta"{name=ar4}] \ar[d, Rightarrow, "\ \lambda'_{\theta_a}"]  & & \Ka( a) \\
 & &1'\Ka (a)\ar[rru,swap, "\lambda'"] & &
\ar[from=ar1, to=ar3, Rightarrow, "\ \mit\Gamma\!_\Ha",shorten <=2pt,shorten >=5pt]
\ar[from=ar2, to=ar4, Rightarrow, "\ \theta",shorten <=2pt,shorten >=5pt]
\end{tikzcd}
\end{equation*}
\LARGE
\[
\parallel
\]
\normalsize
\begin{equation}
\label{EQ:2 Definition s.m. transformation}
\begin{tikzcd}
& \Ka (1)\Ha (a) \ar[rr,"\theta"] \ar[rddd, Rightarrow, shorten <=10pt,shorten >=10pt, "\ \Phi_{\otimes'}"] & &\Ka (1)\Ka(a) \ar[rd,"\chi_\Ka"]  & \\
\Ha (1)\Ha (a)\ar[ur,"\theta"]&\ar[l,Leftarrow, "{\footnotesize \ \ \  M^{-1}\otimes'\id}",shorten <=13pt, pos=0.8] & & & \Ka (1a)\ar[d,"\Ka(\lambda)"] \\
1'\Ha (a) \ar[u,"\iota_\Ha"]\ar[ruu, swap,bend right,"\iota_\Ka"] \ar[rrd,swap,"\theta"]& &  &\ar[uu, Leftarrow,shorten >=10pt,swap, "\ \mit\Gamma\!_\Ka",pos=0.2] & \Ka (a) \\
 & &1'\Ka (a)\ar[rru,swap,"\lambda'"]\ar[ruuu,"\iota_\Ka"] & &
\end{tikzcd}
\end{equation}
\begin{equation*}
\begin{tikzcd}
& \Ka (a)1' \ar[r,"\iota_\Ha"{name=ar1}] & \Ka (a)\Ha (1)\ar[dr,"\theta"] & \\
\Ka (a) \ar[ru,"\rho'"] \ar[r,Rightarrow,shorten <=5pt,shorten >=5pt, "\rho'_\theta"] & \Ha (a) 1'\ar[u,"\theta"] \ar[r,swap,"\iota_\Ha"{name=ar2}] & \Ha (a)\Ha (1)\ar[rd, Rightarrow,shorten <=5pt,shorten >=5pt,"\ \mit\Pi"]  \ar[d,"\chi_\Ha"] \ar[u,swap,"\theta"]  & \Ka (a)\Ka (1) \ar[d,"\chi_\Ka"] \\
\Ha (a)\ar[u,"\theta"]\ar[rd,swap,"\theta"] \ar[ru,"\rho'"] \ar[rr,"\Ha(\rho)"]& \ar[u,Leftarrow,shorten <=10pt,shorten >=2pt,swap, "\ {\mit\Delta}_\Ha^{-1}",pos=0.6]\ar[d,Rightarrow,shorten <=2pt,shorten >=2pt,"\ \theta_\rho"] &\Ha (a1)\ar[r,"\theta"]  & \Ka (a1) \\
& \Ka (a) \ar[rru,swap,"\Ka(\rho)"]&  & 
\ar[from=ar1, to=ar2, Rightarrow,shorten <=5pt,shorten >=5pt,"\ \iota_\theta"]
\end{tikzcd}
\end{equation*}
\LARGE
\[
\parallel
\]
\normalsize
\begin{equation}
\label{EQ:3 Definition s.m. transformation}
\begin{tikzcd}
& \Ka (a)1' \ar[rrd, swap, bend right=15,"\iota_\Ka"{name=ar1}] \ar[r,"\iota_\Ha"] & \Ka (a)\Ha (1)\ar[dr,"\theta"]  & \\
\Ka (a) \ar[ru,"\rho'"]\ar[rrrd,swap,"\Ka(\rho)"{name=ar2}] & \ar[dd,Rightarrow,shorten <=10pt,shorten >=10pt,"\ \id", pos=0.6] & \ar[u,Leftarrow, swap,shorten <=5pt,shorten >=2pt,"\ \id\otimes'M^{-1}"] & \Ka (a)\Ka (1) \ar[d,"\chi_\Ka"] \\
\Ha (a)\ar[u,"\theta"]\ar[rd,swap,"\theta"]&  &  & \Ka (a1) \\
& \Ka (a) \ar[rru,swap,"\Ka(\rho)"]&  & 
\ar[from=ar1, to=ar2, Rightarrow,shorten <=2pt,shorten >=10pt," \!{\mit\Delta}_\Ka^{-1}", pos=0.2]
\end{tikzcd}
\end{equation}
and
\begin{equation}
\label{EQ:4 Definition s.m. transformation}
\small
\begin{tikzcd}
 \Ha (b) \Ha (a) \ar[rr,"\chi_\Ha"] & \ar[d, Leftarrow, swap,shorten <=2pt,shorten >=2pt, "{\mit\Upsilon}_\Ha^{-1}\ " ] & \Ha (ba)\ar[d,"\theta"]  \\
\Ha (b)\Ha (a) \ar[d,swap,"\theta\circ'\beta'"]\ar[u,"\id"]\ar[r,"\chi_\Ha\circ'\beta'"] & \Ha (ab)\ar[ur,"\Ha(\beta)"]\ar[dr,"\theta"] \ar[r, Rightarrow,shorten <=2pt,shorten >=2pt,"\theta_\beta"]  & \Ka (ba) \\
\Ka (a)\Ka (b) \ar[rr,swap,"\chi_\Ka"]& \ar[u, Rightarrow, shorten <=2pt,shorten >=2pt, "\mit\Pi\ "] & \Ka (ab)\ar[u,swap,"\Ka(\beta)"] 
\end{tikzcd}
 \ \ = \ \ 
\begin{tikzcd}
 \Ha (a)\Ha (b) \ar[rr,"\chi_\Ha\circ'\beta'"]\ar[dr,"\theta\circ'\beta'"] & \ar[d, Leftarrow,shorten <=3pt,shorten >=3pt, swap,"\mit\Pi\ "] & \Ha (ba)\ar[d,"\theta"]  \\
\Ha (a)\Ha (b) \ar[d,swap,"\theta"]\ar[u,"\id"] \ar[r, Rightarrow,shorten <=2pt,shorten >=2pt, "\beta'_{\theta\otimes' \theta} "] & \Ka (b)\Ka (a) \ar[r,"\chi_\Ka"] & \Ka (ba)   \\
\Ka (a)\Ka (b) \ar[rr,swap,"\chi_\Ka"]\ar[ur,"\beta'"]& \ar[u, Rightarrow,shorten <=2pt,shorten >=2pt, "{\mit\Upsilon}_\Ka^{-1}\ "] & \Ka (ab)\ar[u,swap,"\Ka(\beta)"] 
\end{tikzcd}
\end{equation}
\normalsize
In \eqref{EQ:1 Definition s.m. transformation}, the unlabelled 2-morphisms in the first diagram are constructed from naturality of $\alpha^\star$ and 2-functoriality of $\otimes$, while the unlabelled 2-morphism in the second diagram is induced by the equivalence $\alpha^\star \circ\alpha\Longrightarrow  \id$. 
\end{definition}

\begin{definition}
A \emph{symmetric monoidal modification} between two symmetric
monoidal 2-trans{-}formations $\theta, \theta'\colon
\mathcal{H}\Longrightarrow  \mathcal{K}$ consists of a modification $m
\colon \theta \Rrightarrow \theta'$ of the underlying natural 2-transformations satisfying
\begin{equation}
\begin{footnotesize}
\begin{tikzcd}
\mathcal{H}(a)\otimes' \mathcal{H}(b)\ar[dd, swap, "{\theta \otimes' \theta}"] \ar[rr,"{\chi_\mathcal{H}}"] & & \mathcal{H}(a\otimes b)\ar[dd, "{\theta}"] \ar[dd, bend left, "{\theta'}",out=90, in=90] & \\ 
\; \ar[rr, Rightarrow, "{\mit\Pi}",shorten >=0.5cm,shorten <=0.5cm] & \; & \; \ar[r, Rightarrow, "m",shorten >=0.3cm,shorten <=0.2cm, pos=0.45] & \; \\
\mathcal{K}(a)\otimes' \mathcal{K}(b) \ar[rr,swap,"{\chi_\mathcal{K}}"] & & \mathcal{K}(a\otimes b) &
\end{tikzcd}
 \ \ = \ \ 
\begin{tikzcd}
& \mathcal{H}(a)\otimes' \mathcal{H}(b)\ar[dd, "{\theta' \otimes' \theta'}"] \ar[dd, bend right ,swap,"{\theta \otimes' \theta}", out=-90, in=-90] \ar[rr,"{\chi_\mathcal{H}}"] & & \mathcal{H}(a\otimes b)\ar[dd, "{\theta'}"]   \\ 
\; \ar[r, Rightarrow, "m\otimes m",shorten <=0.4cm, pos=0.7] & \; \ar[rr, Rightarrow, "{\mit\Pi'}",shorten >=0.5cm,shorten <=0.8cm] & \; & \;  \\
& \mathcal{K}(a)\otimes' \mathcal{K}(b) \ar[rr,swap,"{\chi_\mathcal{K}}"] & & \mathcal{K}(a\otimes b) 
\end{tikzcd}
\end{footnotesize}
\end{equation}
and
\begin{equation}
\begin{tikzcd}
 & & \mathcal{H}(1) \ar[dd, "\theta"] \ar[dd, bend left, "{\theta'}",out=90, in=90] & \\
 1' \ar[rru, "{\iota_\mathcal{H}}"] \ar[rrd, swap,"{\iota_\mathcal{K}}"] \ar[rr, Rightarrow, "M",shorten >=0.2cm,shorten <=0.8cm, pos=0.65]& &\; \ar[r, Rightarrow, "m",shorten >=0.3cm,shorten <=0.2cm, pos=0.45] &\; \\
  & & \mathcal{K}(1) &
\end{tikzcd}
 \ \ = \ \ 
\begin{tikzcd}
 & & \mathcal{H}(1) \ar[dd, "\theta'"]  & \\
 1' \ar[rru, "{\iota_\mathcal{H}}"] \ar[rrd, swap,"{\iota_\mathcal{K}}"] \ar[rr, Rightarrow, "M'",shorten >=0.2cm,shorten <=0.8cm, pos=0.65]& &\;  &\; \\
  & & \mathcal{K}(1) &
\end{tikzcd}
\end{equation}
\end{definition}
\chapter{Homotopy theory for groupoids, stacks and integration}\label{Sec: Homotop Groupoid}
In this appendix we review some basic facts and definitions related to
groupoids used in this thesis. We start by defining a model structure 
on the category of groupoids. We give concrete formulas for homotopy (co)limits
in this model structure. 

Afterwards, we define stacks as a categorification of sheaves and state some of there basic properties. 
The final section of this appendix is dedicated to the integration of gauge invariant
functions over groupoids.     

\section{A model category for groupoids}
In many situation requiring two objects of a category $\cC$ to be 
isomorphic is to restrictive and there exists a class of morphisms called weak equivalences 
which should replace isomorphisms. 
For example from a homotopical point of view two topological spaces are `the same' if there exist 
a weak homotopy equivalence between them. 
Let $\W\subset \Mor(\cC)$ denote the collection of weak equivalences. In this case the 
localization $\cC[\W^{-1}]$ of $\cC$ at the weak equivalences is a natural object to consider. 
Furthermore, one should replace categorical structures, such as limits, Kan extensions and 
equivalences of categories with constructions which are compatible with weak equivalences.  

To get a better technical handle on weak equivalences the introduction of two additional classes
of morphisms called fibrations and cofibrations is useful. We are mostly interested in the 
case where $\cC$ is a category enriched over the category $\sSet$ of simplicial sets. The relevant
definition in this situation is that of an simplical model category, see for example \cite[Section 11.4]{Riehl}.
We do not spell out the definition in detail here, but rather explain to what it boils down in the 
case of groupoids. 
There is a simplicial model structure on the category of (small) groupoids $\Grpd$
which we review following~\cite{Hollander}.
The category $\Grpd$ becomes enriched over $\sSet$ as follows:
the simplex category $\Delta$ of finite ordered sets and order preserving maps embeds into the 
category $\Cat$ of categories by sending a finite ordered set $S$ to the category which has one object
for every element of $S$ and exactly one morphism from $s_1\in S$ to $s_2\in S$ if and only if $s_1\leq 
s_2$. The \emph{simplicial nerve} of a category $\cC$ is the simplicial set $N\cC_\bullet \coloneqq \Hom_{\Cat}(\bullet, C)$. The simplicial nerve functor admits a left adjoint $h\colon \sSet \longrightarrow \Cat$ which can be constructed via left Kan extension 
\begin{equation}
\begin{tikzcd}
\Delta \ar[r] \ar[d] & \Cat \\
\sSet \ar[ru, "h", swap] & 
\end{tikzcd}
\end{equation}  
We denote by $\Pi \colon \sSet \longrightarrow \Grpd$ the composition of $h$ with the functor $\Cat \longrightarrow \Grpd$ which sends a category to the groupoid constructed by inverting all morphisms.            
The simplicial mapping space between two groupoids $\cG$ and $\cG'$ is $N([\cG,\cG'])$ where we denote 
by $[\cG,\cG']$ the category of functors from $\cG$ to $\cG'$.
Now we can describe the simplicial model category on $\Grpd$.
We denote by $\Delta_1$ and $\Delta_0$ the image of $[1]$ and $[0]$ under the embedding $\Delta \hookrightarrow \Cat$, respectively. Furthermore, let $I$ be the groupoid with two objects and
one isomorphisms between them. 
\begin{definition}
A functor $\cG \longrightarrow \cG'$ between groupoids is a 
\begin{itemize}
\item \emph{weak equivalence} if it is an equivalence of categories,

\item \emph{cofibration} if it is injective on objects,

\item \emph{fibration} if every diagram of the form 
\begin{equation}\label{Diagram: B. lifts}
\begin{tikzcd}
\Delta_0 \ar[r] \ar[d,"0", hookrightarrow]& \cG \ar[d] \\
\Delta_1 \ar[r] \ar[ru, dotted] & \cG' 
\end{tikzcd}
\end{equation}
admits a lift $\Delta_1\longrightarrow \cG$. 
\end{itemize} 
\end{definition}
This defines a simplical model structure on $\Grpd$, see e.g.~\cite{Hollander}.
\begin{remark}
The lift in \eqref{Diagram: B. lifts} is not required to be unique. There is a special class of functors
between groupoids which admits unique lifts: 
A surjective fibration $\cG\longrightarrow \cG'$ of groupoids for which all lifts in \eqref{Diagram: B. lifts} are unique is called a \emph{covering of groupoids}. Coverings will play an essential role in the study of integration over finite groupoids in Section \ref{Sec: Integration over finite groupoids}. Let $n$ be a natural number. An \emph{$n$-folded} covering is a covering where the fibre over every point contains $n$ elements.
\end{remark}

Being a simplicial model category, implies in particular that $\Grpd$ is  tensored and cotensored (see \cite[Section 3.7]{Riehl}) over $\sSet$, i.e.\
there are functors
\begin{align}
\begin{split}
\otimes \colon \sSet \times \Grpd & \longrightarrow \Grpd \\ 
S\times \cG & \longmapsto S\otimes \cG \coloneqq \Pi(S)\times \cG
\end{split}
\end{align}
and 
\begin{align}
\begin{split}
-^\cdot \colon \sSet^{\opp}\times \Grpd& \longrightarrow \Grpd \\ 
S\times G & \longmapsto G^S \coloneqq [\Pi(S),G]
\end{split}
\end{align}
such that there are natural isomorphism
\begin{align}
N([S\otimes \cG, \cG'])\cong \Map_\sSet (S, N([\cG,\cG'])) \cong N([\cG, \cG'^S])
\end{align}
defining a two-variable adjunction (~see e.g.\ \cite[Definition 4.1.12]{Hovey} for the definition).

(Co)limits in $\Grpd$ are in general not compatible with equivalences of categories: consider the equivalence of diagrams
\begin{equation}
\begin{tikzcd}
 & \Delta_0 \sqcup \Delta_0 \ar[dd] \ar[rr] & & \Delta_0 \\
\Delta_0 \sqcup \Delta_0 \ar[rr] \ar[ru] \ar[dd] & & \Delta_0 \ar[ru] & \\ 
 & \Delta_0 & \\
I \ar[ru] & & &
\end{tikzcd}
\end{equation}
The pushout of the diagram containing $I$ is the groupoid $*\DS \Z$, while the pushout of the 
other diagram is $\Delta_0$.  
Homotopy (co)limits solve this problem\footnote{Alternatively, the problem
could be solved by working in a 2-categorical setting; replacing (co)limts 
with their 2-categorical analogues.} by replace the categorical concept of (co)limits with a homotopy invariant definition. One way to
define homotopy (co)limits is via the introduction of a model structure on diagram categories,
such as the projective, injective or Reedy model structure \cite{Hovey, Riehl}. 
These model structures only exist under certain
conditions on the model category or the shape of the diagram category. 

The approach we use is via a concrete definition using the 
two-sided bar construction~\cite{Riehl}.       
Let $\cD$ be a small category. We now construct homotopy limits and colimits over $\cD$ as functors 
$\holim, \hocolim \colon \Grpd^{\cD} \longrightarrow \Grpd$, where we denote by $\Grpd^{\cD}$ the category
of diagrams of shape $\cD$ in $\Grpd$. Every object
in $\Grpd$ is fibrant  and cofibrant, meaning that the unique map to the terminal object
$1$ is a fibration and the unique map from the initial object $\emptyset$ is a cofibration. 
For this reason our formulas will not involve pointwise (co)fibrant replacements of diagrams. Let 
$F\colon \cD \longrightarrow \Grpd$ be a diagram in $\Grpd$.
We introduce the simplicial groupoid 
\begin{align}
B_n(\star, \cD, F) \coloneqq \coprod_{\vec{d}\colon [n] \rightarrow \cD} F(d_0) \ \ ,
\end{align}
where $\vec{d}$ is a string of morphisms $d_0\rightarrow d_1 \rightarrow \dots \rightarrow d_n$ in $\cD$. The degeneracy maps are induced 
from the degeneracy maps of the nerve of $\cD$ and reindexing the 
coproduct. The face maps are constructed from the face maps in the
nerve of $\cD$ and the map $F(d_0\rightarrow d_1)\colon F(d_0)\longrightarrow
F(d_1)$ for $\partial_0$. The \emph{homotopy colimit} of $F$ 
is the geometric realisation of $B_\bullet(\star, \cD, F)$, i.e.\ the coend
\begin{align}
\hocolim F = \int^{n\in \Delta} \Delta_n \otimes B_n(\star, \cD, F)
\end{align}   
where $\Delta_n$ corresponds to the Yoneda embedding $\Delta \longrightarrow
[\Delta^{\opp},\Set]$. 

To define homotopy limits we introduce the cosimplicial groupoid 
\begin{align}
C^n(\star, \cD, F) \coloneqq \prod_{\vec{d}\colon [n] \rightarrow \cD} F(d_n)
\end{align} 
with face and degeneracy maps dual to those of $B_\bullet (\star, \cD,F)$.
The \emph{homotopy limit} is the totalization of $C^\bullet(\star, \cD, F)$, 
i.e.\ the end 
\begin{align}
\holim F \coloneqq \int_{n\in \Delta} C^n(\star, \cD, F)^{\Delta_n} \ \ .
\end{align}
We conclude this section with an example relevant in the main 
text.

Consider the diagram $F$
\begin{equation}\label{Eq: Diagram pullback}
\begin{tikzcd} 
 & \cG \ar[d,"f_2"] \\
 \cG' \ar[r,"f_1",swap] & \cB
\end{tikzcd}
\end{equation}
in $\Grpd$. The nerve of the underlying diagram category $\cD$ is 1-skeletal.
This implies that the cosimplical object $C^\bullet(\star,\cD,F)$ is 1-skeletal~\cite[Example 6.5.2.]{Riehl} and hence we can compute its totalization as the
end over $\Delta[1]\subset \Delta$. To compute the end we use
\begin{align}
C^0(\star, \cD, F) = \cG\times \cB \times \cG' \text{ and } C^1(\star, \cD, F) = \cG\times \cB \times \cG'\times \cB_{f_1} \times \cB_{f_2} \ \ ,  
\end{align}
where the subscript denotes the morphism indexing the product (components without index
correspond to identity maps). The homotopy limit is the equalizer of the following
diagram (using Proposition~\ref{Prop: End as equalizer}) 
\begin{align}
C^0(\star, \cD, F)\times C^1(\star, \cD, F)^I \rightrightarrows \prod_{d\colon [i]\rightarrow [j]}^{i,j\leq 1} C^j(\star, \cD, F)^{\Pi(\Delta_i)} 
\end{align}
An element of $C^0(\star, \cD, F)\times C^1(\star, \cD, F)^I$ consists of elements $g\in \cG$, $b \in \cB$, $g'\in \cG'$
and morphisms $h\colon g_1\longrightarrow g_2 \in \cG$, $h_b \colon b_1\longrightarrow b_2 \in \cB$, $h'\colon g'_1\longrightarrow g'_2 \in \cG'$, $h_{f_1}\colon b_3 \longrightarrow b_4 \in \cB$ and 
$h_{f_2}\colon b_5 \longrightarrow b_6 \in \cB$. 
There are 3 morphisms in $\Delta[1]$ contributing non-trivially to 
the equalizer: $s_0 \colon [1]\longrightarrow [0]$, $d_0 \colon [0]\longrightarrow [1]$ and $d_1 \colon [0]\longrightarrow [1]$.
Evaluating the two morphisms corresponding to $s_0$ implies the relations $h=\id_g$, $h'=\id_{g'}$ and $h_b=\id_b$. 
Evaluating the maps corresponding to $d_0$ and $d_1$ implies $b_6=b_4=b$,  $f_1(g)=b_3$ and 
$f_2(g')=b_5$. After imposing all these conditions an element in the equalizer can be represented by the diagram
\begin{equation}\label{Eq: Element in homotopy pullback}
\begin{tikzcd}
 & f_2(g') \ar[d, "{h_{f_2}}"] \\ 
 f_1(g) \ar[r, "{h_{f_1}}",swap] & b 
\end{tikzcd} 
\end{equation}
A morphisms consists of three maps $g_1\longrightarrow g_2 $, $g'_1\longrightarrow g'_2 $ and $b_1\longrightarrow b_2$ 
making the obvious diagram commute. The information contained in the element $b$ is obsolete. We can define an equivalent
groupoid with objects pairs of elements $g \in \cG$ and $g'\in \cG'$ together with an isomorphism $f_2(g')\longrightarrow f_1(g)$. The equivalence between the groupoids sends \eqref{Eq: Element in homotopy pullback} 
to $(g,g', h_{f_1}^{-1} \circ h_{f_2})$.
We call this simplified groupoid the \emph{homotopy pullback} of the diagram \eqref{Eq: Diagram pullback} and denote it by $\cG \times_B \cG'$. There are projection maps $\cG \times_B \cG' \longrightarrow \cG$ and $\cG \times_B \cG' \longrightarrow \cG'$ fitting into a homotopy commutative square
\begin{equation}
\begin{tikzcd}
\cG \times_B \cG' \ar[r] \ar[d] & \cG' \ar[d, "{f_2}"] \\
\cG \ar[r, "{f_1}",swap] \ar[ru, Leftarrow, "h"] & B
\end{tikzcd} 
\end{equation}
The homotopy (natural isomorphism) $h$ is constructed from the morphisms $f_2(g')\longrightarrow f_1(g)$. 

Homotopy fibres can be defined as special cases of homotopy pullbacks: let $F\colon \cG \longrightarrow \cG'$ be a functor 
between groupoids and $g'\in \cG'$. The \emph{homotopy fibre of $g'$} $F^{-1}[g']$ is the homotopy pullback 
\begin{equation}
\begin{tikzcd}
F^{-1}[g'] \ar[r] \ar[d] & \cG \ar[d, "{f_2}"] \\
\Delta_0 \ar[r, "{g'}",swap] \ar[ru, Leftarrow, "h"] & \cG'
\end{tikzcd} 
\end{equation}

\begin{definition}
	Let $\cG$ be a groupoid. $\cG$ is \emph{essentially finite} if the set of isomorphism classes of objects
	$\pi_0(\cG)$ and all automorphism groups are finite. We denote by $\FinGrpd$ the category of essentially
	finite groupoids.  
\end{definition}   
For later use we record the following observation.
\begin{lemma}\label{Lemma: Coverings}
Let $F\colon \cG \longrightarrow \cG'$ be a functor between essentially finite groupoids, $g'\in \cG'$ and $\cG_{g'}$ the subgroupoid consisting of elements $g$ in $\cG$ such that $\F(g)$ isomorphic to $g'$. Then $F^{-1}[g']\longrightarrow \cG_{g'}$ is an $|\Aut(g')|$-fold covering. 
\end{lemma}

\section{Stacks}\label{Sec: Stacks}
Let $G$ be a Lie group and $M$ a manifold.
Principal $G$-bundles on $M$ are local in the following sense:
let $\{ \mathcal{U}_i \}_{i\in I}$ be an open covering of $M$. Then from
principal bundles $P_i$ on $\mathcal{U}_i$ together with gauge 
transformations $\varphi_{ij}\colon P_i|_{\mathcal{U}_{ij}} \longrightarrow 
P_j|_{\mathcal{U}_{ij}}$ satisfying $\varphi_{ik}= \varphi_{jk}\circ \varphi_{ij}$ on 
$\mathcal{U}_{ijk}$ with $\mathcal{U}_{ij}\coloneqq \mathcal{U}_i \cap \mathcal{U}_j$ and 
$\mathcal{U}_{ijk}\coloneqq \mathcal{U}_i \cap \mathcal{U}_j \cap \mathcal{U}_k$, 
we can construct a principal $G$-bundle on $M$. Furthermore, all principal bundles 
can be constructed up to gauge
transformation from local data. 
This can be captured by the statement that the natural functor
\begin{align}
\Bun_G(M) \longrightarrow \Desc (\{ \mathcal{U}_i \})
\end{align}
is an equivalence of groupoids, where $\Desc (\{ \mathcal{U}_i \})$ is the groupoid of
local data $\{ P_i, \varphi_{ij} \}$ with respect to the open cover $\mathcal{U}_i$.
A stack is an abstraction of this locality property. Stacks are defined over categories 
with a notion of ``coverings". These categories are called sites.   
\begin{definition}
Let $\cC$ be a category. A \emph{Grothendieck topology}\footnote{What we define is 
sometimes called a pretopology.} on $\cC$ is a class of morphisms
in $\cC$ called coverings such that
\begin{itemize}
\item every isomorphism is a covering,

\item 
the pullback of a covering along an arbitrary morphisms in $\cC$ exist and is again a covering, and

\item the composition of coverings is a covering.
\end{itemize} 
A category together with the choice of a Grothendieck topology is called a \emph{site}.
\end{definition}   

\begin{example}
We are mostly interested in the following example. Let $n$ be a positive integer and
$\Man_n$ the category of $n$-dimensional manifolds (with corners). 
There is a Grothendieck topology on $\Man_n$ where a covering of a manifold
$M\in \Man_n$ is a map of the form $Y\overset{\cong}{\longrightarrow} U=\coprod_{i\in I} \mathcal{U}_i \longrightarrow M$ where the first morphism is an isomorphism and $\{ \mathcal{U}_i \}_{i\in I}$ is an open covering of $M$. 

There is a different Grothendieck topology on $\Man_n$ where the coverings are
surjective submersions. These two Grothendieck topology are equivalent in an 
appropriate sense~\cite{Vistoli}. This implies in particular that the notion of a stack 
with respect to both Grothendieck topology agrees. 
\end{example}

Let $f\colon M \longrightarrow N$ be a smooth map between smooth manifolds. The map 
$f$ induces a pullback functor $f^* \colon \Bun_G(N) \longrightarrow \Bun_G(M)$. 
For two composable morphisms $f$ and $g$ the functors $(f\circ g)^*$ and $g^* \circ f^*$
are only isomorphic, or in different words 
$\Bun_G(\cdot)\colon \Man^{\opp}\longrightarrow \Grpd$ is only a 2-functor, where 
$\Man^{\opp}$ is considered as a 2-category with only identity 2-morphisms.

\begin{definition}
Let $\cC$ be a category. A \emph{pre-stack} $\F$ on $\cC$ is a 2-functor 
\begin{align}
\F \colon \cC^{\opp} \longrightarrow \Grpd \ \ .
\end{align}  
\end{definition} 
The following definition should be understood as a generalization of the 
category $\Desc (\{ \mathcal{U}_i\})$ to pre-stacks on a site.

\begin{definition}
Let $\cC$ be a site, $\F\colon \cC^{\opp}\longrightarrow \Grpd$ a pre-stack on $\cC$
and $\pi \colon Y\longrightarrow X$ a covering. The covering $\pi$ allows us to 
define the following simplicial object
\begin{equation}
\begin{tikzcd}
\cdots Y^{[3]} \ar[r, shift left] \ar[r, shift right] \ar[r] &  Y^{[2]}\ar[r, shift left, "\partial_0"] \ar[r, shift right, "\partial_1",swap]  & Y 
\end{tikzcd}
\end{equation}
where $Y^{[n]}$ is the iterated fibre product $Y\times_X Y \dots \times_X Y $. 
The \emph{descent category $\Desc_\F(Y)$} has as objects pairs $(f_Y, \varphi \colon 
\partial_1^* f_Y \longrightarrow \partial_0^*f_Y)$, where $f_Y$ is an element
of $\F(Y)$ and $\varphi$ is a morphism $\F(\partial_1)[f_Y]\longrightarrow \F(\partial_0)[f_Y]$
in $\F[Y^{[2]}]$ satisfying (suppressing coherence isomorphisms) $ \F(\partial_2)[\varphi] \circ \F(\partial_0)[\varphi] 
= \F(\partial_1)[\varphi] $ in $\F(Y^{[3]})$.

A morphism  $(f_Y, \varphi \colon 
\partial_1^* f_Y \longrightarrow \partial_0^*f_Y) \longrightarrow 
 (f'_Y, \varphi' \colon 
\partial_1^* f'_Y \longrightarrow \partial_0^*f'_Y)$ in $\Desc_\F(Y)$ consists of 
a morphism $g\colon f \longrightarrow f'$  such that (suppressing coherence isomorphisms)
$\F(\partial_0)[g] \circ \varphi= \varphi \circ  \F(\partial_1)[g]$ holds in $\F(Y^{[2]})$.
\end{definition}

\begin{remark}
For a pre-stack $\F \colon \cC^{opp} \longrightarrow \Grpd$, which is a strict 2-functor, the 
homotopy limit of the diagram 
\begin{equation}
\begin{tikzcd}
\cdots \F(Y^{[3]}) \ar[r, shift left, leftarrow] \ar[r, shift right, leftarrow] \ar[r, leftarrow] & \F( Y^{[2]})\ar[r, shift left, leftarrow] \ar[r, shift right, leftarrow]  & \F(Y) 
\end{tikzcd}
\end{equation}
is the descent category $\Desc_\F(Y)$. It is possible to strictify a pre-stack, hence for 
theoretical discussions one can restrict to strict functors. In this set up it is possible to 
define a model structure on $[\cC^{\opp}, \Grpd]$ related to stacks~\cite{Hollander}.
However, the disadvantage of working with strict pre-stacks is that most examples
do not naturally appear as strict functors. For this reason, we continue to work with 2-functors.
\end{remark}

Now we can define stacks.
\begin{definition}
Let $\cC$ be a site. A pre-stack $\F\colon \cC^{\opp} \longrightarrow \Grpd$ is a \emph{stack}
if for every covering $\pi \colon Y \longrightarrow X$ the canonical map
$\F(X) \longrightarrow \Desc(Y)$ is an equivalence of categories. 
\end{definition}

\begin{example}
\begin{itemize}
\item Every sheaf $\F\colon \cC^{\opp} \longrightarrow \Set$ is a stack by considering 
a set as a groupoid with only identity morphisms. 

\item 
Let $\Man_n$ be the category of $n$-dimensional manifolds and $G$ a Lie group.
Principal $G$-bundles with and without connections define stacks 
$\Bun_G(\cdot)\colon \Man_n^{\opp} \longrightarrow \Grpd$ and 
$\Bun_G^\nabla(\cdot)\colon \Man_n^{\opp} \longrightarrow \Grpd$.

\item For some types of geometric structure it might be necessary 
to restrict to a subclass of morphisms of $\Man_n$. For example, 
orientations can only be pulled back along local diffeomorphisms. 
Throughout this thesis we do not specify the kind of morphisms 
we restrict to explicitly. They should be clear from the context.  

\end{itemize}
\end{example}

\begin{remark}
When considering a stack $\F$ we implicitly pick the following
additional structure: for every surjective submersion $\pi \colon Y\longrightarrow M$, weak adjoint inverses to the canonical map $\mathscr{F}(M) \longrightarrow \mathsf{Desc}_\mathscr{F}(Y)$ where $\mathsf{Desc}_\mathscr{F}(Y)$ is the category of descent data associated to $\pi$. For every refinement 
\begin{equation}
\begin{tikzcd}
Y_1 \ar[dd,swap,"{f}"] \ar[dr,"{\pi_1}"]  & \\
 & M \\
Y_2 \ar[ru,swap, "{\pi_2}"] &
\end{tikzcd}
\end{equation} 
we get a natural functor $f^\ast \colon \mathsf{Desc}_\mathscr{F}(Y_2)\rightarrow \mathsf{Desc}_\mathscr{F}(Y_1)$ for which we pick a weak adjoint inverse. The adjointness condition is essential for ensuring naturality of constructions using descent properties.   
\end{remark}

\section{Integration over essentially finite groupoids}\label{Sec: Integration over finite groupoids}
The space of field configuration in a gauge theory is a groupoid. Morphisms correspond to gauge transformations. When performing the path integral one has to carefully take these internal symmetries
into account. For essentially finite groupoids, there exists a well defined integration theory, which we 
use in Chapter \ref{Chapter: t Hooft} to construct finite gauge theories. 
In the following we review the integration over essentially groupoids following Appendix A. of \cite{OFK}.

For essentially finite groupoids there exists a canonical notion of cardinality.
\begin{proposition/definition}
There exists a unique map $|\cdot|\colon \Obj(\FinGrpd)\longrightarrow \Q$ satisfying
\begin{itemize}
\item[(G1)] The equation $|\Delta_0|=1$ holds.
\item[(G2)] For equivalent essentially finite groupoids $\cG$ and $\cG'$. The equation $|\cG|=|\cG'|$ holds. 
\item[(G3)] For groupoids $\cG, \cG'\in \FinGrpd$ we denote by $\cG \sqcup \cG'$ there disjoint union. The equation $|\cG\sqcup \cG'|=|\cG|+|\cG'|$ holds. 

\item[(G4)]
Let $\cG\longrightarrow \cG'$ be an $n$-fold covering of groupoids. The equation $|\cG|=n\cdot |\cG'|$
holds.
\end{itemize}
We call $|\cdot|$ the \emph{groupoid cardinality}. Concretely, $|\cG|$ is given by 
\begin{align}\label{Eq: Formel Groupoid Card}
|\cG|= \sum_{x\in \pi_0(\cG)} \frac{1}{|\Aut(x)|} ,
\end{align}
where $|\Aut (x)|$ denotes the cardinality of the automorphism group of an arbitrary representative for the isomorphism class $x$. 
Moreover, the groupoid cardinality satisfies $|\cG\times \cG'|=|\cG|\times|\cG'|$.
\end{proposition/definition}
\begin{proof}
Condition $(G2)$ and $(G3)$ imply that $|\cdot|$ is completely determined by its values on groupoids of the from $\star \DS \G $ with one object $\star$ and a finite group $\G$ as automorphisms. Consider the
$|\G|$-fold covering $\EG \longrightarrow \star \DS \G$, where $\EG$ is the action groupoid corresponding
to the action of $\G$ on itself via left multiplication. $\EG$ is contractible and hence we find 
$|\star\DS G|= \frac{1}{|G|}$. \\
The equation $|\cG\times \cG'|=|\cG|\times|\cG'|$ follows from a direct calculation using \eqref{Eq: Formel Groupoid Card}.   
\end{proof}

The groupoid cardinality induces a natural counting measure which can be used to integrate gauge invariant functions over essentially finite groupoids.
\begin{definition}
Let $\cG$ be an essentially finite groupoid. A function $f\colon \Obj (\cG)\longrightarrow \C$ is \emph{gauge invariant} if it is constant on isomorphism classes. The \emph{integral of $f$ over $\cG$} is
\begin{align}
\int_\cG f = \int_\cG f(x) \dd x  \coloneqq \sum_{x\in \pi_0(\cG)}\frac{f(x)}{|\Aut (x)|} \ \ .
\end{align}
\end{definition}
There are results for integration over essentially finite groupoids which are analogues of statements in Lebesgue integration theory. We conclude this section by proving these results.
\begin{proposition}[Cavalieri’s principle]\label{Prop: Cavalieri principle}
Let $\phi\colon \cG \longrightarrow \cG'$ be a functor of essentially finite groupoids. Then 
\begin{align}
|\cG|= \int_{\cG'}|\phi^{-1}[x]| \dd x \ \ .
\end{align} 
\end{proposition}
\begin{proof}
Without loss of generality we can assume that $\cG'$ is $\star \DS \G$. By Lemma \ref{Lemma: Coverings} the forgetful functor $\phi^{-1}[\star]\longrightarrow \cG$ is an $|G|$ fold covering. (G3) now implies
\begin{align}
\int_{\star\DS \G} |\phi^{-1}[\star]| = \frac{|\phi^{-1}[\star]|}{|\G|}= \frac{|\cG||G|}{|\G|}=|\cG| \ \ .
\end{align}
\end{proof}

There is a slight generalisation of Cavalieri’s principle, which turns out to be useful in practice.

\begin{proposition}[Generalised Cavalieri’s principle]\label{Prop: Generalised Cavalieri principle}
Let $\phi\colon \cG \longrightarrow \cG'$ be a functor of essentially finite groupoids and $f\colon \Obj (\cG)\longrightarrow \C$ a gauge invariant function. Then 
\begin{align}
\int_\cG f(x) \dd x = \int_{\cG'}\left(\int_{\phi^{-1}[y]}q_y^*f(x)\dd x \right) \dd y \ \ ,
\end{align} 
where $q_y\colon \phi^{-1}[y] \longrightarrow \cG$ is the obvious forgetful functor.
\end{proposition}
\begin{proof}
Every gauge invariant function can be written as a linear combination of delta functions 
\begin{align}
\delta_x\colon \cG \longrightarrow \C \ , \ \ \
g \longmapsto \begin{cases}
1 \ , \ \text{if } g\cong x \\
0 \ , \ \text{otherwise}
\end{cases} \ \ .
\end{align}
Hence it is enough to prove the statement for delta functions. In this case the statement reduces 
to Proposition \ref{Prop: Cavalieri principle} applied to the groupoid $\cG_x$ of elements isomorphic to $x$. 
\end{proof}

\bibliography{Quellen.bib}  

\bibliographystyle{latexeu}

\clearpage

\end{document}